\documentclass{lmcs} 

\usepackage[utf8]{inputenc}
\usepackage[T5,T1]{fontenc} \usepackage[english]{babel}
\usepackage[super]{nth}

\usepackage{url}
\usepackage{stmaryrd}
\usepackage{amsmath}
\usepackage{amssymb}
\usepackage{cleveref}
\usepackage{mdframed}
\usepackage{mathbbol}
\usepackage{mathtools}
\usepackage{amsthm}
\usepackage{thmtools}
\usepackage{mathpartir}
\usepackage{thm-restate}
\usepackage{bm,accents}
\usepackage{appendix}
\usepackage{xfrac}
\usepackage{relsize}
\usepackage{faktor}
\usepackage{cmll}
\usepackage{array}
\usepackage[all]{xy}
\usepackage{csquotes}
\usepackage{centernot}

\usepackage{tikz}
\usetikzlibrary{arrows,automata}
\usetikzlibrary{decorations.pathreplacing}
\usetikzlibrary{decorations.shapes}
\usetikzlibrary{matrix}
\usetikzlibrary{positioning}
\usetikzlibrary{fit}
\usetikzlibrary{arrows}
\usetikzlibrary{shapes}
\usetikzlibrary{decorations.markings}
\tikzstyle{item} = [ellipse, draw, node distance=1.5cm, align=center]

\DeclareSymbolFont{AMSb}{U}{msb}{m}{n}
\DeclareSymbolFontAlphabet{\mathbb}{AMSb}
\DeclareSymbolFontAlphabet{\mathbbl}{bbold}

\usepackage{todonotes}

\ifdefined\cA
\else
\newcommand\cA{\mathcal{A}}
\fi
\ifdefined\cB
\else
\newcommand\cB{\mathcal{B}}
\fi
\ifdefined\cC
\else
\newcommand\cC{\mathcal{C}}
\fi
\ifdefined\cD
\else
\newcommand\cD{\mathcal{D}}
\fi
\ifdefined\cE
\else
\newcommand\cE{\mathcal{E}}
\fi
\ifdefined\cF
\else
\newcommand\cF{\mathcal{F}}
\fi
\ifdefined\cG
\else
\newcommand\cG{\mathcal{G}}
\fi
\ifdefined\cH
\else
\newcommand\cH{\mathcal{H}}
\fi
\ifdefined\cI
\else
\newcommand\cI{\mathcal{I}}
\fi
\ifdefined\cJ
\else
\newcommand\cJ{\mathcal{J}}
\fi
\ifdefined\cK
\else
\newcommand\cK{\mathcal{K}}
\fi
\ifdefined\cL
\else
\newcommand\cL{\mathcal{L}}
\fi
\ifdefined\cM
\else
\newcommand\cM{\mathcal{M}}
\fi
\ifdefined\cN
\else
\newcommand\cN{\mathcal{N}}
\fi
\ifdefined\cO
\else
\newcommand\cO{\mathcal{O}}
\fi
\ifdefined\cP
\else
\newcommand\cP{\mathcal{P}}
\fi
\ifdefined\cQ
\else
\newcommand\cQ{\mathcal{Q}}
\fi
\ifdefined\cR
\else
\newcommand\cR{\mathcal{R}}
\fi
\ifdefined\cS
\else
\newcommand\cS{\mathcal{S}}
\fi
\ifdefined\cT
\else
\newcommand\cT{\mathcal{T}}
\fi
\ifdefined\cU
\else
\newcommand\cU{\mathcal{U}}
\fi
\ifdefined\cV
\else
\newcommand\cV{\mathcal{V}}
\fi
\ifdefined\cW
\else
\newcommand\cW{\mathcal{W}}
\fi
\ifdefined\cX
\else
\newcommand\cX{\mathcal{X}}
\fi
\ifdefined\cY
\else
\newcommand\cY{\mathcal{Y}}
\fi
\ifdefined\cZ
\else
\newcommand\cZ{\mathcal{Z}}
\fi
\ifdefined\bA
\else
\newcommand\bA{\mathbb{A}}
\fi
\ifdefined\bB
\else
\newcommand\bB{\mathbb{B}}
\fi
\ifdefined\bC
\else
\newcommand\bC{\mathbb{C}}
\fi
\ifdefined\bD
\else
\newcommand\bD{\mathbb{D}}
\fi
\ifdefined\bE
\else
\newcommand\bE{\mathbb{E}}
\fi
\ifdefined\bF
\else
\newcommand\bF{\mathbb{F}}
\fi
\ifdefined\bG
\else
\newcommand\bG{\mathbb{G}}
\fi
\ifdefined\bH
\else
\newcommand\bH{\mathbb{H}}
\fi
\ifdefined\bI
\else
\newcommand\bI{\mathbb{I}}
\fi
\ifdefined\bJ
\else
\newcommand\bJ{\mathbb{J}}
\fi
\ifdefined\bK
\else
\newcommand\bK{\mathbb{K}}
\fi
\ifdefined\bL
\else
\newcommand\bL{\mathbb{L}}
\fi
\ifdefined\bM
\else
\newcommand\bM{\mathbb{M}}
\fi
\ifdefined\bN
\else
\newcommand\bN{\mathbb{N}}
\fi
\ifdefined\bO
\else
\newcommand\bO{\mathbb{O}}
\fi
\ifdefined\bP
\else
\newcommand\bP{\mathbb{P}}
\fi
\ifdefined\bQ
\else
\newcommand\bQ{\mathbb{Q}}
\fi
\ifdefined\bR
\else
\newcommand\bR{\mathbb{R}}
\fi
\ifdefined\bS
\else
\newcommand\bS{\mathbb{S}}
\fi
\ifdefined\bT
\else
\newcommand\bT{\mathbb{T}}
\fi
\ifdefined\bU
\else
\newcommand\bU{\mathbb{U}}
\fi
\ifdefined\bV
\else
\newcommand\bV{\mathbb{V}}
\fi
\ifdefined\bW
\else
\newcommand\bW{\mathbb{W}}
\fi
\ifdefined\bX
\else
\newcommand\bX{\mathbb{X}}
\fi
\ifdefined\bY
\else
\newcommand\bY{\mathbb{Y}}
\fi
\ifdefined\bZ
\else
\newcommand\bZ{\mathbb{Z}}
\fi

\newenvironment{notations}[1][Notations]{\begin{trivlist}
\item[\hskip \labelsep {\bfseries #1}]}{\end{trivlist}}

\newcommand{\LinOrd}{{\sf LO}}
\newcommand{\partto}{\rightharpoonup}

\newcommand{\dom}{{\sf{dom}}}

\newcommand{\Finset}{{\sf FinSet}}
\newcommand{\Partfinset}{{\sf PartFinSet}}
\newcommand{\Fincoh}{{\sf FinCoh}}

\newcommand{\id}{\mathrm{id}}
\newcommand{\In}{{\sf in}}

\newcommand{\Cat}{{\sf Cat}}
\newcommand{\Set}{{\sf Set}}

\newcommand\SymAff{\sf Aff}
\newcommand\mCat{\sf MCat}
\newcommand{\bbtwo}{\mathbbl{2}}
\newcommand{\bigslant}[2]{{\left.\raisebox{.2em}{$#1$}\middle/\raisebox{-.2em}{$#2$}\right.}}
\newcommand{\occ}[2]{| #2 |_{#1}} \newcommand{\inl}{\mathsf{in}_1}
\newcommand{\inr}{\mathsf{in}_2}
\newcommand{\normalSST}[2]{{#1}^* \to_\text{SST} {#2}^*}
\newcommand{\SST}[3]{{#2}^* \to_\text{$#1$-SST} {#3}}
\newcommand{\basety}{\mathbbl{o}}
\newcommand{\retty}{\Bot}
\newcommand{\initty}{\text{{\raisebox{\depth}{\scalebox{1}[-1]{$\Bot$}}}}}
\newcommand\laml{\lambda\ell^{\oplus\with}}
\newcommand{\lamp}{\lambda\wp}
\newsavebox{\foobox}
\newcommand{\slantbox}[2][.5]
  {    \mbox
      {        \sbox{\foobox}{#2}        \hskip\wd\foobox
        \pdfsave
        \pdfsetmatrix{1 0 #1 1}        \llap{\usebox{\foobox}}        \pdfrestore
      }    }
\newcommand\fancya{\kern -0.05em \slantbox[0.4]{\text{\normalfont a}}}    
\newcommand\lama{\lambda\fancya}
\newcommand\bnfeq{\mathrel{::=}}
\newcommand\bnfalt{\; | \;}
\newcommand\len[1]{{|#1|}}
\newcommand\card[1]{{|#1|}}
\newcommand\lin{\multimap}
\newcommand\tensor{\otimes}
\newcommand\bigtensor{\bigotimes}

\newcommand\Sr{\mathcal{SR}}
\newcommand\SR{\mathfrak{SR}}
\newcommand\Register{\mathcal{TR}}
\newcommand\Registermc{\mathcal{TR}^{\mathrm{m}}}
\newcommand\Registerleonemc{\mathcal{TR}^{\mathrm{m},\le 1}}
\newcommand\Registerconflict{\mathcal{TR}^{\incoh}}
\newcommand\REGISTERconflict{\mathfrak{TR}^{\incoh}}
\newcommand\Registerleone{\mathcal{TR}^{\le 1}}
\newcommand\Registerleoneconflict{\mathcal{TR}^{\incoh, \le 1}}
\newcommand\REGISTERleoneconflict{\mathfrak{TR}^{\incoh, \le 1}}
\newcommand\REGISTER{\mathfrak{TR}}
\newcommand\REGISTERleone{\mathfrak{TR}^{\le 1}}
\newcommand\interp[1]{{\llbracket #1 \rrbracket}}
\newcommand\curlyinterp[1]{{\llparenthesis #1 \rrparenthesis}}
\newcommand\Obj{\mathsf{Obj}}
\newcommand\Hom[3]{\mathsf{Hom}_{#1}\left(#2,#3\right)}
\newcommand\varHom[3]{\mathsf{Hom}_{#1}\left(#2\text{{\Huge\textit{,}}}\;#3\right)}
\newcommand\mfC{\mathfrak{C}}
\newcommand\mfD{\mathfrak{D}}
\newcommand\unit{\mathbf{I}}
\newcommand\ev{\mathrm{ev}}

\newcommand\op{\mathrm{op}}
\newcommand\powerset{\cP}
\newcommand\powerfin{\cP_{<\infty}}
\newcommand\Ty{\mathrm{Ty}}
\newcommand\Tree{\mathbf{Tree}}
\newcommand\Treety{\mathrm{Tree}}
\newcommand\normal{\mathrm{reduce}}
\newcommand\normalmap{reduce}
\newcommand\prune{\mathrm{prune}}
\newcommand\prunemap{{prune}}
\newcommand\contract{\mathrm{contract}}
\newcommand\contractmap{{contract}}

\newcommand\normalforests[2]{\mathrm{NF}\left({#1},{#2}\right)}
\newcommand\embeddinjunk{\mathtt{padwithjunk}}
\newcommand\arity{\mathrm{ar}}
\newcommand\LTree{\mathbf{LTree}}
\newcommand\restr[2]{#1 \upharpoonright #2}
\newcommand\forestify[1]{\mathrm{Grph}(#1)}
\newcommand\forestifymorph[1]{\widehat{#1}}
\newcommand\lam{\lambda}
\newcommand\Lamcat{\mathcal{L}}
\newcommand\tuple[1]{\left\langle #1 \right\rangle}
\newcommand\tlet[3]{{\mathsf{let}} \; {#1} = {#2} \; {\mathsf{in}} \; {#3}}
\providecommand\case{{\mathsf{case}}}
\newcommand\abort{{\mathsf{abort}}}

\newcommand\mfLam{\mathfrak{L}}
\newcommand\mfM{\mathfrak{M}}
\newcommand\Str{\mathrm{Str}}
\newcommand\Nat{\mathrm{Nat}}
\newcommand\Bool{\mathrm{Bool}}
\newcommand\NF{\mathrm{NF}}
\newcommand\NE{\mathrm{NE}}
\newcommand\alphtorankedalph[1]{\overline{#1}}
\newcommand\lamlequiv{{\beta\eta}}
\newcommand\mcattoaff[1]{#1_{\mathrm{aff}}}
\newcommand\tomcat[1]{#1_{\mathrm{mcat}}}
\newcommand\toconflict[1]{#1_{\mathrm{coh}}}
\newcommand\longto\longrightarrow

\newcommand\web[1]{\|#1\|}
\newcommand\clique{\mathrm{Cl}}
\newcommand{\ttl}{\triangleleft}
\newcommand{\ttr}{\triangleright}

\newcommand{\BinTree}{\mathrm{BinTree}}
\newcommand{\dBinTree}{\partial\mathrm{BinTree}}
\newcommand{\ExprBT}{\mathrm{ExprBT}}
\newcommand{\ExprdBT}{\mathrm{Expr}\partial\mathrm{BT}}
\newcommand{\assi}{\mathcal{A}}

\newcommand{\Conf}{\mathrm{Conf}}

\newcommand{\trpair}[2]{\lceil #1 , #2 \rceil}
\newcommand{\trnode}{\lceil\cdot\rceil}
\newcommand{\BT}{\mathtt{BT}}
\newcommand\extr{\varepsilon}
\newcommand\cocoeq{\approx_c}

\begin{document}

\title[Streaming transducers vs categorical semantics]{Implicit automata in typed $\lambda$-calculi II:\\ streaming transducers vs categorical semantics}

\author[L.~T.~D.~{\fontencoding{T5}\selectfont{}Nguyễn}]{{\fontencoding{T5}\selectfont{}Lê Thành Dũng (Tito) Nguyễn}}	\address{Laboratoire
    d'informatique de Paris Nord, Villetaneuse, France}	\email{nltd@nguyentito.eu}  
\author[C.~Noûs]{Camille Noûs}
\address{Laboratoire Cogitamus}
\urladdr{\url{https://www.cogitamus.fr/camilleen.html}}

\author[P.~Pradic]{Pierre Pradic}
\address{Department
    of Computer Science, University of Oxford, United Kingdom}
\email{pierre.pradic@cs.ox.ac.uk}

\begin{abstract}
We characterize regular string transductions as programs in a linear
$\lambda$-calculus with additives. One direction of this equivalence is proved
by encoding copyless streaming string transducers (SSTs), which compute regular
functions, into our $\lambda$-calculus. For the converse, we consider
a categorical framework for defining automata and transducers over words, which
allows us to relate register updates in SSTs to the semantics of the linear
$\lambda$-calculus in a suitable monoidal closed category.

To illustrate the relevance of monoidal closure to automata theory, we leverage this notion
to give abstract generalizations of the arguments showing that copyless SSTs may be determinized
and that the composition of two regular functions may be implemented by a copyless SST.

Our main result is then generalized from strings to trees using a similar
approach. In doing so, we exhibit a connection between a feature of streaming
tree transducers and the multiplicative/additive
distinction of linear logic.

\end{abstract}

\maketitle

\section{Introduction}

We recently initiated~\cite{aperiodic} a series of works at the interface of
programming language theory and automata. As the title suggests, the present
paper is the second installment; it starts with an introduction to this research
programme, meant to be accessible with a general computer science background
(\Cref{sec:intro-intro}). After stating a main theorem, we shall argue, in two
mostly independent subsections, that these connections between two fields that
we investigate:
\begin{itemize}
\item are relevant to natural questions on the
  \mbox{$\lambda$-calculus} (\Cref{sec:motiv-lambda});
\item provide new conceptual insights into automata theory
  (\Cref{sec:motiv-automata}).
\end{itemize}
\Cref{sec:intro-cat} exposes some of our key technical ideas, stressing the role
of category theory as a mediating language.
A table of contents is provided after this introduction.

\subsection{What is this all about?}
\label{sec:intro-intro}

\subsubsection{From proofs-as-programs to implicit complexity}

One of the central principles in the contemporary theory of programming
languages is a close relationship between constructive logics and statically
typed functional programming, known as the \emph{proofs-as-programs
  correspondence}, also known as the \emph{formulae-as-types} or
  \emph{Curry--Howard} correspondence. The idea is that, in certain logical
systems, proofs admit a \enquote{normalization procedure} that can be seen as
the execution of a program. According to this analogy, a proof is thus a
program, and the \emph{formula} that it proves is the \emph{type} of the
program. A remarkable empirical fact is that this manifests as several concrete
isomorphisms between (theoretical) languages and proof systems that were
designed independently.

An important point is that \emph{termination} on the programming side is highly
desirable in this context since it entails \emph{consistency} on the logical
side. Take for instance the \emph{untyped \mbox{$\lambda$-calculus}}, one of the models
of computation that led to the birth of computability theory in the 1930s,
nowadays used as a theoretical foundation for functional programming. It allows
non-terminating programs. The \emph{simply typed $\lambda$-calculus} adds a type
system on top of it; one can then rule out this possibility of non-termination
by only allowing well-typed programs, thus ensuring the consistency of the
corresponding logical system (here, intuitionistic propositional logic) at the
price of losing Turing-completeness.

Such a termination guarantee might even come with quantitative time complexity
bounds. For instance, Hillebrand et al.~\cite{HKMairson} show that programs in
the simply typed \mbox{$\lambda$-calculus} operating over certain data encodings
and returning booleans can compute all functions\footnote{This does \emph{not}
  mean that a given algorithm with elementary complexity must admit a direct
  implementation in the simply typed \mbox{$\lambda$-calculus}; instead, what
  must exist is a $\lambda$-calculus program computing the same function from
  inputs to outputs, with potentially different inner workings.} in the
complexity class $\mathsf{ELEMENTARY}$ (i.e.\ those with a time complexity
bounded by a tower of exponentials), and only those. This result illustrates the
type-theoretic approach to \emph{implicit computational complexity}, a
well-established field concerned with machine-free characterizations of
complexity classes via high-level programming languages\footnote{We refer to the
  introduction of our previous paper \cite{aperiodic} for a discussion of the
  difference between implicit computational complexity and descriptive
  complexity.}. Many works in this area have taken inspiration from \emph{linear
  logic}~\cite{girardLL} to design more sophisticated type systems, starting
with two characterizations of polynomial time~\cite{girardBLL,girardELL}. As
another example, Linear Logic by Levels~\cite{LLLevels} also characterizes
$\mathsf{ELEMENTARY}$ and admits a translation from the simply typed
$\lambda$-calculus~\cite{MELLbyLevel}.

\subsubsection{Implicit automata}

Let us consider \emph{strings} over some finite alphabet as input. Functions
mapping these strings to booleans are equivalent to sets of strings, and the
latter are called \emph{languages}. Complexity classes (of decision problems)
are often defined as \emph{sets of languages}, but such sets also arise in
\emph{automata theory}. A typical example is the class of \emph{regular
  languages}, that can be defined by regular expressions or equivalently by
finite automata (we assume here that the reader knows about those): usually, it
is not considered a complexity class, although this statement is sociological
rather than formal\footnote{One possible technical argument is that the class of
regular languages is not closed under $\mathsf{ALOGTIME}$ (i.e.\ uniform
$\mathsf{AC}^0$) reductions.}.

Our research programme aims to provide for automata what (type-theoretic)
implicit complexity has done for complexity classes. A characterization of
regular languages had already been obtained by Hillebrand and Kanellakis in the
simply typed $\lambda$-calculus~\cite[Theorem~3.4]{HillebrandKanellakis}.
Starting from this, we characterized~\cite{aperiodic} the smaller class of
\emph{star-free languages} by relying on a richer type system that supports
so-called \emph{non-commutative} types. As mentioned in~\cite[\S7]{aperiodic},
some other results of this kind already exist, but not many; and as far as we
know, the idea of \enquote{implicit automata} as a topic worthy of systematic
investigation had not been put forth before in writing.

\subsubsection{Transducers}

Here, our goal is to go beyond languages and to consider \emph{string-to-string
  functions} instead. There is a wide variety of classes of such functions that
appear in automata theory, and several of them collapse to regular languages
when we restrict them to a single output bit (this is the case for the so-called
sequential functions, rational functions, regular functions\ldots{} see the
surveys~\cite{siglog,MuschollPuppis}). In other words, many automata models that
recognize regular languages are no longer equivalent when extended with the
ability to produce an output string. Such automata with output are called
\emph{transducers}. We could therefore expect fine distinctions between the
feature sets of various $\lambda$-calculi to be reflected in differences between
the string functions that they can express.

Yet we only know of two precedents for string-to-string transduction classes
captured using typed functional programming: sequential functions in a cyclic
proof system~\cite{SubstructuralAutomata} and polyregular functions in a
$\lambda$-calculus with primitive data types~\cite[\S4]{polyregular}. Both
are discussed further in the prequel paper~\cite[\S7]{aperiodic}.

This brings us to our first main theorem:
\begin{restatable}{thm}{mainstring}
\label{thm:main-string}
  A function $\Sigma^* \to \Gamma^*$ is \emph{$\laml$-definable} if and only if
  it is \emph{regular}.
\end{restatable}
By \enquote{$\laml$-definable}, we mean expressible in the $\laml$-calculus (a
system based on Intuitionistic Linear Logic) in the specific but mostly standard way described in Definition~\ref{def:laml-definable}.
As for \emph{regular functions}, they are a well-studied class with many equivalent
definitions, for instance two-way transducers or monadic second-order
logic~\cite{EngelfrietHoogeboom}. We may also point to several recent
formalisms~\cite{regularcombinators,RTE,Daviaud} for regular functions using
combinators as belonging to \enquote{implicit complexity for automata}, albeit
not of the type-theoretic kind (implicit complexity is an umbrella term which
traditionally also includes tools such as recursive function algebras or term
rewriting).

\subsection{Internal motivations from typed $\lambda$-calculi}
\label{sec:motiv-lambda}

We mentioned earlier a characterization of $\mathsf{ELEMENTARY}$ in the simply
typed $\lambda$-calculus by Hillebrand et al.~\cite{HKMairson}. It uses a
somewhat unusual (though perfectly justified) representation for the inputs. The
conventional choice would have been the \emph{Church encoding} of strings. They
are indeed the usual tool to represent all computable functions in the untyped
$\lambda$-calculus, and in some terminating type systems, any
\enquote{reasonable} function can still be programmed over these encodings (for
example, this is the case for System F~\cite{WadlerGR}). But
in the simply typed case, some earlier results by Statman had suggested a hopeless lack of
expressiveness (see e.g.\ the discussion in~\cite{FortuneLeivant} after its
Theorem~4.4.3.). Then Hillebrand and Kanellakis's aforementioned
result~\cite{HillebrandKanellakis}
showed that, surprisingly, one gets the class of regular languages by using
Church-encoded strings in the simply typed $\lambda$-calculus!

Recently, the present paper's first author reused their ideas in~\cite{ealreg}
to solve an open problem from \enquote{standard} implicit complexity, concerning
a characterization of polynomial time by Baillot et al.~\cite{Benedetti}
that makes use of Church encodings. The question was
whether a certain feature (namely type fixpoints) was necessary for this result.
It turns out~\cite{ealreg} that when this feature is removed, the class of
languages obtained is regular languages instead of
$\mathsf{P}$.\footnote{
Digression: in~\cite{logspace}, we explored the input representation
of~\cite{HKMairson} transposed into a language similar to that
of~\cite{Benedetti}, without these type fixpoints. We gathered some evidence
suggesting that one gets a characterization of deterministic logarithmic space
(though the upper bound that we manage to prove is a bit weaker).
}

The moral of the story, for us, is that the use of Church-encoded strings can
lead naturally to connections with automata theory. Admittedly, this naturality
judgment is inherently subjective. But concretely, it translates into a
methodological commitment: we explore the expressiveness of typed
$\lambda$-calculi that already exist (perhaps up to a few minor details),
whereas it is usual in implicit
complexity to engineer some (potentially ad-hoc) new type system to achieve
desired complexity effects. Most of the time, the features of these preexisting
\mbox{$\lambda$-calculi} have original motivations that are entirely unrelated
to complexity or automata (for instance, the non-commutative types that we used
in~\cite{aperiodic} originate from the study of natural language~\cite{Lambek}
and the topology of proofs (see e.g.~\cite[\S{}II.9.]{TowardsGoI})).

In the case of the simply typed $\lambda$-calculus, characterizing the definable
string-to-string functions in the style
of~\cite[Theorem~3.4]{HillebrandKanellakis} (again!) is in fact an old open
problem (while a more restrictive notion of $\lambda$-definability
is well understood~\cite{Zaionc}). As we are not yet able to solve it, we
instead tackle a version where \emph{linearity} constraints have been added,
resulting in Theorem~\ref{thm:main-string}. Recall that a function is said to be linear
(in the sense of linear logic) when it uses its argument only once. The system
that we use to express these constraints is Dual Intuitionistic Linear
Logic~\cite{Barber} with additive connectives (called here the
$\laml$-calculus).

\subsection{Conceptual interest for (categorical) automata theory}
\label{sec:motiv-automata}

This notion of linearity in programming language theory has a counterpart in the
old theme of \emph{restricting the copying power} of automata models (see
e.g.~\cite{ERS}). The latter is manifested in one of the possible definitions of
the regular functions mentioned in Theorem~\ref{thm:main-string}:
\emph{copyless\footnote{The adjective \enquote{copyless} does not appear in the
    original paper~\cite{SST} but is nowadays commonly used to distinguish them
    from the later \emph{copyful SSTs}~\cite{FiliotReynier}.} streaming string
  transducers} (SSTs)~\cite{SST}. An SST is roughly speaking an automaton whose
internal memory consists of a state (in a finite set) and some string-valued
registers, and its transitions are copyless when they compute new register
values without duplicating the old ones.

Put this way, Theorem~\ref{thm:main-string} seems unsurprising. But there is more going
on behind the scenes. In particular, while it is trivial that the composition of
two $\laml$-definable functions is itself $\laml$-definable, composing copyless
SSTs requires intricate combinatorics as can be seen
in~\cite[Chapter~13]{Toolbox} for example.
As it turns out,
the tools developed in order to prove Theorem~\ref{thm:main-string} also yield a clean
proof of the closure of copyless streaming string transducers under composition,
which it even generalizes using the language of \emph{category theory}, see
below.

Another subtlety comes from our extension of Theorem~\ref{thm:main-string} to
\emph{ranked trees}:
\begin{restatable}{thm}{maintree}
  \label{thm:main-tree}
  Let ${\bf\Sigma}$ and ${\bf\Gamma}$ be ranked alphabets. A function
  $\Tree({\bf \Sigma}) \to \Tree({\bf \Gamma})$ is $\laml$-definable if and only
  if it is regular.
\end{restatable}

The class of \emph{regular tree functions} is obtained by generalizing the
definition for strings based on monadic second-order logic (MSO, see
e.g.~\cite{FOTree}). There is also an automata model adapting SSTs to trees,
namely the \emph{bottom-up ranked tree transducers} (BRTT)~\cite{BRTT}. However,
it is conjectured that some regular functions \emph{cannot} be computed by
\emph{copyless} BRTTs. Instead, a more sophisticated linearity condition, called
the \emph{single use restriction}\footnote{The expression \enquote{single use
    restriction} already appears in much earlier automata models for regular
  tree functions: attributed tree transducers~\cite{AttributedMSO} and macro
  tree transducers~\cite{MacroMSO}. This suggests that some kind of linearity
  condition is at work in those models, though we have not investigated this
  point further.}, is imposed on BRTTs in~\cite{BRTT} in order to characterize
the regular tree functions. The additional flexibility\footnote{Alternative
  options to restore enough expressive powers are copyless BRTTs with regular
  look-ahead (considered in~\cite[\S3.4]{BRTT} for streaming transducers
  over unranked trees) or preprocessing by MSO
  relabelings~\cite[\S4]{FOTree}. Beware: in the latter reference, the
  term \enquote{single use} refers to copyless assignments.} thus afforded,
compared to copyless BRTTs, turns out to correspond directly to an important
feature of linear type systems, namely the \emph{additive conjunction}.

As a concrete manifestation of this correspondence, we conjecture that all tree
functions definable in the $\lama$-calculus of our previous paper~\cite{aperiodic}
can be computed by copyless BRTTs. This $\lama$-calculus does not contain the additive connectives
$\with/\oplus$ of linear logic; to compensate, it has an affine type system, instead
of a linear one (whence the $\fancya$).
We leave the proof of this fact for future work.

In the case of strings, single-use-restricted streaming string transducers are
very close to copyless \emph{non-determinstic} SSTs. (That additive connectives in linear logic have something to do with non-determinism
has previously been observed in other settings, for instance in~\cite{MairsonTerui}.)
Their equivalence with copyless SSTs thus corresponds to a
\emph{determinization} theorem, that already has an indirect proof via
MSO~\cite{NSST}. We provide here a direct construction, whose main technical
ingredient is the \enquote{transformation forest} data structure applied to
copyless SSTs in~\cite[Chapter~13]{Toolbox} and reminiscent of the
Muller--Schupp determinization~\cite{muller_schupp} for automata over infinite
words. Most importantly, this determinization result is again formulated in a
general \emph{category-theoretic} setting.

Together with the aforementioned analysis of the composition of SSTs, those are
our two contributions to \enquote{categorical automata theory}. This kind of use
of categories to understand the essence of various constructions on automata --
such as determinization or minimization -- and to generalize them to other
settings has a long history, see for instance~\cite{TreeAutomataEndofunctor} and
the many references therein\footnote{There are also connections between
  categories and algebraic language theory~\cite{Tilson}, which however seemed
  less relevant to our work here.}.

\subsection{Transducers over monoidal closed categories}
\label{sec:intro-cat}

Another example of categorical automata theory is the work of Colcombet and
Petrişan~\cite{ColcombetPetrisanSIGLOG,ColcombetPetrisan} on minimization, whose
direct relevance to us lies in the \emph{categorical framework for automata
  models} that it introduces: objects serve as state spaces and morphisms as
transitions. Our technical development takes place in a very similar framework.
\begin{itemize}
\item We first define a category $\Sr$ that corresponds to \emph{single-state}
  copyless SSTs.
\item Since copylessness and the so-called single use restriction morally differ
  by the presence of the additive conjunction `$\with$' of linear logic, we
  \enquote{add `$\with$' freely} to achieve a similar effect: automata in the
  resulting category $\Sr_\with$, although not identical to single-state single
  use restricted SSTs, are easily seen to be equally expressive.
\item Finally, we perform another \enquote{completion} denoted $(-)_\oplus$ to
  incorporate a finite set of states, so that $\Sr_\oplus$ (resp.\
  $(\Sr_\with)_\oplus$) corresponds to usual -- i.e.\ stateful -- copyless
  (resp.\ single use restricted) SSTs. (Similar completions by certain colimits
  have been previously exploited~\cite{ColcombetPetrisanVect} within Colcombet
  and Petrişan's framework.)
\end{itemize}
The linchpin on which the various results previously mentioned rely is:
\begin{thm}
  \label{thm:smcc-string}
  $(\Sr_\with)_\oplus$ is a \emph{symmetric monoidal closed} category (SMCC).
\end{thm}

On the one hand, from the point of view of categorical automata theory, an SMCC
provides a setting in which constructions relying on \emph{function spaces}
(i.e.\ internal homsets) can be carried out (this is typically the case when one
exploits the finiteness of $Q^Q$ for any finite set of states $Q$). This is the
case for our composition result, whose general version is stated over arbitrary
SMCCs. While $(\Sr_\with)_\oplus$ is an SMCC, $\Sr_\oplus$ is \emph{not}, and
this explains why composing copyless SSTs (that correspond to $\Sr_\oplus$)
directly is difficult: instead, a detour through $(\Sr_\with)_\oplus$ allows us
to apply Theorem~\ref{thm:smcc-string}. This move from $\Sr_\oplus$ to
$(\Sr_\with)_\oplus$ is made possible by our abstract determinization argument,
which itself relies on the existence of \emph{some} (but not all) function
spaces in $\Sr_\oplus$.

On the other hand, for the programming language theorist, the notion of SMCC is
an axiomatization of the \emph{denotational semantics} for
the \enquote{purely linear} fragment of our \mbox{$\laml$-calculus}.
We can therefore apply a \emph{semantic evaluation} argument,
following a long tradition in implicit complexity
(cf.~\cite{HillebrandKanellakis,Terui}), to deduce Theorem~\ref{thm:main-string}
from Theorem~\ref{thm:smcc-string};
similarly, Theorem~\ref{thm:main-tree} follows from the
monoidal closure of a category $(\Register_\with)_\oplus$ for trees.
(Semantics of linear logic have also been applied to \emph{higher-order
  recursion schemes}, a topic at the interface with automata,
in~\cite{grellois,HOParity,MAHORS}, as well as to the purely automata-theoretic
\emph{Church synthesis problem} in some publications~\cite{LMSO,LMSO2} coauthored
by the present paper's third author.) That said, to create suitable conditions
for semantic evaluation, a quite lengthy syntactic analysis is required, with
the presence of positive connectives in the $\laml$-calculus causing some
complications.

At this point, we must mention the kinship of this $(\Sr_\with)_\oplus$ with one
of the earliest denotational models of linear logic, the \emph{Dialectica
  categories}~\cite{PaivaDialectica} (originating as a categorical account of
Gödel's \enquote{Dialectica} functional interpretation~\cite{godeldial}).
Composing the \emph{free finite coproduct completion} $(-)_\oplus$
with its dual product completion $(-)_\with$ is
indeed reminiscent of a factorization into free sums and free products of a
generalized Dialectica construction~\cite{HofstraDialectica}.
Thus, \mbox{Theorem~\ref{thm:smcc-string}} holds for reasons similar to those
for the monoidal closure of
Dialectica categories (with a function space formula that resembles the
interpretation of implication in~\cite{godeldial}).
Such Dialectica-like structures have appeared
in quite varied contexts in the past few years, such as \emph{lenses} from
functional programming and \emph{compositional game theory} (see
e.g.~\cite[\S4]{Hedges18} for both), and, more in line with the topic of this
paper, the aforementioned works on Church's synthesis~\cite{LMSO,LMSO2} (and
a closely related work on automata over infinite trees~\cite{Colin}).

To wrap up this introduction, let us mention that as a bridge between
$\laml$-definability and those automata models, we also define within our
categorical framework a notion of transducer whose memory is made up of (purely
linear) $\laml$-terms. This idea of using linear $\lambda$-terms inside a
transducer model also appears in a recent characterization of regular
tree functions~\cite{LambdaTransducer}.

\section*{Acknowledgment}
\noindent We thank Zeinab Galal for her comments on free (co)completions,
Sylvain Salvati for discussions on connections between transducers and simply
typed $\lambda$-calculi, and Gabriel Scherer for his advice regarding the
intricacies of normalization of the $\laml$-calculus.

Some of the ideas presented here were developed concurrently
with, and are inextricably linked to, those of our previous
paper~\cite{aperiodic}; therefore, we also express our gratitude again to the
many people cited in the latter's acknowledgments.

\tableofcontents
\newpage
\section{Preliminaries}

\subsection{Notations \& elementary definitions}
\label{subsec:notations}

\subsubsection{Sets and categories}

The cardinality of a set $X$ is written $\card{X}$.
We sometimes consider a family $(x_i)_{i \in I}$ as a
map $i \mapsto x_i$, which amounts to treating $\prod_{i \in I} X_i$ as a
dependent product. Consistently with this, we make use of the dependent sum
operation
\[ \sum_{i \in I} X_i = \{(i,x) \mid i \in I,\, x \in X_i\} \]
We (seldom) write numerals $n$ for the underlying sets $\{0, \ldots, n-1\}$ for conciseness.

Given a category $\cC$, we write $\Obj(\cC)$ for its class of objects and
$\Hom{\cC}{A}{B}$ for the set of arrows (or morphisms) from $A$ to $B$ (for
$A,B \in \Obj(C)$).
The composition of two morphisms $f \in \Hom{\cC}{A}{B}$ and $g \in
\Hom{\cC}{B}{C}$ is denoted by $g \circ f$.
Following the traditional notations of linear logic, products and coproducts
will be customarily written using `$\with$' and `$\oplus$' respectively
-- except in the category of sets where we use the
notations `$\times$' and `$+$' as usual  -- and we reserve $\top$ for the terminal
object. We sometimes use basic combinators such as $\tuple{-}$/$[-]$ for pairing/copairing and
$\pi_i$/$\In_i$ for projections/coprojections. With these notations, recall that
the binary coproduct of sets is the tagged union
\[ X + Y = \{\inl(x) \mid x \in X\} \cup \{\inr(y) \mid y \in Y\} \]
The injection $\inl/\inr$ may be omitted by abuse of notation when it is clear
from the context, that is, for $x \in X$, we allow ourselves to write $x$ for
$\inl(x)$ when it is understood that this refers to an element of $X+Y$.

Finally, if we are given a binary operation $\square$ over the objects of a
category,
we freely use the corresponding \enquote{$I$-ary} operation, with a notation of
the form $\bigbox_{i \in I} A_i$, over families
indexed by a finite set $I$. Concretely speaking, this depends on a fixed total
order over $I = \{ i_1 < \ldots < i_{\card{I}}\}$ to unfold as $A_{i_1} \,\square\,
(A_{i_2} \,\square\, (\ldots \,\square\, A_{i_{\card{I}}})\ldots)$ -- for convenience,
the reader may consider that a choice of such an order for every finite
set is fixed once and for all for the rest of the paper. In practice,
the particular order does not matter since we will deal with operations $\square
\in \{\oplus,\with,\otimes, \dots\}$ that
are symmetric in a suitable sense. Those operations also have units (i.e.\
identity elements), giving a canonical meaning to $\bigoplus_{i\in\varnothing} A_i$,
$\bigwith_{i\in\varnothing} A_i$, etc.

Finally, as is usual when dealing with categories, we sometimes allow ourselves to
implicitly use the axiom of choice for classes to pick objects determined by their
universal properties to build functors (for instance, given an object $A$ in a
category $\cC$ with cartesian products, we shall speak of the functor
$- \with A$ without first mentioning that a choice of cartesian products $X \with A$
exists for for every $X$ in $\cC$).
This is merely for convenience; the reader may check that in all of our concrete
examples of interest, canonical choices can be made without appealing to choice.

\subsubsection{Strings and ranked trees}

Alphabets designate finite sets and are written using the variable names
$\Sigma, \Gamma$. The set of strings (or words) over an alphabet $\Sigma$ is
denoted by $\Sigma^*$. The concatenation of two strings $u,v \in \Sigma^*$ is
written $uv$ (or sometimes $u \cdot v$ for clarity);
recall that $\Sigma^*$ endowed with this operation is the free monoid
over the set of generators $\Sigma$, and its identity element is the empty
string $\varepsilon$. We write $\len{w}$ for the length of a word
$w \in \Sigma^*$, and given a letter $c \in \Sigma$, the notation $\occ{c}{w}$
refers to the number of occurrences of $c$ in $w$.

Ranked alphabets are pairs $(\Sigma, \arity)$ such that $\Sigma$ is an alphabet and
the \emph{arity} $\arity$ is a family of finite sets\footnote{This is slightly non-standard; the more usual notion
would be that $\arity$ be only a family of numbers $\Sigma \to \bN$. To talk more freely about function
spaces, we work with finite sets rather than numbers in several instances, which motivates departing from the
usual notion.} indexed by $\Sigma$;
they are written using $\bf \Sigma, \Gamma$. We may write $\{a_1 : A_1,\; \ldots,\; a_n : A_n\}$
for the ranked alphabet $(\{a_1, \ldots, a_n\}, \arity)$ with $\arity(a_i) = A_i$.

Given a ranked alphabet $\bf \Sigma$, the set $\Tree({\bf \Sigma})$ of
trees/terms over a ranked alphabet $\bf \Sigma$ is defined as usual: if $a$ is a
letter of arity $X$ in $\bf \Sigma$ and $t$ a family of $\bf \Sigma$-trees, we
write $a(t)$ for the corresponding tree. Examples of such trees are pictured in
Figure~\ref{fig:tree-ex}.

\begin{rem}
  \label{rem:strings-as-trees}
  Given an alphabet $\Sigma$, define $\alphtorankedalph{\Sigma}$ to be the
  ranked alphabet $(\Sigma + \{\varepsilon\}, \arity)$ such that
  $\arity(\inl(a)) = \{*\}$ and $\arity(\inr(\varepsilon)) = \varnothing$. This
  gives a isomorphism $\Tree(\alphtorankedalph{\Sigma}) \simeq \Sigma^*$,
  illustrated on the right of Figure~\ref{fig:tree-ex}.
\end{rem}

\begin{figure}
\center
\includegraphics{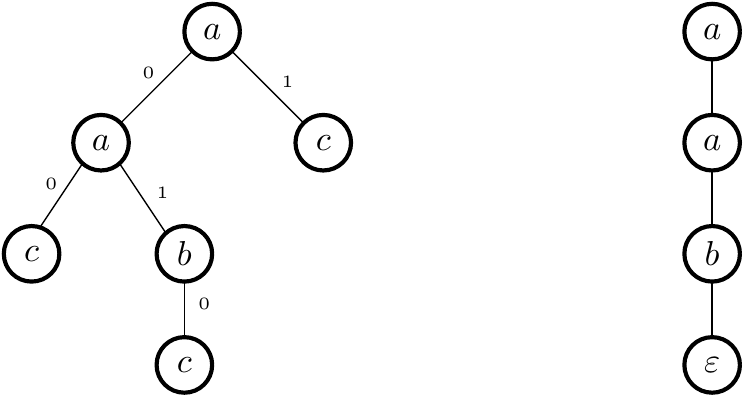}
\vspace{1.5em}
\hrule

\caption{Graphical representations of the tree $a(a(c,b(c)),c)$ over the ranked
  alphabet $\{a : \{0,1\},\, b : \{0\},\, c : \varnothing\}$ (left) and of the word
  $aab \in \{a,b\}^*$ seen as an element of $\Tree(\alphtorankedalph{\{a,b\}}) =
  \Tree({\{a : \{*\},\, b : \{*\},\, \varepsilon : \varnothing\}})$ (right).}
\label{fig:tree-ex}
\end{figure}

\subsection{Transducer models for regular functions over strings and trees}

\subsubsection{Strings}

Let us first recall the machine model that provides our reference definition for
regular functions: copyless streaming string transducers~\cite{SST} (SSTs). A
SST is an automaton whose internal memory contains, additionally to its control
state, a finite number of string-valued registers. It processes its input in a
single left-to-right pass. Each time a letter is read, the contents of the
registers may be recombined by concatenation to determine the new register
values. Formally:

\begin{defi}
  \label{def:register-transition}
  Fix a finite alphabet $\Gamma$. Let $R$ and $S$ be finite sets; we shall
  consider their elements to be ``register variables''.

  A \emph{$\Gamma$-register transition}\footnote{Sometimes called a
    \emph{substitution} in the literature, e.g.\
    in~\cite{SSTomega,AperiodicSST}.} from $R$ to $S$ is a function $t : S \to
  (\Gamma + R)^*$. Such a transition is said to be \emph{copyless} when for
  every $r \in R$, there is at most one occurrence of $\inr(r)$ among all the words
  $t(s)$ for $s \in S$ (i.e.\ when $\sum_{s \in S} \occ{\inr(r)}{t(s)} \le
  1$). We write $[R \to_{\Sr(\Gamma)} S]$ for the set of copyless
  $\Gamma$-register transitions from $R$ to $S$, or $[R \to_\Sr S]$ when
  $\Gamma$ is clear from context.

  Let $t \in [R \to_{\Sr(\Gamma)} S]$. For $x = (x_r)_{r \in R} \in
  (\Gamma^*)^R$ and $s \in S$, we denote by $(t^\dagger(x))_s \in \Gamma^*$ the
  word obtained from $t(s)$ by substituting every occurrence of a register
  variable $r \in R$ by the string $x_r$
  -- formally, by applying the morphism of free
  monoids $(\Gamma + R)^* \to \Gamma^*$ that maps $\inl(c)$ to $c$ and $\inr(r)$
  to $x_r$. This defines a set-theoretic map $t^\dagger : (\Gamma^*)^R \to
  (\Gamma^*)^S$, describing \emph{how $t$ acts on tuples of strings}.
\end{defi}

For instance, $t : z \mapsto axby$ (where we omitted $\inl/\inr$) is in
$[\{x,y\} \to_{\Sr(\{a,b\})} \{z\}]$ (it is copyless since $x$ and $y$
appear only once), and $t^\dagger(x \mapsto b,\, y \mapsto aa) = (z \mapsto
abbaa)$.

\begin{defi}[\cite{SST}]
  \label{def:sst}
    A (deterministic) \emph{copyless
    streaming string transducer} (SST) with input alphabet $\Sigma$ and output
  alphabet $\Gamma$ is a tuple $\cT = (Q, q_0, R, \delta, i, o)$ where
\begin{itemize}
\item $Q$ is a finite set of states and $q_0 \in Q$ is the initial state;
\item $R$ is a finite set of registers;
\item $\delta : \Sigma \times Q \to Q \times [R \to_{\Sr(\Gamma)} R]$ is
  the transition function;
\item $i \in [\varnothing \to_{\Sr(\Gamma)} R] \cong (\Gamma^*)^R$ describes
  the initial register values;
\item $o : Q \to [R \to_{\Sr(\Gamma)} \{\bullet\}]$ (where $\bullet$ is an
  arbitrary element) describes how to recombine the final values of the
  registers, depending on the final state, to produce the output.
\end{itemize}
(The SSTs that we consider in this paper are always copyless.)

The function $\Sigma^* \to \Gamma^*$ computed by $\cT$ maps an input string $w =
w_1 \ldots w_n \in \Sigma^*$ to the output string $o(q_n)^\dagger \circ
t^\dagger_n \circ \ldots \circ t^\dagger_1 \circ i^\dagger(\varnothing) \in
\Gamma^*$ where
\begin{itemize}
\item the empty family $\varnothing$ is indeed the unique element of
  $(\Gamma^*)^\varnothing$;
\item the codomain $(\Gamma^*)^{\{\bullet\}}$ of $o^\dagger(q_n)$ is identified with
  $\Gamma^*$;
\item the register transitions $(t_i)_{1 \leq i \leq n}$ and the final state
  $q_n \in Q$ are inductively defined, starting from the fixed initial state
  $q_0$, by $(q_i,t_i) = \delta(w_i,q_{i-1})$.
\end{itemize}
The functions that can be computed by copyless SSTs are called \emph{regular
  string functions}.
\end{defi}

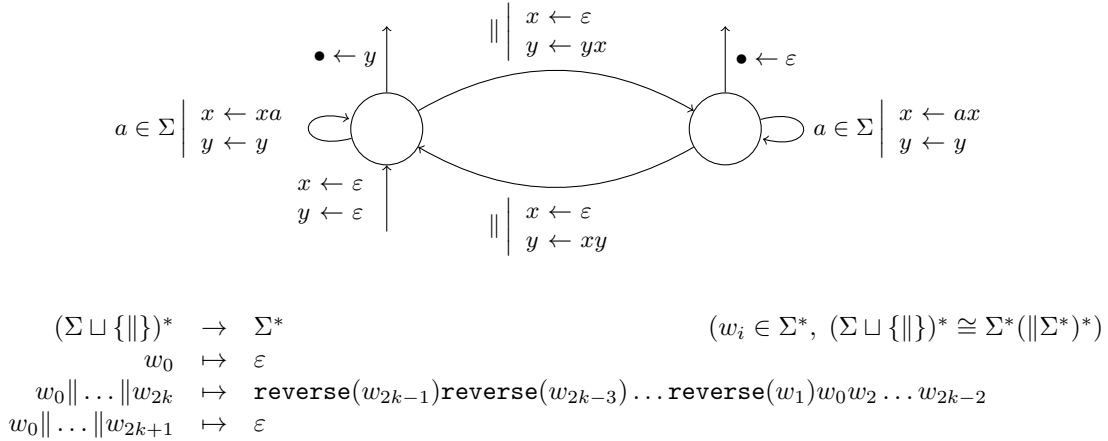
\begin{figure}
\begin{center}
\begin{tikzpicture}[->,node distance=1.5cm]
      \node[state] (A) {} ;
  \node (D) [right of=A] {};
  \node (E) [right of=D] {};
  \node[state] (B) [right of=E] {};
  \node (I) [below of=A] {};
  \node (OA) [above of=A] {};
  \node (OB) [above of=B] {};

\path
  (I) edge node[left, style={font=\footnotesize}]
 {$\begin{array}{l@{~}c@{~}l} x &\leftarrow& \varepsilon \\ y &\leftarrow& \varepsilon \end{array}$} (A) ;
\path
  (A) edge node[left, style={font=\footnotesize}]
 {$\bullet \leftarrow y$} (OA) ;
\path
  (B) edge node[right, style={font=\footnotesize}]
 {$\bullet \leftarrow \varepsilon$} (OB) ;
\path
  (A) edge [bend left] node[above, style={font=\footnotesize}]
 {$\|\left| \begin{array}{l@{~}c@{~}l} x &\leftarrow& \varepsilon \\ y &\leftarrow& yx \end{array} \right.$} (B) ;
\path (B) edge [loop right] node[right,style={font=\footnotesize}]
{$a \in \Sigma \left| \begin{array}{l@{~}c@{~}l} x &\leftarrow& ax \\ y &\leftarrow& y \end{array} \right.$}
(A) ;
\path (B) edge [bend left] node[below,style={font=\footnotesize}]
 {$\|\left| \begin{array}{l@{~}c@{~}l} x &\leftarrow& \varepsilon \\ y &\leftarrow& xy \end{array} \right.$}
(A) ;
\path (A) edge [loop left] node[left,style={font=\footnotesize}]
{$a \in \Sigma \left| \begin{array}{l@{~}c@{~}l} x &\leftarrow& xa \\ y &\leftarrow& y \end{array} \right.$}
(A) ;
\end{tikzpicture}

\[
\small
\begin{array}{rcl!{\hspace{15em}} l}
(\Sigma \sqcup \{ {\|} \})^* &\to& \Sigma^* 
& {(w_i \in \Sigma^*,\; (\Sigma \sqcup \{\|\})^* \cong \Sigma^*(\|\Sigma^*)^*)}\\
w_0 &\mapsto& \varepsilon \\
w_0\|\ldots\|w_{2k} &\mapsto&
\multicolumn{2}{l}{
\mathtt{reverse}(w_{2k-1})
\mathtt{reverse}(w_{2k-3})
\ldots
\mathtt{reverse}(w_1)
w_0
w_2
\ldots
w_{2k-2}}
\\
w_0\|\ldots\|w_{2k+1} &\mapsto& \varepsilon \\
\end{array}
\]
\end{center}
\vspace{0.5em}
\hrule
\caption{An informal depiction of a SST and the induced map $(\Sigma \sqcup \{
  \| \})^* \to \Sigma^*$.}
\label{fig:sst-charles}
\end{figure}

\begin{exa}
\label{ex:sst-rev}
  Let us describe a simple copyless SST with $\Sigma = \Gamma$ and a single
  state, so that $\delta : \Sigma \to [R \to_\Sr R]$. We take $R =
  \{x,y\}$; both $x$ and $y$ are initialized with the empty
  string~$\varepsilon$, and the register transition $t_c = \delta(c) \in [R
  \to_\Sr R]$ associated to $c \in \Sigma$ is $(x \mapsto xc,\, y \mapsto
  cy)$ (to be pedantic, one should write $(x \mapsto \inr(x)\inl(c),\, y \mapsto
  \inl(c)\inr(y))$). Then for $w = w[1] \ldots w[n] \in \Sigma^*$, we have:
  \[t_{w[n]}^\dagger \circ \dots \circ t_{w[1]}^\dagger(x \mapsto
    \varepsilon,\,y\mapsto\varepsilon) = (x\mapsto
    w,\,y\mapsto\mathtt{reverse}(w))\]
  If we take $o = xy$ (via the canonical
  isomorphism $[R \to_{\Sr(\Gamma)} 1]^{\{q\}} \cong (\Gamma + R)^*$), the
  function computed by the SST is $w \in \Sigma^* \mapsto w \cdot
  \mathtt{reverse}(w)$.
\end{exa}
Figure~\ref{fig:sst-charles} shows a more sophisticated SST with two states and
the associated regular function.

\subsubsection{Trees as output}

Let us briefly discuss the challenges that arise when extending this model to
handle \emph{ranked trees} instead of strings. We will revisit this material in
more detail in \Cref{sec:trees}.

The notion of \emph{regular tree-to-tree function} is defined by generalizing
the characterization of regular string functions by Monadic Second-Order
Logic~\cite{MacroMSO,AttributedMSO,EngelfrietHoogeboom}, in a way that is
compatible with the above isomorphism. There are two orthogonal difficulties
that have to be overcome to define a SST-like model for regular tree functions:
one comes from producing trees as output, while the other comes from taking
trees as input. \emph{Bottom-up ranked tree transducers}\footnote{The name
  \enquote{streaming tree transducer} is used in~\cite{BRTT} for a transducer
  model operating over unranked trees. BRTTs are proposed in the same paper as a
  simpler, equally expressive variant for the special case of ranked trees.}
(BRTTs)~\cite{BRTT} (and the similar model of register tree transducers
in~\cite[\S4]{FOTree}) provide solutions for both.

String-to-tree regular functions require a modification of the kind of data
stored in the registers of an SST. Tree-valued registers are \emph{not enough},
for the following reasons: to recover the flexibility of string concatenation,
one should be able to perform operations such as grafting the root of some tree to
a leaf of another tree; but then the latter should be a tree with a
distinguished leaf, serving as a \enquote{hole} waiting to be substituted by a
tree. (This is fundamental in the theory of forest algebras, which proposes
various counterparts for trees to the monoid of strings with concatenation,
see~\cite{ForestAlgebras}.) By allowing both trees and \enquote{one-hole trees}
as register values, with the appropriate notion of copyless register transition
(cf.\ \Cref{subsec:registercat-tree}),
one gets the \emph{copyless streaming string-to-tree transducers}, whose
expressive power corresponds exactly to the regular
functions~\cite[Theorem~3.16]{BRTT}.

\subsubsection{Trees as inputs}
\label{subsubsec:prelim-trees-inputs}

To compute tree-to-tree regular functions, the first idea would be to blend the
notion of copyless SST with the classical bottom-up \emph{tree automata}. One
would then get \emph{copyless bottom-up ranked tree transducers}. However, this
model is believed to be too weak to express all regular tree functions (even in
the case of tree-to-string functions). An explicit counterexample is conjectured
in~\cite[\S2.3]{BRTT}, in the case of regular functions on unranked trees; we
adapt it here into a function from ranked trees to strings.

In the example below, for a ranked letter $a$ of arity $\bbtwo =
\{\triangleleft,\triangleright\}$, we use the abbreviation $a(t,u)$ for
$a(\triangleleft \mapsto t,\, \triangleright \mapsto u)$.

\begin{exa}[\enquote{Conditional swap}]
  \label{ex:condswap}
  Define $f : \Tree(\{a : \bbtwo, b : \bbtwo, c : \varnothing\}) \to \{a,b,c\}^*$
  by
  \[ f(a(t,u)) = f(u) \cdot a \cdot f(t) \qquad f(t) = \mathtt{inorder}(t)\
    \text{if $t$ doesn't match the previous pattern} \]
  where $\mathtt{inorder}$ prints the nodes of $t$ following a depth-first
  in-order traversal. In other words, $f = \mathtt{inorder} \circ g$ where
  $g(a(t,u)) = a(g(u),g(t))$ and $g(t) = t$ otherwise (i.e.\ when the root of
  $t$ is either $b$ or $c$).
\end{exa}
\begin{conj}[{adapted from~\cite[\S2.3]{BRTT}}]
  The above $f$ cannot be computed by a copyless BRTT.
\end{conj}

One must then allow more register transitions than the copyless ones. This
cannot be done haphazardly, for arbitrary register transitions would lead to a
much larger class of functions than regular tree functions. Alur and D'Antoni
call their relaxed condition~\cite{BRTT} the \emph{single use restriction} (it
will only be formally defined in \Cref{subsec:with-trees}); the following
single-state BRTT for~$f$ provides a typical example of the new possibilities
allowed.

\begin{exa}[Non-copyless BRTT for conditional swap]
\label{ex:brtt-condswap}
  Take $R = \{x,y\}$, initialized at the $c$-labeled leaves with $(x \mapsto
  c,\, y \mapsto c)$. At a subtree $a(u,v)$, we need to combine the registers
  $x_\ttl,y_\ttl$ (resp.\ $x_\ttr,y_\ttr$) coming from the left (resp.\ right)
  child $u$ (resp.\ $v$) to produce the values of the registers $x,y$ at this
  node: this is performed by a register transition
  \[ t_a \in [\{x_\ttl,y_\ttl,x_\ttr,y_\ttr\} \to_\Sr \{x,y\}] \qquad
    t_a(x) = x_\ttr a x_\ttl \qquad t_a(y) = y_\ttl a y_\ttr \]
  The idea is that the register values produced by processing a subtree $u$ are
  $f(u)$ for $x$ and $\mathtt{inorder}(u)$ for $y$. The register transition for
  a $b$-labeled node is then $t_b(x) = t_b(y) = y_\ttl b y_\ttr$, reflecting the
  fact that $f(b(u,v)) = \mathtt{inorder}(b(u,v))$.

  This $t_b$ is not copyless since $y_\ttl$ occurs twice: once in $t_b(x)$ and
  once in $t_b(y)$. The observation at the heart of the single use restriction
  is that the values of $x$ and $y$ for a given subtree can never be combined in
  the same expression in the remainder of the BRTT's run, so that allowing this
  duplication of $y_\ttl$ will never lead to having two copies of $y_\ttl$
  inside the value of a single register. We will see much later in
  Example~\ref{ex:condswap-sru} that this BRTT is indeed single-use-restricted.
\end{exa}

\subsection{The $\laml$-calculus, encodings of strings/trees, and definability of
functions}

\subsubsection{Types \& terms}

We consider a linear $\lambda$-calculus which we dub the
\emph{$\laml$-calculus}, based (via the Curry--Howard correspondence)
on propositional intuitionistic linear logic with both multiplicative and additive connectives (IMALL) together with a base linear type $\basety$.
The grammar of types is as follows:
\[\tau, \sigma \bnfeq \basety \bnfalt \tau \lin \sigma \bnfalt \tau \tensor
\sigma \bnfalt \unit \bnfalt \tau \to \sigma \bnfalt \tau \with \sigma \bnfalt \tau \oplus \sigma \bnfalt \top \bnfalt 0\]
A typing context $\Psi$ is a finite set of declarations $x_1 : \tau_i, \ldots, x_k : \tau_k$ where the $x_i$ are pairwise distinct variables (which
constitute the set of free variables of $\Psi$)
and the $\tau_i$ are types.
Typed $\laml$-terms are given in Figure~\ref{fig:laml} along with the inductive definition of the typing judgment $\Psi; \; \Delta \vdash t : \tau$,
where $\Psi$ and $\Delta$ are contexts (with disjoint sets of free variables), $\tau$ is a type and $t$ is a term.
In such a judgment, $\Psi$ is called the \emph{non-linear} context and $\Delta$ the \emph{linear} context; the basic idea is that
variables in $\Psi$ may be used arbitrarily many times, while those in $\Delta$ must be used \emph{exactly once}.
This is formally more restrictive than an \emph{affineness condition}, where we would rather restrict variables in $\Delta$ to occur \emph{at most once}
in $t$.

In practice, $\laml$ is not less expressive than its affine
variant\footnote{Which would be obtained by adjoining the following
  \emph{weakening rule} to the system presented in Figure~\ref{fig:laml}:
\[\dfrac{\Psi; \; \Delta \vdash t : \tau}{\Psi ; \; \Delta, \; \Delta' \vdash t : \tau}\]} since it features additives: the basic idea is that the affineness can be encoded at the level
of types by using the linear type $\tau \with \unit$ instead of the affine type
$\tau$ (as argued for instance in~\cite[\S1.2.1]{LLSS}).

The simply typed $\lambda$-calculus admits an embedding into $\laml$.
Conversely, there is a mapping from $\laml$ to the simply typed
$\lambda$-calculus with products and sums by ``forgetting linearity'' (and 
replacing the tensorial product eliminator $\tlet{x\tensor y}{t}{u}$ by the variant
based on projections $u[\pi_1(t)/x,\pi_2(t)/y]$).

\begin{figure}
\begin{mdframed}
\[
\begin{array}{c !{\quad} c}
\dfrac{}{\Psi ; \; x : \tau \vdash x : \tau}
&
\dfrac{}{\Psi, x : \tau ; \; \cdot \vdash x : \tau}
\\\\
 \dfrac{\Psi; \; \Delta, \; x : \tau \vdash t : \sigma}{\Psi; \; \Delta \vdash \lambda x. t : \tau \lin \sigma}
&
\dfrac{\Psi; \; \Delta \vdash t : \tau \lin \sigma \qquad \Psi; \; \Delta'' \vdash u : \tau}{\Psi ; \; \Delta, \; \Delta' \vdash t \; u : \sigma}
  \\\\
\dfrac{\Psi,\; x : \tau; \; \Delta \vdash t : \sigma}{\Psi; \; \Delta \vdash \lambda^\oc x. t : \tau \to \sigma}
  &
\dfrac{\Psi; \; \Delta \vdash t : \tau \to \sigma \qquad \Psi; \; \cdot \vdash u : \tau}{\Psi ; \; \Delta \vdash t \; u : \sigma}
  \\\\
\dfrac{\Psi; \; \Delta \vdash t : \tau \qquad \Psi; \; \Delta' \vdash u : \sigma}{\Psi; \; \Delta, \;\Delta' \vdash t \tensor u : \tau \tensor \sigma}
&
\dfrac{
\Psi; \; \Delta' \vdash u : \tau \tensor \sigma
\qquad
\Psi; \; \Delta, x : \tau, y : \sigma \vdash t : \kappa
}{\Psi ; \; \Delta, \; \Delta' \vdash \tlet{x \tensor y}{u}{t} : \kappa}
\\\\
\dfrac{}{\Psi; \; \cdot \vdash () : \unit}
&
\dfrac{\Psi; \; \Delta \vdash t : \unit \qquad \Psi; \; \Delta' \vdash u : \tau}{\Psi; \; \Delta, \; \Delta' \vdash \tlet{()}{t}{u} : \tau}
\\\\\\
\dfrac{\Psi; \; \Delta \vdash t : \tau \qquad
\Psi; \; \Delta \vdash u : \sigma
}{\Psi ; \; \Delta \vdash \tuple{t, u} : \tau \with \sigma}
&
\dfrac{\Psi ;\; \Delta \vdash t : \tau \with \sigma}{\Psi ; \; \Delta \vdash \pi_1(t) : \tau}
\qquad
\dfrac{\Psi ;\; \Delta \vdash t : \tau \with \sigma}{\Psi ; \; \Delta \vdash \pi_2(t) : \sigma}
\\\\\\
\multicolumn{2}{c}{
\dfrac{\Psi ;\; \Delta \vdash t : \tau}{\Psi ; \; \Delta \vdash \inl(t) : \tau \oplus \sigma}
\qquad
\dfrac{\Psi ;\; \Delta \vdash t : \sigma}{\Psi ; \; \Delta \vdash \inr(t) : \tau \oplus \sigma}}
\\\\
\multicolumn{2}{c}{
\dfrac{
\Psi; \; \Delta, x : \tau \vdash u : \kappa \qquad
\Psi; \; \Delta, x : \tau \vdash v : \kappa \qquad
\Psi; \; \Delta' \vdash t : \tau \oplus \sigma
}{\Psi ; \; \Delta, \Delta' \vdash \case(t, x.u,x.v) : \kappa}
}
\\\\
\dfrac{}{\Psi; \; \Delta \vdash \tuple{ } : \top}
&
\dfrac{\Psi; \; \Delta \vdash t : 0}{\Psi; \; \Delta, \; \Delta' \vdash \abort(t) : \tau}
\end{array}
\]
\end{mdframed}
\caption{Typing rules of $\laml$.}
\label{fig:laml}
\end{figure}

\begin{figure}
\begin{mdframed}
\[
\footnotesize
\begin{array}{l !\; c !{\;\;} c}
\text{$\beta$-equivalence} &
(\lam x. t) \; u ~=_\beta~ t[u/x]
&
(\lam^\oc x. t) \; u ~=_\beta~ t[u/x]
\\ &
\pi_1(\tuple{t,u}) ~=_\beta~ t
&
\pi_2(\tuple{t,u}) ~=_\beta~ u
\\ &
\case(\inl(t), x.u,x.v) ~=_\beta~ u[t/x]
&
\case(\inr(t), x.u,x.v) ~=_\beta~ v[t/x]
\\ &
\tlet{x \tensor y}{t \tensor u}{v} ~~=_\beta~~ v[t/x][u/y]
&
\tlet{()}{()}{t} ~=_\beta~ t
\\\\
\text{$\eta$-equivalence} &
\lam x. t \; x
~=_\eta~
t
&
\lam^\oc x. t \; x
~=_\eta~
t
\\
&
\tlet{x \tensor y}{t}{u[x \tensor y/z]} ~~=_\eta~~ u[t/z]
&
\tuple{\pi_1(t), \pi_2(t)} ~=_\eta~ t
\\
&
\tlet{x \tensor y}{t}{v[u/z]} ~=_\eta~ v[\tlet{x \tensor y}{t}{u}/z]
&
x ~=_\eta~ \tuple{ }
\\
&
\tlet{()}{t}{u[()/z]} ~~=_\eta~~ u[t/z]
\\
&
\tlet{()}{t}{v[u/z]} ~=_\eta~ v[\tlet{()}{t}{u}/z]
\\ &
\case(t, x.u[\inl(x)/z], y.u[\inr(y)/z]) ~=_\eta~ u[t/z]
&
\abort(t) ~=_\eta~ u 
\end{array}
\]
\end{mdframed}
\caption{Equations for $\laml$-terms (relating terms that have matching types).}
\label{fig:laml-betaeta}
\end{figure}

As usual, we identify $\laml$-terms up to renaming of bound variables ($\alpha$-equivalence) and admit the standard definition of the capture-avoiding substitution. For the purpose of this paper,
since we are not interested in the fine details of their operational semantics,
we usually consider $\laml$-terms up to $\lamlequiv$-equivalence $=_{\lamlequiv}$ as generated by the equations in Figure~\ref{fig:laml-betaeta} and congruence. Note that those equations are implicitly typed and that typing is invariant under $\lamlequiv$-equivalence. 

Much like any $\lambda$-calculus, $\laml$ can be seen as a programming language by considering a
reduction relation $\to_{\beta\extr}$, which happens to be included in $=_{\beta\eta}$. One
property that we shall use is that $\laml$ is \emph{normalizing}, i.e., that the relation
$\to_{\beta\extr}$ is terminating. This allows to consider terms of very specific shape when
working up to $\beta\eta$. While the argument is routine, we need this result, as well as a
fine-grained understanding of the normal forms to discuss further preliminary syntactic lemmas,
so we give an outline in Appendix~\ref{sec:laml-normalization}.

We now isolate an important class of types and terms for the sequel.
\begin{defi}
We call a type \emph{purely linear} if it does not have any occurrence of the `$\to$' connective.
A $\laml$-term $t$ is also called \emph{purely linear} if there is a typing derivation
$\Psi; \; \Delta \vdash t : \tau$ where any type occurring must be purely linear.
\end{defi}

Intuitively, purely linear terms are those which are not allowed to duplicate any arguments
involving $\basety$. For any type derivation $\Psi; \; \Delta \vdash t : \tau$, if the types
occurring in $\Psi$ and $\Delta$, as well as $\tau$, are purely linear, then so is $t$;
this is a consequence of normalization.

\subsubsection{Church encodings}
In order to discuss string (and tree) functions in $\laml$, we need to discuss how they are encoded.
Recall that in the pure (i.e.\ untyped) $\lambda$-calculus, the canonical way to
encode
inductive types\footnote{Including the natural numbers, if one wants for
  instance to show that the untyped $\lambda$-calculus captures all computable
  functions.
We should also mention that the generalization of Church encodings to trees is
actually due to Böhm and Berarducci~\cite{BohmBerarducci}.}
is via \emph{Church encodings}. Such encodings are typable in the simply-typed $\lambda$-calculus. For instance, for natural numbers and strings over $\{a,b\}$, writing $\underline w$ for the
Church encoding of $w$, we have
\[
\begin{array}{clc}
\Nat^\oc ~=~ (\basety \to \basety) \to \basety \to \basety &\qquad \qquad \qquad& \Str^\oc_{\{a,b\}} ~=~ (\basety \to \basety) \to (\basety \to \basety) \to \basety \to \basety \\\\
\underline{3} = \lambda^\oc s. \lambda^\oc z. s \; (s \; (s \; z)) & & \underline{aab} = \lambda^\oc a. \lambda^\oc b. \lambda^\oc \varepsilon. a \; (a \; (b \; \varepsilon)) \\
\end{array}
\]
Conversely, one may show that any closed simply typed $\lambda$-term of type $\Nat^\oc$ (resp.\ $\Str^\oc_{\{a,b\}}$) is $\lamlequiv$-equivalent to the Church encoding of some number (resp.\ string).
In the rest of this paper, we will use a less common, but more precise
$\laml$-type for Church encodings of strings of trees, first introduced
in~\cite[\S5.3.3]{girardLL}.

\begin{defi}
  Let $\Sigma$ be an alphabet. We define $\Str_\Sigma$ as $(\basety \lin
  \basety) \to \ldots \to (\basety \lin \basety) \to \basety \to \basety$ where
  there are $\card{\Sigma}$ occurrences of $\basety \lin \basety$. Note in
  particular that thanks to the isomorphism\footnote{We keep this notion
    informal, but suffices to say that this is intended to be definable
    internally to $\laml$.} $(A \with B) \to C \;\cong\; A \to B \to C$ (non-linear
  currying), we have\footnote{In this encoding, the unique constructor of arity
    $0$ is treated non-linearly, while in the prequel~\cite{aperiodic}, it was
    treated linearly. We chose non-linearity here in order to be consistent with
    the definition for ranked trees (cf.\
    Remark~\ref{rem:strings-as-trees-lambda}): indeed, while strings have a
    single end-marker, trees may have multiple leaves, so non-linearity is
    necessary in their case. This apparent inconsistency with our previous work
    is actually unproblematic as both string encodings are interconvertible, see
    e.g.~\cite[Remark~5.7]{aperiodic}. }
\[\Str_\Sigma ~~\cong~~ \left(\bigwith_{a \in \Sigma} (\basety \lin \basety)\right) \to \basety \to \basety\]
\end{defi}

It can be checked that $\Str_\Sigma$ has the same (up to $\beta\eta$ equality) closed inhabitants as
the usual $\Str^\oc_\Sigma$ presented above, but one should keep in mind
that this choice is not entirely innocuous. It is in large part motivated by our main result (Theorem~\ref{thm:main-string}), which might no longer
hold when taking $\Str^\oc_\Sigma$ instead of $\Str_\Sigma$.

This situation generalizes to trees.
For instance, the Church encoding
of the tree depicted in Figure~\ref{fig:tree-ex} is
$\lam^\oc a. \lam^\oc b. \lam^\oc c.\; a \; (a \; c \; (b \; c)) \; c ~:~ 
(\basety \to \basety \to \basety) \to (\basety \to \basety) \to \basety \to \basety$.

\begin{defi}
Given a ranked alphabet ${\bf \Sigma} = (\Sigma, \arity)$, the $\laml$ type $\Treety_{\bf \Sigma}$ is defined
as
\[\Treety_{\bf \Sigma} ~~=~~ (\basety \lin \ldots \lin \basety) \to \ldots
\to (\basety \lin \ldots \lin \basety) \to \basety\]
where there are $\card{\Sigma}$ top-level arguments, and, within the component
corresponding to the letter
$a \in \Sigma$, there are $\card{\arity(a)}$.
In other words, we have the isomorphism
\[\Treety_{\bf \Sigma} ~~\cong~~ \left(\bigwith_{a \in \Sigma}
  (\basety^{\tensor\arity(a)} \lin \basety)\right) \to \basety\]
\end{defi}

\begin{rem}
  \label{rem:strings-as-trees-lambda}
The isomorphism of Remark~\ref{rem:strings-as-trees}
translates to an equality $\Str_{\Sigma} = \Treety_{\alphtorankedalph{\Sigma}}$.
\end{rem}

Church encodings give a map from trees in $\Tree({\bf \Sigma})$ to $\laml$-terms of type $\Treety_{\bf \Sigma}$
in the empty context.
This map is in fact a bijection if terms are considered up to $\beta\eta$-equality:
normalization of the $\laml$-calculus enforces surjectivity, and one may use a set-theoretic semantics
of $\laml$ to build a left inverse (see the proof of Proposition~\ref{prop:laml-tree-nf} in the appendix for further details).

\subsubsection{Computing with Church encodings}
We are now ready to give our notion of computation for our string (and tree) functions.
First, we need an operation of type substitution in $\laml$,
which allow to substitute an arbitrary type $\kappa$ for $\basety$.
\[
\begin{array}{l@{~}c@{~}l!\qquad l@{~}c@{~}ll}
\basety[\kappa] &=& \kappa & (\tau \lin \sigma)[\kappa] &=& \tau[\kappa] \lin \sigma[\kappa]
& \qquad\ldots
\end{array}
\]
Type substitution extends in the obvious way to typing contexts as well, and even
to \emph{typing derivations}, so that 
\[\Psi; \; \Delta \vdash t : \tau \qquad \Rightarrow \qquad
\Psi[\kappa]; \; \Delta[\kappa] \vdash t : \tau[\kappa]\]

In particular, it means that a Church encoding $t : \Treety_{\bf \Sigma}$ is
also of type $\Treety_{\bf \Sigma}[\kappa]$ for any type $\kappa$.
This ensures that the following notion of definable tree functions (strings
being a special case) in the  $\laml$-calculus
makes sense.

\begin{defi}
\label{def:laml-definable}
A function $f : \Tree({\bf \Sigma}) \to \Tree({\bf \Gamma})$ is called \emph{$\laml$-definable} when
there exists a \emph{purely linear} type $\kappa$ together with a $\laml$-term
\[\mathtt{f} \quad:\quad \Treety_{\bf \Sigma}[\kappa] \lin \Treety_{\bf \Gamma}\]
such that $f$ and $\mathtt{f}$ coincide up to Church encoding; i.e.,
for every tree $t \in \Tree({\bf \Sigma})$ 
\[\underline{f(t)} ~~ =_{\lamlequiv} ~~ \mathtt{f} \; \underline{t}\]
In particular, a string function $\Sigma^* \to \Gamma^*$ is $\laml$-definable
when the corresponding unary tree function
$\Tree\left(\alphtorankedalph{\Sigma}\right) \to
\Tree\left(\alphtorankedalph{\Gamma}\right)$ (cf.\
Remark~\ref{rem:strings-as-trees-lambda}) is $\laml$-definable. Note that
\[\Treety_{\alphtorankedalph{\Sigma}}[\kappa] \lin \Treety_{\alphtorankedalph{\Gamma}}
  \quad=\quad \Str_\Sigma[\kappa] \lin \Str_\Gamma \]
\end{defi}

\begin{rem}
Once again, our set-up, summarized in Definition~\ref{def:laml-definable}, is biased toward making our main theorem true;
there are many non-equivalent alternatives which also make perfect sense. For instance,
changing the following would be reasonable:
\begin{itemize}
\item allow $\kappa$ to be be arbitrary (i.e.\ to contain $\oc$) or with some restrictions.
\item consider the non-linear arrow $\to$ instead of $\lin$ at the toplevel.
\item change the type of Church encodings (recall the distinction $\Str^\oc_\Sigma$/$\Str_\Sigma$).
\end{itemize}
Most of these alternatives share the good structural properties outlined below.
Giving more automata-theoretic characterizations for those and comparing them lies beyond the scope of
this paper, but would be interesting.

The two first choices above will turn out to have a clear operational meaning:
the pure linearity of $\kappa$ corresponds to single-use-restricted assignment
(as mentioned in the introduction), whereas the use of the linear function arrow
`$\lin$' corresponds to the fact that a streaming tree transducer
\emph{traverses its input in a single pass}.
\end{rem}

As our main theorems claim, $\laml$-definable functions and regular functions coincide, so all
our examples of regular functions can be coded in $\laml$, as we show concretely
below.
\begin{exa}
\label{ex:laml-rev}
The $\mathtt{reverse}$ function $\Sigma^* \to \Sigma^*$ from Example~\ref{ex:sst-rev}
is $\laml$-definable. Supposing that we have $\Sigma = \{a_1, \dots, a_k\}$,
one $\laml$-term that implements it is
\[\lam s. \lam^\oc a_1. \ldots \lam^\oc a_k. \lam^\oc \varepsilon.\; s \; (\lam x.
  \lam z. x \; (a_1 \; z)) \ldots (\lam x. \; (a_k \; z)) \; (\lam x.x) \;
  \varepsilon
  \;:\; \Str_\Sigma[\basety \lin \basety] \lin \Str_\Sigma \]
\end{exa}
\begin{exa}
\label{ex:laml-charles}
The SST of Figure~\ref{fig:sst-charles} is computed by a $\laml$-term of type
$\Str_{\Sigma \sqcup \{\|\}}\left[
\tau
\right] \lin \Str_{\Sigma}$
with $\tau = \Bool \tensor ((\basety \lin \basety) \with \unit) \tensor ((\basety \lin \basety) \with \unit)$.
Intuitively, $\Bool$ corresponds to the current state of the SST while each component $(\basety \lin \basety) \with \unit$ corresponds to a register.
Define the auxiliary terms $\delta : (\basety \lin \basety) \lin \tau \lin \tau$,
$\delta_{\|} : \basety \lin \tau \lin \tau$ and $o : \basety \lin (\tau \lin \tau) \lin \basety$ as
\begin{figure}
\[
\footnotesize
\begin{array}{ll}
\begin{array}{l@{\;}l}
\delta ~ = ~ \lam a \; z. &
\mathsf{let} \; (b,z') \; = z \; \mathsf{in} \\
&\mathsf{let} \; (x,y) \; = z' \; \mathsf{in} \\
&\mathsf{if} \; b \; \mathsf{then} \\
&\quad (\mathtt{tt}, \tuple{\lam u.\; \pi_1(x) \; (a \; u), \tlet{()}{\pi_2(y)}{\pi_2(x)}}, y)\\
&\mathsf{else} \\
&\quad(\mathtt{ff}, \tuple{\lam u. \; a \; (\pi_1(x) \; u), \tlet{()}{\pi_2(y)}{\pi_2(x)}}, y)\\
\end{array}
\\
\\
\begin{array}{l@{\;}l}
\delta_{\|} ~ = ~ \lam z. &
\mathsf{let} \; (b,z') \; = z \; \mathsf{in} \\
&\mathsf{let} \; (x,y) \; = z' \; \mathsf{in} \\
&\mathsf{if} \; b \; \mathsf{then} \\
&\quad (\mathtt{ff}, \tuple{\lam u.u , \tlet{()}{\pi_2(y)}{\pi_2(x)}}, \tuple{\lam v. \pi_1(y) \; (\pi_1(x) \; u),
\tlet{()}{\pi_1(x)}{\pi_2(y)}})\\
&\mathsf{else} \\
&\quad (\mathtt{tt}, \tuple{\lam u.u , \tlet{()}{\pi_2(y)}{\pi_2(x)}}, \tuple{\lam v. \pi_1(x) \; (\pi_1(y) \; u),
\tlet{()}{\pi_1(x)}{\pi_2(y)}})\\
\end{array}
\\
\\
\begin{array}{l@{\;}l}
o ~ = ~ \lam \varepsilon \; z. &
\mathsf{let} \; (b,z') \; = z \; (\mathtt{tt}, \tuple{\lam u. u, ()}, \tuple{\lam u. u}) \; \mathsf{in} \\
&\mathsf{let} \; (x,y) \; = z' \; \mathsf{in} \\
&\mathsf{let} \; () \; = \pi_2(x) \; \mathsf{in} \\
&\mathsf{if} \; b \; \mathsf{then} \\
&\quad \pi_1(y) \; \varepsilon \\
&\mathsf{else} \\
&\quad \tlet{()}{\pi_2(y)}{\varepsilon}
\end{array}
\\\\\hline
\end{array}
\]
\caption{Auxiliary terms for Example~\ref{ex:laml-charles}
($\mathtt{tt} = \inl(())$, $\mathtt{ff} = \inr()$ and $\mathsf{if} \; t \; \mathsf{then} \; u \; \mathsf{else} \; v$ is a notation for $\case(t, x.\tlet{()}{x}{u}, y. \tlet{()}{y}{v})$).
}
\label{fig:aux-terms}
\end{figure}
Supposing that we have $\Sigma = \{a_1, \ldots, a_k\}$,
and that the letter $\|$ corresponds to the first constructor in the input string, the $\laml$-definition is
given by
\[\lam s.\lam^\oc a_1.\ldots \lam^\oc a_k. \lam^\oc \varepsilon. \; o \; (s \; \delta_\| \; (\delta \; a_1) \; \ldots \; (\delta \; a_k))\]
where the terms $\delta$, $\delta_{\|}$ and $o$ are defined in Figure~\ref{fig:aux-terms}.
\end{exa}
\begin{exa}
\label{ex:laml-condswap}
Consider the ranked alphabet ${\bf \Sigma} = \{a : \bbtwo, b : \bbtwo, c :
\varnothing\}$ (where $\bbtwo = \{\triangleleft,\triangleright\}$) and
the alphabet $\Gamma = \{a, b, c\}$.
The conditional swap of Example~\ref{ex:condswap} is $\laml$-definable as a term of type
\[\Treety_{\bf \Sigma}[(\basety \lin \basety) \with (\basety \lin \basety)] \to \Str_{\Gamma}\]
reminiscent of the BRTT given in Example~\ref{ex:brtt-condswap}. Observe the use
of an \emph{additive conjunction}~`$\with$' (that is not of the form
$(-\with\unit)$ meant to make data discardable), reflecting the fact that this BRTT is single-use-restricted but
not copyless. To wit, setting $\tau = (\basety \lin \basety) \with (\basety
\lin \basety)$ and assuming free variables $a, b : \basety \lin \basety$, define
the auxiliary terms
\[
\begin{array}{lcll}
\delta_a &=& \lam l. \lam r. \; \tuple{\pi_1(l) \circ a \circ \pi_1(r), \pi_1(r) \circ a \circ \pi_1(l)}
&: \tau \lin \tau \lin \basety \lin \basety\\
\delta_b &=& \lam l. \lam r. \; (\lam x. \tuple{x,x}) \; (\pi_1(l) \circ b \circ \pi_1(r))
&: \tau \lin \tau \lin \basety \lin \basety\\
\end{array}
\]
where $f \circ g$ stands for the composition $\lam z.\; f \; (g \; z)$.
The conditional swap is then coded as
\[\lam t. \lam^\oc a. \lam^\oc b. \lam^\oc c. \lam^\oc \varepsilon. \; \pi_2\;(t \; \delta_a \; \delta_b \; (\lam x.\; c \; x)) \; \varepsilon\]
\end{exa}

\subsection{Monoidal categories and related concepts}
\label{subsec:prelim-cat}

Our use of category theory, while absolutely essential, stays at a fairly elementary level.
We assume familiarity with the notions of category, functor, natural
transformation, (cartesian) product
and coproduct (and their nullary cases, terminal and initial objects), but not much more than that; the remaining categorical prerequisites are summed up here for
convenience. The reader familiar with monoidal closed categories can safely
skip directly to \S\ref{subsubsec:prelim-affine}.

\subsubsection{Monoidal categories, symmetry and functors}

The idea of categorical semantics is to interpret the types of a programming
language -- in our case, the purely linear fragment of the $\laml$-calculus --
as objects, and the programs (terms) as morphisms. (A formal statement tailored
to our purposes will be given later in Lemma~\ref{lem:laml-initial}.) In this
perspective, the additive conjunction `$\with$' of the $\laml$-calculus is
interpreted as a categorical \emph{cartesian product}, while the additive
disjunction `$\oplus$' corresponds to a \emph{coproduct}; this justifies our use
of the notations $\with/\oplus$ for products/coproducts. We now define monoidal
products, which are meant to interpret the multiplicative conjunction
`$\otimes$'.

\begin{defi}[{\cite[Sections~4.1 to~4.4]{mellies09ps}}]
Let $\cC$ be a category. A \emph{monoidal product} $\tensor$ over $\cC$
is given by the combination of
\begin{itemize}
\item a bifunctor $- \tensor - : \cC \times \cC \to \cC$
\item a distinguished object $\unit$
\item natural isomorphisms $\lambda_A : \unit \tensor A \to A$ (\emph{left
    unitor}), $\rho_A : A \tensor \unit \to A$ (\emph{right unitor}),  and
$\alpha_{A,B,C} : (A \tensor B) \tensor C \to A \tensor (B \tensor C)$
(\emph{associator}) subject to the following
coherence conditions:
\[
\xymatrix@C=-5mm@W=10mm{
 & & (A \tensor B) \tensor (C \tensor D) \ar[drr]^{\qquad \alpha_{A,B,C\tensor D}}                 \\
((A \tensor B) \tensor C) \tensor D \ar[urr]^{\alpha_{A \tensor B, C, D} \qquad} \ar[dr]_{\alpha_{A,B,C} \tensor \id_D \qquad}  && &&
A \tensor (B \tensor (C \tensor D)) \\
&  (A \tensor (B \tensor C)) \tensor D \ar[rr]_{\alpha_{A,B \tensor C,D}} &&  A \tensor ((B \tensor C) \tensor D) \ar[ur]_{\qquad \id_A \tensor \alpha_{B,C,D}}
}
\]

\[
\xymatrix{
(A \tensor \unit) \tensor B \ar[dr]_{\rho_A \tensor \id_B} \ar[rr]^{\alpha_{A,\unit,B}} & & A \tensor (\unit \tensor B) \ar[dl]^{\qquad \id_A \tensor \lambda_B} \\
& A \tensor B
}
\]
\end{itemize}

Such a monoidal product is called \emph{symmetric} if it comes with natural isomorphisms $\gamma_{A,B} : A \tensor B \to B \tensor A$ subject to the following coherences

\[
\xymatrix@C=15mm{
& A \tensor (B \tensor C) \ar[r]^{\gamma_{A,B \tensor C}} & (B \tensor C) \tensor A \ar@/^1pc/[dr]^{\alpha_{B,C,A}}
\\
(A \tensor B) \tensor C \ar@/^1pc/[ur]^{\alpha_{A,B,C}} \ar@/_1pc/[dr]_{\gamma_{A,B} \tensor \id_C \qquad} & & & B \tensor (C \tensor A) \\
& (B \tensor A) \tensor C \ar[r]^{\alpha_{B,A,C}} & B \tensor (A \tensor C) \ar@/_1pc/[ur]_{\qquad \id_B \tensor \gamma_{A,C}}
}
\]
\[
\xymatrix{
A \tensor B \ar[r]^{\gamma_{A,B}} & B \tensor A \ar[d]^{\gamma_{B,A}} \\
& A \tensor B \ar@{=}[ul]
}
\]
\end{defi}

In the sequel, we use the name \emph{(symmetric) monoidal category} for a category $\cC$ that
comes equipped with a (symmetric) monoidal structure $\tensor, \unit, \ldots$.
We write such structures $(\cC, \tensor, \unit)$ for short\footnote{Which is slightly abusive, as $\lambda, \rho, \alpha$ and $\gamma$ are also part of the
structure (and not uniquely determined from the triple $(\cC, \tensor, \unit)$).}.
Of course, if a category $\cC$ has products $\with$ and a terminal object $\top$, then $(\cC, \with, \top)$ is a symmetric monoidal category, and similarly for coproducts and intial objects.

We shall sometimes need to refer to morphisms between monoidal categories, which are
essentially functors together with natural transformations witnessing that the monoidal
structure is preserved.
\begin{defi}[{\cite[Section 5.1]{mellies09ps}}]
\label{def:monfunctor}
Let $(\cC, \tensor, \unit)$ and $(\cD, \widehat{\tensor}, \widehat{\unit})$ be two monoidal categories.
A \emph{lax monoidal functor} is given by a functor $F : \cC \to \cD$ together with natural transformations
\[m_0 : \widehat{\unit} \to F(\unit) \qquad \qquad \qquad m_{A,B} : F(A) \mathrel{\widehat{\tensor}} F(B) \to F(A \tensor B)\]
making the following diagrams commute.
\[
\xymatrix@C=30mm@R=10mm{
(F(A) \mathrel{\widehat{\tensor}} F(B)) \mathrel{\widehat{\tensor}} F(C)
\ar[r]^{\alpha_{F(A),F(B),F(C)}}
\ar[d]_{m_{A,B} \mathrel{\widehat{\tensor}} \id_{F(C)}}
&
F(A) \mathrel{\widehat{\tensor}} (F(B) \mathrel{\widehat{\tensor}} F(C))
\ar[d]^{\id_{F(A)} \mathrel{\widehat{\tensor}} m_{B,C}}
\\
F(A \tensor B) \mathrel{\widehat{\tensor}} F(C)
\ar[d]_{m_{A \tensor B, C}}
&
F(A) \mathrel{\widehat{\tensor}} F(B \tensor C)
\ar[d]^{m_{A, B \tensor C}}
\\
F((A \tensor B) \tensor C)
\ar[r]^{F(\alpha_{A,B,C})}
&
F(A \tensor (B \tensor C))
\\
}
\]

\[
\xymatrix@C=15mm{
F(A) \mathrel{\widehat{\tensor}} \widehat\unit
\ar[r]^{\rho_{F(A)}}
\ar[d]_{\id_{F(A)} \mathrel{\widehat{\tensor}} m_0}
& F(A) \\
F(A) \mathrel{\widehat{\tensor}} F(\unit)
\ar[r]^{m_{A,\unit}} &
F(A \tensor \unit) \ar[u]_{F(\rho_A)} \\
}
\qquad
\xymatrix@C=15mm{
\widehat\unit \mathrel{\widehat{\tensor}} F(A)
\ar[r]^{\lambda_{F(A)}}
\ar[d]_{m_0 \mathrel{\widehat{\tensor}} \id_{F(A)}}
& F(A) \\
F(\unit) \mathrel{\widehat{\tensor}} F(A)
\ar[r]^{m_{\unit,A}} &
F(\unit \tensor A) \ar[u]_{F(\lambda_A)} \\
}
\]
A lax monoidal functor is called \emph{strong monoidal} if the
natural transformations $m_0$ and $m_{A,B}$ are isomorphisms.
\end{defi}

Let us note that while every concrete
instance of monoidal functor in the paper, save for the ultimate example in
Appendix~\ref{sec:app-with}, is also going to be a \emph{symmetric} monoidal
functor (i.e., satisfy additional coherence diagrams involving $\gamma$),
we do not make use of that fact.

\subsubsection{Function spaces}

Our next definition concerns the categorical semantics of the linear function arrow
`$\lin$'. (Since we will only need a semantics for the \emph{purely linear}
fragment of the $\laml$-calculus, we will not discuss the non-linear arrow
`$\to$' here.)

\begin{defi}[{\cite[Sections~4.5 to~4.7]{mellies09ps}}]
\label{def:monoidalclosure}
Let $(\cC, \tensor , \unit)$ be a (symmetric) monoidal category and $A, B \in \Obj(\cC)$.
An \emph{internal homset} from $A$ to $B$ is an object $A \lin B \in \Obj(\cC)$
with a prescribed arrow $\ev_{A,B} : (A \lin B) \tensor A \to B$
(the \emph{evaluation map})
such that, for every other arrow $f : C \tensor A \to B$, there is a unique map
$\Lambda(f)$ (called the \emph{curryfication of $f$})
making the following diagram commute:
\[
\xymatrix@C=20mm{
(A \lin B) \tensor A \ar[r]^-{\ev_{A,B}} & B \\
C \tensor A \ar@{-->}[u]^{\Lambda(f) \tensor \id} \ar[ur]_f
}
\]
When there exists an internal homset for every pair objects in $\cC$, we say that
$(\cC, \tensor, \unit)$ is a \emph{(symmetric) monoidal closed} category.
\end{defi}

As for (co)products, internal homsets are determined up to unique isomorphism,
so we may talk somewhat loosely about \emph{the} internal homset later on. While
we work with the universal property given in
Definition~\ref{def:monoidalclosure} when the construction of an internal homset
involves a bit of combinatorics, we will also sometimes use the following
characterization.

\begin{prop}
  \label{prop:lin-hom-hom}
  The object $A \lin B$ is an internal homset for $A,B \in \Obj(\cC)$ if and
  only if there is a family of isomorphisms
\[\Hom{\cC}{C \tensor A}{\,B} \quad\cong\quad \Hom{\cC}{C}{\,A \lin B}\]
which is \emph{natural} in the parameter $C$ varying contravariantly over $\cC$
(in other words, if $\Hom{\cC}{- \tensor A}{\,B}$ and $\Hom{\cC}{-}{\,A \lin B}$
are naturally isomorphic as functors $\cC^{\op} \to \Set$).
\end{prop}
\begin{proof}
  This is an instance of \cite[Chapter~III, Section~2, Proposition~1]{CWM}.
\end{proof}

\subsubsection{Affineness and quasi-affineness}
\label{subsubsec:prelim-affine}

Given a monoidal product $\otimes$, morphisms from $A$ to $A \otimes A$
need not exist in general; this accounts for the linearity constraints in $\laml$.
But monoidal categories do not incorporate the ability of register transitions in SSTs to discard the content
of a register, a behavior more aligned with the \emph{affine} $\lambda$-calculi.
This notion thus plays a role in our development, so we discuss its incarnation in
categorical semantics.

\begin{defi}
\label{def:smc-affine}
A (symmetric) monoidal category $(\cC, \tensor, \unit)$ is called 
\emph{affine}\footnote{Such categories are also sometimes called
  \emph{semi-cartesian}~\cite{nlabsemicartesian}. We rather chose affine here for conciseness and because we will have to handle categories which have both cartesian products $\with$ and an additional affine monoidal product $\tensor$.} if $\unit$ is a terminal object of $\cC$.
\end{defi}

Most symmetric monoidal categories are not affine. However, there is a generic way
of building an affine monoidal category from a monoidal category.
Recall that if $\cC$ is a category and $X$ is an object of $\cC$, one may consider
the \emph{slice category} $\bigslant{\cC}{X}$
\begin{itemize}
\item whose objects are morphisms $A \to X$ ($A \in \Obj(\cC)$),
\item and such that $\Hom{{\footnotesize\bigslant{\cC}{X}}}{f : A \to X}{\,g : B \to X} = \{h \in
  \Hom{\cC}{A}{B} \mid g \circ h = f\}$.
\end{itemize}
If $\cC$ has a monoidal structure $(\tensor, \unit)$, this structure
can be lifted to $\bigslant{\cC}{\unit}$ by taking the identity $\unit \to \unit$ as the unit
and
\[
\left(
\xymatrix{
A \ar[r]^{f} &
\unit}
\right) 
\tensor
\left(
\xymatrix{
B \ar[r]^{g} & 
\unit}
\right) 
~~=~~
\left(
\xymatrix@C=15mm{
A \tensor B
\ar[r]^-{f \tensor g} &
\unit \tensor \unit \ar[r]^-{\lambda_\unit ~=~ \rho_\unit} &
\unit}\right)
\]
as the monoidal product. This gives rise to an affine monoidal structure over $\bigslant{\cC}{\unit}$,
and a strong monoidal structure for the forgetful functor $\dom : \bigslant{\cC}{\unit} \to \cC$.

In the converse direction, one can sometimes turn an object $A$ from $\cC$ into
one of $\bigslant{\cC}{\unit}$. This is the case when $A$ admits a cartesian
product with $\unit$, which may be written $A\with\unit$ (note
that if $\cC$ is affine, $A$ itself is such a cartesian product). We are then
led to consider the projection $\pi_2 : A\with\unit \to \unit$ as an object of
the slice category.
\begin{defi}
\label{def:quasi-aff}
A (symmetric) monoidal category $(\cC, \tensor, \unit)$ is called
\emph{quasi-affine}  if every $A \in \Obj(\cC)$ has a cartesian product $A \with
\unit$ with the monoidal unit.
\end{defi}
\begin{rem}
  \label{rem:quasi-aff-adjoint}
  We have a map $A \in \Obj(\cC) \mapsto \left(A\with\unit
    \xrightarrow{\;\pi_2\;} \unit\right) \in
  \Obj\left(\bigslant{\cC}{\unit}\right)$ in any
  quasi-affine category, according to the above discussion. It turns out that it
  extends to a functor $J$ which embeds $\cC$ into this affine slice category;
  moreover, $J$ is \emph{right adjoint} to the forgetful functor $\dom$. The
  interested reader may even check (although we will not make use of this) that
  the existence of a right adjoint to $\dom$ is \emph{equivalent} to
  quasi-affineness.
\end{rem}

\subsubsection{Monoids}

Since we are interested in string transductions, the free monoids $\Sigma^*$ are going to
make an appearance. Let us thus conclude this section by recalling the notion of monoid
\emph{internal to a monoidal category}.

\begin{defi}[{\cite[Section 6.1]{mellies09ps}}]
\label{def:monoid-object}
  Given a monoidal category $(\cC, \tensor, \unit)$, an
\emph{internal monoid} (or a \emph{monoid object}) is a triple $(M,\mu,\eta)$
where $M \in \Obj(\cC)$ and $\mu : M \tensor M \to M$, $\eta : \unit \to M$
are morphisms making the following unitality and associativity diagrams commute
\[
\xymatrix{
\unit \tensor M \ar[d]_{\eta \tensor \id} \ar[r]^{\lambda_\unit}
& M
& M \tensor \unit \ar[l]_{\rho_\unit} \ar[d]^{\id \tensor \eta}
\\
M \tensor M \ar[ur]_-\mu & & M \tensor M \ar[ul]^-\mu}
\xymatrix@C=15mm{
\ar[d]_{\mu \tensor \id} (M \tensor M) \tensor M \ar[r]^{\alpha_{M,M,M}} & M \tensor (M \tensor M) \ar[r]^-{\id \tensor \mu} & M \tensor M \ar[d]^\mu \\
M \tensor M \ar[rr]_{\mu} & & M
}
\]
\end{defi}

A useful example of this notion is the
\enquote{internalization} of the monoid of endomorphisms of $A$ when $A$ is part of a monoidal \emph{closed} category.

\begin{prop}
  \label{prop:internal-endo-monoid}
  Let $(\cC,\otimes,\unit)$ be a monoidal category. Any internal homset $A \lin
  A$ (with $A \in \Obj(C)$) that exists in $\cC$ has an internal monoid
  structure $(A \lin A,\; \eta,\; \mu)$ such that
  \[ \eta = \Lambda'(\id_A) \qquad
\quad
\mu \circ (\Lambda'(f) \tensor \Lambda'(g)) \circ \lambda_\unit =
\Lambda'(f \circ g)
\quad \text{for $f, g \in \Hom{\cC}{A}{A}$}
 \]
 where $\Lambda' : \Hom{\cC}{A}{A} \xrightarrow{\,\sim\,} \Hom{\cC}{\unit}{\,A
   \multimap A}$ is defined as $\Lambda' : h \mapsto \Lambda(h \circ \lambda_A)$
 from the curryfication $\Lambda$ and the left unitor $\lambda_A$.
\end{prop}
\begin{proof}[Proof sketch]
  One can define the monoid multiplication $\mu : (A \lin A) \otimes (A \lin A)
  \to (A \lin A)$ as the curryfication $\mu = \Lambda(\mathsf{app2})$ of the
  morphism $\mathsf{app2}$ built by composing the sequence
  \[ ((A \lin A) \otimes (A \lin A)) \otimes A
    \xrightarrow{\alpha} (A \lin A) \otimes ((A \lin A) \otimes A)
    \xrightarrow{\id\otimes\ev} (A \lin A) \otimes A \xrightarrow{\ev} A \]
  and check that it satisfies the coherence diagrams for internal monoids (that
  also involve the unit $\eta$ defined in the proposition statement) and the
  equation relating $\mu$ to $\Lambda'$.
\end{proof}

Let us conclude our categorical preliminaries on the following.

\begin{prop}
  \label{prop:monoid-withunit}
  Let $(\cC,\otimes,\unit)$ be a monoidal category and let
  $M\with\unit \in \Obj(C)$ be a cartesian product of some $M \in \Obj(C)$ with
  the monoidal unit $\unit$.
  Suppose that $(M,\mu,\eta)$ is a monoid object.
  Then $M\with\unit$ has an internal monoid structure defined by
  \[ \tuple{\mu\circ(\pi_1\otimes\pi_1),\;\lambda_\unit\circ(\pi_2\otimes\pi_2)}
    ~:~ (M\with\unit) \otimes (M\with\unit) \to M\with\unit \qquad
    \tuple{\eta,\id_\unit} ~:~ \unit \to M\with\unit \]
  where $\pi_1 : M\with\unit \to M$ and $\pi_2 : M\with\unit \to \unit$ are the
  projections and
  $\tuple{-,-}$ is the pairing given by the universal property of the cartesian
  product.

  Furthermore, this makes $\left( M\with\unit \xrightarrow{\;\pi_2\;}
    \unit \right) \in \Obj\left( \bigslant{\cC}{\unit} \right)$ into a monoid
  object of $\bigslant{\cC}{\unit}$.
\end{prop}
The routine verification of the required commutations of diagrams is left to the
reader.
\begin{rem}
  Our applications of this proposition will take place in \emph{quasi-affine}
  monoidal categories. For those, it admits a more conceptual proof: the right
  adjoint $J : \bigslant{\cC}{\unit} \to \cC$ to the forgetful functor $\dom$
  (cf.\ Remark~\ref{rem:quasi-aff-adjoint}) is \emph{lax monoidal}, and
  therefore so is $\dom \circ J$ which maps $A$ to $A\with\unit$ on objects;
  furthemore, the image of a monoid object by a lax monoidal functor is itself a
  monoid object in a canonical way~\cite[Section~6.2]{mellies09ps}.
\end{rem}

\section{Regular string functions in the $\laml$-calculus}
\label{sec:strings}

The goal of this section is to prove our main theorem pertaining to string functions.

\mainstring*

To prove Theorem~\ref{thm:main-string}, we introduce a generalized notion of SST parameterized
by a structure that we call a \emph{(string) streaming setting} $\mfC$, which is a structure
whose main component is a category to be thought of as the collection of possible
register transitions. From the point of view of expressiveness, $\mfC$ can be thought of
as a gadget delimiting a class of transition monoids which may be used for computations on top
of finite structure of a $\mfC$-SST.
The reason why we use categories as parameters is to be able to bridge easily
the usual notion of SST and the categorical semantics of $\laml$ in a single framework.

In Section~\ref{subsec:CSSTdef-string}, we define our streaming settings and $\mfC$-SSTs.
We make the connections with usual SSTs and $\laml$, in \S\ref{subsec:def-sr} and \S\ref{subsec:def-lamcat} respectively, through two distinguished streaming settings
$\SR$ and $\mfLam$. This allows to reframe Theorem~\ref{thm:main-string}
as the equivalence between \mbox{$\SR$-SSTs} and single-state $\mfLam$-SSTs.
Then, in Section~\ref{subsec:oplus-string}, we study the free coproduct completion
of categories $(-)_\oplus$, which readily extends to streaming settings.
In particular, properties of $\SR_\oplus$ are explored.
Section~\ref{subsec:with-string} deals with the dual construction $(-)_{\with}$,
the free product completion. A tight link between the expressiveness of $\mfC_\with$-SSTs
and \emph{non-deterministic} $\mfC$-SSTs is established.
Section~\ref{subsec:dial-string} then combines those results to study the composition
$((-)_{\with})_{\oplus}$ of those two completions (which we describe as a
direct construction $(-)_{\oplus\with}$), relying on the previous sections.
In particular, it is shown that the category at the center of $\SR_{\oplus\with}$
is a model of the purely linear fragment of $\laml$.
Finally, Section~\ref{subsec:main-string} briefly summarizes how to combine the results of
the previous sections into a proof of Theorem~\ref{thm:main-string}.

\subsection{A categorical framework for automata: streaming settings}
\label{subsec:CSSTdef-string}

We now introduce 
string streaming settings, which should be seen as a sort of memory framework for transducers
iterating performing a single left-to-right pass over a word.
This is the abstract notion that will allow us to generalize SSTs:

\begin{defi}
Let $X$ be a set.
A \emph{string streaming setting} with output $X$
is a tuple $\mathfrak{C} = (\cC, \initty, \retty, \curlyinterp{-})$ where
\begin{itemize}
\item $\cC$ is a category
\item $\initty$ and $\retty$ are arbitrary objects of $\cC$
\item $\curlyinterp{-}$ is a set-theoretic map $\Hom{\cC}{\initty}{\retty} \to X$
\end{itemize}
\end{defi}

Since the properties of the underlying category of a streaming setting will turn out
to be the most crucial thing in the sequel, we shall abusively apply adjectives
befitting categories to streaming settings, such as ``affine symmetric monoidal'' to
streaming settings in the sequel.

The notion of streaming setting is a convenient tool motivated by our
subsequent development rather than our primary object of study.
A closely related framework in which some of our abstract results can be formulated
is defined in~\cite{ColcombetPetrisan} (see Remark~\ref{rem:colpetri}). 

For the rest of this section, we will refer to string streaming setting simply as streaming
settings; we also fix two alphabets $\Sigma$ and $\Gamma$ for the rest of this section.

\begin{defi}
\label{def:c-sst}
Let $\mathfrak{C} = (\cC, \initty, \retty, \curlyinterp{-})$ be a streaming
setting with output $X$.
A $\mathfrak{C}$-SST with input alphabet $\Sigma$ and output $X$ is a tuple
$(Q, q_0, R, \delta, i, o)$ where
\begin{itemize}
\item $Q$ is a finite set of states and $q_0 \in Q$
\item $R$ is an object of $\cC$
\item $\delta$ is a function $\Sigma \times Q \to Q \times \Hom{\cC}{R}{R}$
\item $i \in \Hom{\cC}{\initty}{R}$ is an initialization morphism
\item $(o_q)_{q \in Q} \in \Hom{\cC}{R}{\retty}^Q$ is a family of output
  morphisms -- alternatively, we will sometimes consider it as a map $o : Q \to
  \Hom{\cC}{R}{\retty}$.
\end{itemize}
We write $\cT : \SST{\mathfrak{C}}{\Sigma}{X}$ to mean that $\cT$ is a
$\mathfrak{C}$-SST with input alphabet $\Sigma$ and output $X$ (the latter
depends only on $\mathfrak{C}$).

The corresponding function $\interp{\cT} : \Sigma^* \to X$
is then computed as for standard SSTs (cf.\ Definition~\ref{def:sst}):
an input word $w$ generates a sequence of states $q_0,\ldots,q_{\len{w}} \in Q$ and
a sequence of morphisms $f_i : R \to R$ in $\cC$,
and the output is then $\curlyinterp{o_{q_{\len{w}}} \circ f_{\len{w}} \circ \dots \circ f_1
  \circ i} \in X$.

An important class of $\mfC$-SSTs are those for which the set of states $Q$ is a
singleton, significantly simplifying the above data. They are called
\emph{single-state $\mfC$-SSTs}.
\end{defi}
\begin{rem}
  \label{rem:colpetri}
  Single-state $\mfC$-SSTs are very close to the $\cC$-automata over words
  defined by Colcombet and Petrişan~\cite[Section~3]{ColcombetPetrisan}, or more
  precisely $(\cC,\initty,\retty)$-automata with our notations. The main
  difference is that the latter's output would juste be an element of
  $\Hom{\cC}{\initty}{\retty}$: there is no post-processing $\curlyinterp{-}$ to
  produce an output.

  As for the addition of finite states, ultimately, it does not increase the
  framework's expressive power: we shall see in Remark~\ref{rem:sst-comparison}
  that $\mfC$-SSTs are equivalent to single-state SSTs over a modified category.
  We chose to incorporate states into our definition for convenience.
\end{rem}

\begin{exa}
  Let $\mathfrak{Set}_X = (\Set, \{\bullet\}, X,
  \curlyinterp{-})$ where $\curlyinterp{-}$ is the canonical isomorphism between
  $\Hom{\Set}{\{\bullet\}}{X} = X^{\{\bullet\}}$ and $X$. Then \emph{any} function
  $\Sigma^* \to X$ can be \enquote{computed} by a single-state
  $\mathfrak{Set}_X$-SST by taking $R = \Sigma^*$.
\end{exa}
\begin{exa}
  \label{exa:finset}
 Let $\mathfrak{Finset}_2 = (\Finset, \{\bullet\}, \{0,1\}, \curlyinterp{-})$ with $\curlyinterp{-}$
 the canonical isomorphism $\Hom{\Finset}{\{\bullet\}}{\{0,1\}} \cong \{0,1\}$. Single-state
 $\mathfrak{Finset}_2$-SST are essentially the usual notion of
  deterministic finite automata\footnote{Actually, \emph{complete} DFA, i.e.\
    DFA with total transition functions.}. Therefore,
 the functions they compute are none other than the indicator functions of
 regular languages.
\end{exa}
\begin{exa}
  Consider the category $\mathcal{POL}_\bQ$ whose objects are natural numbers, whose morphisms are tuples of multivariate polynomials
  over $\bQ$ with matching arities (so that $\Hom{\mathcal{POL}_\bQ}{n}{k} = (\bQ[X_1, \ldots, X_n])^k$) and where composition
  is lifted from the composition of polynomials in the usual way, making $\mathcal{POL}_\bQ$ into a category with (strict) cartesian products.
  Then, taking $\mathfrak{Pol}_\bQ = (\mathcal{POL}_\bQ, 0,1,
  \curlyinterp{-})$ where $\curlyinterp{-}$ is the isomorphism identifying $\bQ$
  and polynomials without variables ($n=0$),
  we can recover the definition of \emph{polynomial automata}
  from~\cite{PolynomialAutomata} as single-state $\mathfrak{Pol}_\bQ$-SSTs.
\end{exa}
\begin{exa}
  \label{ex:hines}
  The core of a paper by Hines~\cite{Hines} can be recast in our framework as
  saying that \emph{two-way automata} are the same thing as single-state SSTs
  over a well-chosen streaming setting, whose underlying category is called
  $\mathsf{Int(pSet)}$ in~\cite{Hines} (we will not give further details here).
  Recall that the transducer counterparts of these devices, namely \emph{two-way
    transducers}, are among the characterizations of regular string
  functions~\cite{EngelfrietHoogeboom}.

  Interestingly, this provides yet another connection with linear logic: this
  category $\mathsf{Int(pSet)}$ belongs to a family of techniques called
  \enquote{geometry of interaction} related to the game semantics of linear
  logic. As an example of recent application of such techniques to the interface
  between $\lambda$-calculus and automata, we refer to the work of Clairambault
  and Murawski~\cite{MAHORS} on higher-order recursion schemes.
  $\mathsf{Int(pSet)}$ is also close in spirit to the construction of free
  symmetric compact closed categories~\cite{kelly1980coherence}.
\end{exa}

Given two streaming settings $\mfC$ and $\mfD$ with a common output set $X$,
$\mfC$-SSTs are said to \emph{subsume} $\mfD$-SSTs if
for every $\mfD$-SST $\cT$ there is a $\mfC$-SST $\cT'$ with $\interp{\cT} = \interp{\cT'}$.
We say that $\mfC$-SSTs and $\mfD$-SSTs are \emph{equivalent} if both classes subsume one another.

There is a straightforward notion of morphism of streaming settings with common output.
\begin{defi}
Let $\mfC = (\cC, \initty_\mfC, \retty_\mfC, \curlyinterp{-}_\mfC)$
and
 $\mfD = (\cD, \initty_\mfD, \retty_\mfD, \curlyinterp{-}_\mfD)$
be streaming settings with the same output set $X$. A morphism of streaming settings is given by
a functor $F : \cC \to \cD$
and $\cD$-arrows $i : \initty_\mfD \to F(\initty_\mfC)$ and $o : F(\retty_\mfC) \to \retty_\mfD$
such that\[ \forall f \in \Hom{\cC}{\initty_\mfC}{\retty_\mfC},\, \curlyinterp{o \circ F(f) \circ i}_\mfD = \curlyinterp{f}_\mfC\]
\end{defi}
This notion is useful to compare the expressiveness of classes of generalized SSTs
because of the following lemma. 

\begin{lem}
\label{lem:morph}
If there is a morphism of streaming settings $\mfC \to \mfD$, then $\mfD$-SSTs
subsume $\mfC$-SSTs and single-state $\mfD$-SSTs subsume single-state
$\mfC$-SSTs.
\end{lem}
\begin{proof}[Proof sketch]
  Given a $\mfC$-SST $(Q, q_0, R, \delta, i, (o_q)_{q \in Q})$ (with the
  notations of Definition~\ref{def:c-sst}) and a morphism of streaming settings
  $(F : \cC \to \cD,\; i' : \initty_\mfD \to F(\initty_\mfC),\; o' :
  F(\retty_\mfC) \to \retty_\mfD)$, one builds a $\mfD$-SST that computes the
  same function as follows. The set of states and initial state are unchanged
  (so our proof applies both to the stateful and the single-state case). The
  memory object becomes $F(R)$, and the $\Hom{\cC}(R,R)$ component of the
  transition $\delta$ function is passed through the functor $F$ to yield a
  $\mfD$-morphism $F(R) \to F(R)$. The new initialization morphism is $F(i)
  \circ i'$ and the new output morphisms are $(o' \circ F(o_q))_{q \in Q}$.
\end{proof}

\begin{rem}
  \label{rem:morph-set}
  For any streaming setting $\mfC$, the functor $\Hom{\cC}{\initty}{-}$ is a
  morphism of streaming settings $\mfC \to \mathfrak{Set}$ with $i =
  \mathrm{id}$ and $o = \curlyinterp{-}_\mfC$.
\end{rem}

In the sequel, we will omit giving the morphisms $i : \top_\cD \to F(\top_\cC)$
and $o : F(\retty_\cC) \to \retty_\cD$ most of the time, as they will be isomorphisms deducible
from the context.
The one exception to this situation will be in Lemma~\ref{lem:string-cps}.

\subsection{The category $\Sr(\Gamma)$ of $\Gamma$-register transitions}
\label{subsec:def-sr}

We now show that usual copyless SSTs are indeed an instance of our general
notion of categorical SSTs. To do so we must arrange copyless register
transitions (Definition~\ref{def:register-transition}) into a category: given $t
\in [R \to_{\Sr(\Gamma)} S]$ and $t' \in [S \to_{\Sr(\Gamma)} T]$, we must be
able to compose them into $t' \circ t \in [R \to_{\Sr(\Gamma)} T]$. Moreover,
this composition should be compatible with the action of register transitions on
tuples of strings, i.e.\ the latter should be \emph{functorial}: $(t' \circ
t)^\dagger = t'^\dagger \circ t^\dagger$.

\begin{defi}[see e.g.~{\cite[Section~C]{SSTomega}}]
  \label{def:reg-trans-compo}
  Let $t \in [R \to_{\Sr(\Gamma)} S]$ and $t' \in [S
  \to_{\Sr(\Gamma)} T]$; recall that $t$ and $t'$ are defined as maps
  between sets $t : S \to (\Gamma + R)^*$ and $t' : T \to (\Gamma + S)^*$.

  We define the \emph{composition of register transitions} $t'
  \circ_{\Sr(\Gamma)} t : T \to (\Gamma + R)^*$ to be the set-theoretic
  composition $t^\ddagger \circ t'$ where $t^\ddagger : (\Gamma + S)^* \to
  (\Gamma + R)^*$ is the unique monoid morphism extending the copairing of
  $\inl$ and $t$ (i.e.\ $(\inl(c) \mapsto \inl(c),\, \inr(s) \mapsto
  t(s)) : \Gamma + S \to (\Gamma + R)^*$).
\end{defi}

\begin{prop}
  \label{prop:Register}
  There is a \emph{category} $\Sr(\Gamma)$ (given a finite alphabet
  $\Gamma$ which we will often omit in the notation) whose objects are finite
  sets of registers, whose morphisms are copyless register transitions --
  $\Hom{\Sr(\Gamma)}{R}{S} = [R \to_{\Sr(\Gamma)} S]$ -- and whose
  composition is given by the above definition. This means in particular that,
  with the above notations, $t' \circ t \in [R \to_{\Sr(\Gamma)} T]$,
  i.e.\ copylessness is preserved by composition. Furthermore:
  \begin{itemize}
  \item This category admits the empty set of registers as the terminal object:
    $\top = \varnothing$.
  \item The action of register transitions on tuples of strings gives rise to a
    functor $(-)^\dagger : \Sr \to \Set$, with $X^\dagger = (\Gamma^*)^X$ on
    objects.
  \end{itemize}
\end{prop}
The above proposition, which the interested reader may verify from the
definitions, is merely a restatement using categorical vocabulary of properties
that are already used in the literature on usual SSTs.

\begin{defi}
  We write $\SR(\Gamma)$ for the streaming setting
  $(\Sr(\Gamma), \top = \varnothing, \Bot = \{\bullet\}, \curlyinterp{-})$
  where $\curlyinterp{-} : [\varnothing \to_{\Sr(\Gamma)} \{\bullet\}] \to
  \Gamma^*$ is the canonical isomorphism $((\Gamma +
  \varnothing)^*)^{\{\bullet\}} \cong \Gamma^*$.
\end{defi}

\begin{fact}
\label{fact:register-sst-sst}
Standard copyless SSTs $\normalSST{\Sigma}{\Gamma}$ are the same thing as
$\SR$-SSTs $\Sigma^* \to \Gamma^*$.
\end{fact}

\begin{rem}
  The functor $\Hom{\cR}{\top}{-}$ mentioned in Remark~\ref{rem:morph-set} is,
  in the case of $\cR$, naturally isomorphic to $(-)^\dagger$. Therefore, the
  latter can be extended to a morphism $\SR \to \mathfrak{Set}$ of string
  streaming settings.
\end{rem}

\begin{prop}
  The category $\Sr$ can be endowed with a symmetric monoidal structure,
  where the monoidal product $R \tensor S$ is the disjoint union of register
  sets $R + S$ and the unit is the empty set of registers. Since the latter is
  also the terminal object of $\Sr$, this defines an \emph{affine
    symmetric monoidal category}.
\end{prop}

Note that given $t \in [R \to_{\Sr(\Gamma)} S]$ and $t' \in [T
\to_{\Sr(\Gamma)} U]$, there is only one sensible way to define a
set-theoretic map $t \otimes t' : U + S \to (\Gamma + (R + T))^*$. The above
proposition states, among other things, that $t \otimes t' \in [R + T
\to_{\Sr(\Gamma)} S + U]$. Checking this, as well as the requisite coherence
diagrams for monoidal categories, is left to the reader.

Next, let us observe that $\{\bullet\} \in \Obj(\Sr)$, representing a single register,
can be equipped with the structure of an internal monoid $(\{\bullet\}, \mu_\bullet, \eta_\bullet)$
by setting
\[\eta_\bullet(\bullet) = \varepsilon \qquad \text{and} \qquad
  \mu_\bullet(\bullet) = \inr(\mathtt{l})\inr(\mathtt{r})
  \quad\text{where}\ \mathtt{l} = \inl(\bullet)\ \text{and}\ \mathtt{r} =
  \inr(\bullet)
\]
so that $\mu_\bullet \in [\{\bullet\} \to_{\Sr(\Gamma)}
\{\bullet\} \otimes \{\bullet\}]$ has the codomain
 $\Gamma + (\{\bullet\} + \{\bullet\}) = \Gamma + \{\mathtt{l,r}\}$ when
 considered as a map between sets. This internal monoid is the key to giving
an inductive characterization of $\Sr$.
Given a
string $w = w_1 \ldots w_n \in \Gamma^*$, let us write $\widehat{w} \in
[\varnothing \to_{\Sr(\Gamma)} \{\bullet\}]$ for the register transition
defined by the map $\widehat{w} : \bullet \mapsto \inl(w_1)\dots\inl(w_n)
\in (\Gamma + \varnothing)^*$.

\begin{thm}
  \label{thm:functor-from-sr}
  Let $(\cC,\otimes,\unit)$ be an \emph{affine} symmetric monoidal category.
  
  For any internal monoid $(M,\mu,\eta)$ of $\cC$ and any family
  $(m_c)_{c\in\Gamma} \in \Hom{\cC}{\unit}{M}$ of morphisms, there exists a
  \emph{strong monoidal functor} $F : \Sr(\Gamma) \to \cC$ such that:
  \begin{itemize}
  \item $F(\varnothing) = \unit$, $F(\{\bullet\}) = M$ and $F(\widehat{c}) =
    m_c$ for every $c \in \Gamma$;
  \item $F(\mu_\bullet) = \mu$ and $F(\eta_\bullet) = \eta$ with the above
    definitions (implying that $F(\{\mathtt{l,r}\}) = M \otimes M$);
  \item the isomorphisms $\unit \to F(\varnothing)$ and $F(\{\bullet\}) \otimes
    F(\{\bullet\}) \to F(\{\bullet\} + \{\bullet\})$ that are part of the strong
    monoidal structure for $F$ are equal to $\id_\unit$ and $\id_{M \otimes M}$
    respectively.
  \end{itemize}
\end{thm}
Note that since $F$ is a monoidal functor, it transports the monoid object
$(\{\bullet\},\mu_\bullet,\eta_\bullet)$ to a structure of internal monoid over
$F(\{\bullet\}) = M$ in a canonical way~\cite[Section~6.2]{mellies09ps}. \mbox{A
  fact} that encapsulates the idea of the second item above -- but which,
strictly speaking, also depends on the third one -- is that the result of this
transport is precisely $(M,\mu,\eta)$.
\begin{proof}
  Although the intuition of the proof is simple, its execution involves a
  significant amount of bureaucracy; in particular, it manipulates canonical
  isomorphisms given by Mac Lane's coherence theorem for symmetric monoidal
  categories~\cite[Section~XI.1]{CWM}. For this reason, we will only illustrate
  the idea here in the concrete case of the cartesian category of sets;
  \emph{the full proof can be found in \Cref{sec:coherence-sr}}.

  For $(\cC, \otimes, \unit) = (\Set, \times, \{*\})$, we can reformulate the
  data given in the statement as a monoid $M$ with a family of elements
  $(m_c)_{c \in \Gamma} \in M^\Gamma$, that can be identified with functions
  $m_c : \{*\} \to M$ for $c \in \Gamma$. The functor that we build then maps an
  object $R$ -- a finite set of registers -- to the set $M^R$. The action on
  morphisms is best illustrated through an example: $(t : z \mapsto axby) \in
  [\{x,y\} \to_{\Sr(\{a,b\})} \{z\}]$ (where $\inl/\inr$ are omitted) becomes
  the map $(u_x,u_y) \in M^{\{x,y\}} \mapsto m_a u_x m_b u_y \in M \cong
  M^{\{z\}}$. Note that when we apply this construction to $M = \Gamma^*$ with
  $m_c = c$, we recover the functor $(-)^\dagger : \Sr(\Gamma) \to \Set$ from
  Proposition~\ref{prop:Register}. 
\end{proof}

\begin{rem}
Informally speaking, we think of $\Sr(\Gamma)$ as the affine
symmetric monoidal category freely generated by an internalization of the free
monoid $\Gamma^*$.
To truly express this, one would need to add to 
Theorem~\ref{thm:functor-from-sr} the uniqueness of $F$ up to natural
isomorphism. This would imply, among other things, that
the morphisms of $\Sr$ are inductively given by
\begin{itemize}
\item the identities $\id$
\item the compositions
\item the structural morphisms associated to the tensor product
\item the unique morphism $\{\bullet\} \to \varnothing$
\item canonical morphisms $\widehat{c} : \top \to \{\bullet\}$ for every individual letter $c \in \Gamma$
\item a canonical morphism $\eta : \top \to \{\bullet\}$ corresponding to the empty word
\item a multiplication morphism $\mu : \{\bullet\} \tensor \{\bullet\} \to
  \{\bullet\}$ corresponding to string concatenation.
\end{itemize}
However, we do not prove this inductive presentation, nor the uniqueness
property, since they are not necessary for our purposes.
\end{rem}

Since the monoid object structure of internal homsets
(Proposition~\ref{prop:internal-endo-monoid}) has a somewhat explicit
description, Theorem~\ref{thm:functor-from-sr} admits a specialized and
simplified formulation for affine symmetric monoidal closed categories.
For our purposes, it will be useful to give a version that also applies in the
quasi-affine case.
\begin{cor}
  \label{cor:sr-to-smcc}
  Let $(\cC,\otimes,\unit)$ be a \emph{quasi-affine} symmetric monoidal \emph{closed} category and $A$ be an
  object in $\cC$.
  For any family $(f_c)_{c \in \Gamma} \in
  \Hom{\cC}{A}{A}^\Gamma$ of endomorphisms, there exists a strong monoidal
  functor $F : \Sr(\Gamma) \to \cC$ such that
  \[ F(\varnothing) = \unit \qquad F(\{\bullet\}) = (A \lin A) \with \unit
    \qquad \forall w \in \Gamma^*,\; F(\widehat{w}) =
    \tuple{\Lambda'\left(f_{w[1]} \circ \dots \circ f_{w[n]}\right),\,\id_\unit} \]
  where we use the notations $\widehat{(-)}$ from
  Theorem~\ref{thm:functor-from-sr} and $\Lambda'(-)$ from
  Proposition~\ref{prop:internal-endo-monoid}.
\end{cor}
Our first application of this result will be to show in the next subsection that
all regular functions are $\laml$-definable. For this purpose, the level of
abstraction that we are working with is unnecessary: it would suffice to encode
copyless SSTs as $\laml$-terms, a programming exercise that is not particularly
difficult. But later, the generalized preservation theorem of
\Cref{subsec:composition} will apply Corollary~\ref{cor:sr-to-smcc} in its full
generality.
\begin{proof}
  Proposition~\ref{prop:internal-endo-monoid} gives a canonical internal monoid
  structure to $A \lin A$, which by Proposition~\ref{prop:monoid-withunit} can
  be lifted to a monoid object $((A \lin A) \with\unit,\, \mu,\,\eta)$. For $f
  \in \Hom{\cC}{A}{A}$, let $\Lambda''(f) = \tuple{\Lambda'(f),\id_\unit} :
  \unit \to (A \lin A) \with\unit$.

  We apply Theorem~\ref{thm:functor-from-sr} to the slice category
  $\bigslant{\cC}{\unit}$ -- which satisfies the affineness assumption -- with
  the internal monoid $\pi_2 : (A \lin A) \with\unit \to \unit$ (cf.\
  Proposition~\ref{prop:monoid-withunit} again) and the family
  $(\Lambda''(f_c))_{c\in\Gamma}$ (each $\Lambda''(f)$ is a morphism in the
  slice category from its unit $\id_\unit$ to this $\pi_2$ since $\pi_2
  \circ \Lambda''(f) = \id_\unit$).
  We compose the resulting functor $\Sr \to \bigslant{\cC}{\unit}$ with the
  forgetful functor $\dom : \bigslant{\cC}{\unit} \to \cC$ to get $F : \Sr \to
  \cC$ such that $F(\varnothing) = \unit$ and $F(\{\bullet\}) = (A \lin A)
  \with\unit$. As a composition of strong monoidal functors, $F$ is also strong
  monoidal. This takes care of all but one of the corollary's conclusions.

  For the remaining one, we need to do some preliminary work. First, let us
  recall two commuting diagrams below. The left one comes from the naturality of
  the family of (iso)morphisms $m_{R,S} : F(R) \otimes F(S) \to F(R + S)$ that
  make $F$ a (strong) monoidal functor, while the right one is among the
  coherence conditions in Definition~\ref{def:monfunctor}.
  \[
  \xymatrix@C=15mm{
    \unit \otimes \unit
  \ar[r]^{m_{\varnothing,\varnothing}}
  \ar[d]_{F(\widehat{u})\otimes F(\widehat{v})}
  & F(\varnothing+\varnothing) \ar[d]_{F(\widehat{u} \otimes \widehat{v})} \\
  F(\{\bullet\}) \otimes F(\{\bullet\})
  \ar[r]^{m_{\{\bullet\},\{\bullet\}}} &
  F(\{\bullet\} + \{\bullet\}) \\
  }
  \qquad
  \xymatrix@C=15mm{
  \unit \tensor F(\varnothing)
  \ar[r]^{\lambda_{F(\varnothing)}}
  \ar[d]_{m_0 \tensor \id_{F(\varnothing)}}
  & F(\varnothing) \\
  F(\varnothing) \tensor F(\varnothing)
  \ar[r]^{m_{\varnothing,\varnothing}} &
  F(\varnothing + \varnothing) \ar[u]_{F(\lambda_\varnothing)} \\
  }
  \]
  Here, $m_0 : \unit \to F(\varnothing)$ is also part of the strong monoidal
  structure of $F$ and $u,v \in \Gamma^*$. The construction of
  Theorem~\ref{thm:functor-from-sr} gives us $m_0 = \id_\unit$ and
  $m_{\{\bullet\},\{\bullet\}} = \id_{F(\{\bullet\}+\{\bullet\})}$; furthermore,
  we have $\varnothing+\varnothing=\varnothing$ and $\lambda_\varnothing =
  \id_\varnothing$ in $\Sr$. Thus, in the end, we can combine the two equalities
  expressed by the above diagrams and simplify them to get
  \[ \forall u,v \in \Gamma^*,\quad F(\widehat{u}) \otimes F(\widehat{v}) =
    F(\widehat{u} \otimes \widehat{v}) \circ \lambda_\unit \]
  Theorem~\ref{thm:functor-from-sr} also guarantees that $F(\mu_\bullet) = \mu$,
  $F(\eta_\bullet) = \eta$ and $F(\widehat{c}) = \Lambda''(f_c)$ for all
  $c\in\Gamma$. At the same time, by combining
  Propositions~\ref{prop:internal-endo-monoid} and~\ref{prop:monoid-withunit},
  one can derive
  \[ \eta = \Lambda''(\id_A) \qquad
    \quad
    \mu \circ (\Lambda''(f) \tensor \Lambda''(g)) \circ \lambda_\unit =
    \Lambda''(f \circ g)
    \quad \text{for $f, g \in \Hom{\cC}{A}{A}$}
  \]
  We use all the above in a proof by induction of the desired conclusion.
  \begin{itemize}
  \item The base case is $F(\widehat{\varepsilon}) = F(\eta_\bullet) = \eta =
    \Lambda''(\id_A)$ (indeed, $\widehat{\varepsilon} = \eta_\bullet = (\bullet
    \mapsto \varepsilon)$ by definition).
  \item For the inductive case, write any word of length $n+1$ in $\Gamma^*$ as
    $wc$ for $w \in \Gamma^*$ with $\len{w}=n$ and $c \in \Gamma$, and suppose
    that by induction hypothesis $F(\widehat{w}) = \Lambda''\left(w\right)$. A
    direct computation of register transitions suffices to check that
    $\widehat{wc} = \mu_\bullet \circ (\widehat{w}\otimes\widehat{c})$. Then
    \begin{align*}
      F(\widehat{wc}) &= F(\mu_\bullet) \circ F(\widehat{w} \otimes \widehat{c}) =
                        \mu \circ \left(F\left(\widehat{w}\right) \otimes
                        F\left(\widehat{c}\right)\right) \circ \lambda_\unit\\
      &= \mu \circ\left(
        \Lambda''\left(f_{w[n]} \circ \dots \circ f_{w[1]}\right)
        \otimes \Lambda''\left(f_c\right) \right)   \circ \lambda_\unit
      = \Lambda''\left(f_{w[1]} \circ \dots \circ
        f_{w[n]} \circ f_c\right) \qedhere
    \end{align*}
  \end{itemize}
\end{proof}

\subsection{The syntactic category $\Lamcat$ of purely linear $\laml$-terms}
\label{subsec:def-lamcat}

Now we relate our notion of generalized SSTs to the $\laml$-calculus.
If $\Gamma = \{b_1, \ldots, b_n\}$ is an alphabet, call $\widetilde{\Gamma}$ the
non-linear typing context
\[ \widetilde\Gamma = (b_1 : \basety \lin \basety,\, \ldots,\, b_n : \basety \lin \basety,\, \varepsilon : \basety)\]
\begin{defi}
\label{def:syntacticcat}
We call $\Lamcat(\widetilde{\Gamma})$ (or just $\Lamcat$ when $\Gamma$ is clear from the context) the category
\begin{itemize}
\item whose objects are purely linear $\laml$ types;
\item whose morphisms from $\tau$ to $\sigma$ are terms $t$ such that
$\widetilde{\Gamma}; \; \cdot \vdash t : \tau \lin \sigma$, considered \emph{up to
$\lamlequiv$-equivalence};
\item whose identity is given by $\lam x. x$ and composition of $f$ and $g$ by $\lam x.\; f \; (g \; x)$.
\end{itemize}
\end{defi}

\begin{rem}
\label{rem:betaeta-conservativity}
In the definition of morphisms, represented by $\laml$ terms, we
only make a restriction on the types of the $\laml$-terms.
Because $\laml$ is normalizing (Theorem~\ref{thm:laml-normalization}), we could have further
assumed the terms to be normal and thus, to only contain subterms whose types are also purely linear.
Therefore, it makes sense to say that $\Lamcat(\tilde{\Gamma})$ is about the \emph{purely linear fragment of $\laml$}
augmented with (inert) constants from $\tilde{\Gamma}$.

Similarly, for the equivalence relation we use to actually define the homsets,
we use the full $\lamlequiv$-equivalence, which could, on the face of it, require to
go oustide of the purely linear fragment of $\laml$ to establish certain equalities (for instance,
consider the rather artifical derivation $\lam x. x =_\lamlequiv \lam x. (\lam^\oc y. y) \; x =_\lamlequiv \lam x.x$).
This is also unnecessary as it can be checked that the normalization argument relies 
on a reduction relation $\to_{\beta\extr}$ which is confluent up to commutative conversions $\cocoeq$ (Theorem~\ref{thm:CR}); since both $\to_{\beta\extr}$ and $\cocoeq$ preserve the purely linear fragment, this is enough to conclude
that we could ignore non-purely linear terms when defining $=_{\beta\eta}$ for the purely linear fragment.
\end{rem} 

$\Lamcat$ is a monoidal closed category with products and coproducts,
which captures the expressiveness of purely linear $\laml$-terms enriched with
constants for the ``empty word'' and prepending letters of $\Gamma$ to
the left of a ``word'' when regarding the type $\basety$ being regarded as the type of such words.
This leads to the expected notion of streaming setting.

\begin{defi}
\label{def:syntacticstreaming}
$\mfLam$ is the streaming setting $(\Lamcat, \unit, \basety, \curlyinterp{-}_{\mfLam})$
with output $\Gamma^*$ such that $\curlyinterp{t}_\mfLam = w$ if and only if\footnote{Recall that this defines
a total function because of the bijection between Church encodings and normal forms; see
Proposition~\ref{prop:laml-tree-nf}.}
$\lam^\oc b_1. \ldots \lam^\oc b_n.\; \lam^\oc \varepsilon.\; t$ is $\lamlequiv$-equivalent to the Church encoding of $w$ (this defines a total function because of Proposition~\ref{prop:laml-tree-nf}).
\end{defi}

Our interest in $\mfLam$ lies in the following lemma. (Recall that a $\mfC$-SST
is said to be \emph{single-state} if its set of states is a singleton.)
\begin{lem}
\label{lem:lamlsst}
A function $\Sigma^* \to \Gamma^*$ is computable by a single-state $\mfLam$-SSTs
if and only if it is $\laml$-definable in the sense of
Definition~\ref{def:laml-definable}.
\end{lem}

To prove one direction of this equivalence, we need a technical lemma on
$\laml$-terms defining string functions. In order to state the lemma in the more
general case of tree functions, so that it may be reused in \Cref{sec:trees}, we
extend the notation $\widetilde{\Gamma}$ above as follows: given a ranked
alphabet ${\bf \Sigma} = \{a_1 : A_1,\; \ldots,\; a_n : A_n\}$ (recall the
notation from \S\ref{subsec:notations}), let
  \[\widetilde{\bf \Sigma} = (a_1 : \basety \lin \ldots \lin \basety,\; \ldots,\; a_n : \basety \lin \ldots \lin \basety)\]
where the type of $a_i$ has $\card{A_i}$ arguments (thus, it
contains $\card{A_i}+1$ occurrences of $\basety$).

\begin{lem}
\label{lem:laml-niceshape}
Let ${\bf \Sigma} = \{a_1 : A_1, \ldots, a_n : A_n\}$
and ${\bf \Gamma} = \{ b_1 : B_1, \ldots, b_k : B_k\}$ be ranked alphabets such that there is some $A_i = \varnothing$ (i.e., $\Tree({\bf \Sigma}) \neq \varnothing$).
Up to $\lamlequiv$-equivalence, every term of type $\Treety_{\bf \Sigma}[\kappa] \lin \Treety_{\bf \Gamma}$ is of the shape
\[ \lam s. \lam^\oc b_1. \ldots \lam^\oc b_k. \; o \; (s \; d_1 \; \ldots \; d_n)\]
such that $o$ and the $d_i$ are purely linear $\laml$-terms with no occurrence of $s$. In particular:
\[\widetilde{\bf \Gamma} ; \cdot \vdash o : \kappa \lin \basety \qquad \qquad \widetilde{\bf\Gamma} ; \cdot \vdash d_i : \kappa \lin \ldots \lin \kappa\]
(with the type of $d_i$ having $\card{B_i}$ arguments).
\end{lem}
We expend considerable effort in proving this lemma; some of it is spent on
routine yet cumbersome bureaucracy, and some on actual technical subtleties. But
since the obstacles are unrelated to the various conceptual points concerning
automata and semantics that we wished to stress, we relegate the proof to
\Cref{sec:laml-niceshape}.

The idea is to analyze the shape of $\laml$-terms in \emph{normal form}. For
this purpose, the naive notion of $\beta$-normal form, that is, the
non-existence of $\beta$-reductions from a term, is inadequate because of the
\emph{positive} connectives $\otimes/\oplus$. We thus start by defining a better
notion of normal form, that can be reached by combining $\beta$-reductions with
applications of $\eta$-conversions to unlock new $\beta$-redexes. This is
done in \Cref{sec:laml-normalization}, where we prove a normalization
theorem. \Cref{sec:laml-niceshape} then provides a lengthy case analysis of
normal forms inhabiting the type $\Treety_{\bf \Sigma}[\kappa] \lin \Treety_{\bf
  \Gamma}$ to establish Lemma~\ref{lem:laml-niceshape}.

Similar issues arise in the literature concerned with deciding
$\beta\eta$-convertibility in $\lambda$-calculi with positive connectives
(see~\cite{gasche-thesis} for a comprehensive and pedagogical overview of this
subject). In our case, we are interested in normal forms only because they
lend themselves to syntactic analysis.

We can now use this to establish the equivalence between $\laml$-definability
and $\mfLam$-SSTs.
\begin{proof}[Proof of Lemma~\ref{lem:lamlsst}]
Before beginning the proof, it should be noted that SSTs process strings from left
to right while Church encodings work rather from right to left. This is not a big issue
in the presence of higher-order functions.

\vspace{0.5em}

$\left(\text{$\mfLam$-SST}~ \subseteq~ \laml\right) ~~~~$ Given a $\mfLam$-SST $\cT = (\{*\},*,\tau,\delta,i,o)$, $\delta$ may be regarded
as family of $\laml$ terms $(t_a)_{a \in \Sigma}$ (with free variables in $\widetilde{\Gamma}$).
Suppose that $\Sigma = \{a_1, \ldots, a_k\}$ and
recall that Example~\ref{ex:laml-rev} provides a $\laml$-term $\mathtt{rev} : \Str_\Sigma[\basety \lin \basety] \lin \Str_\Sigma$ implementing the reversal of its input string.
$\interp{\cT}$ is implemented by the following $\laml$-term of type
$\Str_{\Sigma}[\tau \lin \tau] \lin \Str_{\Gamma}$:
\[\lambda s.\; \lam^\oc b_1.\; \ldots\; \lam^\oc b_n.\; \lam^\oc \varepsilon.\; o \; (\mathtt{rev} \; s \; t_{a_1} \ldots t_{a_k} \; (i \; ()))\]

$\left(\laml ~\subseteq~ \text{$\mfLam$-SST}\right) ~~~~$
Given a term of type $\Str_\Sigma[\tau] \lin \Str_\Gamma$, by Lemma~\ref{lem:laml-niceshape}, it is $\lamlequiv$-equivalent to
\[\lambda s.\; \lam^\oc b_1.\; \ldots\; \lam^\oc b_n.\; \lam^\oc \varepsilon.\; t \; (s \; u_1 \; \ldots \; u_k \; v)\]
where $t$, $v$ and the $u_i$ are some terms typable in $\widetilde{\Gamma}$. The underlying string function is 
computed by the $\mfLam$-SST
\[\cT = (\{*\},*,\tau \lin \tau,\delta,\lam x. x,\lam f. o \; (f \; i))\]
where $\delta(a_i,*) = (*, \lam g.\lam x. a_i\; (g \; x))$.
\end{proof}

Lemma~\ref{lem:lamlsst} therefore enables us to reframe Theorem~\ref{thm:main-string} as statement comparing
the expressiveness of single-state $\mfLam$-SSTs and $\SR$-SSTs. This motivates our abstract development
focused on comparing the expressiveness of various $\mfC$-SSTs.

Toward this goal, we shall construct morphisms of streaming settings from and to
$\mfLam$. One of them is straightforward using our previous technical development.

\begin{lem}
\label{lem:register-to-laml}
There is a morphism of streaming settings $\SR \to \mfLam$.
\end{lem}
We build this morphism using the generic construction of
Corollary~\ref{cor:sr-to-smcc}. Again, this is not strictly necessary; see
Lemma~\ref{lem:register-to-laml-tree} for a sketch of a \enquote{manual}
definition of such a morphism in the case of trees. But
Corollary~\ref{cor:sr-to-smcc} is needed for other purposes anyway
(\Cref{subsec:composition}) and gives us the correctness proof almost
\enquote{for free}.
\begin{proof}
  We first give the construction of the underlying functor. First, note that
  $\Lamcat$ has all cartesian products and internal homsets, given by the
  syntactic connectives on types with the same notations. Thus, it is
  symmetric monoidal closed and quasi-affine; we can therefore apply
  Corollary~\ref{cor:sr-to-smcc} to the base type $\basety$,
  regarded as an object of $\Lamcat$, and to
  the family of endomorphisms $(\widetilde{\Gamma}; \; \cdot \vdash c : \basety \lin
  \basety)_{c\in\Gamma} \in \Hom{\Lamcat}{\basety}{\basety}^\Gamma$. (Indeed, by
  definition, every letter $c \in \Gamma$ serves as a variable $c$ given the
  type $\basety \lin \basety$ in the typing context $\widetilde{\Gamma}$.) This
  gives us $F : \Sr \to \Lamcat$ such that $F(\varnothing) = \unit$,
  $F(\{\bullet\}) = (\basety\lin\basety)\with\unit$ and
  \[ \forall w \in \Gamma^*,\quad F(\widehat{w}) \;=_{\beta\eta}\; \lambda z.\;
    \tlet{()}{z}{\tuple{(\lambda x.\, w_1\,(w_2\,\ldots (w_n\, x) \ldots)),\;()}} \]
  Since $\initty_\SR = \varnothing$ and $\initty_\mfLam = \unit$, we can simply
  take $i = \id_\unit : \initty_\mfLam \to F(\initty_\SR)$ as part of our morphism
  of streaming settings. The map $o : F(\retty_\SR) \to \retty_\mfLam$ is more
  interesting. Since $\retty_\SR = \{\bullet\}$ and $\retty_\mfLam = \basety$,
  we can see that $o$ must be a $\laml$-term of type
  $((\basety\lin\basety)\with\unit) \lin \basety$ in the context
  $\widetilde{\Gamma}$. The choice that works is $o = \lam p. \; \pi_1(p) \;
  \varepsilon$ (recalling that $(\varepsilon : \basety) \in \widetilde{\Gamma}$
  stands for the empty string). Proving that $\curlyinterp{o \circ
    F(\widehat{w}) \circ i}_\mfLam = \curlyinterp{\widehat{w}}_\SR$ for any $w
  \in \Gamma^*$ is merely a matter of performing $\beta\eta$-conversions
  starting from the above equation on $F(\widehat{w})$; we leave this to the
  reader. Since every $f \in \Hom{\Sr}{\initty}{\retty}$ is of the form $f =
  \widehat{w}$ for some $w \in \Gamma^*$, this suffices to show that $(F,i,o)$
  fits the definition of a morphism of streaming settings.
\end{proof}

However, note that this does not alone tells us that $\SR$-SSTs are subsumed by $\laml$, as
Lemma~\ref{lem:lamlsst} only allows us to use single-state $\Lamcat$-SSTs. To circumvent this, we will
later prove that a streaming setting with coproducts allows to simulate
state with single-state SSTs, thus completing a proof of the easier direction of Theorem~\ref{thm:main-string}.

In order to build morphisms from $\mfLam$ to other streaming settings,
we shall make use of the following:
\begin{lem}
\label{lem:laml-initial}
Let $\mfC$ be a streaming setting $(\cC,\unit, \retty,\curlyinterp{-})$ whose
underlying category $\cC$ is symmetric monoidal closed with monoidal unit $\unit$.
Suppose that $\cC$ additionally has chosen finite products and coproducts.
Let $(f_b)_{b \in \Gamma}$ be a family of morphisms $\Hom{\cC}{\retty}{\retty}^\Gamma$
and $e \in \Hom{\cC}{\initty}{\retty}$ a distinguished morphism such that $\curlyinterp{e}$ is the empty word and, for every
$g \in \Hom{\cC}{\initty}{\retty}$, we have $\curlyinterp{f_b \circ g} = b\curlyinterp{g}$ (that is, $f_b$ acts by concatenating the single-letter word $b$ on the left).

Then there is a canonical morphism $\mfLam \to \mfC$ of streaming settings. Moreover the underlying functor
is strong monoidal for $\oplus$, $\with$, $\tensor$ and preserves $\lin$.
\end{lem}

Without spelling out the details, this Lemma essentially states that $\Lamcat$
is \emph{initial} among symmetric monoidal closed categories with products and
coproducts.
We do not offer a proof of this statement, which we consider folklore. The interested
reader may find a similar development in \cite[Chapter~4]{Bierman} for the
case of $\Lamcat(\varnothing)$ (i.e., where the $\laml$-terms have no free variables).
Let us note that, because of the specific way we defined $\Lamcat$,
Remark~\ref{rem:betaeta-conservativity}, showing the conservativity of the
congruence $=_{\beta\eta}$ of $\laml$ over the purely
linear fragment should be the first step in this proof.
This is because the notion of symmetric monoidal closed category
does not require the existence of an exponential modality $\oc$, and thus,
all the equations in the initial symmetic monoidal closed category
with products and coproducts should only
satisfy those equations mentioning those constructs.

\subsection{The free coproduct completion (or finite states)}
\label{subsec:oplus-string}

\subsubsection{Definition and basic properties} We give here an elementary definition of ``the'' free finite coproduct completion
of categories $\cC$ and some of its basic properties. The construction consists essentially in considering finite families of objects of $\cC$ as ``formal coproducts''
(equivalently, one could use finite lists as in~\cite[Definition~3]{Galal20}).

\begin{defi}
Let $\cC$ be a category. The \emph{free finite coproduct completion} $\cC_\oplus$ is defined as follows:
\begin{itemize}
\item An object of $\cC_\oplus$ is a pair $(U,(C_u)_{u \in U})$ consisting of a finite set
$U$ and a family of objects of $\cC$ over $U$. We write those as formal sums $\bigoplus_{u \in U} C_u$
in the sequel.
\item A morphism from
$\bigoplus_{u \in U} C_u \to \bigoplus_{v \in V} C_v$ is a $U$-indexed family of pairs $(v_u,g_u)_{u \in U}$
with $v_u \in V$ and $g_u : C_u \to C_{v_u}$ in $\cC$. In short,
\[\varHom{\cC_\oplus}{\bigoplus_{u \in U} C_u}{\bigoplus_{v \in V} C_v} \; = \; \prod_{u \in U} \sum_{v \in V} \Hom{\cC}{C_u}{C_{v}}\]
\item The identity at object $\bigoplus_{u \in U} C_u$ is the family $(u, \id_{C_u})_{u \in U}$.
Given two composable maps
\[(w_v, h_v)_{v \in V} : \bigoplus_{v \in V} C_v \to \bigoplus_{w \in W} C_w \qquad
\qquad \text{and} \qquad \qquad
(v_u, g_u)_{u \in U} : \bigoplus_{u \in U} C_u \to \bigoplus_{v \in V} C_v\] 
the composite is defined to be the family
\[(w_{v_u}, h_{v_u} \circ g_u) : \bigoplus_{u \in U} C_u \to \bigoplus_{w \in W} C_w\]
\end{itemize}
\end{defi}

There is a full and faithful functor $\iota_\oplus : \cC \to \cC_\oplus$ taking an object $C \in \Obj(\cC)$
to the one-element family $\bigoplus_1 C \in \Obj(\cC_\oplus)$. Objects lying in the image of this
functor will be called \emph{basic} objects of $\cC_\oplus$.
The formal sum notation reflects that families $\bigoplus_{u \in U} C_u$
should really be understood as coproducts of those basic objects $C_u$.
More generally, it is straightforward to check that, for any finite set $I$ and family
$\bigoplus_{u \in U_i} C_u$ over $i \in I$, canonical coproducts in $\cC_\oplus$ can be computed as follows
\[\bigoplus_{i \in I} \bigoplus_{u \in U_i} C_{i,u} \quad = \quad \bigoplus_{(i,u) \in \sum_{i\in I} U_i} C_{i,u}\]

As advertised, this is a free finite coproduct completion in the following sense: for any functor $F : \cC \to \cD$
to a category $\cD$ with finite coproducts,
there is an extension $\widetilde{F} : \cC_\oplus \to \cD$ preserving finite coproducts making the following diagram commute:
\[
\xymatrix{
\cC \ar[d]_{\iota_\oplus} \ar[r]^F & \cD
\\
\cC_\oplus \ar@{-->}[ur]_{\widetilde{F}}
}
\]
and it is unique up to unique natural isomorphism under those conditions.

Finally, suppose that we are a monoidal structure on $\cC$. Then, it is possible to
extend it to a monoidal structure over $\cC_\oplus$ in a rather canonical way: we require that $\tensor$ distributes
over $\oplus$, i.e., that $A \tensor (B \oplus C) \cong (A \tensor B) \oplus (A \tensor C)$.
Formally speaking, we set\footnote{For readers more familiar with the free cocompletion $\Set^{\cC^\op}$
of $\cC$, note that the coproduct-preserving functor $E$ determined by
\[
\xymatrix@R=4mm{
& \cC \ar[dl]_{\iota_\oplus} \ar[dr]^{{\bf y}} \\
\cC_\oplus \ar@{-->}[rr]_E & &  \Set^{\cC^\op}}
\]
is full and faithful, as well as strong monoidal when $\Set^{\cC^\op}$ is
equipped with the \emph{Day convolution} as monoidal product.
The latter is computed as the following coend using a monoidal product $\tensor$ in $\cC$
\[(P \tensor Q)(U) = \int^{V,W} \Hom{\cC}{U}{V \tensor W} \times P(V) \times Q(W)\]}

\[ \left(\bigoplus_{u \in U} C_u \right) \tensor
 \left(\bigoplus_{v \in V} C_v \right)
\quad = \quad
\bigoplus_{(u,v) \in U \times V} C_u \tensor C_v
\]

If $\unit$ is the unit of the tensor product in $\cC$, then the basic object $\bigoplus_1 \unit$ is taken to be the unit of the tensor product in $\cC_\oplus$.
An affine symmetric monoidal structure on $\cC$ can be lifted in a satisfactory manner to this new tensor product (in particular $\iota_\oplus(\top)$
is still a terminal object).

The following property, which is arguably obvious from the definition of
$\cC_\oplus$, will turn out to be quite important later.
\begin{prop}\label{prop:yondistr}
  For any $A, B \in \Obj(\cC_{\oplus})$, there is a natural isomorphism
  \[\Hom{\cC_\oplus}{\iota_\oplus(-)}{A \oplus B} \quad \cong \quad
\Hom{\cC_\oplus}{\iota_\oplus(-)}{A} + \Hom{\cC_\oplus}{\iota_\oplus(-)}{B} \]
\end{prop}

\begin{rem}\label{rem:coprod-univ}
  Compare with the following natural isomorphism, which holds for any category
  $\cD$ and $A,B \in \Obj(\cD)$:
  \[\Hom{\cD}{A \oplus B}{-} \quad \cong \quad \Hom{\cD}{A}{-} \times \Hom{\cD}{B}{-} \]
  In fact, this characterizes the coproduct of $A$ and $B$. Just like
  Proposition~\ref{prop:lin-hom-hom}, this is an instance of a general
  equivalence between universal properties and natural bijections involving
  homsets.
\end{rem}

\subsubsection{Conservativity over affine monoidal settings}

First, note that the coproduct completion can be lifted at the level of streaming settings.

\begin{defi}
Given a streaming setting $\mfC = (\cC, \initty, \retty, \curlyinterp{-}_\cC)$, define
$\mfC_\oplus$ as the tuple
\[(\cC_\oplus, \iota_\oplus(\initty), \iota_\oplus(\retty), \curlyinterp{-}_{\cC_\oplus})\]
where $\curlyinterp{-}_{\cC_\oplus}$ is obtained by precomposing the canonical isomorphism
(recalling that $\iota_\oplus$ is full and faithful)
\[\Hom{\cC_\oplus}{\iota_\oplus(\initty)}{\iota_\oplus(\retty)}
\cong
\Hom{\cC}{\initty}{\retty}\]
\end{defi}

Before moving on, let us make the following definition:
an object $A$ in a monoidal category $(\cC, \tensor, \unit)$ is
said to have \emph{unitary support}if there exists a map $\unit \to A$.
This is quite useful in affine categories for transductions, as it ensures the following.

\begin{lem}
\label{lem:sumtensorembed}
Let $\cC$ be a symmetric affine monoidal category. 
Then, for any pair of finite families $(C_u)_{u \in U}$ and $(C_v)_{v \in V}$ of objects of $\cC$ such that all $C_u$ and $C_v$ have unitary support,
we have a $U \times V$-indexed family of embeddings
\[\embeddinjunk_{u,v} : \Hom{\cC}{C_u}{C_v} \; \to \; \varHom{\cC}{\bigtensor_{u \in U} C_u}{\bigtensor_{v \in V} C_v}\]
\end{lem}

The basic idea behind Lemma~\ref{lem:sumtensorembed} can be pictured using string diagrams as in Figure~\ref{fig:junk}: a morphism
$C_u \to C_v$ can be pictured as a single string, which is to be embedded in a diagram with $U$-many inputs and $V$-many outputs. The fact that $\cC$ is affine allows us to cut all input strings for $u' \neq u$ using a weakening node, and
unitary support allow us to create some ``junk'' strings with no input to connect to those $v' \neq v$.
This might fail for arbitrary symmetric affine monoidal categories: take for instance the
category of finite sets and surjections between them, with the coproduct as a monoidal product.

\begin{figure}
\centering
\includegraphics[scale=0.75]{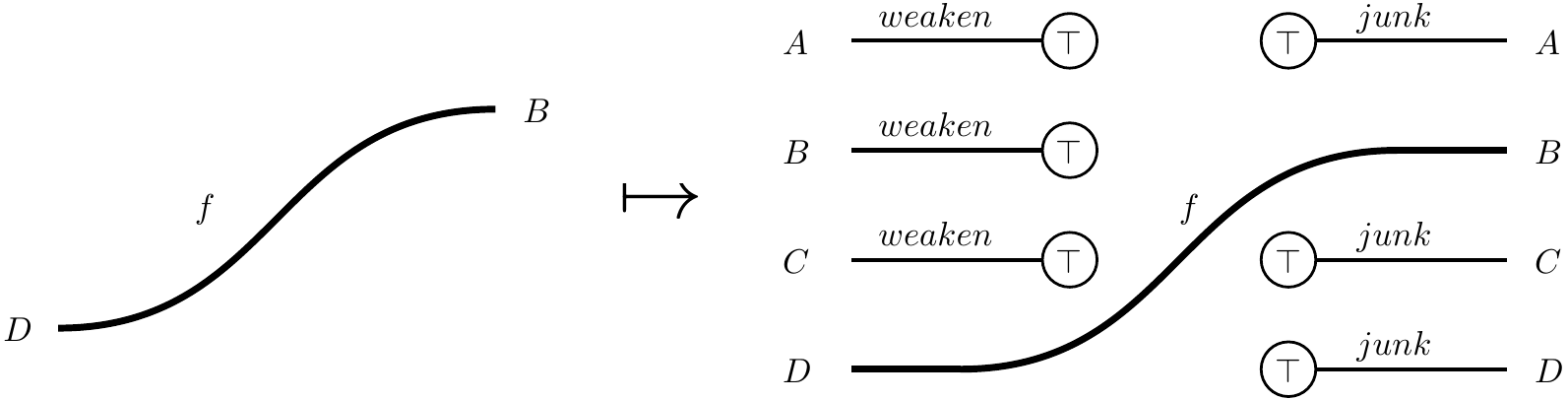}
\vspace{1em}
\hrule
\caption{$\embeddinjunk_{D,B} : \Hom{\cC}{D}{C} \to \Hom{\cC}{A \tensor B \tensor C \tensor D}{A \tensor B \tensor C \tensor D}$}
\label{fig:junk}
\end{figure}

We are now ready to state our first theorem asserting that, in those
favorable circumstances, coproduct completions do not give rise to more expressive SSTs.

\begin{thm}
\label{thm:oplus-sst-conservative}
Let $\mfC$ be an affine symmetric monoidal streaming setting where
all objects $C$ such that $\Hom{\cC}{\initty}{C} \neq \varnothing \neq \Hom{\cC}{C}{\retty}$ have unitary support.

$\mathfrak{C}$-SSTs are equivalent to $\mathfrak{C}_\oplus$-SSTs.
\end{thm}
\begin{proof}
Since $\iota_\oplus$ extends to a morphism of streaming settings,
$\mathfrak{C}_\oplus$-SSTs subsume $\mathfrak{C}$-SSTs.

Conversely, let $\cT = \left(Q, q_0, \bigoplus_{u \in U} C_u, \delta, i, o\right)$ be a
$\mfC_\oplus$-SST with input $\Sigma^*$ and $C_u$ basic objects.
Then, we construct a $\mfC$-SST
\[\cT' = \left(Q \times U,\; (q_0, u_0),\; \bigtensor_{u \in U} C_u,\; \delta',\; i_{u_0},\; o'\right)\]
such that $\interp{\cT} = \interp{\cT'}$. We define successively $(u_0, i_{u_0})$, $\delta'$ and $o'$.
\begin{itemize}
\item
We have $i \in \Hom{\cC_\oplus}{\iota_\oplus(\top)}{\bigoplus_{i \in I} C_i}$ which can
be rewritten as a factorization
\[\xymatrix@C=30mm{
\displaystyle\iota_\oplus(\top) = \bigoplus_1 \initty \ar[r]^{\quad (\In_{u_0}, * \mapsto \id)} &
\displaystyle\bigoplus_{U} \initty \ar[r]^{(\id, u \mapsto i_u)} & \displaystyle\bigoplus_{u \in U} C_u
}\]
for some $u_0 \in U$ (the $i_u$ for $u \neq u_0$ are taken arbitrarily thanks to the assumption that the $C_u$ have unitary support).
This $u_0$ is the second component of the initial state of $\cT'$.
\item We set $\delta'(a,(q,u)) = (q', \alpha_u(f'))$ if $\delta(a,q) = (q', f')$, where
\[(\alpha_u)_{u \in U} \; : \; \prod_{u \in U} \left[ \varHom{\cC_\oplus}{\bigoplus_{u \in U} C_u}{\bigoplus_{u \in U} C_u} \to \varHom{\cC_\oplus}{\bigtensor_{u \in U} C_u}{\bigtensor_{u \in U} C_u} \right]\]
is defined by taking the pointwise composite of 
\[
\begin{array}{lrccc}
& {\widetilde{\alpha}_u}:
&\displaystyle\varHom{\cC_\oplus}{\bigoplus_{u \in U} C_u}{\bigoplus_{u' \in U} C_{u'}}
&\rightarrow&
\displaystyle\sum_{u' \in U} \Hom{\cC_\oplus}{C_u}{C_{u'}} \\
&
\beta_u
: &
\displaystyle\sum_{u' \in U} \Hom{\cC_\oplus}{C_u}{C_{u'}}
&\rightarrow&
\displaystyle U \times \varHom{\cC_\oplus}{\bigtensor_{u \in U} C_u}{\bigtensor_{u \in U} C_u} \\
&
\pi : &
\displaystyle U \times \varHom{\cC_\oplus}{\bigtensor_{u \in U} C_u}{\bigtensor_{u \in U} C_u}
&\rightarrow&
\displaystyle\varHom{\cC_\oplus}{\bigtensor_{u \in U} C_u}{\bigtensor_{u \in U} C_u}
\end{array}
\] 
where
$\widetilde{\alpha}_u$ is obtained by evaluating its input $f \in \prod_{u \in U} \sum_{u' \in U} \Hom{\cC_\oplus}{C_u}{C_{u'}}$ at $u$,
$\beta_u = {\sum_{u'}\embeddinjunk_{u,u'}}$ (with $\embeddinjunk$ given as per Lemma~\ref{lem:sumtensorembed}) and $\pi$ taken to be the second projection.
\item Finally,
we set $o'(q,u) \in \Hom{\cC}{\bigtensor_{v \in U} C_v}{\retty}$
to $\widetilde{o}_u \in \Hom{\cC}{C_u}{\retty}$ precomposed with the projection $\pi_u \in \Hom{\cC}{\bigtensor_{v \in U} C_v}{C_u}$.
\end{itemize}
To conclude, one can check by induction that $\interp{\cT}(w) =
\interp{\cT'}(w)$ for every $w \in \Sigma^*$.
\end{proof}

\begin{cor}
\label{cor:oplus-sst-reg-conservative}
$\SR_\oplus$-SSTs are equivalent to $\SR$-SSTs.
\end{cor}
\begin{proof}
All objects of $\Sr$ have unitary support via an induction:
tensor with the map $\varepsilon : \unit \to \Gamma^*$ corresponding to the empty word
at the recursive step.
\end{proof}

\subsubsection{State-dependent memory SSTs}
The free coproduct completion encourages us to define the notion of
\emph{state-dependent memory SST}, generalizing usual copyless SSTs as follows:
instead of taking a single object $C \in \Obj(\cC)$ as an abstract infinitary memory,
we allow to take a family $(C_q)_{q \in Q} \in \Obj(\cC)^Q$ indexed by the states of the
SST.
\begin{defi}
\label{def:sdm-sst}
A \emph{state-dependent memory} $\mfC$-\emph{SST} (henceforth abbreviated sdm-$\mfC$-SST or sdmSST when $\mfC$ is clear form context)
with input $\Sigma^*$ is a tuple $(Q, q_0, \delta, (C_q)_{q \in Q}, i, o)$  where
\begin{itemize}
\item $Q$ is a finite set of states
\item $q_0 \in Q$ is some initial state
\item $\displaystyle \delta : \Sigma \to \prod_{q \in Q} \sum_{r \in Q} \Hom{\cC}{C_q}{C_r}$ is a transition function
\item $i \in \Hom{\cC}{\initty}{C_{q_0}}$ is the initialization morphism
\item $o \in \prod_{q \in Q} \Hom{\cC}{C_{q}}{\retty}$ is the output family of morphism
\end{itemize}
\end{defi}

In the sequel, we shall often use sdmSSTs because we find them convenient to give more elegant
constructions that produce little ``junk'', as is encoded in Lemma~\ref{lem:sumtensorembed}.
They essentially give the full power of coproducts in any given situation as shown below.

\begin{lem}
\label{lem:sdmSST-oplus}
Let $\mfC$ be a streaming setting. State-dependent memory $\mfC$-SSTs are
exactly as expressive as $\mfC_\oplus$-SSTs.
\end{lem}
\begin{proof}[Proof sketch]
Given a $\mfC_\oplus$-SST $\left(Q,q_0, \bigoplus_{u \in U} C_u, \delta, i,o\right)$
where $i(*) = \In_{u_0}(i')$ one may check that the following sdm-$\mfC$-SST
computes the same function:
\[\left(Q \times U, (q_0,u_0), (C_u)_{(q,u) \in Q \times U}, \delta', i', (o(q)_u)_{(q,u) \in Q \times U}\right)\]
where $\delta'(a)_{q,u} = ((r,u'), f)$ if and only if $\delta(a,q) = (r, v)$ and $v_u = (u', f)$.

Conversely, letting $(Q, q_0, (C_q)_{q \in Q}, \delta, i, o)$ be a sdm-$\mfC$-SST, an equivalent $\mfC_\oplus$-SST
is given by $(Q, q_0, \bigoplus_{q \in Q} C_q, \delta', \In_{q_0}(i), o')$, where it is sufficient to define
$o'(q)$ as $(o_q)_q$ and to ensure that if $\delta'(a,q) = (r, (r_{q'}, f_{q'})_{q' \in Q})$, then
$\delta(a)_q = (r, f_q)$ and $r_q = r$. This can be done.
\end{proof}

Finally, let us remark that the notions of single-state, ``normal'' and state-dependent memory $\mfC$-SSTs coincide if
$\mfC$ has all coproducts.

\begin{lem}
\label{lem:coprod-whatever}
If $\mfC$ is a streaming setting with coproducts, single-state $\mfC$-SSTs are as expressive as general $\mfC$-SSTs
and sdm-$\mfC$-SSTs.
\end{lem}
\begin{proof}
Take a sdmSST $(Q, q_0, (C_q)_{q \in Q}, \delta, i, o)$ to the single-state SST
\[\left(\{\bullet\},\, \bullet,\, \bigoplus_{q \in Q} C_q,\, \delta',\, \In_{q_0} \circ i,\, [o(q)]_{q \in
      Q}\right)\]
where $\delta'$ is defined from $\delta$ through the maps
\[
\xymatrix@C=3mm{
\displaystyle \left[\prod_{q \in Q} \sum_{r \in Q} \Hom{\cC}{C_q}{C_r}\right]^\Sigma \ar[r]
& \displaystyle \left[\prod_{q \in Q} \Hom{\cC}{C_q}{\bigoplus_{r \in Q} C_r}\right]^\Sigma \ar[r]^\sim
& \displaystyle \left[\Hom{\cC}{\bigoplus_{q \in Q} C_q}{\bigoplus_{r \in Q} C_r}\right]^\Sigma 
}
\]
\end{proof}

\begin{rem}
  \label{rem:sst-comparison}
The comparison between single-state, standard and state-dependent memory $\mfC$-SSTs can be summed up
in terms of completion with the following ``equalities'':
\[\text{$\mfC$-SSTs ~~=~~ single-state $\mfC_{\oplus\mathrm{const}}$-SSTs} \qquad \qquad \text{sdm-$\mfC$-SSTs ~~=~~ single-state $\mfC_\oplus$-SSTs}\]
where $\mfC_{\oplus\mathrm{const}}$ designates the following restriction of $\mfC_\oplus$: the category $\cC_\oplus$ is restricted to the full subcategory $\cC_{\oplus\mathrm{const}}$
whose objects are constant formal sums $\bigoplus_{i \in I} C$ for some $C \in \Obj(\cC)$.
\end{rem}

\subsubsection{Some function spaces in $\Sr_\oplus$}
Now we study $\Sr_\oplus$ in some more detail. This category
is unfortunately not able to interpret even the $\tensor/\lin$ fragment of $\laml$, because, like $\Sr$,
it lacks internal homsets $A \lin B$ for every pair $(A,B) \in \Obj(\Sr_\oplus)^2$.
However, they exist when $A$ lies in the image of $\iota_\oplus$.
It will turn out to be very useful later on.

Now we reframe a useful technical argument, typically made when dealing with determinization and composition
of standard copyless SSTs to obtain the internal homsets we desire. The core of this argument (sketched in~\cite[p.206-207]{Toolbox}) is that register transitions may be effectively coded using a combination of state and
a larger set of registers. Here, the intuition is that the free coproduct completion allows the category $\Sr_\oplus$ to integrate all the features of the additional finite set of states provided by the SSTs. 

\begin{lem}
\label{lem:hom-reg-oplus}
Let $R, S \in \Obj(\Sr)$. There is an internal homset $\iota_\oplus(R) \lin \iota_\oplus(S)$ in $\Sr_\oplus$.
\end{lem}
In Example~\ref{ex:sroplus-hom}, we work through the proof below in a concrete case.
\begin{proof}
First, recall that $\iota_\oplus$ is full and faithful, and that it is thus pertinent to focus our preliminary analysis on
morphisms in $\Sr$.
Recall that a register transition $f : R \to S$, which is a set-theoretical map $S \to (\Gamma + R)^*$
where for every $r \le R$, $\sum_{s \in S} \occ{\inr(r)}{f(s)} \le 1$ (i.e., it is copyless).
Consider the map $(\Gamma + R)^* \to R^*$ erasing the letters of $\Gamma$. Then, the image of the induced
map $p : \Hom{\Sr}{R}{S} \to [S \to R^*]$ is clearly finite because of copylessness.
In fact, letting $\LinOrd(X)$ be the set of
all total orders over some set $X$, we have an isomorphism between the image of $\Hom{\Sr}{R}{S}$ under $p$
and the following dependent sum
\[\cO(R,S) \qquad = \qquad \sum_{\hat{f} : R \partto S} \prod_{s \in S} \LinOrd(\hat{f}^{-1}(s))\]
The intuition is that $\hat{f}$ tracks where register variables in $R$ get affected and the additional data encode in which order
they appear in an affectation.
Once this crucial finitary information $(\hat{f}, (<_s)_{s \in S})$ is encoded in the internal homset using coproducts, it only remains to recover the
information we erased with $p$, i.e. what words in $\Gamma^*$ located between occurrences of register variables. This information cannot
be bounded by the size of $R$ and $S$, but the number of intermediate words can; we may index them by $S + \dom(\hat{f})$.

Putting everything together, it means that we take
\[\iota_\oplus(R) \lin \iota_\oplus(S) \qquad = \qquad \bigoplus_{(\hat{f}, <) \in \cO(R,S)} \iota_\oplus(S + \dom(\hat{f}))\]
Now, we need to define the evaluation map $\ev_{R,S} : [\iota_\oplus(R) \lin \iota_\oplus(S)] \tensor \iota_\oplus(R) \to \iota_\oplus(S)$.
Recall that the tensor distributes over $\oplus$, so we really need to exhibit $\ev_{R,S}$ in
\[\varHom{\Sr_\oplus}{\bigoplus_{(\hat{f}, <)} \iota_\oplus(S + \dom(\hat{f}) + R)}{\iota_\oplus(S)} ~~ \cong ~~ \prod_{(\hat{f}, <)} \varHom{\Sr}{S + \dom(\hat{f}) + R}{S}\]
where the indices $(\hat{f}, <)$ ranges over $\cO(R,S)$ on both sides.
Call $\ev_{R,S,\hat{f}, <}$ the corresponding family of $\Sr$-morphisms, whose members are set-theoretic maps $S \to (S + \dom(\hat{f}) + R)^*$.
Calling $\{r_1, \ldots, r_k\}$ the subset of $\dom(\hat{f})$ ordered by $r_1 < \ldots < r_k$, we set
\[\ev_{R,S,\hat{f}, <}(s) ~~=~~ \In_0(s)\In_2(r_1)\In_1(r_1) \ldots \In_2(r_k)\In_1(r_k)\]
This concludes the definition of $\ev$.
We now leave checking that this satisfies the required universal property to the reader.
\end{proof}

\begin{exa}
  \label{ex:sroplus-hom}
Let us illustrate this construction in a simple case. Consider
the following register transition for the concrete base alphabet $\{a,b\}$ and register names $x, y, z, u, r, s$:
\[
\begin{array}{lcl}
r &\leftarrow& zaxabyaa \\
s &\leftarrow& bab \\
\end{array}
\]
Up to the evident isomorphism $\{x,y,z,u\} \cong \{x\} \tensor\{y,z,u\}$, this determines a morphism $f \in \Hom{\Sr}{\{x\} \tensor \{y,z,u\}}{\{r,s\}}$.
Let us describe the unique map $\Lambda(f)$ in
\[ \Hom{\Sr_\oplus}{\iota_\oplus(\{x\})}{\;\iota_\oplus(\{y,z,u\}) \lin \iota_\oplus(\{r,s\})} \cong \sum_{\hat{f}, <} \Hom{\Sr}{\{x\}}{\{r,s\} + \dom(\hat{f})}\]
such that $\ev \circ (\Lambda(f) \tensor \id) = f$.
On its first component, we set $\hat{f}(y) = \hat{f}(z) = r$ and leave it undefined on $u$;
as for the order, we set $z <_r y$.
The last component corresponds to the register transition $h$ depicted below on the left, where,
for legibility, we write $\bf r, s, z$ and $\bf y$ for
$\inl(r), \inl(s), \inr(z)$ and $\inr(y)$.
Using the same notations, we also display on the right
the relevant component of $\ev_{\{y,z,u\},\{r,s\}, \hat{f}, <} \in \Hom{\Sr}{(\{r,s\} + \{z,y\})+\{z,y\}}{\{r,s\}}$,
so that the reader may convince themself that the composite $\ev_{\{y,z,u\},\{r,s\}, \hat{f}, <} \circ (h \tensor \id_{\{y,z,u\}})$
is indeed $f$.
\[
\begin{array}{lcl}
{\bf r} &\leftarrow& \varepsilon \\
{\bf z} &\leftarrow& axab \\
{\bf y} &\leftarrow& aa \\
{\bf s} &\leftarrow& bab \\
\end{array}
\qquad
\qquad
\begin{array}{lcl}
r &\leftarrow& \inl({\bf r})\;{\bf y}\;\inl({\bf y})\;{\bf z}\;\inl({\bf z}) \\
s &\leftarrow& \inl({\bf s}) \\
\end{array}
\]
\end{exa}

Lemma~\ref{lem:hom-reg-oplus} can be extended to define internal homsets
$\iota_\oplus(R) \lin C$ for arbitrary $C \in \Obj(\Sr_\oplus)$ through the
natural isomorphism of Proposition~\ref{prop:yondistr}. However, extending this to all
homsets (i.e.\ allowing any object of $\Sr_\oplus$ in the left-hand side) seems
impossible: the lack of \emph{products} prevents us from doing so.

\subsection{The product completion (or non-determinism)}
\label{subsec:with-string}
The above point (among others) leads us to study the free finite product completion of streaming settings.
As for coproducts, we first discuss the categorical construction before turning to the expressiveness.
Here, the situation is more intricate as it turns out that sdm-$\mfC_\with$-SSTs of interest will roughly
have the power of \emph{non-deterministic} sdm-$\mfC$-SSTs. We make that
connection precise.

Thankfully, a determinization theorem for usual copyless SSTs, i.e.\ $\SR$-SSTs,
exists in the literature~\cite{NSST} (the proof in the given reference is
indirect and goes through monadic second-order logic). It could be applied
without difficulty to show that non-determinism does not increase the power of
\emph{sdm}-$\SR$-SSTs.

To keep the exposition self-contained \emph{and} illustrate our framework, we
give in \Cref{subsec:uniformization} a direct determinization argument generalized
to our setting by using, crucially, the concept of internal homsets (and Lemma~\ref{lem:hom-reg-oplus} for
the desired application).

\subsubsection{Definition and basic properties}

The product completion can, of course, be defined as the dual of the coproduct completion.
\begin{defi}
Let $\cC$ be a category. Its \emph{free finite product completion} is $\cC_\with =((\cC^\op)_\oplus)^\op$.
\end{defi}

While conceptually immaculate, this definition merits a bit of unfolding.
The objects of $\cC_\with$ are still finite families $(C_x)_{x \in X}$ -- in this context,
we write them as formal products $\bigwith_{x \in X} C_x$. As for homsets, we have the dualized situation
\[\varHom{\cC_\with}{\bigwith_{x \in X} C_x}{\bigwith_{y \in Y} C_y} \quad \cong \quad \prod_{y \in Y} \sum_{x \in X} \Hom{\cC_\with}{C_{x}}{C_y}\]
We also a full and faithful functor $\iota_\with : \cC \to \cC_\with$ with a similar universal property as for the coproduct completion.

As with the coproduct completion, one may want to produce a tensor product in $\cC_\oplus$ if the underlying category $\cC$ has one.
The very same recipe can be applied: we define the tensor so that the distributivity $A \tensor (B \with C) \cong (A \tensor B) \with (A \tensor C)$ holds.

\[ \left(\bigwith_{x \in X} C_x \right) \tensor
 \left(\bigwith_{y \in Y} C_y \right)
\quad = \quad
\bigwith_{(x,y) \in X \times Y} C_x \tensor C_y
\]

\begin{rem}
One might be disturbed by this distributivity
of $\tensor$ over $\with$, which goes against the non-linear intuition of thinking of $\with$ and $\tensor$ as ``morally the same''.
This feeling may also be exacerbated by the familiar iso
\[\Hom{\Sr_\oplus}{\top}{R \tensor S} \cong \Hom{\Sr_\oplus}{\top}{R} \times \Hom{\Sr_\oplus}{\top}{S}\]
This indeed becomes false when going from $\Sr$ to $\Sr_\with$.
The useful mnemonic here (which is untrue for pure LL!) is that the \emph{multiplicative} connective distributes over both \emph{additive} connectives in the same way.
\end{rem}

\begin{rem}
While $\cC_\with$ inherits a monoidal product from $\cC$ much like with the
coproduct completion, it does \emph{not} preserve the \emph{affineness} of monoidal products.
The product completion indeed adds a new terminal object, namely the
empty family, which can \emph{never} be isomorphic to the singleton family $\iota_\with(\top)$.
More generally, $\iota_\with$ preserves colimits rather than
limits.
\end{rem}

It is also straightforward to extend the product completion at the level of streaming settings $\mfC \mapsto \mfC_\with$.

\subsubsection{Relationship with non-determinism}
At this juncture, our goal is to prove the equivalence between sdm-$\SR_\with$-SSTs and sdm-$\SR$-SSTs.
One direction is trivial; for the other, we actually prove that
sdm-$\SR_\with$-SSTs are subsumed by sdm-$\SR_\oplus$-SSTs.

This result involves some non-trivial combinatorics. We prove it via \emph{uniformization} of non-deterministic sdm-$\SR$-SSTs, a mild generalization of determinization\footnote{We work with uniformization here
we find it slightly more convenient to handle. This choice of uniformization over
determinization is rather inessential.}.

Our non-deterministic devices make \emph{finitely branching} choices, following
the case of the usual non-deterministic SSTs~\cite[Section~2.2]{NSST}. To
express this, we use the notation $\powerfin(X)$ for the set of \emph{finite}
subsets of a set $X$. (Note that if $Q$ is some finite set of states,
we have the simplification $\powerset(X) = \powerfin(X)$.)

\begin{defi}
  Let $\mfC$ be streaming setting with output $X$.
A \emph{non-deterministic} sdm-$\mfC$-SST with input $\Sigma^*$ and output $X$ is a tuple $\left(Q, I, (C_q)_{q \in Q}, \Delta, i, o\right)$ where
\begin{itemize}
\item $Q$ is a finite set of states and $I \subseteq Q$
\item $(C_q)_{q \in Q}$ is a family of objects of $\cC$
\item $\Delta$ is a finite transition relation: $\displaystyle\Delta \in \powerfin\left( \Sigma \times \sum_{(q,r) \in Q^2} \Hom{\cC}{C_q}{C_r}\right)$
\item $i \in \displaystyle\prod_{q_0 \in I} \Hom{\cC}{\top}{C_{q_0}}$
  is a family of input morphisms
\item $o \in \displaystyle\prod_{q \in Q} (\Hom{\cC}{C_q}{\retty} + 1)$
  is a family of partial output morphisms
\end{itemize}
A \emph{partial} sdm-$\mfC$-SST $(Q, q_0, (C_q)_{q \in Q}, \delta, i, o)$ has
the same definition as a (deterministic) sdm-$\mfC$-SST
(Definition~\ref{def:sdm-sst}), except for the output $o$, which is allowed to
be partial, just as in the last item above.
\end{defi}

By transposing the usual notion of run of a non-deterministic finite automaton,
one sees that a non-deterministic sdm-$\mfC$-SST $\cT$ gives rise to a function
$\interp{\cT} : \Sigma^* \to \powerfin(X)$ (for an input alphabet $\Sigma$ and
output set $X$). Similarly, a partial sdm-$\mfC$-SST $\cT'$ is interpreted as a
partial function $\interp{\cT} : \Sigma^* \partto X$.

\begin{rem}
  In line with Remark~\ref{rem:sst-comparison}, we may describe
  non-deterministic sdm-$\mfC$-SSTs as single-state (deterministic) SSTs over an
  enriched streaming setting.

  Let $\mfC = (\cC, \initty_\mfC, \retty_\mfC, \curlyinterp{-}_\mfC)$, with
  output $X$. We first define the category $\mathsf{NFA}(\cC)$:
  \begin{itemize}
  \item its objects consist of a finite set $Q$ with a family $(C_q)_{q\in Q}
    \in \Obj(\cC)^Q$;
  \item its morphisms are
    $\displaystyle\Hom{\mathsf{NFA}(\cC)}{(C_q)_{q\in Q}}{(C'_r)_{r\in R}} =
    \powerfin\left( \sum_{(q,r) \in Q \times R} \Hom{\cC}{C_q}{C_r}\right)$
  \item the composition of morphisms extends that of binary relations:
    \[ \varphi \circ \psi = \{(q,\, s,\, f \circ g) \mid (q,\; r,\; f : C_q \to
      C'_r) \in \varphi,\; (r,\; s,\; g : C'_r \to C''_s) \in \psi\} \]
  \end{itemize}
  To lift this to a construction $\mathsf{NFA}(\mfC)$ on streaming settings, we
  take:
  \begin{itemize}
  \item $\initty_{\mathsf{NFA}(\mfC)}$ and $\retty_{\mathsf{NFA}(\mfC)}$ are the
    $\{\bullet\}$-indexed families containing respectively $\initty_\mfC$ and
    $\retty_\mfC$, so that
    $\Hom{\mathsf{NFA}(\cC)}{\initty_{\mathsf{NFA}(\mfC)}}{\retty_{\mathsf{NFA}(\mfC)}}
    = \powerfin\left( \{(\bullet,\bullet)\} \times
      \Hom{\cC}{\initty_{\mfC}}{\retty\mfC} \right)$
  \item $\curlyinterp{\varphi}_{\mathsf{NFA}(\mfC)} = \{\curlyinterp{f}_\mfC
    \mid (\bullet,\bullet,f) \in \varphi \} \subseteq X$, so that the output set
    of a $\mathsf{NFA}(\mfC)$-SST is $\powerfin(X)$.
  \end{itemize}
  There is a slight mismatch between the above definition of non-deterministic
  sdm-$\mfC$-SSTs and single-state $\mathsf{NFA}(\mfC)$-SSTs: in the latter, two
  distinct input (resp.\ output) morphisms may correspond to the same initial
  (resp.\ final) state. However, the former can encode such situations by
  enlarging the set of states (to keep it finite, the use of $\powerfin(-)$ in
  the definitions is crucial).
  
  One can also give an analogous account of partial sdm-$\mfC$-SSTs; we leave
  the interested reader to work out the details.
\end{rem}

Coming back to our main point, we now state the slight variation of the
determinization theorem for copyless SSTs~\cite{NSST} that fits our purposes.
\begin{defi}
  Given an arbitrary function $F : X \to \powerset(Y)$, we say that $f : X
  \partto Y$ \emph{uniformizes} $F$ if and only if $\dom(f) = X \setminus
  F^{-1}(\varnothing)$ and $\forall x \in \dom(f).\; f(x) \in F(x)$.
\end{defi}
\begin{thm}
\label{thm:uniformization-literature}
For every non-derministic $\SR$-SSTs $\cT$ with input $\Sigma^*$, there exists a partial deterministic $\SR$-SST $\cT'$ such that
$\interp{\cT'}$ uniformizes $\interp{\cT}$.
\end{thm}

We now show that studying the uniformization property between classes of sdmSSTs parameterized by streaming setting $\mfC$ and $\mfD$
is morally the same as comparing the expressiveness of sdm-$\mfC_\with$-SSTs and sdm-$\mfD$-SSTs. 

\begin{lem}
\label{lem:with-uniformization}
Non-deterministic sdm-$\mfC$-SSTs are uniformizable by partial sdm-$\mfD$-SSTs if and only if sdm-$\mfD$-SSTs subsume sdm-$\mfC_\with$-SSTs.
\end{lem}

\begin{proof}
First, let us assume that sdm-$\mfC$-SSTs uniformize sdm-$\mfD$-SSTs and let
\[\cT \; = \; \left(Q, q_0, (A_q)_{q \in Q}, \delta, i, o\right) \qquad \text{with} \qquad A_q \; = \; \bigwith_{x \in X_q} C_{q, x}\]
be a $\mfC_\with$-SST with input $\Sigma$.
We first define a non-deterministic sdm-$\mfC$-SST
$\cT'$ by setting 
\[
\begin{array}{c}
\cT' = \left(I + Q', I, (C_{p(m)})_{m \in I + Q'}, \Delta, i', o'\right) \\\\
\end{array}
\]
\[\begin{array}{c}
\begin{array}{lcl}
Q' &=& \sum_{q \in Q} X_q \\
I &=& \sum_{x \in X_{q_0}} \{ f \mid i_* = (x, f) \}
\end{array}
\qquad\qquad\begin{array}{lrcl}
p :& I + Q' &\to& Q' \\
&\inl(x,f) &\mapsto& (q_0, x) \\
&\inr(q,x) &\mapsto& (q, x) \\
\end{array}
\end{array}
\]\[
\begin{array}{c}
i'_{\inl(x,f)} = f \qquad \qquad
o'_m = \inl\left(\pi_2\left((o_{\pi_1(p(m))})_{\pi_2(p(m))}\right)\right)
\end{array}\]
\[\begin{array}{c}
\begin{array}{lcl}
\Delta &=& \left\{ (a, ((m,\inr(r,y)), f)_{m}) \left|
\begin{array}{llcl}
\multicolumn{4}{l}{\forall (x,q) \in Q'. \forall m \in p^{-1}(x,q).} \\
& \pi_1(\delta(a)_q) &=& r \\
\wedge& \pi_2(\delta(a)_q)_{y} &=& (x, f)
\end{array}
 \right. \right\} \\\\
\end{array}
\end{array}
\]
Taking $\cT''$ to be a sdm-$\mfD$-SST uniformizing $\cT'$, we have $\{-\} \circ \interp{\cT''} = \interp{\cT'} = \{-\} \circ \interp{\cT}$, so we are done.

For the converse, assume that sdm-$\mfD$-SSTs subsume sdm-$\mfC_\with$-SSTs and suppose we have some non-deterministic sdm-$\mfC$-SST $\cT = (Q, I, (A_q)_{q \in Q}, \Delta, i, o)$
to uniformize. Fix a total order $\preceq$ over the morphisms of $\cC$ occurring in $\Delta$ (recall that there are finitely many of them).
Consider a partial deterministic sdm-$\mfC_\with$-SST $\cT'$ obtained from $\cT$ by a powerset construction
\[\cT' \; \quad = \; \left(\powerset(Q),\; I,\; \left(\bigwith_{r \in R}
      A_r\right)_{R \subseteq Q},\; \delta,\; i,\; o'\right)\]
\[\text{where}~ \delta(a)_R ~~=~~ (S, (r_s, f_s)_{s \in S}) \quad \text{if and only if} \quad \forall s \in S ~\left[ \begin{array}{ll} (a, ((r_s,s),f_s)) \in \Delta \\ \forall g. (a,((r_s,s),g)) \in \Delta \Rightarrow f_s \preceq g \end{array}\right. \]
and $o'_R = \restr{o}{r}$ for some arbitrary $r$ such that the right hand-side is defined; if there is no such $r$, $o'_R$ is undefined and we call $R$ a \emph{dead} set of states.
By padding $o'$ with some arbitrary values on such dead states, we may extend $\cT'$ to a non-partial deterministic $\mfC$-SST $\cT''$ so that $\interp{\cT'} \subseteq \interp{\cT''}$. We may then consider a sdm-$\mfD$-SST
$\cT''' = (Q', q_0, (D_q)_{q \in Q'}, \delta, i'', o'')$ such that $\interp{\cT'}(w) = \interp{\cT''}(w) = \interp{\cT'''}$ for $w \in \dom(\interp{\cT'})$.
This $\cT'''$ is almost our uniformizer; we only need to restrict the domain of its output function.
This can be achieved by adding a
$\powerset(Q)$ component to the state space corresponding to the set of states reached by $\cT$ and forcing the output function to be undefined if this component contains a dead set of states.
\end{proof}

Putting Lemma~\ref{lem:with-uniformization} together with Theorem~\ref{thm:uniformization-literature} yields the desired result.

\begin{thm}
\label{thm:with-reg-conservativity}
sdm-$\SR_\with$-SSTs are subsumed by sdm-$\SR$-SSTs.
\end{thm}

We also provide a direct self-contained proof of the following statement generalizing Theorem~\ref{thm:with-reg-conservativity}.

\begin{restatable}{thm}{myuniformization}
\label{thm:myuniformization}
Let $\mfC$ and $\mfD$ be streaming settings such that there is a morphism of streaming settings $\mfC \to \mfD$,
whose underlying functor is $F : \cC \to \cD$.
Assume further that $\cD$ carries a symmetric monoidal affine structure
and has internal homsets $F(C) \lin F(C')$ for every pair of objects $C, C' \in \Obj(\cC)$.

Then, sdm-$\mfC_\with$-SSTs are subsumed by sdm-$\mfD$-SSTs.
\end{restatable}

Let us check that our technical development until now
allows deriving Theorem~\ref{thm:with-reg-conservativity} from the above result.
We instantiate $\mfC = \SR$ and $\mfD = \SR_\oplus$, with the
functor $\iota_\oplus : \Sr \to \Sr_\oplus$. We have already seen
that $\Sr_\oplus$ is a symmetric monoidal affine category. The assumption
on internal homsets is exactly Lemma~\ref{lem:hom-reg-oplus}. The theorem then
tells us that sdm-$\SR_\with$-SSTs are subsumed by
sdm-$\SR_\oplus$-SSTs, and the latter are no more expressive than
sdm-$\SR$-SSTs by Corollary~\ref{cor:oplus-sst-reg-conservative}.

Note that, conversely, Theorem~\ref{thm:with-reg-conservativity} also implies
Theorem~\ref{thm:uniformization-literature} through
Lemma~\ref{lem:with-uniformization}. Therefore, while
Theorem~\ref{thm:uniformization-literature} is a variant of the previously known
determinization of copyless SSTs, we generalized it
in a more abstract setting. Our main contribution is identifying the notion of internal
homsets as one of the key components which make the direct determinization proof
(that does not appear in~\cite{NSST}, but might be part of the folklore, see
e.g.~\cite[Problem~139 (p.~226)]{Toolbox}) work.
The proof itself is thus rather unsuprising, but rather involved, so we postpone it to Subsection~\ref{subsec:uniformization}.

\subsection{The $\oplus\with$-completion (a Dialectica-like construction)}
\label{subsec:dial-string}
We now consider the composition of the coproduct completion with the product completion
$(\cC_\with)_\oplus$.
Unraveling the formal definition and distributing sums and products at the right spots,
we define an isomorphic category $\cC_{\oplus\with}$ which is a bit less cumbersome to
manipulate in practice.

\begin{thm}
Given an arbitrary category $\cC$, there is an \emph{isomorphism of categories} (not
just an equivalence) between $(\cC_\with)_\oplus$ and the category
$\cC_{\oplus\with}$ defined below.
\begin{itemize}
\item The objects of $\cC_{\oplus\with}$ are triples $(U, (X_u)_{u}, (C_{u,x})_{(u,x)})$ where $U$ is a finite set,
$(X_u)_{u \in U}$ is a family of finite sets and $C_{u, x}$ is a family of objects of
$\cC$ indexed by $(u, x) \in \sum_{u \in U} X_u$.
We drop the first index $u$ when it is determined by $x \in X_u$ from context and write those objects $\bigoplus_{u \in U} \bigwith_{x \in X_u} C_{x}$ for short.
\vspace{2mm}
\item Its homsets are defined as
\[\varHom{\cC_{\oplus\with}}{\bigoplus_{u \in U} \bigwith_{x \in X_u} C_{x}}{\bigoplus_{v \in V} \bigwith_{y \in Y_v} C_{y}}
\quad = \quad
\prod_{u \in U} \sum_{v \in V} \prod_{y \in Y_v} \sum_{x \in X_u} \Hom{\cC}{C_{x}}{C_{y}}
\]
\item Its identities are maps
\[
\begin{array}{c@{}c@{}c@{}c@{}c@{}c@{}c@{}cl}
\displaystyle\prod_{u \in U}
& &\displaystyle\sum_{u' \in U}
& &\displaystyle\prod_{x' \in X_u}
& &\displaystyle\sum_{x \in X_{u'}}
& & \Hom{\cC}{C_{x}}{C_{,x'}}
\\ \\
u &\mapsto & \big[ u &, &\big(x &\mapsto &(x&,& \id_{C_{x}} )\big)\big]
\end{array}
\]
\item Composition is defined as in Figure~\ref{fig:dialcomp}. The
  interesting steps of this computation are those involving $v$ and $y$; since
  they are similar, let us focus on $v$. A map of the form
  \[ \begin{array}{c@{}c@{}c@{}c@{}c@{}c@{}c@{}c@{}c@{}c}
\displaystyle\prod_{v \in V} A_v
& \quad\times &\quad\displaystyle\sum_{v \in V}
& \;B_v & \quad\longrightarrow  & \quad\displaystyle\sum_{v \in V} & \;A_v & \;\times & \;B_v
\\ \\
(a_v)_{v \in V} & \quad, & (v', & b) & \quad\mapsto & (v', & (a_{v'} & \;, & b))
\end{array}
\]
is applied. This makes the two $\cC_{\oplus\with}$-morphisms interact (an interaction
represented in Remark~\ref{rem:dial-game} below as a move in a game): the $v'$
provided by the right one selects $a_{v'}$ among all the possibilities $(a_v)_v$
proposed by the left one.
\end{itemize}
\end{thm}
\begin{proof}
  By mechanical unfolding of the definitions.
\end{proof}

\begin{figure}
\[
\xymatrix@R=6mm{
\displaystyle\varHom{{}}{\bigoplus\limits_{v} \bigwith\limits_{y} C_{y}}{\bigoplus\limits_{w} \bigwith\limits_{z} C_{z}}
\times
\varHom{{}}{\bigoplus\limits_{u} \bigwith\limits_{x} C_{x}}{\bigoplus\limits_{v} \bigwith\limits_{y} C_{y}}
\ar@{=}[d] \\
\displaystyle\prod_{v} \sum_{w} \prod_{z} \sum_{y} \Hom{\cC}{C_{y}}{C_{z}}
\quad\times
\quad\prod_{u} \sum_{v} \prod_{y} \sum_{x} \Hom{\cC}{C_{x}}{C_{y}}
\ar[d]
 \\
\displaystyle\prod_{u} \left(
\prod_{v} \sum_{w} \prod_{z} \sum_{y} \Hom{\cC}{C_{y}}{C_{z}}
\;\times
\;\sum_{v} \prod_{y} \sum_{x} \Hom{\cC}{C_{x}}{C_{y}}
\right)
\ar[d]
\\
\displaystyle\prod_{u}\sum_v \left(
\sum_{w} \prod_{z} \sum_{y} \Hom{\cC}{C_{y}}{C_{z}}
\quad\times
\quad
\prod_{y} \sum_{x} \Hom{\cC}{C_{x}}{C_{y}}
\right)
\ar[d]^\sim
\\
\displaystyle\prod_{u}\sum_{v,w} \left(
\prod_{z} \sum_{y} \Hom{\cC}{C_{y}}{C_{z}}
\quad\times
\quad
\prod_{y} \sum_{x} \Hom{\cC}{C_{x}}{C_{y}}
\right)
\ar[d]
\\
\displaystyle\prod_{u}\sum_{v,w}
\prod_z
 \left(
\sum_{y} \Hom{\cC}{C_{y}}{C_{z}}
\quad\times
\quad
\prod_{y} \sum_{x} \Hom{\cC}{C_{x}}{C_{y}}
\right)
\ar[d]
\\
\displaystyle\prod_{u}\sum_{v,w}
\prod_z
\sum_y
\left(
\Hom{\cC}{C_{y}}{C_{z}}
\quad
\times
\quad
\sum_{x} \Hom{\cC}{C_{x}}{C_{y}}
\right)
\ar[d]^\sim
\\
\displaystyle\prod_{u}\sum_{v,w}
\prod_z
\sum_{y,x}
\left(
\Hom{\cC}{C_{y}}{C_{z}}
\times
\Hom{\cC}{C_{x}}{C_{y}}
\right)
\ar[d]^-{\text{composition in $\cC$}}
\\
\displaystyle\prod_{u}\sum_{v,w}
\prod_z
\sum_{y, x}
\Hom{\cC}{C_{x}}{C_{z}}
\ar[d]^-{\text{project away $v, y$}}
\\
\displaystyle\prod_{u}\sum_{w}
\prod_z
\sum_x
\Hom{\cC}{C_{x}}{C_{z}}
\ar@{=}[d]
\\
\displaystyle\varHom{\cC_{\oplus\with}}{\bigoplus_{u} \bigwith_{x} C_{x}}{\bigoplus_{w} \bigwith_{z} C_{z}}
}
\]

\hrule

\caption{Composition in $\cC_{\oplus\with}$ ($- \in -$ are omitted from indices)}
\label{fig:dialcomp}
\end{figure}

\begin{rem}
The reader may notice that composition in $\cC_{\oplus\with}$
is very similar to the interpretation of cuts in G\"odel's Dialectica interpretation~\cite{godeldial}
and/or composition in categories of \emph{polynomial functors}~\cite{gkpolynomial, glehnmoss18}.
This intuition can be made formal. In particular, see~\cite{HofstraDialectica} for a
decomposition of a general version of the categorical Dialectica construction
into free completions with
\emph{simple} sums and products. In our context, the completion with simple coproducts
would be the $(-)_{\oplus\mathrm{const}}$ of
Remark~\ref{rem:sst-comparison}; conversely, a \enquote{dependent Dialectica}
could be defined in the fibrational setting of~\cite{HofstraDialectica}
analogously to our $(-)_{\oplus\with}$-completion by removing the simplicity
restriction.

\end{rem}
\begin{rem}
\label{rem:dial-game}
For the uninitiated, it can be helpful to compute this completion on the trivial
category $\mathbbl{1}$ with
one object and only its identity morphism. In this case, objects consists of a pair of a finite
set $U$ together with a family $(X_u)_{u \in U}$ of finite sets that can be regarded as
a two-move sequential game (with no outcome) between player $\oplus$ and $\with$:
first $\oplus$ plays some $u \in U$ and then $\with$ plays some $x \in X_u$.
One can then consider \emph{simulation games} between $(U,(X_u)_u)$ (the ``left hand-side'') and $(V,(Y_v)_v)$ (the ``right hand-side'')
proceeding as follows:

\begin{tabular}{ll}
\parbox{0.5\textwidth}{
\begin{itemize}
\item first, $\with$ plays some $u \in U$ on the left
\item then, $\oplus$ plays some $v \in V$ on the right
\item $\with$ answers with some $y \in Y_v$ on the right
\item finally $\oplus$ answers with $x \in X_u$ on the left.
\end{itemize}}
  &
\hspace{-4cm}\parbox{0.5\textwidth}{\[
\begin{array}{l| l @{~}c@{~} ll  l @{~}c@{~} @{~}c}
& U &,& (X_u)_u & \to & V ,& (Y_v)_v \\
\hline
\with & u & & & & & \\ 
\oplus &     & & & & v  \\
\with &     & & & & &  y \\
\oplus &     & & x \\
\end{array}
\]}
\end{tabular}
Morphisms in $\mathbbl{1}_{\oplus\with}$ are $\oplus$-strategies in such games. Identities
correspond to copycat strategies and composition is (a simple version) of an usual scheme in game semantics.
As for $\cC_{\oplus\with}$, one may consider that once this simulation game is played, $\oplus$ needs to provide a datum in some $\Hom{\cC}{C_{x}}{C_{y}}$ which depends on the outcome of the game.
\end{rem}

We write $\iota_{\oplus\with} : \cC \to \cC_{\oplus\with}$ for the (full and faithful) embedding
sending an object $C$ to the one-element family $\bigoplus_1 C$.
As the coproduct completion preserves limits, it means that $\cC_{\oplus\with}$ always boasts both
binary cartesian products and coproducts. Concretely, products are computed using the distributivity
of products over coproducts
\[\bigwith_{i \in I} \bigoplus_{j \in J_i} A_{j} \quad \cong \quad \bigoplus_{f \in \prod_{i \in I} J_i} \bigwith_{i \in I} A_{f(i)}\]
If $\cC$ has a symmetric monoidal structure $(\tensor, \unit)$, the lifting is computed
in $\cC_{\oplus\with}$ as
\[\bigtensor_{i \in I} ~~ \bigoplus_{u \in U_i} ~~ \bigwith_{x \in X_{u}} ~~ C_{x}
\qquad = \qquad \bigoplus_{f \in \prod\limits_{i \in I} U_i} ~~ \bigwith_{g \in \prod\limits_{i \in I} X_{i,f(i)}} ~~ \bigtensor_{i \in I} ~~ C_{g(i)}\]

\subsubsection{The monoidal closure theorem}

Recall from the introduction that a result of central importance in this paper
is the fact that the category $(\Sr_{\with})_{\oplus}$ is symmetric monoidal
closed (Theorem~\ref{thm:smcc-string}). This will be a consequence of the
general property below.

\begin{thm}
\label{thm:dial-haslin}
Let $(\cC,\tensor,\unit)$ be a symmetric monoidal category. Assume that its
coproduct completion $\cC_{\oplus}$ admits internal homsets
$\iota_{\oplus}(A) \lin \iota_{\oplus}(B)$ for every $A,B \in \Obj(\cC)$.
Then its Dialectica-like completion $\cC_{\oplus\with}$ is monoidal closed.
\end{thm}
Indeed, since Lemma~\ref{lem:hom-reg-oplus} gives us precisely the assumption on
internal homsets in the above statement for $\cC=\Sr$, we immediately get:
\begin{cor}[equivalent to Theorem~\ref{thm:smcc-string}]
\label{cor:register-dial-closed}
$\Sr_{\oplus\with}$ is monoidal closed.
\end{cor}

\begin{figure}
\[
\begin{array}{l}
\Hom{\cC_{\oplus\with}}{A \tensor B}{C} \\
\begin{array}{lcl@{\kern0em}c}
&=& {\displaystyle \varHom{\cC_{\oplus\with}}{\left[\bigoplus_{w \in W} \bigwith_{z \in Z_w} \iota_{\oplus\with}(A_{z})\right] \tensor \left[\bigoplus_{u \in U} \bigwith_{x \in X_u} \iota_{\oplus\with}(B_{x})\right]}{\bigoplus_{v \in V} \bigwith_{y \in Y_v} \iota_{\oplus\with}(C_{y})}}
\\
&=& {\displaystyle \varHom{\cC_{\oplus\with}}{\bigoplus_{(w,u) \in W \times U} \bigwith_{(z,x) \in Z_w \times X_u} \iota_{\oplus\with}(A_{z} \tensor B_{x})}{\bigoplus_{v \in V} \bigwith_{y \in Y_v} \iota_{\oplus\with}(C_{y})}}
\\
&\cong& {\displaystyle \prod_{w \in W} \prod_{u \in U} \sum_{v \in V} \prod_{y \in Y_v} \sum_{z \in Z_w} \sum_{x \in X_u} \varHom{\cC_{\oplus\with}}{\iota_{\oplus\with}(A_{z}) \tensor \iota_{\oplus\with}(B_{x})}{\iota_{\oplus\with}(C_{y})}} \\
&\cong& {\displaystyle \prod_{w \in W} \prod_{u \in U} \sum_{v \in V} \prod_{y \in Y_v} \sum_{z \in Z_w} \sum_{x \in X_u} \varHom{\cC_{\oplus\with}}{\iota_{\oplus\with}(A_{z})}{\iota_\with^\oplus(\iota_{\oplus\with}(B_{x}) \lin \iota_{\oplus\with}(C_{y}))}}
\\
&\cong& {\displaystyle \prod_{w \in W} \prod_{u \in U} \sum_{v \in V} \prod_{y \in Y_v} \sum_{z \in Z_w} \varHom{\cC_{\oplus\with}}{\iota_{\oplus\with}(A_{z})}{\bigoplus_{x \in X_u} \iota_{\with}^\oplus(\iota_{\oplus\with}(B_{x}) \lin \iota_{\oplus\with}(C_{y}))}}
&
\text{\raisebox{2em}{$(\dagger)$}}\\
&\cong& {\displaystyle \prod_{w \in W} \prod_{u \in U} \sum_{v \in V} \prod_{y \in Y_v} \varHom{\cC_{\oplus\with}}{\bigwith_{z \in Z_w} \iota_{\oplus\with}(A_{z})}{\bigoplus_{x \in X_u} \iota_{\with}^\oplus(\iota_{\oplus\with}(B_{x}) \lin \iota_{\oplus\with}(C_{y}))}}
&
\text{\raisebox{2em}{$(\heartsuit)$}}\\
&\cong& {\displaystyle \prod_{w \in W} \prod_{u \in U} \sum_{v \in V} \varHom{\cC_{\oplus\with}}{\bigwith_{z \in Z_w} \iota_{\oplus\with}(A_{z})}{\bigwith_{y \in Y_v} \bigoplus_{x \in X_u} \iota_{\with}^\oplus(\iota_{\oplus\with}(B_{x}) \lin \iota_{\oplus\with}(C_{y}))}}\\
&\cong& {\displaystyle \prod_{w \in W} \prod_{u \in U} \varHom{\cC_{\oplus\with}}{\bigwith_{z \in Z_w} \iota_{\oplus\with}(A_{z})}{\bigoplus_{v \in V} \bigwith_{y \in Y_v} \bigoplus_{x \in X_u} \iota_\with^\oplus(\iota_{\oplus\with}(B_{x}) \lin \iota_{\oplus\with}(C_{y}))}}
&
\text{\raisebox{2em}{$(\dagger)$}}\\
&\cong& {\displaystyle \prod_{w \in W} \varHom{\cC_{\oplus\with}}{\bigwith_{z \in Z_w} \iota_{\oplus\with}(A_{z})}{B\lin C}}\\
&\cong& \Hom{\cC_{\oplus\with}}{A}{B \lin C} \\
\end{array}
\\\\
\hline
\end{array}
\]
\caption{Monoidal closure of $\cC_{\oplus\with}$ (Theorem~\ref{thm:dial-haslin}).}
\label{fig:dial-haslin}
\end{figure}

Let us now sketch how the proof of Theorem~\ref{thm:dial-haslin} goes.
To this end, assume that we have
$A = \bigoplus_{u \in U} \bigwith_{x \in X_u} \iota_{\oplus\with}(A_{x})$ and
$B = \bigoplus_{v \in V} \bigwith_{y \in Y_v} \iota_{\oplus\with}(B_{y})$ be objects of $\cC_{\oplus\with}$
and assume that we have internal homsets $\iota_{\oplus\with}(A_{x}) \lin \iota_{\oplus\with}(B_{y})$
for every $(u,x) \in \sum_{u \in U} X_u$ and $(v, y) \in \sum_{v \in V} Y_v$.
Let
\[\iota_\with^\oplus : \cC_\oplus \to \cC_{\oplus\with} \quad\text{such that}\quad  \iota_{\oplus\with} = \iota_\with^\oplus \circ \iota_{\oplus} \]
be the full and faithful embedding be obtained by applying the universal
property of $\cC_\oplus$ to the coproduct-preserving functor
$\iota_{\oplus\with}$.
The linear arrow $A \lin B$ can then be defined as follows
\[A \lin B \quad = \quad \bigwith_{u \in U} \bigoplus_{v \in V} \bigwith_{y \in Y_v} \bigoplus_{x \in X_u} \iota_{\with}^\oplus(\iota_{\oplus}(A_{x}) \lin \iota_{\oplus}(B_{y}))\]
It is clear that this is meant to mimick the definition of $\Hom{\cC_{\oplus\with}}{-}{-}$
by substituting external hom, products and coproducts with internal ones, as is customary with
Dialectica. To show that $B \lin -$ is right adjoint to $- \tensor B$, one may directly
establish that we have an isomorphism of homsets
\[ \Hom{\cC_{\oplus\with}}{A \tensor B}{C} \cong \Hom{\cC_{\oplus\with}}{A}{B \lin C} \]
\emph{natural in $A$} using the recipe of Figure~\ref{fig:dial-haslin}.
Unmarked steps follow from definitions or standard abstract nonsense (characterization of
internal (co)products in terms of external ones).
The two $(\dagger)$ steps involve instances of the isomorphism
\[\sum_{i \in I} \varHom{\cC_{\oplus\with}}{\bigwith_{x \in X}\iota_{\oplus\with}(Z_x)}{D_i} ~~\cong~~
\varHom{\cC_{\oplus\with}}{\bigwith_{x \in X}\iota_{\oplus\with}(Z_x)}{\bigoplus_{i \in I} D_i}\]
which holds for any family $(Z_x)_{x \in X}$ of objects of $\cC$ and family $(D_i)_{i \in I}$ of objects of $\cC_{\oplus\with}$.
This is a formal computation corresponding to the tail end of Proposition~\ref{prop:yondistr},
which applies because of the equivalence $\cC_{\oplus\with} \cong \left(\cC_\with\right)_\oplus$.

The most subtle step, labelled with $(\heartsuit)$, necessitates a more detailed computation that relies
crucially on the fact that $\iota_{\with}^\oplus(\iota_{\oplus}(A_{x}) \lin \iota_{\oplus}(B_{y}))$ is
$\with$-free. This can be regarded as a generic isomorphism for arbitrary families
of objects $(E_u)_{u \in U}$ and $(Z_x)_{x \in X}$ of $\cC$
\[ \Hom{\cC_{\oplus\with}}{\bigwith_{x \in X} \iota_{\oplus\with}(Z_x)}{\bigoplus_{u \in U}\iota_{\oplus\with}(E_u)} \quad\cong\quad \sum_{x \in X} \Hom{\cC_{\oplus\with}}{Z_x}{\bigoplus_{u \in U} \iota_{\oplus\with}(E_u)} \]
which is easy to compute explicitly.

In addition to checking that there is a chain of isomorphisms as per
Figure~\ref{fig:dial-haslin}, it should also be checked that they are natural
in $A$. This requires identifying precisely what presheaves over $\cC_{\oplus\with}$
are involved at every step and arguing that the maps are indeed natural.
This is not necessarily so straightforward, so we offer a more detailed proof of
Theorem~\ref{thm:dial-haslin} in Appendix~\ref{sec:app-dial}.

\subsection{Proof of the main result on strings}
\label{subsec:main-string}

We can now give the proof of Theorem~\ref{thm:main-string}, which can be summarized as the equality
\[\text{$\laml$-definable} ~~=~~ \text{$\SR$-SSTs}\]
thanks to Fact~\ref{fact:register-sst-sst}. We start with the consequences of
our syntactic analysis
\[ \text{$\laml$-definable}
  \;\underset{\overset{\uparrow}{\text{Lemma~\ref{lem:lamlsst}}}}{=}\;
  \text{single-state $\mfLam$-SSTs}
  \;\underset{\overset{\uparrow}{\text{Lemma~\ref{lem:coprod-whatever}}}}{=}\;
  \text{$\mfLam$-SSTs}
\]
reducing our goal to
\[ \text{$\mfLam$-SSTs} ~~=~~ \text{$\SR$-SSTs} \]
For the above equality, the right-to-left inclusion is simpler than its
converse: the existence of a morphism of streaming settings from $\SR$ to
$\mfLam$ (Lemma~\ref{lem:register-to-laml}) entails that $\mfLam$-SSTs subsume
$\SR$-SSTs (by Lemma~\ref{lem:morph}). (Were we only interested only proving
that regular functions are $\laml$-definable, our setting would be a complete
overkill.)

On the other hand, the other direction mobilizes almost the whole development.
First, our semantic evaluation argument combines Lemma~\ref{lem:laml-initial}
with Corollary~\ref{cor:register-dial-closed} to get
\[ \text{$\mfLam$-SSTs} ~~\subseteq~~ \text{$\SR_{\oplus\with}$-SSTs} \]
We then finish proving Theorem~\ref{thm:main-string} with automata-theoretic
considerations:
\[ \text{$\SR_{\oplus\with}$-SSTs}
  \;\underset{\overset{\uparrow}{\text{Lemma~\ref{lem:sdmSST-oplus}}}}{=}\;
  \text{sdm-$\SR_{\with}$-SSTs}
  \;\underset{\overset{\uparrow}{\text{Theorem~\ref{thm:with-reg-conservativity}}}}{=}\;
  \text{sdm-$\SR$-SSTs}
\]

\[
  \text{sdm-$\SR$-SSTs}
  \;\underset{\overset{\uparrow}{\text{Corollary~\ref{cor:oplus-sst-reg-conservative}}}}{=}\;
  \text{$\SR_\oplus$-SSTs}
  \;\underset{\overset{\uparrow}{\text{Lemma~\ref{lem:coprod-whatever}}}}{=}\;
  \text{$\SR$-SSTs}
 \]

\section{Some transducer-theoretic applications of $\mfC$-SSTs and internal homsets}

This section is devoted to showing that the notion of monoidal closure can be used
to give a satisfying self-contained description of two important transformations of
usual streaming string transducers, both of which are not entirely straightforward:
the composition of two copyless SSTs and the uniformization of non-deterministic
copyless SSTs.

We take advantage of our abstract setting to give proofs for categorical generalizations of
those two statements, similarly in spirit to the generalized minimization argument
found in~\cite{ColcombetPetrisan}. The specialized theorem then follows easily from
previous theorems in our developments having to do with monoidal closure, especially
Theorem~\ref{thm:dial-haslin} and Lemma~\ref{lem:hom-reg-oplus}.

The ideas behind our arguments are not new; our main goal is to vindicate the view
that the notion of internal homset is the key to showing those results, even if it does
not appear explicitly in previous arguments.

\subsection{On closure under precomposition by regular functions}
\label{subsec:composition}

First, let us recall right off the bat that the closure under composition
of functions definable by copyless streaming string transducers follows
from our main result on strings (Theorem~\ref{thm:main-string}) together with basic considerations
on typed $\lambda$-calculi (see e.g.~\cite[Lemma~2.8]{aperiodic} that applies
\textit{mutatis mutandis} to the $\laml$-calculus) that entail the closure under
composition of $\laml$-definable functions. However, note that Theorem~\ref{thm:main-string}
relies on the syntactic analysis of Lemma~\ref{lem:laml-niceshape}, which is arguably a non-trivial
result about the $\laml$-calculus (as demonstrated by the size of
Appendices~\ref{sec:laml-normalization} and~\ref{sec:laml-niceshape}, which are
both necessary to prove it).

The argument we give in this section circumvents that difficulty and does not
appeal to Theorem~\ref{thm:main-string} nor to Lemma~\ref{lem:laml-niceshape}.
That said, it still shares some (non-syntactic) ingredients with our proof of
Theorem~\ref{thm:main-string}, namely:
\begin{itemize}
\item the monoidal closure and quasi-affineness (see below) of $\Sr_{\oplus\with}$;
\item the fact that $\SR_{\oplus\with}$-SSTs are no more expressive than $\SR$-SSTs.
\end{itemize}
These results still require substantial developments -- indeed, this composition
property is quite non-trivial as mentioned in the introduction -- but bypass the need of mentioning
the $\laml$-calculus.

Beyond this simplification, the main advantage of our approach here is that we
get a more general theorem, that applies to many streaming settings; in
particular, the final output does not have to be a string.
Without further ado, let us state it.

\begin{thm}
  \label{thm:precomp-string}
  Let $\mfC$ be a string streaming setting with output set $X$. Suppose that the
  underlying category $\cC$ is \emph{symmetric monoidal closed} and \emph{quasi-affine}.
  Furthermore, let us assume that $\initty_\mfC$ is equal to the monoidal unit
  $\unit$.
  
  Then for any $f : \Gamma^* \to X$ computed by some $\mfC$-SST, and any regular
  $g : \Sigma^* \to \Gamma^*$, the function $f \circ g : \Sigma^* \to X$ is
  computed by some (stateful) $\mfC$-SST.
  In other words, the class of functions defined by
  $\mfC$-SSTs is \emph{closed under precomposition by regular functions}.
\end{thm}

Before proving the above theorem, let us check that it entails known
preservation and composition properties.

\begin{cor}
  Let $L \subseteq \Gamma^*$ be a regular language and $g : \Sigma^* \to
  \Gamma^*$ be a regular function. Then the language $g^{-1}(L) \subseteq
  \Sigma^*$ is regular.
\end{cor}
\begin{proof}
  That $L$ is regular is equivalent to its indicator function $\chi_L : \Gamma^*
  \to \{0,1\}$ being computed by some single-state $\mathfrak{Finset}_2$-SST,
  see Example~\ref{exa:finset}. The underlying category of finite sets is
  \emph{cartesian closed}, and the monoidal structure given by a cartesian
  product is automatically symmetric and affine. According to
  Theorem~\ref{thm:precomp-string}, $\chi_L \circ g$ can therefore be computed
  by some $\mathfrak{Finset}_2$-SST. Observing that the category of finite sets
  has coproducts, and applying Lemma~\ref{lem:coprod-whatever}, we even have a
  \emph{single-state} $\mathfrak{Finset}_2$-SST for $\chi_L \circ g$. Finally,
  the latter is none other than the indicator function of $g^{-1}(L)$.
\end{proof}

\begin{cor}
  Let $f : \Gamma^* \to \Delta^*$ and $g : \Sigma^* \to \Gamma^*$ be regular
  functions. Then $f \circ g$ is also a regular function.
\end{cor}
\begin{proof}
  This is just the application of Theorem~\ref{thm:precomp-string} to
  $\SR_{\oplus\with}$-SSTs -- indeed, we saw at the very end of
  \Cref{sec:strings} that the functions computed by $\SR_{\oplus\with}$-SSTs are
  exactly the regular functions. By Theorem~\ref{thm:smcc-string} /
  Corollary~\ref{cor:register-dial-closed}, the underlying category
  $\Sr_{\oplus\with}$ is symmetric monoidal closed. Finally, $\Sr_{\oplus\with}$
  is quasi-affine since it has all cartesian products by construction.
\end{proof}

We now come to the proof of this generalized preservation theorem.

\begin{proof}[Proof of Theorem~\ref{thm:precomp-string}]
  We give below a proof assuming that $f$ is defined by some \emph{single-state}
  $\mfC$-SST (but beware: the $\mfC$-SST computing $f \circ g$ will still be
  stateful!).  The general case can be applied by considering the streaming
  setting $\mfC_{\oplus\mathrm{const}}$-SST, as was briefly mentioned in
  Remark~\ref{rem:sst-comparison}, and using the fact that single-state
  $\mfC_{\oplus\mathrm{const}}$-SSTs, stateful
  $\mfC_{\oplus\mathrm{const}}$-SSTs and stateful $\mfC$-SSTs are equally
  expressive. We leave it to the reader to check that the symmetric monoidal
  structure and the cartesian products in $\cC$ can be lifted to
  $\cC_{\oplus\mathrm{const}}$, making this generalization possible.

  Therefore, we may assume without loss of generality that $f$ is computed by a single-state $\mfC$-SST
 $T_f = (\{\bullet\}, A_f,\delta_f,i_f,o_f)$ where $A_f$ is an object of $\cC$ and
  \[ \delta_f : \Gamma \to \Hom{\cC}{A_f}{A_f} \qquad i_f \in
    \Hom{\cC}{\initty}{A_f} \qquad o_f \in \Hom{\cC}{A_f}{\retty} \]
  Let $T_g = (Q, q_0, R_g, \delta_g, i_g, o_g)$ be an usual copyless SST (i.e., a $\Sr$-SST)
  computing the regular function $g$, where $Q$ and $R_g$ are finite sets and
  \[ q_0 \in Q \qquad \delta_g : \Sigma \times Q \to Q \times [R_g \to_\Sr
    R_g] \]
  \[ i_g \in [\varnothing \to_\Sr R_g] \qquad o_g : Q \to [R_g
    \to_\Sr \{\bullet\}] \]
  We will write $A = A_f$ and $R = R_g$ for short.

  We want to build from this data a $\mfC$-SST $T$ defining $f \circ g$. Since
  $\cC$ is quasi-affine and symmetric monoidal closed, we can apply
  Corollary~\ref{cor:sr-to-smcc} to the object $A \lin A$ and to a family of
  morphisms $(\widetilde{\delta_f}(c))_{c \in \Gamma} \in \Hom{\cC}{A \lin
    A}{\;A \lin A}^\Gamma$ that will be defined later. This gives us a functor
  $F_{\delta_f} : \Sr \to \cC$, enjoying various properties that will be
  progressively recalled, which is at the heart of our construction.

  The set of states of our new $\mfC$-SST $T$ is $Q$, with initial state $q_0$,
  and its memory object is $F_{\delta_f}(R)$. The initialization morphism is
  defined as $i = F_{\delta_f}(i_g) \in \Hom{\cC}{\unit}{F_{\delta_f}(R)}$ -- we
  use the assumption $\initty = \unit$, and the fact that
  $F_{\delta_f}(\varnothing) = \unit$ (by Corollary~\ref{cor:sr-to-smcc}). The
  transition function is
\[
\begin{array}{c@{}c@{}c@{}c@{}c@{}c@{}c@{}c}
\delta : 
& \quad\Sigma & \;\times
& \;Q & \quad\longrightarrow\quad
& Q &\times
& \quad\Hom{\cC}{F_{\delta_f}(R)}{F_{\delta_f}(R)}
\\ \\
 &\quad c & \;, & q &\quad\mapsto\quad &\quad\pi_1(\delta_g(c,q)) &,& F_{\delta_f}(\pi_2(\delta_g(c,q)))
\end{array}
\]
  Finally, using $j_f : F_{\delta_f}(\{\bullet\}) \to A$ to be defined later,
  we take as our new output function
  \[ o : q \in Q \mapsto o_f \circ j_f \circ F_{\delta_f}(o_g(q))
    \in \Hom{\cC}{F_{\delta_f}(R)}{\Bot} \]

  Let us now sketch the verification that this defines the intended function $f
  \circ g : \Sigma^* \to X$. In the process, we will fill the missing
  definitions to make everything work out.

  Let $w = w_1 \ldots w_n \in \Sigma^*$ be an input string. The sequence
  $q_0,\ldots,q_n \in Q$ of states visited by both $T_g$ and $T$ when fed this
  input is obtained by $q_{i+1} = \pi_1(\delta_g(w_i,q_i))$ from the initial
  $q_0$. By definition of the output of $T_g$, we have:
  \[ \widehat{g(w)} \quad=\quad o_g(q_n) \circ \pi_2(\delta_g(w_n,q_{n-1}))
    \circ \dots \circ \pi_2(\delta_g(w_1,q_{0})) \circ i_g \]
  where $\widehat{g(w)} \in [\varnothing \to_{\Sr(\Gamma)} \{\bullet\}]$
  corresponds to $g(w) \in \Gamma^*$ (cf.\ Theorem~\ref{thm:functor-from-sr})
  and the `$\circ$' denotes a composition of register transitions (i.e.\ of
  $\Sr$-morphisms). Similarly, the output $T(w)$ of the $\mfC$-SST $T$ that we
  built on the input $w$ is defined as
  \[ T(w) \quad=\quad \curlyinterp{o(q_n) \circ
      F_{\delta_f}(\pi_2(\delta_g(w_n,q_{n-1}))) \circ \dots \circ
      F_{\delta_f}(\pi_2(\delta_g(w_1,q_{0}))) \circ i} \]
  which, by unfolding the definitions of $o$ and $i$, applying the functoriality
  of $F_{\delta_f}$ and comparing with the previous equality, one can simplify
  into
  \[ T(w) \quad=\quad \curlyinterp{o_f \circ j_f \circ
      F_{\delta_f}\!\left(\widehat{g(w)}\right)}
  \]

  It is now time to define $j_f \in \Hom{\cC}{F_{\delta_f}(\{\bullet\})}{A}$.
  To do so, let us first introduce
  \[ \mathsf{appto}(\varphi) \;:\; B \lin C \xrightarrow{\;\sim\;} (B \lin C)
    \tensor \unit \xrightarrow{\;\id \otimes \varphi\;} (B \lin C) \tensor B
    \xrightarrow{\;\ev\;} C \]
  where the last arrow is the evaluation map $\ev_{B,C}$, for any
  $B,C\in\Obj(\cC)$ and $\varphi : \unit \to B$. An useful property, whose
  verification we leave to the reader, is
  \[ \mathsf{appto}(\varphi) \circ \Lambda'(\psi) \quad=\quad \psi \circ \varphi
    \qquad\text{for any}\ \psi : B \to C\]
  where $\Lambda' : \Hom{\cC}{B}{C} \xrightarrow{\sim} \Hom{\cC}{\unit}{\,B \lin
    C}$ is defined in Proposition~\ref{prop:internal-endo-monoid}. We then take
  \[ j_f \;:\; F_{\delta_f}(\{\bullet\}) \xrightarrow{\;\pi_1\;} (A \lin A) \lin
    (A \lin A) \xrightarrow{\;\mathsf{appto}(\Lambda'(\id_A))} A \lin A
    \xrightarrow{\;\mathsf{appto}(i_f)\;} A \]
  where $\pi_1$ is the left projection from $F_{\delta_f}(\{\bullet\}) = ((A
  \lin A) \lin (A \lin A)) \with\unit$ (this equality is guaranteed by
  Corollary~\ref{cor:sr-to-smcc}) and $i_f$ is the initialization morphism of
  $T_f$. Using the equation
  \[ F_{\delta_f}\!\left( \widehat{g(w)} \right) \quad=\quad
    \tuple{\Lambda'\left(\widetilde{\delta_f}(g(w)_1) \circ \dots
        \circ \widetilde{\delta_f}(g(w)_m)\right),\;\id_\unit}\quad\text{where}\
    m = \len{g(w)} \]
  coming from Corollary~\ref{cor:sr-to-smcc}, we then have
  \[ j_f \circ F_{\delta_f}\!\left( \widehat{g(w)} \right) \quad=\quad
    \mathsf{appto}(i_f) \circ \widetilde{\delta_f}(g(w)_1) \circ \dots
    \circ \widetilde{\delta_f}(g(w)_m) \circ \Lambda'(\id_A)\]
  
  Next, we define $\widetilde{\delta_f}(c) = \Lambda(\ev_{A,A} \circ (\id_{A
    \lin A} \otimes \delta_f(c))) \in \Hom{\cC}{A \lin A}{\;A \lin A}$ for $c
  \in \Gamma$. In other words, $\widetilde{\delta_f}(c)$ is the curryfication of
  \[ (A \lin A) \tensor A \xrightarrow{\;\id \otimes \delta_f(c)\;} (A \lin A)
    \tensor A \xrightarrow{\;\ev\;} A \]
  One can then check that $\widetilde{\delta_f}(c) \circ \Lambda'(\psi) =
  \Lambda'(\psi\circ(\delta_f(c)))$ for any $\psi : A \to A$. Putting everything
  together, we finally have
  \begin{align*}
    T(w) &= \curlyinterp{o_f \circ \mathsf{appto}(i_f) \circ \widetilde{\delta_f}(g(w)_1) \circ \dots \circ \widetilde{\delta_f}(g(w)_m) \circ \Lambda'(\id_A)}\\
         &= \curlyinterp{o_f \circ \mathsf{appto}(i_f) \circ \Lambda'(\delta_f(g(w)_m) \circ \dots \circ \delta_f(g(w)_1))}\\
         &= \curlyinterp{o_f \circ \delta_f(g(w)_m) \circ \dots \circ \delta_f(g(w)_1) \circ i_f}
  \end{align*}
  and this final expression is precisely the definition of the output of $T_f$
  on $g(w)$. Since $T_f$ computes $f$, we end up with $T(w) = f(g(w))$, as we
  wanted.
\end{proof}

To conclude this section, let us note that an analogous result
for precomposition by regular tree functions can be shown by leveraging
the results of Section~\ref{sec:trees}; we leave it as an exercise. An important
subtlety: since the presence of the additive conjunction is important to compute
regular tree functions (as we stressed in the introduction), one must consider
tree streaming settings whose underlying categories \emph{have finite cartesian
  products} (which entails quasi-affineness).

\subsection{Uniformization through monoidal closure}
\label{subsec:uniformization}

We recall below the categorical uniformization theorem that we used in
\Cref{subsec:with-string} and provide its proof.
\myuniformization*

Recall that according to Lemma~\ref{lem:with-uniformization}, the conclusion
amounts to saying that non-deterministic sdm-$\mfC$-SSTs are uniformizable by
partial sdm-$\mfD$-SSTs.

\begin{proof}We do not prove the statement in excruciating details, but provide key formal definitions
so that a reader familiar with a modicum of automata theory and category theory should be able
to reconstitute a fully formal argument with ease. Let us stress once again that
all of the combinatorics may be regarded as adaptation of known arguments.

The argument is based on a
notion of \emph{transformation forest}, a name that we borrow from~\cite[Chapter~13]{Toolbox} for an extremely
similar concept\footnote{There are two formal differences between our notions, which
are not very big but worth mentioning for readers of~\cite{Toolbox}.
First, edges of a transformation forest are intended to be associated
with (elements of) a monoid, while ours should be associated with (``elements of'') internal homsets
$F(C_u) \lin F(C_v)$. Were we trying to prove $\mfD$-uniformization for $\mfC$-SSTs, we
would have necessarily $C_u = C_v$ and the aforementioned object would have a monoid structure
internal to $\cD$, so this distinction is more an artefact of our settings rather than an essential one.
Second, what~\cite{Toolbox} calls transformation forests refers to a class of forest with ``no junk'', such as
dangling leaves not referring to an intended output or spurious internal nodes, while we allow those in
an initial definition; we add the adjective ``normalized'' for those containing ``no junk'' as we shall
see later, so this is merely a terminological detail.}.
This gadget is also reminiscent of trees used in determinization procedures like
the Muller-Schupp construction for automata over
$\omega$-words~\cite{muller_schupp} (another exposition can be found in~\cite[Chapter~1]{Toolbox}),
and of the sharing techniques used in the original paper on SSTs~\cite[\S5.2]{SST}.
In determinization procedures, this constitutes an elaboration of powerset constructions
recalling not only reachable states, but also crucial information on \emph{how} those states are reached.
Here, the vertices $v$ of such forests will be labelled by objects $C_v$ of $\cC$ and each edge $(u,v)$
will be correspond to a ``register containing a value of type $F(C_u) \lin F(C_v)$''.

We decompose this proof sketch in three parts: first, we introduce transformation forests,
their semantic interpretation as families of maps in $\cD$; we explain how they may be composed
and that maps in $\cC_\with$ may be regarded as depth-1 transformation forests.
Then, we explain how to reduce the size of transformation forests in a sound way (this is the crucial
part ensuring that the resulting sdm-$\mfD$-SSTs will have finitely many states). Finally, we explain how
to put all of this together to uniformize sdm-$\mfC$-SSTs.

\subsubsection{Transformation forests and their semantics}
A transformation forest is defined as a tuple $\cF = (V, E, O, (C_v)_{v \in V})$ where
\begin{itemize}
\item $V$ is a non-empty finite set of vertices
\item $E \subseteq V^2$ is a set of directed edges, pointing from parents to children
\item $O$ is a non-empty subset of $V$
which we cal the set of \emph{output nodes}
\item every $C_v$ is an object of $\cC$
\end{itemize}

When a transformation forest $\cF$ is fixed, we call $I_\cF$ its set of roots (which we may sometimes \emph{input nodes}; we drop the subscript
when there is no ambiguity).
Given a transformation forest $\cF = (V, E, (v_o)_{o \in O}, (C_v)_{v\in V})$, we assign the following object of $\cD$:\[\Ty(\cF) \quad = \quad \bigtensor_{(u,v) \in E} F(C_u) \lin F(C_v)\]
An example of a transformation forest $\cF$ and a computation of its type $\Ty(\cF)$ is pictured in Figure~\ref{fig:trans-forest-ex}.

\begin{figure}
\center
\includegraphics{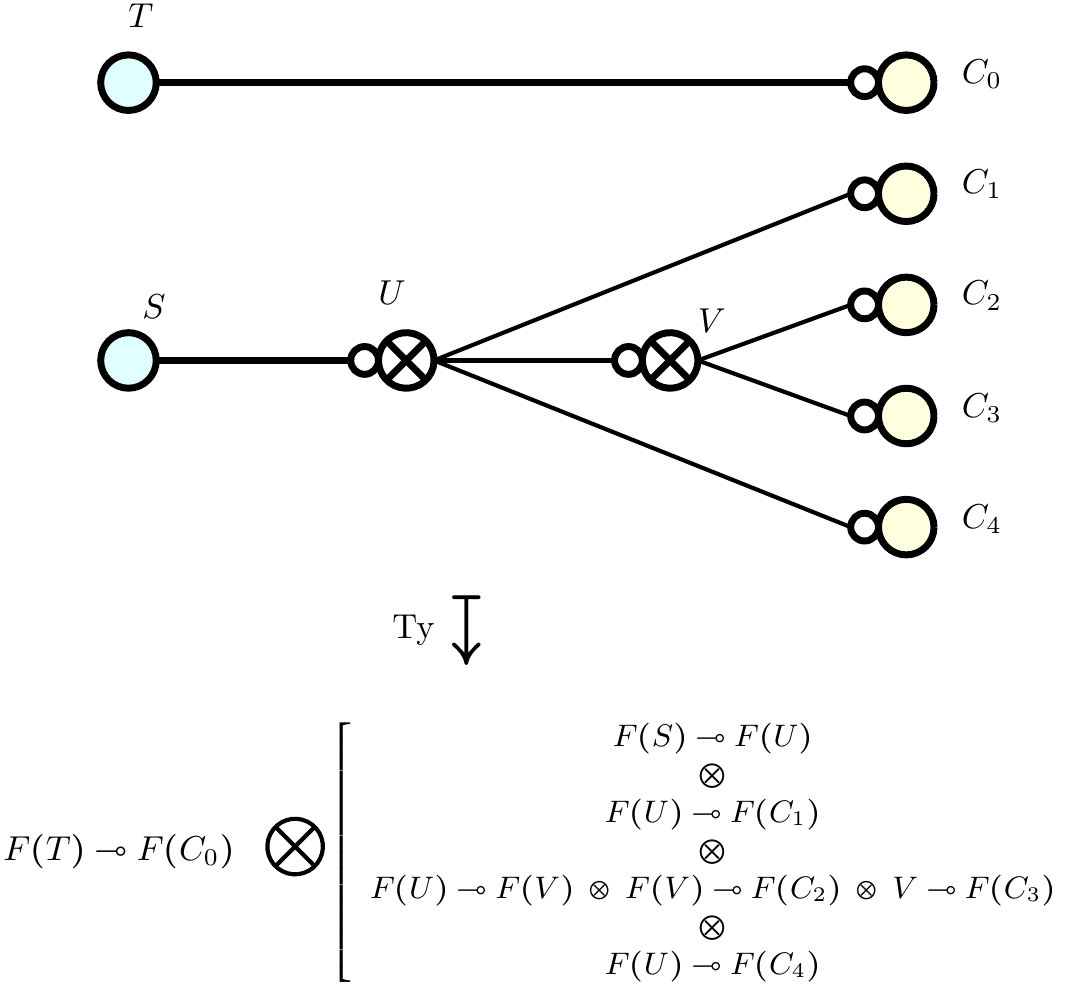}
\vspace{2em}
\hrule
\caption{A transformation forest $\cF : T \with S \to \displaystyle\bigwith_{i =
    0}^4 C_i$\; ($S,T,U,V,C_0,\ldots,C_4 \in \Obj(\cC)$).}
\label{fig:trans-forest-ex}
\end{figure}

To guide the intuition, one may note that there is an embedding\footnote{This becomes an isomorphism when $\Hom{\cD}{\top}{A \tensor B} \cong \Hom{\cD}{\top}{A} \times \Hom{\cD}{\top}{B}$;
this is the case for our intended application $\cD = \Sr_\oplus$.}
of a set of suitable labellings of the forest $\cF$ into $\Hom{\cD}{\top}{\Ty(\cF)}$.
\[
\xymatrix@C=15mm{
\prod_{(u,v) \in E} \Hom{\cD}{F(C_u)}{F(C_v)} \quad \ar[r] & \quad \Hom{\cD}{\top}{\Ty(\cF)}}
\]

We will now call abusively the input of $\cF$ the object $A = \bigwith_{i \in I} C_i$ and the output $B = \bigwith_{o \in O} C_o$; we write $\cF : A \to B$ in the sequel.
The justification for this notation is that there is a family of maps
\[\interp{\cF} \; \in \; \prod_{o \in O} \sum_{i \in I} \Hom{\cD}{\Ty(\cF)}{F(C_i) \lin F(C_{o})}\]
obtained by internalizing the composition of morphisms along branches of $\cF$, which we call the \emph{semantics of $\cF$}.

We also note that this allow to interpret arbitrary $\cC_\with$ morphisms: such a morphism $f : A \to B$ can be mapped to a pair $(\forestify{f}, \forestifymorph{f})$
consisting of
\begin{itemize}
\item a depth-1 transformation forest $\forestify{f} : \bigwith_{i \in I} C_i \to \bigwith_{o \in O} C_o$ (an example is depicted in Figure~\ref{fig:trans-forest-ex1})
\item a morphism $\forestifymorph{f} \in \Hom{\cD}{\top}{\Ty(\cF)}$
\end{itemize}
so that, if $f = (i_o, \alpha_o)_{o \in O}$, we have $\interp{\forestify{f}}_o \circ \forestifymorph{f} = (i_o, \Lambda(\alpha_o \circ \rho^{-1}))$.

\begin{figure}
\center
\includegraphics{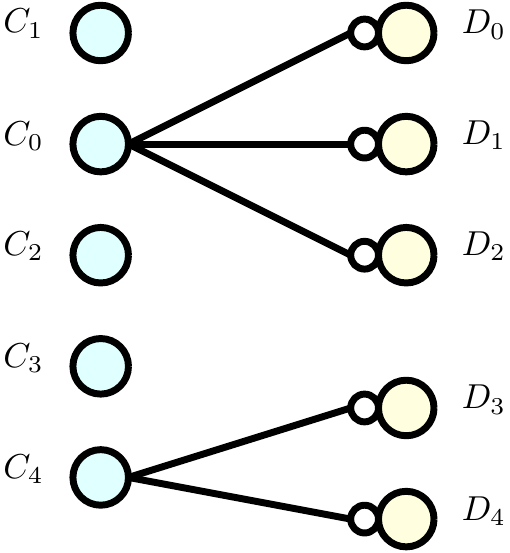}
\vspace{2em}
\hrule
\caption{A depth-$1$ transformation forest.}
\label{fig:trans-forest-ex1}
\end{figure}

Given two forests $\cF : A \to B$ and $\cF' : B \to C$, there is a composition $\cF' \circ \cF$
obtained by gluing the input nodes of $\cF'$ along the output nodes of $\cF$.
A crucial point is that the semantics of the composite $\interp{-}$ may be computed as follows\[\interp{\cF' \circ \cF}(o) \; = \; \left( \interp{\cF}_1(\interp{\cF'}_1(o')), \interp{\cF'}_2(o') \circ \interp{\cF}_2(\interp{\cF'}_1)\right)\]

\subsubsection{Reducing transformation forests}
We now introduce two elementary transformations $\cF \mapsto \cF'$ over transformation forests, together
with associated morphisms $\Ty(\cF) \to \Ty(\cF')$:
\begin{itemize}
\item {\bf Pruning}: if $v \in V_\cF$ is a leaf which is not an output node in the forest $\cF : A \to B$, call $\prune(\cF, v) : A \to B$ the forest obtained by
removing $v$ and adjacent edges from $\cF$.

This induces a canonical map
\[\prunemap(v) : \Ty(\cF) \to \Ty(\prune(\cF, v))\]
by using the weakening maps on the components corresponding to the deleted edges.
\item {\bf Contraction}: If there is a vertex $v$ with a unique child $c$ and a (unique) parent $p$ in the forest $\cF : A \to B$, call $\contract(\cF,v) : A \to B$
the forest obtained by replacing the edges $(p,v)$ and $(v,c)$ with a single edge $(p, c)$ and removing $v$.

There is a canonical map
\[\contractmap(v) : \Ty(\cF) \to \Ty(\contract(\cF, v))\]
induced by the internal composition map
\[\Hom{\cD}{[F(C_p) \lin F(C_v)] \tensor [F(C_v) \lin F(C_c)]}{F(C_p) \lin F(C_c)}\]
\end{itemize}

The auxiliary maps $\prunemap$ and $\contractmap$ operations are compatible with the semantic map $\interp{\cdot}$ in the sense
that for every $o \in O$ and suitable vertices $u,v$ of $\cF$ , we have 
\[i = \pi_1(\interp{\cF}(o)) = \pi_1(\interp{\prune(\cF,u)}(o)) = \pi_1(\interp{\contract(\cF,v)}(o))\]
and the following diagrams commute in $\cD$
\[
\xymatrix@C=15mm{
\Ty(\cF) \ar[r]^{\interp{\cF}(o)} \ar[d]_{\prunemap(v)} & F(C_i) \lin F(C_o) \\
\Ty(\prune(\cF,v)) \ar[ur]_{\qquad \interp{\prune(\cF,v)}(o)}
}
\xymatrix@C=15mm{
\Ty(\cF) \ar[r]^{\interp{\cF}(o)} \ar[d]_{\contractmap(v)} & F(C_i) \lin F(C_o) \\
\Ty(\contract(\cF,v)) \ar[ur]_{\qquad \interp{\contract(\cF,v)}(o)}
}
\]

Consider the rewrite system $\leadsto$ over forests $\cF : A \to B$ induced by those two operations, and call $\leadsto^=$
its reflexive closure and $\leadsto^*$ the transitive closure of $\leadsto^=$.
It can be shown that the reflexive closure $\leadsto^=$ satisfies the \emph{diamond lemma}, i.e.,
that for every $\cF, \cF'$ and $\cF''$ such that $\cF \leadsto^= \cF'$ and $\cF \leadsto^= \cF''$, there exists $\cF'''$ such that
$\cF' \leadsto^= \cF'''$ and $\cF'' \leadsto^= \cF'''$. This ensures that $\leadsto$ is confluent.
Furthermore, given a rewrite $\cF \leadsto^* \cF'$, there is a map $\Ty(\cF) \to \Ty(\cF')$
obtained by composing maps $\prunemap(v)$ and $\contractmap(v)$
(and identities for trivial rewrites $\cF \leadsto^* \cF$).
It can be shown that this map does \emph{not} depend on the rewrite path.
This is done first by arguing that if we have a rewrite square for $\leadsto^=$,
the associated diagram in $\cD$ is commutative, and then proceed by induction over the
rewrite paths using the diamond lemma.
\[
\xymatrix@C=20mm{
\cF \ar@{~>}[r] \ar@{~>}[d] & \cF' \ar@{~>}[d] \\
\cF'' \ar@{~>}[r] & \cF''' \\
}
\qquad\qquad
\begin{array}{c}
\\
\\
\\
\longmapsto
\end{array}
\qquad\qquad
\xymatrix@C=20mm{
\Ty(\cF) \ar@{->}[r] \ar@{->}[d] & \Ty(\cF') \ar@{->}[d] \\
\Ty(\cF'') \ar@{->}[r] & \Ty(\cF''') \\
}
\]

Defining the size of a forest as its number of vertices, it is clear that $\prune(\cF,u)$ and $\contract(\cF,v)$
have size strictly less than $\cF$, so the rewrite system is also terminating.
With confluence, it means that for every forest $\cF : A \to B$, there is a unique forest $\normal(\cF) : A \to B$
such that $\cF \leadsto^* \normal(\cF)$ and there is no $\cF'$ such that $\normal(\cF) \leadsto \cF'$.
We call $\normal(\cF)$ the \emph{normal form} of $\cF$; a forest $\cF$ is called \emph{normal} if $\normal(\cF) = \cF$.
By the discussion above, there are canonical maps $\normalmap_\cF : \Ty(\cF) \to \Ty(\normal(\cF))$ coming from rewrites.

\begin{figure}
\center
\includegraphics[scale=0.75]{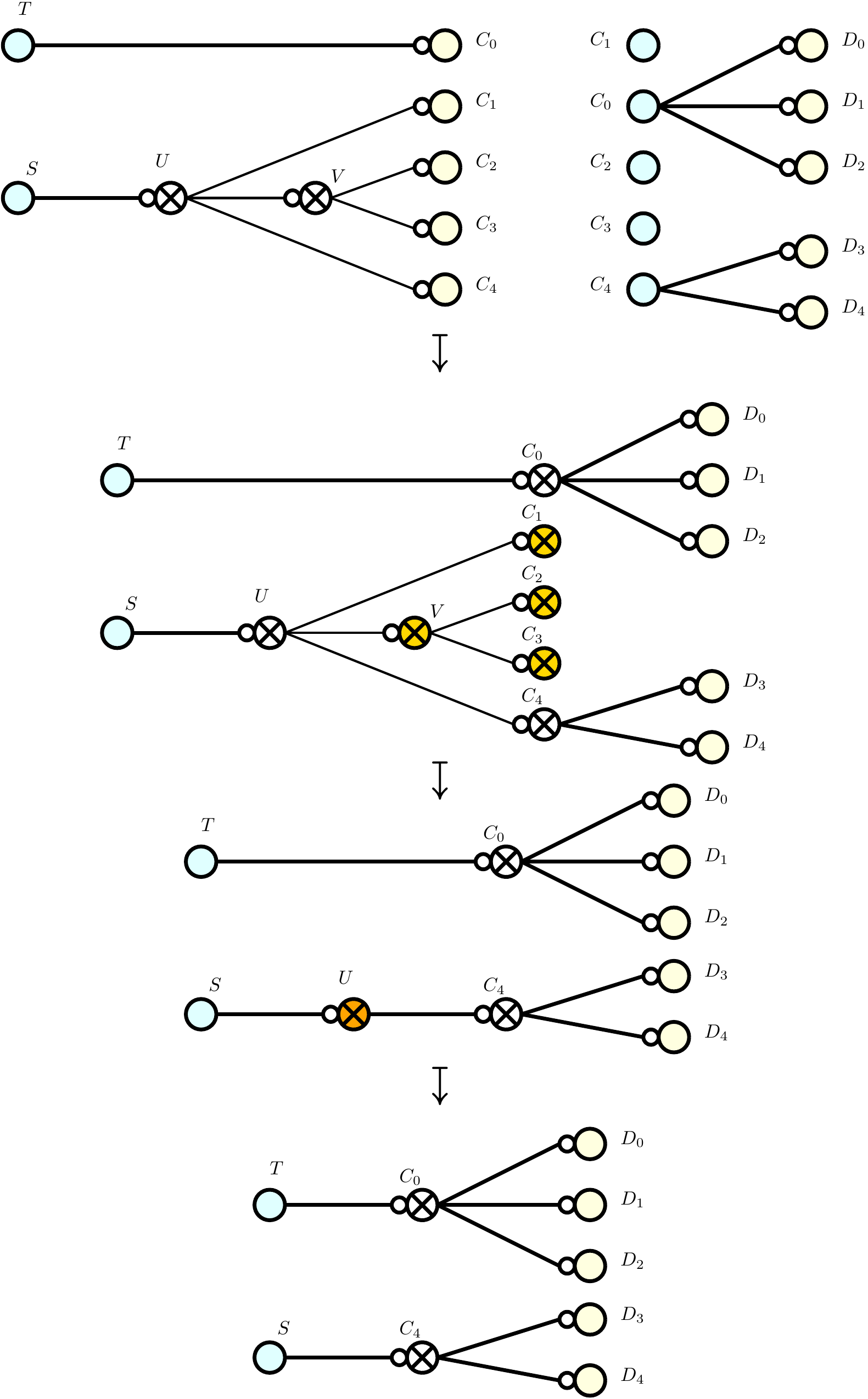}

\vspace{2em}
\hrule
\caption{An example of composing normal transformation forests, where first the usual composition and then two big steps of reduction, corresponding respectively to pruning and contracting, are carried out in succession.}
\label{fig:trans-forest-compose-ex}
\end{figure}

The last important thing to note is that, if the input $A = \bigwith_{i \in I}
C_i$ and output $B = \bigwith_{o \in O} C_o$ are fixed, up to isomorphism,
there are finitely many normal forests $\cF : A \to B$ as their size is bounded by $2(\card{I} + \card{O})$.
We write $\normalforests{A}{B}$ for a finite set of representative for all normal forests $A \to B$, and given $\cF \in \normalforests{B}{C}$ and $\cF' \in \normalforests{A}{B}$,
we write $\cF \circ_\cN \cF'$ for the unique forest in $\normalforests{A}{C}$ which is isomorphic to $\normal(\cF \circ \cF')$; an example of the full computation of a $\cF \circ_\cN \cF'$ is given in Figure~\ref{fig:trans-forest-compose-ex}.
Similarly, we assume that $\forestify{f} \in \normalforests{A}{B}$ for every morphism $f \in \Hom{\cC_\with}{A}{B}$.

\subsubsection{Putting everything together}
Let $\cT = \left(Q, q_0, \left(A_q\right)_{q \in Q}, \delta, i, o\right)$
be a sdm-$\mfC_\with$-SST with input $\Sigma^*$. We define the sdm-$\mfD$-SST
\[\cT' ~~ = ~~ \left( \sum_{q \in Q} \normalforests{A_{q_0}}{A_{q}}, (q_0, \forestify{\id}), (\Ty(\cF))_{q, \cF}, \delta', i' o' \right)\]
where
\[
\begin{array}{c}
\displaystyle \delta' ~~:~~ \Sigma \to \prod_{q, \cF} \sum_{r, \cF'} \Hom{\cD}{\Ty(\cF)}{\Ty(\cF')} \\ \\
\displaystyle i' \in \Hom{\cD}{\initty}{\forestify{\id}}
\qquad\qquad
o' \in \prod_{q, \cF} \Hom{\cD}{\Ty(\cF)}{\retty}
\end{array}
\]
are given as follows, assuming that $A_q = \displaystyle\bigwith_{x \in X_q} C_{q,x}$
\begin{figure}
\[
\begin{array}{|c@{\qquad}|@{\qquad}c|}
\hline
\delta^\cD(a)_{q,\cF} &
o'_{q,\cF} \\
\hline
& \\
\xymatrix{
\\
\Ty(\cF) \ar[d]^\sim \\
\top \tensor \Ty(\cF) \ar[d]^{\forestifymorph{\delta^{\cC_\with}(a)_q} \tensor \id} \\
\Ty(\forestify{\delta^{\cC_\with}(a)_q}) \tensor \Ty(\cF) \ar[d]^\sim \\
\Ty(\forestify{\delta^{\cC_\with}(a)_q} \circ \cF) \ar[d]^\normalmap \\
\Ty(\normal(\forestify{\delta^{\cC_\with}(a)_q} \circ \cF)) \ar@{=}[d] \\
\Ty(\delta^{\mathrm{NF}}(a)_{q,\cF})  \\
}&
\xymatrix{
\Ty(\cF) \ar[d]^{f} \\
{F(C_{q_0,x_0})} \lin {F(C_{q,x})}
\ar[d]^\sim \\
({F(C_{q_0,x_0})} \lin {F(C_{q,x})}) \tensor \initty
\ar[d] \\
({F(C_{q_0,x_0})} \lin {F(C_{q,x})}) \tensor F(\initty)
\ar[d]^{\id \tensor F(i_{x_0})} \\
({F(C_{q_0,x_0})} \lin {F(C_{q,x})}) \tensor F(C_{q_0,x_0})
\ar[d]^\ev \\
F(C_{q,x}) \ar[d]^{F(o^\cC)} \\
F(\retty) \ar[d] \\
\retty
}
\\
\hline
\end{array}
\]
\caption{Definition of $\delta^{\cD}(a)_{q,\cF}$
and $o'_{q,\cF}$.}
\label{fig:trans-output-uniformizer}
\end{figure}
\begin{itemize}
\item Notice that $\Ty(\forestify{\id}) = \bigtensor_{x \in X_{q_0}} F(C_{q_0, x}) \lin F(C_{q_0, x})$;
we simply take the constant map corresponding to the $X_{q_0}$-fold tensor of $\Lambda(\id_{C_{q_0,x}})$ for $i'$.
\item Fix $a \in \Sigma$ and recall that $\delta(a)$ is a family of pairs
\[
(\delta^Q(a)_q, \delta^{\cC_\with}(a)_q)_{q \in Q} ~~ \in ~~
\prod_{q \in Q} \sum_{r \in Q} \Hom{\cC_\with}{A_q}{A_r}\]
$\delta'$ is then defined by the equation 
\[\delta'(a)_{q, \cF} ~~=~~
\left(\left(
\delta^Q(a)_{q},
\delta^{\mathrm{NF}}(a)_{q, \cF}\right),
\delta^\cD(a)_{q, \cF}
\right)\]
where we set $\delta^{\mathrm{NF}}(a)_{q,\cF} = \forestify{\delta^{\cC_\with}(a)_q} \circ_\cN \cF$ and
$\delta^\cD(a)_{q, \cF}$ is obtained as in Figure~\ref{fig:trans-output-uniformizer}
\item Finally, for $q, \cF$ ranging over states of $\cT'$, we define $o'_{q, \cF}$.
Recall that $\cF$ determines a canonical family
\[\interp{\cF} \in \prod_{x \in X_q} \sum_{x_0 \in X_{q_0}} \Hom{\cD}{\Ty(\cF)}{F(C_{q_0,x_0}) \lin F(C_{q,x})}\]
First, as $\retty^{\mfC_\with} = \iota_\with(\retty^\mfC)$, note that
there is a unique $x \in X_q$ such that $o_q = (x, o^\cC)_*$ for some $o^\cC \in \Hom{\cC}{C_{q,x}}{\retty^\mfC}$.
Writing the pair $\interp{\cF}_{x}$ as $(x_0, f)$, $o'_{q,\cF}$ is depicted in Figure~\ref{fig:trans-output-uniformizer}.
\end{itemize}

We omit the proof that $\interp{\cT} = \interp{\cT'}$ by induction
over the length of an input word.
While spelling it out may be a bit notation-heavy, there is no particular
difficulty considering the remarks above linking $\interp{-}$, the composition
of transformation forests and their normal forms.
\end{proof}

\section{Regular tree functions in the $\laml$-calculus}
\label{sec:trees}

The goal of this section is now to prove our main theorem for tree functions
(that we recall below), extending the case of strings
(Theorem~\ref{thm:main-string}).

\maintree*

To prove this theorem, we follow a similar approach as in Section~\ref{sec:strings}.

\begin{rem}
While the result on strings follows as a corollary of Theorem~\ref{thm:main-tree},
this relies on the equivalence between regular string functions and regular tree functions when
strings are regarded as particular trees. Given that we start from transducer-based characterization
of these notions, this means that we would use e.g.~\cite[Theorem 3.16-3.17]{BRTT}, which are themselves
non-trivial results\footnote{A self-contained proof is also possible by building on our development
by exploiting Theorem~\ref{thm:with-reg-conservativity}, arguably the most intricate argument presented here.}.
\end{rem}

We thus first define a generalized notion of bottom-up ranked tree transducers (BRTT)
parameterized by a notion of \emph{tree streaming setting}.

\begin{defi}
Let $X$ be a set.
A \emph{tree streaming setting} with output $X$
is a tuple $\mathfrak{C} = (\cC, \tensor, \unit, \retty, \curlyinterp{-})$ where
\begin{itemize}
\item $(\cC, \tensor, \unit)$ is a \emph{symmetric monoidal} category
\item $\retty$ is an object of $\cC$
\item $\curlyinterp{-}$ is a set-theoretic map $\Hom{\cC}{\unit}{\retty} \to X$
\end{itemize}
\end{defi}

This notion essentially differs by asking that the underlying category be equipped
with a symmetric monoidal product, which is used in defining the semantics of
$\mfC$-BRTTs. The tensor product is used to fit the branching structure of trees
and $\unit$ is used for terminal nodes (so there is no need of a distinguished initial object
$\initty$ as in string streaming settings).

\begin{defi}
Let $\mathfrak{C} = (\cC, \tensor, \unit, \retty, \curlyinterp{-})$ be a tree streaming setting with output $X$.
A $\mathfrak{C}$-BRTT with input (ranked) alphabet $\bf \Sigma$ and output $X$ is a tuple
$(Q, R, \delta, o)$ where
\begin{itemize}
\item $Q$ is a finite set of states
\item $R$ is an object of $\cC$
\item $\delta$ is a function $\displaystyle \prod_{a \in \Sigma} \left( Q^{\arity(a)} \to Q \times \varHom{\cC}{\bigtensor_{\arity(a)} R}{R} \right)$
\item $o \in \Hom{\cC}{R}{\retty}$ is an output morphism
\end{itemize}
Its semantics is a set-theoretic map
$\Tree({\bf \Sigma}) \to X$ defined as follows:  writing $\delta_Q(a,(t_i)_{i \in \arity(a)})$ for
$\pi_1(\delta(a,(t_i)_{i \in \arity(a)}))$ and $\delta_\cC(a, (t_i)_{i \in \arity(a)})$  for $\pi_2(\delta(a,(t_i)_{i \in \arity(a)}))$, define auxiliary functions $\delta^*_Q : \Tree({\bf \Sigma}) \to Q$
and $\delta^*_\cC : \Tree({\bf \Sigma}) \to \Hom{\cC}{\unit}{R}$ by iterating $\delta$:
\[
\begin{array}{lcl}
\delta_Q^*\left(a\left((t_i)_{i \in \arity(a)}\right)\right) &~~=~~& \delta_Q\left(a,(\delta^*_Q(t_i))_{i \in \arity(a)}\right)
\\
\delta_\cC^*\left(a\left((t_i)_{i \in \arity(a)}\right)\right) &~~=~~& \displaystyle \delta_\cC\left(a,(\delta^*_Q(t_i))_{i \in \arity(a)}\right) \circ \left(\bigtensor_{i \in \arity(a)} \delta^*_\cC(t_i)\right) \circ \varphi_{\arity(a)}
\end{array}
\]
where $\varphi_{\arity(a)}$ is the unique isomorphism
$\unit \xrightarrow{\,\sim\,} \bigotimes_{i \in \arity(a)} \unit$
generated by the associator and unitors of $(\cC,\tensor,\unit)$.
The output function $\interp{\cT} : \Tree({\bf \Sigma}) \to X$ is defined as $\interp{\cT} = \curlyinterp{o \circ \delta^*_\cC}$.
\end{defi}

\begin{rem}
  \label{rem:brtt-sym}
Strictly speaking, we do not need the monoidal product to be symmetric for the notion
to make sense, but it would require using a fixed order over the input ranked alphabet $\bf \Sigma$.
Although one could choose an arbitrary total order over $\bf \Sigma$, different orders might define different classes of
functions $\Tree({\bf \Sigma})$ when the monoidal product is not symmetric. This
is why we work in symmetric monoidal categories.
\end{rem}

As with string streaming settings, it is convenient to define a notion of morphism
of tree streaming settings to compare the expressiveness of classes of BRTTs.

\begin{defi}
Let $\mfC = (\cC, \unit_\cC, \tensor_\cC, \retty_\cC, \curlyinterp{-}_\mfC)$
and
 $\mfD = (\cD, \unit_\cD, \tensor_\cC, \retty_\cD, \curlyinterp{-}_\mfD)$
be tree streaming settings with output $X$. A morphism of tree streaming settings is given by
a lax monoidal functor $F : (\cC, \tensor_\cC, \unit_\cC) \to (\cD,\tensor_\cD,\unit_\cD)$
and a $\cD$-arrow $o : F(\retty_\cC) \to \retty_\cD$
such that, for every $f \in \Hom{\cC}{\initty_\cC}{\retty_\cC}$, we have
\[\curlyinterp{o \circ F(f) \circ i}_\cD = \curlyinterp{f}_\cC\]
where $i : \unit_\cD \to F(\unit_\cC)$ is obtained as part of the lax monoidal functor structure over $F$.
\end{defi}

Observe that we do not require those functors to commute with the symmetry
morphisms for the monoidal products, as promised in \Cref{subsec:prelim-cat}.
This is consistent with the fact that the symmetries are not really involved in
computing the image of a tree by a $\mfC$-BRTT, according to
Remark~\ref{rem:brtt-sym}: it is only their mere existence that matters.

\begin{lem}
\label{lem:morph-BRTT}
If there is a morphism of tree streaming settings $\mfC \to \mfD$, then $\mfD$-BRTTs subsume $\mfC$-BRTTs.
\end{lem}
\begin{proof}
Let us give the proof for single-state BRTTs; the proof for general BRTTs would be notationally heavier
but not much more insightful.
Suppose that we have some single-state $\mfC$-BRTT $\cT = (R,\delta,o)$ -- where
we leave the only state implicit and regard the transition function $\delta$ and
output function $o$ as elements of
$\prod_{a \in \Sigma} \varHom{\cC}{\bigtensor_{\arity(a)} R}{R}$ and
$\Hom{\cC}{R}{\retty}$ respectively --
and that the morphism under consideration is composed of a lax monoidal functor $F : \cC \to \cD$
and a $\cD$-arrow $o' : F(\retty_{\cC}) \to \retty_\cD$.

Since $F$ is lax monoidal, we have a family of natural transformations
\[ m_{I,A} ~:~~~~ \bigtensor_{I} F(A) ~~\longto~~ F\left(\bigtensor_{I}
    A\right)\]
where $I$ ranges over all finite sets\footnote{Recall from
  \Cref{subsec:notations} that an operation $\bigotimes_{i \in I} (-)$ -- here,
  a functor -- is associated to every finite indexing set $I$ by choosing an
  arbitrary total order over $I$.} and $A$ over objects of $\cC$; this family is
compatible with the associators and unitors in $\cC$ and $\cD$. Furthermore,
$m_{\varnothing,A} : \unit \to F(\unit)$ is the same for all $A \in \Obj(C)$, so
we shall abbreviate it as $m_\varnothing$.

We claim that $\interp{\cT} = \interp{\cT'}$, where $\cT'$ is the single-state $\mfD$-BRTT $(F(R), \delta', o' \circ F(o))$ with the same typing conventions and
$\delta'_a = F(\delta_a) \circ m_{\arity(a), R}$.
Let us prove this. To do so we first consider the iterations of the transition functions
\[\delta^* : \Tree({\bf \Sigma}) \to \Hom{\cC}{\unit}{R} \qquad \text{and} \qquad \delta'^* : \Tree({\bf \Sigma}) \to \Hom{\cD}{\unit}{F(R)}\]
and show that $\delta'^*(t) = F(\delta(t)) \circ m_{\varnothing}$ by induction over $t \in \Tree({\bf \Sigma})$.
So suppose that $t = a((u_x)_{x \in \arity(a)})$ and that the inductive
hypothesis holds (this also takes care of the base case: when
$\arity(a) = \varnothing$, the inductive hypothesis is vacuous). In such a case, let
us show that each face in the following diagram commutes (where all tensor products have arity $\arity(a)$):
\[\xymatrix@C=20mm@R=10mm{
\unit \ar@{}[rddd]|-{\mathtt{(b)}} \ar@/_1pc/[rdddd]_-{m_{\varnothing}} \ar[dr]^\sim \ar@/^3pc/[rrrdd]^{\delta'^*(t)}
&
\ar@{}[dr]|{\mathtt{(a)}}
\\
&\bigtensor \unit \ar[d]_{\bigtensor m_{\varnothing}} \ar@/^0pc/[dr]^{\bigtensor\limits_x \delta'^*(u_x)}
& \\
&\bigtensor F(\unit) \ar@{}[ur]|<<<<<<{\mathtt{(c)}}
\ar[r]_{\bigtensor\limits_x F(\delta^*(u_x))} \ar[d]_{m_{\arity(a),\unit}} &
\bigtensor F(R) \ar[r]^{\delta'_a} \ar[d]_{m_{\arity(a),R}}
\ar@{}[dr]|<<<<<<{\mathtt{(e)}}
& F(R) \\
&F\left(\bigtensor\unit\right) \ar@{}[ur]|<<<<<<<<<{\mathtt{(d)}}
\ar[r]_{F\left(\bigtensor\limits_x \delta^*(u_x)\right)} & F\left(\bigtensor R\right) \ar[ur]_{F(\delta_a)} &
\\
&F(\unit) \ar@{}[urr]|{\mathtt{(f)}} \ar[u]_{F(\sim)} \ar@/_3pc/[rruu]_{F(\delta^*(t))}
}\]
Faces \texttt{(a)} and \texttt{(f)} commute by definition of iterated transition functions and
face \texttt{(e)} corresponds to the definition of $\delta'$. Face \texttt{(d)} commutes because
of the naturality of $m_{-,-}$ and face \texttt{(b)} because of its compatibility with associators.
Finally, face \texttt{(c)} corresponds to
the inductive hypothesis.
Therefore, the topmost and bottommost paths coincide, so we have $\delta'^*(t) = F(\delta(t)) \circ m_{\varnothing}$,
which concludes our inductive argument. We can then conclude since we have, for every tree $t$,
\[
\begin{array}{lcll}
\interp{\cT'}(t) &=& \curlyinterp{o' \circ F(o) \circ (\delta'^*(t))}_\mfD & \text{by definition} \\
&=& \curlyinterp{o' \circ F(o) \circ F(\delta^*(t)) \circ m_\varnothing}_\mfD & \text{inductive argument}\\
&=& \curlyinterp{o' \circ F(o \circ \delta^*(t)) \circ m_\varnothing}_\mfD & \text{by functoriality}\\
&=& \curlyinterp{o \circ (\delta^*(t))}_\mfC & \text{since $(F, o')$ is part of a morphism $\mfC \to \mfD$}\\
&=& \interp{\cT}(t) & \text{by definition}\\
\end{array}
\]
\end{proof}

From now on, we may omit the ``tree'' in when discussing streaming settings.

As for SSTs, linearity is an important concept for traditional BRTTs over ranked alphabets.
Rather than starting from the category corresponding exactly to usual BRTTs, we will first study a
more convenient, albeit less expressive, streaming setting based on the idea of trees with multiple holes.
For this, it will be convenient to introduce the notion of \emph{multicategory}, which is essentially
a notion of category where morphisms are allowed to have multiple input objects.
We will thus first devote Section~\ref{subsec:multicat} to some preliminaries describing how
to (freely) generate affine monoidal categories from multicategories.

After spelling out the universal properties that will allow us to easily define functors and
a quick review of the generalization of the results of Section~\ref{subsec:oplus-string} pertaining to
coproduct completions in Section~\ref{subsec:oplus-trees}, we present a basic multicategory of registers
$\Registermc$ containing ``trees with holes'' in Section~\ref{subsec:registercat-tree}, and the corresponding
streaming setting $\REGISTER$ on top of a category $\Register$.
The usual restriction to registers containing ``trees with at most one hole''
is also discussed and shown to be no less expressive thanks to our basic results on the coproduct completion.
Then, we explain how the usual notion of BRTT presented in~\cite{BRTT} can be shown to have the same expressiveness
as $\REGISTER_\with$-BRTTs in Section~\ref{subsec:with-trees}.

Finally, in Section~\ref{subsec:dial-trees}, we show that $\Register_\oplus$ has internal homsets $\iota_\oplus(R) \lin \iota_\oplus(S)$ and conclude that $\Register_{\oplus\with}$ is monoidal-closed. We conclude the
proof of Theorem~\ref{thm:main-tree} in Section~\ref{subsec:main-trees}.

For the rest of this part, we fix a ranked alphabet $\bf \Gamma$ so that we may focus on outputs
contained in $\Tree({\bf \Gamma})$, much like we focussed on outputs in some fixed $\Gamma^*$ before.

\subsection{Multicategorical preliminaries}
\label{subsec:multicat}

This section is devoted to spelling out the formal definition of the notion of
multicategory that we use in the sequel, and their relation to symmetric
affine monoidal category. While technically necessary for the sequel, it is rather
dry and should maybe only be skimmed over at first reading.

\begin{defi}
A (weak symmetric) multicategory $\cM$ consists of
\begin{itemize}
\item a class of objects $\Obj(\cM)$
\item a class of \emph{multimorphisms} going from pairs $(I, (A_i)_{i \in I})$ of a finite index set $I$ and a family $(A_i)_{i \in I}$ of objects to
objects $B$. We omit the first component of the source and write $\Hom{\cM}{(A_i)_{i \in I}}{B}$ for the set of these multimorphisms.
\item for every object $A$, a distinguished identity multimorphism $\id_A \in \Hom{\cM}{(A)_{* \in 1}}{A}$.
\item for every set-theoretic map $f : I \to J$, families $(A_i)_{i \in I}$, $(B_j)_{j \in J}$ and object $C$, a composition operation
\[
\begin{array}{ccccl}
\Hom{\cM}{(B_j)_{j \in J}}{C} &\times& \left[\prod_{j \in J} \Hom{\cM}{(A_i)_{i \in f^{-1}(j)}}{B_j}\right] &\longto& \Hom{\cM}{(A_i)_{i \in I}}{C} \\
\alpha&,& (\beta_j)_{j \in J} &\longmapsto& \alpha *_f \beta
\end{array}
\]
\item for every bijection $\sigma : I' \to I$ between finite sets, a family of actions

\[\sigma^* : \Hom{\cM}{(A_i)_{i \in I}}{B} \to \Hom{\cM}{(A_{\sigma(i')})_{i' \in I'}}{B}\]
correspond to reindexing along symmetries. 
\end{itemize}
Furthermore, the above data is required to obey the following laws.
\begin{itemize}
\item The identity morphism be a neutral for composition:
for any $\alpha \in \Hom{\cM}{(A_i)_{i \in I}}{B}$,
\[\id_B *_! (\alpha)_{* \in 1} \quad = \quad \alpha \quad = \quad \alpha *_{\id_I} (\id_{A_i})_{i \in I}\]
\item Composition is associative: for any finite sets $I, J$ and $K$, functions $f : K \to J$ and $g : J \to I$,
families of objects $(A_k)_{k \in K}$, $(B_j)_{j \in J}$, $(C_i)_{i \in I}$, $D$ and families of morphisms
\[\begin{array}{l}
\alpha \in \prod_{j \in J} \Hom{\cM}{(A_k)_{k \in f^{-1}(j)}}{B_j} \\
\beta \in \prod_{i \in I} \Hom{\cM}{(B_j)_{j \in g^{-1}(i)}}{C_i} \\
\gamma \in \Hom{\cM}{(C_i)_{i \in I}}{D}
\end{array}
\]
the following equation holds
\[
\gamma *_g \left(\beta *_f \alpha\right) = \left(\gamma *_{\restr{g}{f^{-1}(j)}} \beta\right)_j *_f \alpha
\]
\item Permutations act functorially: for any $(A_i)_{i \in I}$, $B$ and bijections $\sigma : I' \to I$ and
$\sigma' : I'' \to I'$, the following commute
\[
\xymatrix{
\Hom{\cM}{(A_i)_{i \in I}}{B} \ar[dr] \ar[r] & \Hom{\cM}{(A_{\sigma(i')})_{i' \in I'}}{B} \ar[d] \\
& \Hom{\cM}{(A_{\sigma(\sigma'(i''))})_{i'' \in I''}}{B} \\
}
\]
and $\id_I^* : \Hom{\cM}{(A_i)_{i \in I}}{B} \to \Hom{\cM}{(A_i)_{i\in I}}{B}$ is the identity.
\item Composition is compatible with permutations: for every commuting square
\[
\xymatrix@C=20mm{
J' \ar[d] \ar[r]^g & I' \ar[d]^\sigma \\
J \ar[r]_f & I \\
}\]
in $\Finset$ such that the vertical arrows be bijections, note in particular that for
every $i \in I$, there is a bijection $\sigma_i : g^{-1}(I') \to f^{-1}(I)$; we thus require that
\[ \alpha *_f \left(\beta_{j \in f^{-1}(i)}\right)_{i \in I}
\quad =\quad
 \sigma^*(\alpha) *_g \left(\sigma_i^*\left(\beta_{j \in f^{-1}(i)}\right)\right)_{i \in I}\]
for every suitable $\alpha$ and $\beta_j$s.
\end{itemize}
\end{defi}

Every symmetric monoidal category $\cC$ can be mapped to a multicategory $\tomcat{\cC}$
by taking
\[\Hom{\tomcat{\cC}}{(A_i)_{i \in I}}{B} \quad = \quad \varHom{\cC}{\bigtensor_{i \in I} A_i}{B}\]
We may make this map functorial, provided we equip the class of multicategories
and the class of symmetric monoidal categories with categorical structures.

\begin{defi}
\label{def:multifunctor}
Given two weak multicategories $\cM$ and $\cN$, a functor $F : \cM \to \cN$ consists
of maps of objects $F : \Obj(\cM) \to \Obj(\cN)$ and of multimorphisms
\[F ~:~ \Hom{\cM}{(A_i)_{i \in I}}{B} ~\longto~ \Hom{\cN}{(F(A_i))_{i \in I}}{F(B)}\]
such that $F(\id_A) = \id_{F(A)}$, $F\left(\alpha *_f \left(\beta_{j}\right)_{j \in f^{-1}(i)}\right) = F(\alpha) *_f (F(\beta_j))_{j \in f^{-1}(i)}$ and $F(\sigma*(\alpha)) = \sigma^*(F(\alpha))$ for all suitable objects, index sets, set-theoretic functions and morphisms.
\end{defi}

Definition~\ref{def:multifunctor} gives a class of arrows for a large category $\mCat$ of multicategories.
Calling $\SymAff$ the category whose objects are symmetric affine monoidal categories and morphisms are
strong monoidal functors, the map $\cC \to \tomcat{\cC}$ extends to a functor $\SymAff \to \mCat$.
We are now interested in the inverse process of generating freely a symmetric affine monoidal
category out of a weak multicategory.

\begin{defi}
\label{def:mcattoaff}
Let $\cM$ be a weak multicategory.
The \emph{free affine symmetric monoidal category} generated by $\cM$ is the category $\mcattoaff{\cM}$ such that
\begin{itemize}
\item objects are pairs $(I, (A_i)_{i \in I})$ of a finite set $I$ and a family of objects of $\cM$; write $\bigtensor_{i \in I} A_i$ for such objects.
\item morphisms $\bigtensor_{i \in I} A_i \to \bigtensor_{j \in J} B_j$ are pairs $(f, (\alpha_j)_{j \in J})$ where $f$ is a partial function $I \partto J$ and
$\alpha_j$ is a $\cM$ multimorphism $(A_i)_{i \in f^{-1}(J)} \to B_j$.
\item identities $\bigtensor_{i \in I} A_i \to \bigtensor_{i \in I} A_i$ are the pairs $(\id_I, (\id_{A_i})_{i \in I})$.
\item the composition of
\[(f, (\alpha_j)_{j \in J}) : \bigtensor_{i \in I} A_i \to \bigtensor_{j \in J} B_j~~\text{and}~~ (g, (\beta_k)_{k \in K}) : \bigtensor_{j \in J} B_j \to \bigtensor_{k \in K} C_k\] is 
$\left(g \circ f, \left(\beta_k * \left(\alpha_j\right)_{j \in g^{-1}(k)}\right)_{k \in K}\right)$.
\end{itemize}
\end{defi}

For any bijection $\sigma : I' \to I$, we have canonical isomorphisms $\bigtensor_{i \in I} A_i \to \bigtensor_{i \in I} A_{\sigma(i)}$. We take the binary tensor product to be
\[\left(\bigtensor_{i \in I} A_i \right) \tensor \left( \bigtensor_{j \in J} B_j \right) ~~=~~ \bigtensor_{x \in I + J} 
\left\{ \begin{array}{ll}
A_i & \text{if $x = \inl(i)$} \\ 
B_j & \text{if $x = \inr(j)$} \\
\end{array} \right.\]
and the unit to be the terminal object, which is the nullary family $\bigtensor_\varnothing$.
The associator and symmetries are induced by the isomorphisms $(I + J) + K \cong I + (J + K)$ and $I + J \cong J + I$
respectively, and the units by $I + \varnothing \cong I \cong \varnothing + I$. The axioms of weak symmetric multicategories then imply that this indeed endows $\mcattoaff{\cM}$ with a symmetric affine monoidal structure.
We skip checking the details.

\subsection{The coproduct completion}
\label{subsec:oplus-trees}

Similarly as for strings, the coproduct completion of a category induces a map
$\mfC \mapsto \mfC_\oplus$ over tree streaming settings.
Furthermore, expressiveness of $\mfC$ and $\mfC_\oplus$ remain the same under similar hypotheses as Theorem~\ref{thm:oplus-sst-conservative}.

\begin{thm}
\label{thm:oplus-brtt-conservative}
Let $\mfC$ be a tree streaming setting whose monoidal product is affine and such that
all objects of the underlying category have unitary support. Then, $\mfC$-BRTTs and \mbox{$\mfC_\oplus$-BRTTs}
are equi-expressive.
\end{thm}

The proof is an unsurprising adaptation of the one of Theorem~\ref{thm:oplus-sst-conservative}.
We leave it to the interested reader.
Similarly, we define the notion of a state-dependent memory $\mfC$-streaming tree transducer (sdm-$\mfC$-BRTT),
which can be shown to be as expressive as sdm-$\mfC_\oplus$-BRTT. We only state the definition and
the generalization of Lemma~\ref{lem:sdmSST-oplus}, whose proof we defer to the interested reader.

\begin{defi}
A $\mfC$-\emph{state-dependent memory BRTT} 
with input $\Tree({\bf \Sigma})$ is a tuple $(Q, \delta, (C_q)_{q \in Q}, o)$  where
\begin{itemize}
\item $Q$ is a finite set of states
\item $(C_q)_{q \in Q}$ is a $Q$-indexed family of objects of $\cC$ 
\item $\displaystyle\delta \in \left[\prod_{a \in \Sigma} \prod_{q \in Q^{\arity(a)}} \sum_{r \in Q} \Hom{\cC}{\bigtensor_{x \in \arity(a)} C_{q(x)}}{C_r}\right]$ is a transition function
\item $\displaystyle o \in \prod_{q \in Q} \Hom{\cC}{C_{q}}{\retty}$ is the output family of morphisms
\end{itemize}
\end{defi}

\begin{lem}
\label{lem:sdmBRTT-oplus}
Let $\mfC$ be a streaming setting. State-dependent memory $\mfC$-BRTTs are as expressive as $\mfC_\oplus$-BRTTs.
\end{lem}

\subsection{The combinatorial multicategory $\Registermc$}
\label{subsec:registercat-tree}

We are now ready to give a smooth definition of a category of register transition
for trees, generalizing Proposition~\ref{prop:Register}. 
As announced, we find it more convenient to first give a multicategory $\Registermc$
and then move to monoidal categories by taking $\Register = \mcattoaff{(\Registermc)}$.
We then discuss the restriction consisting of limiting the number of holes in the tree
expressions stored in register to at most one and show that it is not limiting.

\begin{notations}
Recall that we regard ranked alphabets $\bf R$ as pairs $(R, \arity)$ where $R$ is a finite
set of letters and $(\arity(a))_{a \in R}$ is a family $R \to \Finset$ of arities.
Given two ranked alphabets $\bf R$ and $\bf S$, we suggestively
write $\bf R \tensor \bf S$ for the ranked alphabet $(R + S, [\arity,\arity])$.
Given a finite set $U$, call $\cO(U)$ the ranked alphabet $(U, (\varnothing)_{u \in U})$ consisting
of $\card{U}$-many terminal letters and $\cI(U)$ for the ranked alphabet consisting of a single letter of arity $U$.
Given a ranked alphabet ${\bf \Sigma} = (\Sigma, \arity)$ and a subset $X \subseteq \Sigma$, we write $\restr{{\bf \Sigma}}{X}$ for the restriction $(X, \restr{\arity}{X})$.
\end{notations}

\subsubsection{Definition of $\Register$}
Before giving the definition of $\Registermc$, we first need to make formal a notion
of trees with linearly many occurrences of certain constructors.
\begin{defi}
Let $\bf R$ be a ranked alphabet. We define the set $\LTree_{\bf \Gamma}({\bf R})$
of $\bf R$-linear trees as the set of $({\bf \Gamma} \tensor {\bf R})$-trees $\Tree({\bf \Gamma} \tensor {\bf R})$ such that all constructors of $\bf R$ appear exactly once.
\end{defi}

\begin{defi}
\label{def:Registermc}
Define $\Registermc({\bf \Gamma})$ (abbreviated $\Registermc$ in the sequel) as the multicategory
\begin{itemize}
\item whose class of objects $\Obj(\Registermc)$ is $\Finset$.
\item whose class of multimorphisms from $(A_i)_{i \in I}$ to $B$
is the set of linear trees over the joint alphabet $(I, A) \tensor \cO(B)$
(recall that $(I,A)$ can formally be regarded as a ranked alphabet: its set of
letters is $I$ and the arity of $i \in I$ is $A(i) = A_i$).
\[\Hom{\Registermc}{(A_i)_{i \in I}}{B} \quad = \quad \LTree_{\bf \Gamma}((I,A) \tensor \cO(B))\]
\item whose composition operations are given by substitution: given a map $f : I \to J$ and multimorphisms
\[t \in \LTree_{\bf \Gamma}((J,B) \tensor \cO(C)) \qquad \text{and} \qquad u \in \prod_{j \in J} \LTree_{\bf \Gamma}((f^{-1}(j),(A_i)_i) \tensor \cO(B_j))\]
the composite $t *_f u$ is defined by recursion over $t$:
\begin{itemize}
\item if $t = a((t'_k)_k)$ for some $a = \inl(b)$ with $b \in \Gamma$ or $a = \inr(\inr(c))$ for $c \in C$, then
\[t *_f u = a((t'_k *_f u)_k)\]
\item otherwise $t = \inr(\inl(j))((t'_b)_{b \in B_j})$ with $j \in J$ and $t'_b$ in some $\LTree_{\bf \Gamma}((J_b, B_b) \tensor \cO(C_j))$ for $\bigcup_{b \in B_j} J_b = J \setminus \{j\}$, $B_b = \restr{B}{J_b}$ and $\bigcup_{b \in B_j} C_b = C$.   In such a case, we set
\[t *_f u = u_j[(t'_b *_{\id_{J_b}} (u_{j'})_{j' \in J_b})/b]_{b \in B_j}\]
where $[-/-]_{- \in -}$ denotes the more usual substitution of leaves by subtrees (recall that every $b \in B$, and a fortiori $B_j$ has arity $\varnothing$).
\end{itemize}
\end{itemize}
\end{defi}

While the definition of composition of multimorphisms in $\Registermc$
looks daunting, we claim it is rather natural. Figure~\ref{fig:tree-map-comp}
depicts the composition $\alpha *_f (\beta_x)_{x \in \{t, u \}}$ with
\[\alpha = t(a(u(c()),*()),c()) \qquad \beta_t = a(x(p(),b(q())), y()) \quad \text{and} \quad \beta_u = z(c(),r(),c())\]

\begin{figure}
\center

\includegraphics[scale=0.75]{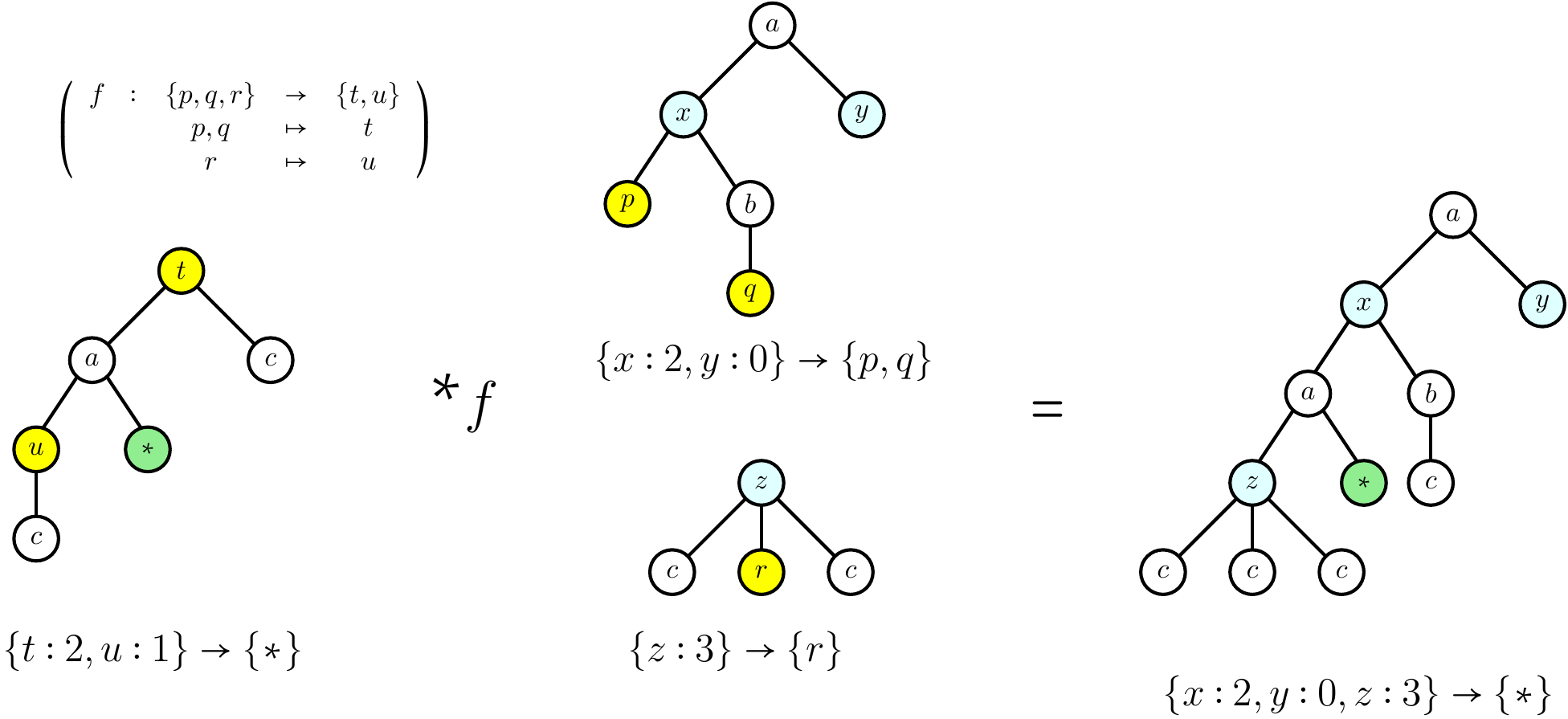}
\vspace{1.5em}
\hrule
\caption{Composition of some multimorphisms of $\Registermc$.}
\label{fig:tree-map-comp}
\end{figure}

We now set $\Register = \mcattoaff{(\Registermc)}$;
while objects of $\mcattoaff{(\Registermc)}$ are supposed to be families of finite sets $(A_i)_{i \in I}$,
in the sequel, we sometimes identify them with the ranked alphabets $(I,A)$ in $\Register$ for notational convenience.
As such, the notation ${\bf R} \tensor {\bf S}$ corresponds to the expected tensorial product in $\Register$.

We are now ready to define our first tree streaming setting.

\begin{defi}
$\REGISTER$ is the tree streaming setting $(\Register, \tensor, \top, \cI(\varnothing), \curlyinterp{-})$, where $\curlyinterp{-}$ is the canonical isomorphism
$\Hom{\Register}{\top}{\cI(\varnothing)} \cong 
 \Hom{\Registermc}{()_{\varnothing}}{\varnothing} \cong \LTree_{\bf
   \Gamma}(\varnothing) \cong \Tree({\bf \Gamma})$
 (where $()_\varnothing$ is the empty family).
\end{defi}

Call $\Registerleonemc$ the full submulticategory of $\Registermc$ whose objects
are empty or singleton sets, and
$\Registerleone \cong \mcattoaff{\Registerleonemc}$ to be the corresponding full subcategory of $\Register$.
The monoidal structure of $\Register$ restricts to $\Registerleone$ without any difficulty, and that $\cI(\varnothing)$ is an
object of $\Registerleone$. This means that $\REGISTER$ has a restriction to a streaming setting $\REGISTERleone$.

While $\REGISTERleone$ turns out to be more elementary and a good building block toward the definition of
usual BRTTs, it is easier to show the monoidal closure of $\REGISTER_{\oplus\with}$ than $\REGISTERleone$.
Thankfully, it turns out that the expressiveness of BRTTs over $\REGISTER$ and
$\REGISTERleone$ is the same.

For one direction, there is a morphism $\REGISTERleone \to \REGISTER$ corresponding to the embedding $\Registerleone \to \Register$.
For the other direction, we exploit Theorem~\ref{thm:oplus-brtt-conservative}.
The proof involves some combinatorics, but nothing surprising as it
amounts to the classical decomposition of multi-hole trees into families of single-hole trees as found in e.g.~\cite[\S3.5]{BRTT}.

\begin{lem}
\label{lem:leone}
There is a morphism of streaming settings $\REGISTER \to \REGISTERleone_\oplus$.
\end{lem}
\begin{proof}[Proof sketch]
We focus on giving enough ingredients to define the underlying (strong) monoidal functor $F : \Register \to \Registerleone_\oplus$,
which is going to preserve $\retty$ (i.e., we will have $F(\cI(\varnothing)) \cong \iota_\oplus(\cI(\varnothing))$).

Rather than giving a direct explicit construction of $F$ (which is rather tedious over morphisms), we obtain
it as a composition of two strong monoidal functors:
the strong monoidal embedding $\iota_\oplus : \Register \to \Register_\oplus$ and a functor
$R : \Register_\oplus \to \Registerleone_\oplus$ right adjoint to
the inclusion $I : \Registerleone_\oplus \to \Register_\oplus$.
\[\xymatrix@C=20mm{
\Register \ar[r]^{\iota_\oplus} & \Register_\oplus \ar@/^1pc/[r]^R \ar@{}[r]|\top & \Registerleone_\oplus \ar@/^1pc/[l]^I
}\]
$I$ is strong symmetric monoidal. Therefore,
by~\cite[Proposition 14, Section 5.17]{mellies09ps}, once we construct
$R$ right adjoint to $I$, it comes equipped with a canonical lax monoidal structure.
Furthermore, since we want $I \dashv R$, we can use the implicit characterization of adjoints
given in~\cite[item~(iv), Theorem~2, Section~IV.1]{CWM}:
to define $R$,
it suffices to give the value of $R\left(A\right)$ for every object $A \in \Obj(\Register_\oplus)$
and counit maps $\epsilon_A : I(R(A)) \to A$ such that, for every object $B \in \Obj(\Registerleone_\oplus)$
and map $h \in \Hom{\Register_\oplus}{I(B)}{A}$, there is a unique $\widetilde{h} \in \Hom{\Registerleone_\oplus}{R(A)}{B}$ such that the following diagram commutes
\[
\xymatrix@C=20mm{
I(B) \ar[dr]_{h} \ar@{-->}[r]^{I(\widetilde{h})} & I(R(A)) \ar[d]^{\epsilon_A}\\
& A
}\]
So we only need to define $R(A)$ and $\epsilon_A$ to obtain our functor $R$; once those are defined,
we leave checking that the universal property holds to the reader.
We first focus on the case where $A = \iota_{\oplus}(\cI(U))$ for some finite set $U$.
Recall that a single-letter alphabet $\cI(U)$, when seen as an object of $\Register$, should be
should be intuitively regarded as a register containing a tree with $U$-many holes.
If $\card{U} \le 1$, we may simply take $R(\iota_\oplus(\cI(U))) = \iota_\oplus(\cI(U))$.
Otherwise, $\card{U} \ge 2$ and $\cI(U)$ is not an object $\Registerleone$; in that case,
we use the following recursive definition
\[R(\iota_\oplus(\cI(U))) ~~~=~~~ \cI(1) ~~\tensor~~ \bigoplus\limits_{b \in \Gamma} \bigoplus\limits_{\substack{f : U \to \arity(b) \\\text{nonconstant}}} \bigtensor\limits_{x \in \arity(b)} R(\iota_\oplus(\cI(f^{-1}(x))))\]
Note that this definition is well-founded because the function $f$ in the second sum is taken to be non constant, so
that $\card{f^{-1}(x)} < \card{U}$ for every $x$. While this suffices as a definition of $R(\iota_\oplus(\cI(U)))$, this
might be a bit opaque without having the definition of $\epsilon_{R(\iota_\oplus(\cI(U)))}$. Before giving that, let us attempt to give
an intuitive rationale behind this definition: there is an isomorphism\footnote{Which we may later on define formally as a composite
\[
\Hom{\Registerleone_\oplus}{\top}{R(\iota_\oplus(\cI(U)))} \xrightarrow{\;I\;}
\Hom{\Register_\oplus}{\top}{R(\iota_\oplus(\cI(U)))} 
\to
\Hom{\Register_\oplus}{\top}{\cI(U)} \xrightarrow{\;\sim\;}
\LTree_{\bf \Gamma}(U)
\]
where the mediating arrow is the post-composition by $\epsilon_{\iota_\oplus(\cI(U))}$.}
\[\Hom{\Registerleone_\oplus}{\top}{R(\iota_\oplus(\cI(U)))} \cong \LTree_{\bf \Gamma}(U)\]
which can be nicely pictured, provided we actually compute recursively $R(\iota_\oplus(\cI(U)))$ and spell out a normal form
\[R(\iota_\oplus(\cI(U))) ~~~\cong~~~ \bigoplus_{t \in {\bf PT}(U)} \bigtensor_{n \in N(t)} A_n\]
with all $A_n = \cI(\varnothing)$ or $A_n \cong \cI(1)$. It is always possible to build a suitable set ${\bf PT}(U)$ simply because
all objects of $\Registerleone_{\oplus}$ have this shape, but an intuitive definition of what one might call a set of
\emph{partitioning trees over $U$} is also possible for $\bf PT$, and $N(t)$ would then correspond to the nodes of the trees.
We skip defining this notion formally, but note that the announced bijection would then match trees with $U$-many holes with
pairs $(t, (u_i)_{i \in N(t)})$ of a partitioning tree $t$ and a family of trees with at most one-hole $(u_i)_{i \in N(t)}$.
This bijective correspondence is pictured in Figure~\ref{fig:ltree-ex-decompose}.
\begin{figure}
\center
\includegraphics[scale=0.65]{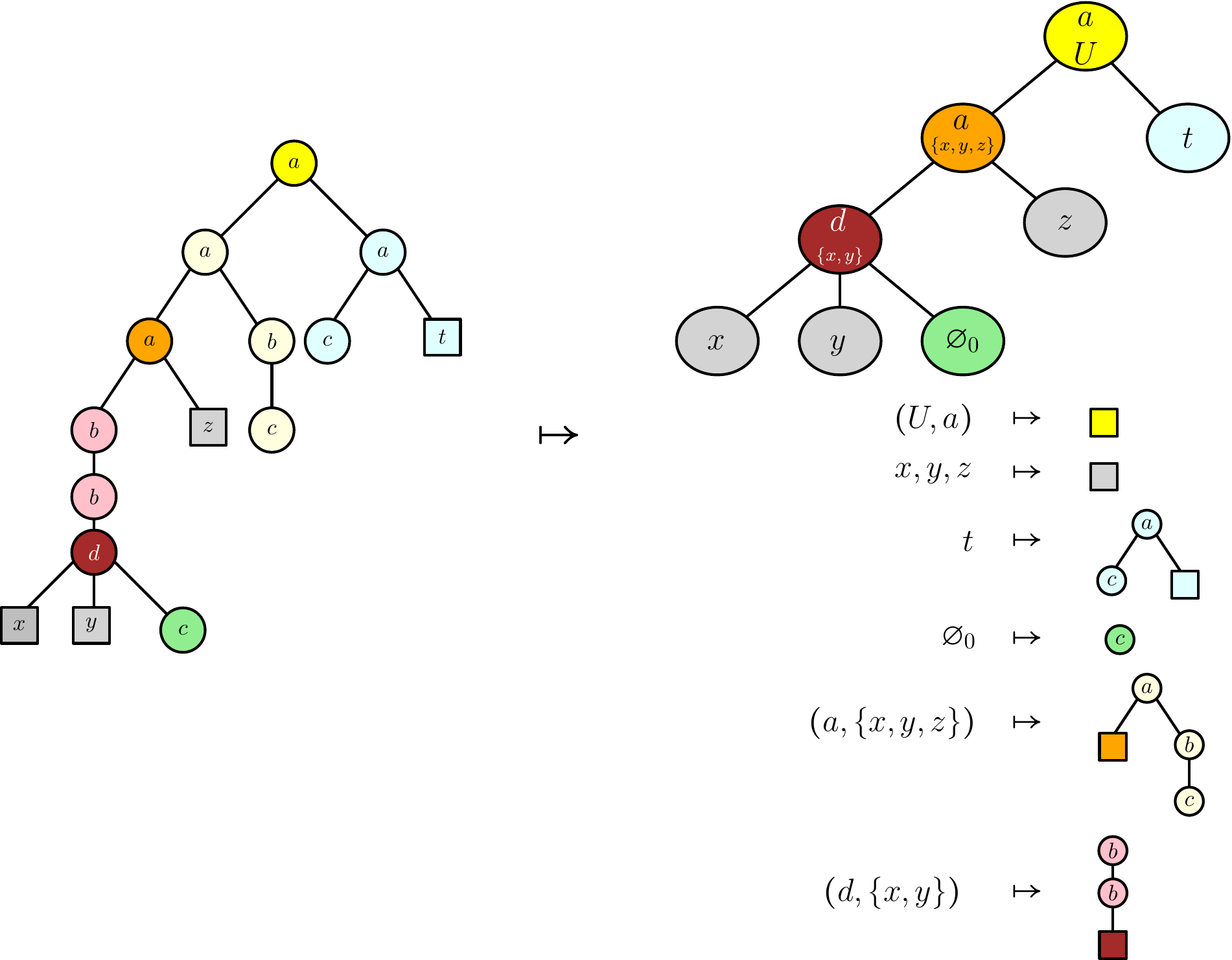}
\vspace{1.5em}
\hrule

\caption{Decomposition of a multi-hole tree $\LTree_{\{a : 2, b : 1, c : 0, d : 3\}}(\{x,y,z,t\})$ as a tuple consisting of a partitioning tree and trees with at most one hole.}
\label{fig:ltree-ex-decompose}
\end{figure}

Now, we define $\epsilon_{R(\iota_\oplus(\cI(U)))} \in \Hom{\Registerleone_\oplus}{I(R(\iota_\oplus(\cI(U))))}{\iota_\oplus(\cI(U))}$, by induction over the
size of $U$. If $\card{U} \le 1$, we take $\epsilon_{R(\iota_\oplus(\cI(U)))} : \iota_\oplus(\cI(U)) \to \iota_\oplus(\cI(U))$ to be the identity.
Otherwise, we need to define a map
\[
g_U ~~:~~ I\left(\bigoplus_{b \in \Gamma} \bigoplus_{\substack{f : U \to \arity(b) \\\text{nonconstant}}} \bigtensor_{x \in \arity(b)} R(\iota_\oplus(\cI(f^{-1}(x)))) \right)
~~\longrightarrow~~ \iota_\oplus(\cI(U))
\]
or, equivalently, a family of $\Register_\oplus$-maps indexed by $b \in \Gamma$ and $f : U \to \arity(b)$ non-constant
\[
g_{b,f} ~~:~~ \bigtensor_{x \in \arity(b)} I(R(\iota_\oplus(\cI(f^{-1}(x)))))
 ~~\longrightarrow~~ \iota_\oplus(\cI(U))
\]
By the induction hypothesis, we have a family of $\Register_\oplus$-maps $(g_x)_{x \in \arity(b)}$
\[
g_x ~~:~~ I(R(\iota_\oplus(\cI(f^{-1}(x))))) ~~\longrightarrow~~ \iota_\oplus(\cI(f^{-1}(x)))
\]
We define $g_{b,f}$ as the composite
\[
\bigtensor_{x \in \arity(b)} I(R(\iota_\oplus(\cI(f^{-1}(x)))))
~\xrightarrow{~~\bigtensor_x g_x~~}~
\displaystyle
\bigtensor_{x \in \arity(b)} \iota_\oplus(\cI(f^{-1}(x)))
~\xrightarrow{~~\overline{b}~~}~
\iota_\oplus(\cI(U))
\]
where $\overline{b}$ is obtained from the map of $\Registermc$ which intuitively takes a family $(t_x)_{x \in \arity(b)}$ of trees into a
single tree $b((t_x)_{x \in \arity(b)})$ (officially, the tree $b((*)_{x \in \arity(b)}) \in \LTree_{\bf \Gamma}(\cO(U))$.

Now, $R$ is defined on objects of the shape $\iota_\oplus(\cI(U))$, as well as $\epsilon$, so we need to extend this to the whole category
$\Register_\oplus$. Recall that every object $A$ of $\Register_\oplus$ can be written as $A = \bigoplus_{v \in V} \bigtensor_{j \in J_u} \iota_\oplus(\cI(U_j))$.
In the end, the functor $R$ is expected to be strong monoidal, and we may force it to preserve coproducts, so we set
\[R(A)
~~ = ~~
\bigoplus_{v \in V} \bigtensor_{j \in J_u} R(\iota_\oplus(\cI(U_j))) \qquad\qquad\qquad
\epsilon_A ~~=~~ \bigoplus_{v \in V} \bigtensor_{j \in J_u} \epsilon_{\iota_\oplus(\cI(U_j))}
\]
\end{proof}

\subsubsection{Relationship to $\laml$}
Recall that if ${\bf \Gamma} = \{ a_1 : A_1, \ldots, a_k : A_k \}$ is an output alphabet, 
we call $\widetilde{\bf \Gamma}$ the context $a_1 : \basety \lin \ldots \lin \basety, \ldots, a_k : \basety \lin \ldots \lin \basety$ where the
type of $a_i$ has $\card{A_i}$ arguments. Definition~\ref{def:syntacticcat} provides us with a suitable affine monoidal closed category
$\Lamcat(\widetilde{\bf \Gamma})$, which we still call $\Lamcat$ when $\widetilde{\bf \Gamma}$ is clear from context. Since we have a monoidal product,
we may easily adapt Definition~\ref{def:syntacticstreaming} to get a tree streaming setting $\mfLam$.
Then we may relate $\laml$-definability to (single-state) $\mfLam$-BRTT.

\begin{lem}
\label{lem:lamlbrtt}
Computability by single-state $\mfLam$-BRTTs and $\laml$-definability are
equivalent for functions $\Tree({\bf \Sigma}) \to \Tree({\bf \Gamma})$.
\end{lem}
\begin{proof}[Proof idea]
This is proven in a similar way as Lemma~\ref{lem:lamlsst}, based on the syntactic Lemma~\ref{lem:laml-niceshape}.
The proof is even more straightforward as there is no mismatch
between the processing of trees by BRTTs and $\laml$-terms working with Church encodings, contrary to
SSTs for strings (which operate top-down rather than bottom-up when regarding strings as trees).
\end{proof}

We can now notice that $\mfLam$-BRTTs are more expressive than $\Registerleone$-BRTTs thanks to the notion of streaming setting morphisms,
much like with strings (this generalizes Lemma~\ref{lem:register-to-laml}).
\begin{lem}
\label{lem:register-to-laml-tree}
There is a morphism of streaming settings $\REGISTER \to \mfLam$.
\end{lem}
\begin{proof}[Proof sketch]
Let us focus on the underlying functor $F : \Register \to \Lamcat$.
For objects (which are finite sets), we
put
\[F\left(\bigtensor_{i \in I} U_i\right) =  \bigtensor_{i \in I}\left(\left(\basety^{\tensor U_i} \lin \basety\right) \with \unit\right)\]
A multimorphism $f \in \Hom{\Registermc}{(U_i)_{i \in I}}{V}$ is an element of $\LTree_{\bf \Gamma}\left(\left(\bigtensor_{i \in I} \cI(U_i)\right) \tensor \cO(V)\right)$
which has a Church encoding $\underline{f}$ which has a type isomorphic to
$\bigtensor_{i \in I} \left( \basety^{\tensor {U_i}} \lin \basety\right) \lin \basety$, and thus embeds into
$F\left(\bigtensor_{i \in I} U_i\right)$ through well-typed term $\iota$.
We take $F(\mcattoaff{f}) = \lam x. \iota \; \underline{f}$, and extend this definition to
arbitrary morphisms $(f, (\alpha_j)_j) : \bigtensor_{i \in I} U_i \to \bigtensor_{j \in J} V_j$ in $\Register$ by first using the second projection $\pi_2$ to restrict to the case where $\dom(f) = I$, and then by piecing together the $F(\mcattoaff{(\alpha_j)})$.
\end{proof}

\begin{cor}
\label{cor:registeropluswith-to-laml-tree}
There is a morphism of streaming settings $\REGISTER_{\oplus\with} \to \mfLam$.
\end{cor}
\begin{proof}[Proof idea]
Starting from Lemma~\ref{lem:register-to-laml-tree}, we have a functor $\Register \to \Lamcat$. Since $\Lamcat$ has all products and
coproducts, the universal properties of the $(-)_\with$ and $(-)_\oplus$ completion yield a functor $F : \Register_{\oplus\with} \to \Lamcat$.
The monoidal structure of the initial functor $\Register\to\Lamcat$ can be lifted accordingly.
For any finite family of objects $(((A_k)_{k \in K_j})_{j \in J_i})_{i \in I}$ sitting in a symmetric monoidal closed category with
products and coproducts, there are canonical morphisms
\[
\left(\bigoplus_{u \in U} \bigwith_{x \in X_u} A_{x}\right)
\tensor
\left(\bigoplus_{v \in V} \bigwith_{v \in Y_v} B_y\right)
~~\longto~~
\bigoplus_{(u,v) \in U \times V} \bigwith_{(x,y) \in X_u \times Y_v} A_x \tensor B_y
\]
which are not isomorphisms in general, but constitute the non-trivial part of the lax monoidal structure of $F$; $m^0 : \unit \to F(\unit)$ is actually the identity.
\end{proof}

\subsection{$\REGISTER_\with$-BRTTs coincide with regular functions, via
  coherence spaces}
\label{subsec:with-trees}

We define a streaming setting $\REGISTERconflict$ and its restriction
$\REGISTERleoneconflict$ (with respective underlying categories
$\Registerconflict$, $\Registerleoneconflict$) so that
$\REGISTERleoneconflict$-BRTTs coincide with Alur and D'Antoni's
  notion of single-use-restricted BRTT~\cite{BRTT}, which they showed to characterize regular tree functions.
We then show that there are morphisms of streaming settings $\REGISTER_\with \to \REGISTERconflict \to \REGISTER_\with$ and thus establish that $\REGISTER_\with$-BRTTs capture exactly regular tree functions.

Much like $\Register$, the category $\Registerconflict$ is obtained by applying a generic
construction to $\Registermc$, taking weak symmetric multicategories to symmetric affine monoidal
categories. In particular, objects of $\Registerconflict$ will consist of formal tensor products
of objects of $\Registermc$. The main difference is that morphisms of $\Registerconflict$ will induce
a dependency \emph{relation} $D \subseteq I \times J$ over indexing sets, rather than a partial function $J \partto I$.
This corresponds to a relaxation of the copylessness condition.
However, objects of $\Registerconflict$ will also be equipped
with a \emph{conflict relation} $\incoh$ over their indexed sets, and $D$ will be required to satisfy a
linearity constraint. Calling $\coh$ the dual \emph{coherence relation} such that $x \coh y$ is equivalent to $x = y \vee \neg (x \incoh y)$, if we have $(i,j) \in D$ and $(i',j') \in D$, the linearity constraint enforces
\[
j \coh_J j' ~~ \Rightarrow ~~ i \coh_I i'  \qquad \text{and} \qquad i \incoh_I i' ~~\Rightarrow~~ j \incoh_J j'
\]
This corresponds to the \emph{single use restriction} imposed on
BRTTs~\cite[\S2.1]{BRTT}, whose introduction was motivated in
\Cref{subsubsec:prelim-trees-inputs}.
\begin{exa}
  \label{ex:condswap-sru}
  The BRTT that we gave for the \enquote{conditional swap} function in
  Example~\ref{ex:brtt-condswap} is single-use-restricted according to the above
  by taking its two registers (i.e.\ objects of $\Registermc$) to be in
  conflict.
\end{exa}
But as our choice of notation and vocabulary suggests, this is also related to
the category of (finite) \emph{coherence spaces}, the first denotational model
of linear logic~\cite{girardLL} (predated by a similar semantics for
system~F~\cite{girardF}). As far as we know, this observation is new (the
conflict relation is denoted by $\mathbf{\eta}$ in~\cite{BRTT}, while $\incoh$
comes from the linear logic literature). The coherence semantics of the linear
$\lambda$-calculus has been used in particular by Gallot, Lemay and
Salvati~\cite{LambdaTransducer} to analyze a top-down tree transducer model
containing linear $\lambda$-terms. Unlike them, we do not use coherence spaces
\emph{as a semantics} here; what happens here is much closer the use of a
coherence/conflict relation to handle additive connectives -- we will indeed
show a connection with the $\with$-completion -- in proof nets,
see~\cite[Appendix~A.1]{MonomialNets} and~\cite{ConflictNets}.

\begin{defi}[see e.g.~{\cite[\S2.2.3]{LLSS}}]
A coherence space $I$ is a pair $(\web{I}, \coh_I)$ of a set $\web I$, called the \emph{web}, and a binary reflexive symmetric
relation $\coh_I$ over $\web I$ called the \emph{coherence relation}.
As usual, given a coherence relation $\coh$, we write $\incoh$ for the dual defined by $i \incoh i' \Leftrightarrow (i = i' \vee \neg (i \coh i'))$.
Finite coherence spaces are those coherence spaces whose webs are finite.
A \emph{linear map} of coherence spaces $f : I \to J$ is a relation $f \subseteq \web{I} \times \web J$ such that, whenever $(i,j) \in f$ and $(i',j') \in f$, we have
\[
i \coh_I i' ~~ \Rightarrow ~~ j \coh_J j'  \qquad \text{and} \qquad j \incoh_J j' ~~\Rightarrow~~ i \incoh_I i'
\]
Note that these are the \emph{converse implications} of those stated above for BRTTs.

The diagonal $\{(i,i) \mid i \in \web I\}$ is a linear map $I \to I$ and the relational composition of
two linear maps $I \to J \to K$ is again a linear map, so that we have
a category $\Fincoh$ whose objects are coherence spaces and morphisms are linear maps.
\end{defi}

$\Fincoh$, equipped with the tensorial product
\[(\web I, \coh_I) \tensor (\web J, \coh_J) = (\web I \times \web J, \coh_I \times \coh_J)\] and dualizing object $(1, 1 \times 1)$,
is a well-studied $*$-autonomous category with cartesian products and coproducts.
The latter may be defined pointwise as
\[(\web I, \coh_I) \oplus (\web J, \coh_J) = (\web I + \web J,\; \coh_{I\oplus
    J})\]
where $\coh_{I\oplus J}$ is the \emph{smallest} relation such that
\[ \inl(i) \coh_{I \oplus J} \inl(i') \ \;\text{when}\ \;i \coh_I i' \qquad\text{and}\qquad
\inr(j) \coh_{I \oplus J} \inr(j') \ \;\text{when}\ \; j \coh_J j'\]
Dualizing an object corresponds to
moving from $\coh$ to $\incoh$, i.e. $(\web I, \coh_I)^\bot = (\web I,
\incoh_I)$, and the product is $I \with J = (I^\bot \oplus J^\bot)^\bot$.

With this in mind, we can describe how to turn a multicategory into an affine monoidal category
where monoidal products may be indexed by coherence spaces. The construction has
a vague family resemblance with the
\emph{coherence completion of categories} introduced by Hu and Joyal~\cite{HuJoyal},
but appears to have quite different properties.

\begin{defi}
Let $\cM$ be a weak symmetric multicategory. We define $\toconflict{\cM}$ to be the category
\begin{itemize}
\item whose objects are pairs $(X, (R_x)_{x \in \web X})$ where $X$ is a finite coherence space
and $(R_x)_{x \in \web X}$ a family of objects of $\cM$. We suggestively write
them $\bigodot_{x \in X} R_x$.
\item whose morphisms
\[(f, (\alpha_y)_{y \in \web Y}) \in \Hom{\toconflict{\cM}}{\bigodot_{x \in X} R_x}{\bigodot_{y \in Y} S_y}\]
are pairs consisting of a linear map $f \in \Hom{\Fincoh}{Y}{X}$ and a family of
multimorphisms $\alpha_y \in \Hom{\cM}{(R_x)_{x \in f(y)}}{S_y}$.
\item whose identities are pairs $(\id_X, (\id_{R_x})_{x \in \web X})$.
\item where the composition of                                \[(f, (\alpha_y)_{y \in \web Y}) \in \Hom{{}}{\bigodot_{x \in X} R_x}{\bigodot_{y \in Y} S_y} ~~\text{and}~~
(g, (\beta_z)_{z \in \web Z}) \in \Hom{{}}{\bigodot_{y \in Y} S_y}{\bigodot_{z \in Z} T_z}\]
is $(f \circ g, (\beta_z * (\alpha_y)_{y \in g(z)})_{z \in Z})$.
\end{itemize}
\end{defi}

\begin{defi}
  We set $\Registerconflict = \toconflict{(\Registermc)}$ and
  $\Registerleoneconflict$ to be its full subcategory consisting of objects
  $\bigodot_{x \in X} A_x$ where each $A_x$ is either empty or a singleton (so
  that $\Registerleoneconflict$ is isomorphic to
  $\toconflict{(\Registerleonemc)}$).
\end{defi}
\begin{prop}
  \label{prop:brtt-regular}
  BRTTs over the restricted tree streaming setting $\REGISTERleoneconflict$ compute
  exactly the regular tree functions.
\end{prop}
\begin{proof}
  By virtue of being equivalent to Alur and D'Antoni's notion of
  single-use-restricted BRTT~\cite{BRTT}. We point the reader to
  \Cref{sec:app-brtt} for a self-contained definition of those not involving
  categories, and leave it as an exercise to formally match those two
  descriptions. Although~\cite{BRTT} and our \Cref{sec:app-brtt} only consider
  BRTTs over \emph{binary} trees, the proof of equivalence between the latter
  and regular tree functions goes through \emph{macro tree transducers} (with
  regular look-ahead and single use restriction) which are known to compute
  regular functions for trees over arbitrary ranked alphabets~\cite{MacroMSO},
  so everything can be made to work out with arbitrary arities in the end.
  \end{proof}

This being done, the remainder of this section does not depend on $\Registermc$; the arguments apply to any weak symmetric multicategory $\cM$
and designated object $\retty \in \Obj(\cM)$.

Accordingly, fix such an $\cM$ and a $\retty$ for the remainder of this section.
Fix also a set $O$ and a map $\curlyinterp{-} : \Hom{\cM}{(\cdot)_\varnothing}{\retty} \to O$.

\begin{prop}
  $\toconflict{\cM}$ has a terminal object, given by the unique family over the
  empty coherence space, and can be equipped with a symmetric monoidal affine
  structure $(\tensor, \top)$ where
\[\left(\bigodot_{i \in I} A_i \right) \tensor \left( \bigodot_{j \in J} B_j \right) ~~=~~ \bigodot_{x \in I \with J} 
\left\{ \begin{array}{ll}
          A_i & \text{if $x = \inl(i)$} \\
          B_j & \text{if $x = \inr(j)$} \\
        \end{array} \right.\]
      and $I \with J$ designates the cartesian product in $\Fincoh$.
\end{prop}
\begin{proof}
  Left to the reader. Strictly speaking, later developments will depend on the
  precise structure itself and not merely on its existence, but there is a
  single sensible choice of bifunctorial action and structural morphisms making
  the above a monoidal product.
\end{proof}

\begin{rem}
To start making sense of the use of the cartesian product of $\Fincoh$,
there is a useful analogy with $\mcattoaff{\cM}$ here.
There is a projection functor $\mcattoaff{\cM} \to \Partfinset$ where $\Partfinset$
is the category of finite sets and partial functions. The tensorial product of $\mcattoaff{\cM}$
required a coproduct at the level of indices. Here, we have a projection functor $\toconflict{\cM}\to \Fincoh^\op$,
and we again use a coproduct at the level of indices (which becomes a product due to the contravariance).
\end{rem}

We call $\mcattoaff{\mfM}$ the tree streaming setting based on $\mcattoaff{\cM}$, $\retty$ and $\curlyinterp{-})$, and $\toconflict{\mfM}$ the corresponding tree streaming setting based on $\toconflict{\cM}$.
\begin{prop}
\label{prop:afftoconflict}
There is a full and faithful strong monoidal functor $\mcattoaff{\cM} \to \toconflict{\cM}$
extending to a morphism of streaming setting $\mcattoaff{\mfM} \to \toconflict{\mfM}$.
\end{prop}
\begin{proof}
Call $F$ this functor, and, for any set $I$, write $\Delta$ for the functor
$\Partfinset \to \Fincoh$ taking a set $I$ to the discrete coherence space $\Delta(I) = (I, \{(i,i) \mid i \in I\})$.
Note that we have $\Delta(I)^\bot = (I, I \times I)$, which may be regarded as the codiscrete coherence space generated by $I$.
On objects of $\mcattoaff{\cM}$,  we define $F$ as
\[F\left(\bigtensor_{i \in I} A_i\right) = \bigodot_{i \in \Delta(I)^\bot} A_i\]
For morphisms $(f, (\alpha_j)_{j \in J}) \in \Hom{\mcattoaff{\cM}}{\bigtensor_{i \in I} A_i}{\bigtensor_{j \in J} B_j}$, we set
\[F(f, (\alpha_j)_{j \in J}) = (\{(j,i) \mid j = f(i)\}, (\alpha_j)_{j \in J})\]
It is rather straightforward to check that $F$ is indeed full, faithful and strong monoidal, and the extension to
a morphism $\mcattoaff{\mfM} \to \toconflict{\mfM}$ is immediate.
\end{proof}

\begin{prop}
  $\toconflict{\cM}$ also has cartesian products, which may be defined as
\[\left(\bigodot_{i \in I} A_i \right) \with \left( \bigodot_{j \in J} B_j \right) ~~=~~ \bigodot_{x \in I \oplus J} 
  \left\{ \begin{array}{ll}
            A_i & \text{if $x = \inl(i)$} \\
            B_j & \text{if $x = \inr(j)$} \\
          \end{array} \right.\]
\end{prop}

(The proof is left to the reader.)
Therefore, we can extend Proposition~\ref{prop:afftoconflict}:
\begin{cor}
\label{cor:affwithtoconflict}
There is a functor $E : (\mcattoaff{\cM})_\with \to \toconflict{\cM}$ that is
full, faithful and lax (but \emph{not} strong) monoidal,
extending to a morphism of streaming settings $(\mcattoaff{\mfM})_\with \to \toconflict{\mfM}$.
\end{cor}
In the following proof and the rest of this section, we write explicitly $\toconflict\unit$ for the monoidal unit of $\tensor$ in $\toconflict{\cM}$
and ${\mcattoaff\unit}_\with$ for the unit in $(\mcattoaff\cM)_\with$.

\begin{proof}[Proof idea]
The universal property of the free product completion defines $E$ as the unique
product-preserving functor extending the functor of Proposition~\ref{prop:afftoconflict}. It remains to equip it with a lax monoidal functor structure. The map $m^0 : \toconflict\unit \to E({\mcattoaff\unit}_\with)$ is
an obvious isomorphism, while the natural transformation $m^2_{A,B} : E(A)
\otimes E(B) \to E(A \otimes B)$ can be obtained via the canonical map
\[\left(\bigwith_{i \in I} A_i\right) \tensor
\left(\bigwith_{j \in J} B_j\right) \to \bigwith_{(i,j) \in I \times J} A_i \tensor B_j\]
in $\toconflict{\cM}$ (it exists in all monoidal categories with products).
\end{proof}

We can now go the other way around.

\begin{lem}
There is a strong monoidal functor $\toconflict{\cM} \to (\mcattoaff{\cM})_\with$, which extends to
a morphism of streaming settings $\toconflict{\mfM} \to (\mcattoaff{\mfM})_\with$.
\end{lem}
\begin{proof}
For a coherence space $(\web X, \coh_X)$, write $\clique(X) \subseteq \powerset(X)$ the set of
\emph{cliques} of $X$
\[\clique(X) ~~=~~ \{ S \in \powerset(X) \mid \forall x \; y \in S. \; x \coh_X y \}\]

We now define the functor $F: \toconflict{\cM} \to (\mcattoaff{\cM})_\with$ on objects as
\[F\left(\bigodot_{x \in X} A_x\right) ~~=~~ \bigwith_{S \in \clique(X)} \bigtensor_{x \in S} A_x\]
As for morphisms, first recall that a morphism
$(R, \alpha) \in \Hom{\toconflict{\cM}}{\bigodot_{x \in X} A_x}{\bigodot_{y \in Y} B_y}$
consists of a linear map $R \in \Hom{\Fincoh}{Y}{X}$ and a family
\[(\alpha_{y})_{y \in \web Y} \in \prod_{y \in \web Y} \Hom{\cM}{(A_x)_{(y,x) \in R}}{B_y}\]
We set out to define
\[F(R,\alpha) \in
\Hom{(\mcattoaff{\cM})_\with}{
F\left(\bigtensor_{x \in S} A_x\right)}{
F\left(\bigtensor_{y \in S'} B_y\right)}\]
recalling that
\[
\begin{array}{l}
\displaystyle \Hom{(\mcattoaff{\cM})_\with}{
F\left(\bigtensor_{x \in S} A_x\right)}{
F\left(\bigtensor_{y \in S'} B_y\right)}
\\\\\hspace{10em}=~~
\displaystyle \prod_{S' \in \clique(Y)} \sum_{S \in \clique(X)}
\sum_{f : S \partto S'} \prod_{y \in S'}
\Hom{\cM}{(A_x)_{x \in f^{-1}(x)}}{B_y}
\end{array}
\]
So fixing $S' \in \clique(Y)$ and recall that $R$ being linear means that we have
\[
(y,x) \in R \wedge (y',x) \in R \qquad \Rightarrow \qquad \left\{
\begin{array}{lcll}
y \coh_Y y' &\Rightarrow& x \coh_X x' & (1) \\
x = x' &\Rightarrow& y \incoh_Y y' & (2) \\
\end{array}\right.
\]
In particular, $(1)$ implies that $\{ x \in \web X \mid (y,x) \in R \}$ is a clique; we take that to be $S$.
$(2)$, and the fact that $y \coh_Y y' \wedge y \incoh_Y y' \Rightarrow y = y'$, imply that $R$ determines a
(total) function $S \partto S'$, which we take to be $f$. Finally, once $y \in S'$ is fixed, we pick the component $\alpha_y$
to complete the definition, which makes sense as $f^{-1}(y) = \{ x \in S \mid (y,x) \in R\} = \{ x \in \web X \mid (y,x) \in R)\}$.
This completes the definition of $F(R,\alpha)$; we leave checking functoriality to the reader.

Now, we turn to defining a morphism of streaming setting $\toconflict{\mfM} \to (\mcattoaff{\mfM})_\with$ from $F$.
To this end, we must first equip $F$ with a lax monoidal structure, that is to define $(\mcattoaff{\cM})_\with$-maps
\[
m^0 : {\mcattoaff{\unit}}_\with \to F(\toconflict{\unit}) \quad\; m^2 : F\left(\bigodot_{x \in X} A_x\right) \tensor F\left(\bigodot_{y \in Y} B_y\right) \to F\left(\left[ \bigodot_{x \in X} A_x\right] \tensor \left[\bigodot_{y \in Y} B_y\right]\right)\]
satisfying the relevant coherence diagrams. We do not check them here, but indicate how to build those two maps.
$m^0$ arises as an obvious isomorphism $F(\toconflict{\unit}) \cong \mcattoaff{\unit}$.
$m^2$ is also an isomorphism, which may be computed as per Figure~\ref{fig:coh2aff-lax}.
\begin{figure}
\begin{align*}
F\left(\bigodot_{x \in X} A_x\right) \tensor F\left(\bigodot_{y \in Y} B_y\right)
\qquad&\cong\qquad
\left(\bigwith_{S \in \clique(X)} \bigtensor_{x \in S} A_x\right) \tensor
\left(\bigwith_{S' \in \clique(Y)} \bigtensor_{y \in S'} B_y\right)\\
&\cong\qquad
\bigwith\limits_{\substack{S \in \clique(X) \\ S' \in \clique(Y)}} \left(\bigtensor_{x \in X} A_x \right) \tensor \left(\bigtensor_{y \in Y} B_y \right)\\
&\cong\qquad
\bigwith_{S \in \clique(X \with Y)} \bigtensor_{z \in S} C_z \\
&\cong\qquad
F\left(\bigodot_{z \in X \with Y} C_z\right) \\
&\cong\qquad
F\left(\left(\bigodot_{x \in X} A_x\right) \tensor \left(\bigodot_{y \in Y} B_y\right)\right)\\
\\
\hline
\end{align*}
\vspace{-3em}
\caption{The main structural natural isomorphism making the functor
$\toconflict{\cM} \to (\mcattoaff{\cM})_\with$ strong monoidal, writing $A_x$ for $C_{\inl(x)}$
 and $B_y$ for $C_{\inr(y)}$.}
\label{fig:coh2aff-lax}
\end{figure}
Finally, there is a canonical isomorphism $F(\retty) \cong \retty$ which allows to complete the definition of the morphism $\toconflict{\mfM} \to (\mcattoaff{\mfM})_\with$.
\end{proof}

We thus conclude this section by first specializing the above to the case $\cM =
\Registerleone$, and then making a final tangential observation. 
\begin{lem}
\label{lem:conflict-with-equiv}
There are morphisms of streaming settings $\REGISTERleone_\with \to \REGISTERleoneconflict \to \REGISTERleone_\with$.
In particular, $\REGISTERleone_\with$-BRTTs compute exactly the regular functions.
\end{lem}

\begin{rem}
  One idea that one could take from Hu and Joyal's coherence
  completion~\cite{HuJoyal} -- but that we do not explore further here -- is to
  look at objects whose indexing coherence spaces are (up to isomorphism)
  generated from singletons by the $\with/\oplus$ operations of $\Fincoh$.
  (Those are called \enquote{contractible} in~\cite[Section~4]{HuJoyal}, and
  considering coherence spaces as undirected graphs, this corresponds to the
  classical notion of \emph{cograph} in combinatorics.)

  In the case of the coherence completion of some category $\cC$, the full
  subcategory spanned by such objects turns out to be the free completion of
  $\cC$ under finite products and coproducts (which differs from our
  $(-)_{\oplus\with}$ in not making `$\with$' distribute over `$\oplus$'); this
  is formalized as a universal property in~\cite[Theorem~4.3]{HuJoyal}. In the
  same vein, we conjecture that the full subcategory of $\toconflict\cM$
  consisting of cograph-indexed objects -- that is, of objects that are
  generated from those of $\cM$ by means of the operations $\otimes/\with$ in
  $\toconflict\cM$ -- is in some way the \emph{free affine symmetric monoidal
    category with products generated by the multicategory $\cM$}.
\end{rem}

\subsection{$\Register_{\oplus\with}$ is monoidal closed}
\label{subsec:dial-trees}

Now, we consider the category $\Register_{\oplus\with}$ in the context of trees.
Much like with strings, this category is symmetric monoidal monoidal closed with finite products and coproducts,
which makes it an ideal target to compile $\laml$-terms.
This structure over $\Register_{\oplus\with}$ is obtained in the same way as for strings:
the monoidal product over $\Register$ is defined as distributing over formal sums and
products and the usual products and coproducts are created by the $(-)_{\oplus\with}$ completion.
Similarly, monoidal closure can be obtained in a generic way once we show that the objects
coming from $\Register$ have internal homsets $\Register_\oplus$ (echoing Lemma~\ref{lem:hom-reg-oplus}).
This section is mostly dedicated to proving this fact, whose proof relies on decomposing linear trees in
a similar way as in Lemma~\ref{lem:leone}.
\begin{lem}
For any two objects $\bf R$ and $\bf S$ of $\Register$, there is an internal hom $\iota_\oplus(\bf R) \lin \iota_\oplus(\bf S)$ in $\Register_\oplus$.
\end{lem}
\begin{proof}
First, we treat the special case where ${\bf S} = \cI(U)$ for some finite set $U$.
To make sense of the definition of $\iota_\oplus({\bf R}) \lin \iota_{\oplus}(\cI(U))$ it is helpful to notice that
it will ultimately induce an isomorphism
\[\Hom{\Register_\oplus}{\top}{\iota_{\oplus}({\bf R}) \lin \iota_{\oplus}(\cI(U))} \cong \Hom{\Register}{{\bf R}}{\cI(U)} \cong \LTree_{\bf \Gamma}({\bf R} \tensor \cO(U))\]
so, recalling that objects of $\Register_\oplus$ are of the shape $\bigoplus_{i \in I} \bigtensor_{j \in J_i} \cI(V_j)$ for $V_j$ being finite sets,
the operational intuition is that one may code trees with ``holes with arity'' into some bounded finitary data (which we may informally
call a partitioning tree) plus finitely many trees containing holes ``without arity''; this bijection is pictured in Figure~\ref{fig:ltree-ex-decompose-fun}.
\begin{figure}
\center
\includegraphics[scale=0.8]{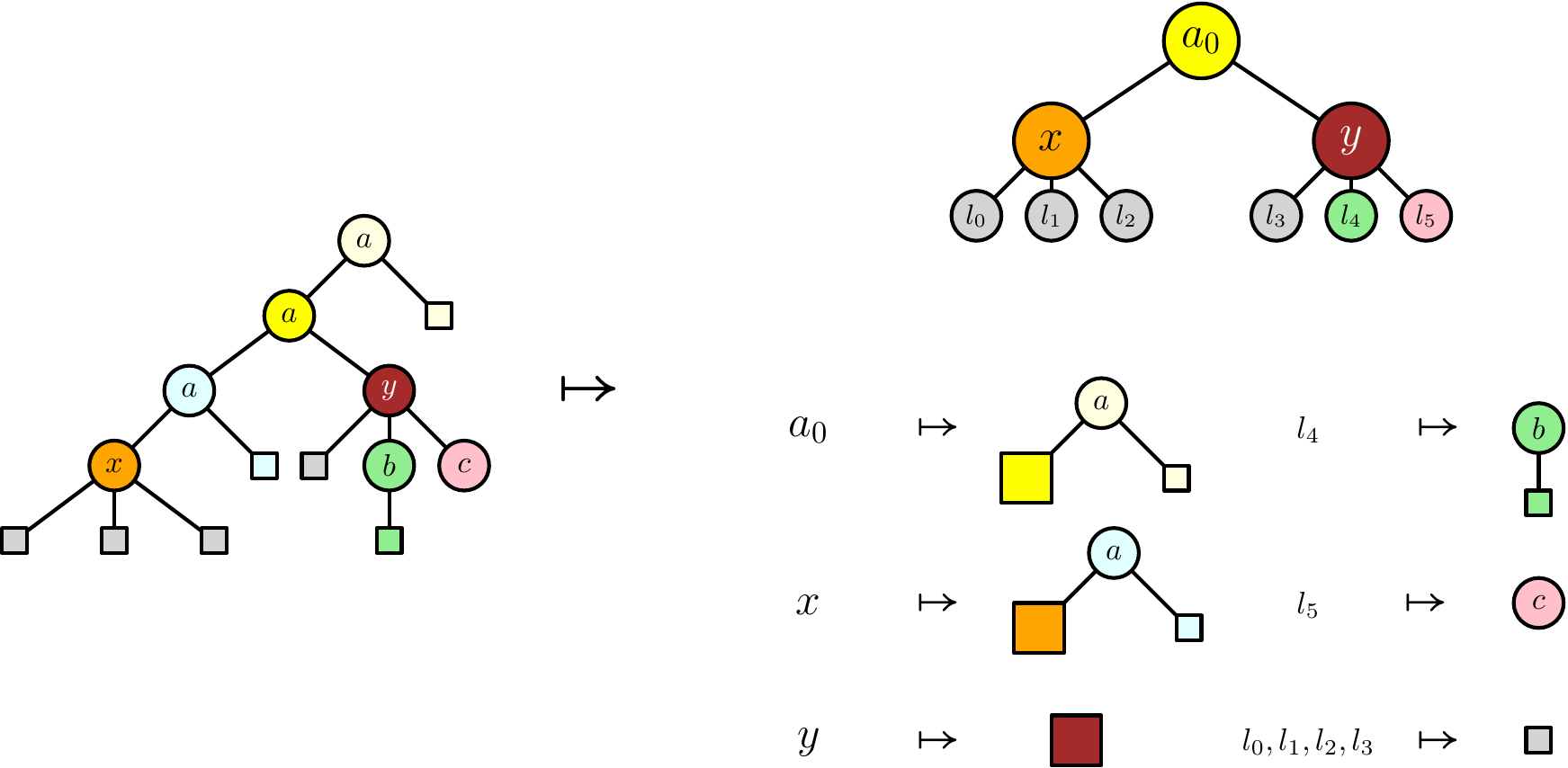}
\vspace{1.5em}
\hrule

\caption{Decomposition of a map of $\Hom{\Register(\{a : 2, b : 1, c : 0\})}{\{x:3,y:3\}}{\cI(7)}$ (which is defined as $\LTree_{\{a : 2, b : 1, c : 0\}}(\{x : 3, y : 3\} \tensor \cO(7))$)
as a tuple consisting of a partitioning tree and trees without the letters $x$ and $y$.}
\label{fig:ltree-ex-decompose-fun}
\end{figure}
As with Lemma~\ref{lem:leone}, we will not use this as our official definition for the internal homset, but rather use the following recursive definition:
\begin{itemize}
\item If ${\bf R} = \top$, set $\iota_\oplus({\bf R}) \lin \iota_\oplus(\cI(U)) ~=~\iota_\oplus(\cI(U))$.
\item Otherwise, define $\iota_\oplus({\bf R}) \lin \iota_\oplus(\cI(U))$ as
\[
\bigoplus_{U = V \uplus W} \iota_\oplus(\cI(V + 1)) \tensor \left( \bigoplus_{b \in {\bf \Gamma}} \left({\bf R} \lin_b \cI(W)\right) \oplus \bigoplus_{r \in R} \left({\bf R} \lin_r \cI(W)\right)\right) \\\\
\]
where ${\bf R} \lin_r \cI(W)$ and ${\bf R} \lin_b \cI(W)$ are auxiliary definitions which correspond to the following situations (recalling that morphisms can be regarded as trees):
\begin{itemize}
\item  $\iota_\oplus(\cI(V + 1)) \otimes \left({\bf R} \lin_r \cI(W)\right)$ correspond to morphisms such that there is a \emph{unique} minimal path leading from the root to a node labelled by a letter in $R$, and that letter is $r$.
The second component ${\bf R} \lin_r \cI(W)$ is meant to include the immediate subtrees of that node while the first $\iota_\oplus(\cI(V+1))$ contains the tree where a nullary node labeled is inserted instead of that node.
The combination of both these data and $r$ allows to recover the original morphism.
\item $\bigoplus_{b \in {\bf \Gamma}} \iota_\oplus(\cI(V+1)) \tensor {\bf R} \lin_b \cI(W)$ correspond to all the other morphisms.
In such a case there is a topmost node labelled by some letter $b$ of $\bf \Gamma$ with which has at least two distinct immediate subtrees which have at least one node of $R$ each.
Similarly to the first subcase, ${\bf R} \lin_b \cI(W)$ is intended to include the immediate subtrees of that node labeled by $b$ while $\iota_\oplus(\cI(V+1))$ contains the tree where a nullary node labeled is inserted instead of that node. The combination of these data and $b$ allows to recover the original morphism.
\end{itemize}
Their formal definition is as follows:
\[
\begin{array}{*4{>{\displaystyle}l}}
&
{\bf R} \lin_b \cI(W) &=&
\bigoplus\limits_{\substack{f : W \to \arity(b) \\
                            g : R \to \arity(b) \\
                            g \text{ nonconstant}}}
\bigtensor\limits_{x \in \arity(\Gamma)} \iota_\oplus(\restr{{\bf R}}{g^{-1}(x)}) \lin \iota_\oplus(f^{-1}(x)) \\\\
&
{\bf R} \lin_r \cI(W) &=&
\bigoplus\limits_{\substack{f : W \to \arity(r) \\
                            g : R \setminus \{r\} \to \arity(r)
                            }}
\bigtensor\limits_{x \in \arity(r)} \iota_\oplus(\restr{{\bf R}}{g^{-1}(x)}) \lin \iota_\oplus(f^{-1}(x))
\end{array}
\]
Note that the definitions of $\iota_{\oplus}({\bf R}) \lin \iota_{\oplus}(\cI(U))$,
${\bf R} \lin_b \cI(W)$ and ${\bf R} \lin_r \cI(W)$ mutually depend on one another.
Still this is is well-defined as the definitions of ${\bf R} \lin_b \cI(W)$ and ${\bf R} \lin_r \cI(W)$ only require $\iota_\oplus({\bf S}) \lin \iota_\oplus(\cI(V))$ for ${\bf S}$ strictly smaller than ${\bf R}$.
\end{itemize}
We now describe the associated evaluation map
\[
\ev_{{\bf R}, \cI(U)} ~~:~~ (\iota_\oplus({\bf R}) \lin \iota_\oplus(\cI(U))) \tensor {\bf R} ~~\longrightarrow~~ \cI(U)
\]
also by recursion over $\bf R$.
\begin{itemize}
\item If ${\bf R} = \top$, it is the identity.
\item Otherwise, we need to provide maps
\[
\left(\cI(V + 1) \tensor \left( \bigoplus_{b \in {\bf \Gamma}} \left({\bf R} \lin_b \cI(W)\right) \oplus \bigoplus_{r \in R} \left({\bf R} \lin_r \cI(W)\right)\right)\right) \tensor {\bf R}
~~\longrightarrow~~ \cI(U)
\]
for every decomposition $U = V \uplus W$, the intuition being that $\cI(V+1)$ is a context containing the top of the tree corresponding to the function
we want to apply. Therefore, once we provide a map
\[
\left( \bigoplus_{b \in {\bf \Gamma}} \left({\bf R} \lin_b \cI(W)\right) \oplus \bigoplus_{r \in R} \left({\bf R} \lin_r \cI(W)\right)\right) \tensor {\bf R}
~~\longrightarrow~~ \cI(W)
\]
we may post-compose it with 
$\cI(V+1) \tensor \cI(W) \to \cI(V + W) \cong \cI(V \uplus W) = \cI(U)$
to define $\ev_{{\bf R} \lin \cI(U)}$.
Recalling that $\tensor$ distributes over $\oplus$, it suffices to provide specialized maps
\[
\ev^b_{{\bf R}, \cI(W)} : {\bf R} \lin_b \cI(W) \tensor {\bf R} \to \cI(W)
\quad \text{and}\quad
\ev^r_{{\bf R}, \cI(W)} : {\bf R} \lin_r \cI(W) \tensor {\bf R} \to \cI(W)
\]
for $b \in \Gamma$ and $r \in R$, which we describe now.
\begin{itemize}
\item For $\ev^b_{{\bf R}, \cI(W)}$, it suffices to define a family of $\Register$-maps
indexed by $f : W \to \arity(b)$ and $g : R \to \arity(b)$ with $g$ non-constant
\[
\displaystyle\left(\bigtensor_{x \in \arity(\Gamma)} (\restr{{\bf R}}{g^{-1}(x)}) \lin \cI(f^{-1}(x))\right) \tensor {\bf R} ~~\longrightarrow~~ \cI(W)
\]
Recall that $b$ can be seen as tree constructor and induces a canonical map
\[
\overline{b} ~~:~~ \bigtensor_{x \in \arity(b)} \cI(f^{-1}(x)) ~~\longrightarrow~~ \cI(W)
\]
Using the induction hypothesis, we have evaluation maps
\[
\left((\restr{{\bf R}}{g^{-1}(x)}) \lin \cI(f^{-1}(x))\right) \tensor (\restr{{\bf R}}{g^{-1}(x)}) ~~\longrightarrow~~ \cI(W)
\]
We can then compose $\overline{b}$ with the product of those maps over $x \in \arity(b)$ and then the isomorphism ${\bf R} \cong \bigtensor_{x \in \arity(b)} \restr{{\bf R}}{f^{-1}(x)}$
to conclude the definition of $\ev^b_{{\bf R}, \cI(W)}$.
\item For $\ev^r_{{\bf R}, \cI(W)}$, it suffices to define a family of $\Register$-maps indexed by $f : W \to \arity(r)$ and $g : R \setminus \{r \} \to \arity(r)$
\[
\left(\bigtensor_{x \in \arity(\Gamma)} (\restr{{\bf R}}{g^{-1}(x)}) \lin \cI(f^{-1}(x))\right) \tensor {\bf R} ~~\longrightarrow~~ \cI(W)
\]
By exploiting the isomorphism ${\bf R} \cong \cI(\arity(r)) \tensor \bigtensor_{x \in \arity(r)} (\restr{{\bf R}}{g^{-1}(x)})$ and using the inductive hypothesis as in the
previous case, we obtain a map
\[
\left(\bigtensor_{x \in \arity(\Gamma)} (\restr{{\bf R}}{g^{-1}(x)}) \lin \cI(f^{-1}(x))\right) \tensor {\bf R} ~~\longrightarrow~~ \displaystyle\cI(\arity(r)) \tensor \bigtensor_{x \in \arity(r)} \cI(f^{-1}(x))
\]
and we may conclude by post-composing by the map
\[\cI(\arity(r)) \tensor \bigtensor_{x \in \arity(r)} \cI(f^{-1}(x)) \to
  \cI(W)\]
which is induced by the depth-$2$ tree whose root
corresponds to the first component, whose children correspond to the successive elements of $\bigtensor_{x \in \arity(r)} \cI(f^{-1}(x))$ and other leaves are in $W$.
\end{itemize}
\end{itemize}
While the definition is a bit wordy, there is then little difficulty in checking that this yields the expected universal property.
\[
\xymatrix@C=30mm{
\iota_{\oplus}({\bf R} \lin \cI(U)) \tensor \iota_\oplus(\cI(U)) \ar[r] & \iota_{\oplus}(\cI(U)) \\
A \tensor \iota_{\oplus}({\bf R}) \ar@{-->}[u]^{\Lambda(h) \tensor \id} \ar[ur]_h
}
\]
One needs then to extend the definition for the general case where $\bf S$ is not necessarily $\cI(U)$; this
is done using a similar approach as for strings, by using a coproduct over partial maps $R \partto S$ tracking
which letter of the input participates in which letter of the output, and employing the particular case where
there is one letter in the output\footnote{This could be factored out as a more general result concerning the existence of internal homsets in categories of the shape $\mcattoaff{\cM}$ for multicategories $\cM$.}
\[\iota_\oplus({\bf R}) \lin \iota_\oplus({\bf S}) ~~=~~ \bigoplus_{f : R \partto S} \bigtensor_{s \in S} \iota_\oplus(\restr{{\bf R}}{f^{-1}(s)}) \lin \iota_\oplus(\cI(\arity(s)))\]
The evaluation function can then be extended and the universal property accordingly lifted to the more general case.
\end{proof}

The above lemma tells us that $\Register_{\oplus\with}$ satisfies the premises
of Theorem~\ref{thm:dial-haslin}, hence:

\begin{thm}
\label{thm:REGISTER-closed}
The category $\Register_{\oplus\with}$ has cartesian products, coproducts and a
symmetric monoidal \emph{closed} structure.
\end{thm}

\begin{rem}
  Given the close relationship between the constructions
  $(\mcattoaff{(-)})_\with$ and $\toconflict{(-)}$, it seems plausible that
  $(\toconflict\Register)_\oplus$ could be monoidal closed (we know that it has
  products, coproducts and a symmetric monoidal structure). We leave this question
  to further work.
\end{rem}

In other words, $\Register_{\oplus\with}$ provides a categorical semantics for
the purely linear fragment of the $\laml$-calculus. Concretely, by a suitable
adaptation of Lemma~\ref{lem:laml-initial} to tree streaming settings, we
finally have:
\begin{cor}
  \label{cor:laml-tree-to-registeropluswith}
  There is a morphism of streaming settings $\mfLam \to \REGISTER_{\oplus\with}$.
\end{cor}

\subsection{Preservation properties of finite completions}

We have now proven all of the combinatorial results leading to the main theorem for trees.
At this stage, we only need to make explicit a few preservation properties of the completions
$(-)_\oplus$ and $(-)_\with$. We only state the minimal requirements that we need to proceed\footnote{
The following, which would entail these requirements, should hold (but we have
not checked it): the completions $(-)_\oplus$ and $(-)_\with$ should behave as pseudo-monads over the
2-category of tree streaming settings, admitting a pseudo-distributive law $((-)_\oplus)_\with \to ((-)_\with)_\oplus$
giving rise to the pseudo-monad $(-)_{\oplus\with}$.}.

\begin{lem}
\label{lem:compl-funct}
Let $\square \in \{ \oplus, \with, \oplus\with\}$ and $\mfC, \mfD$ be streaming settings.
Then if there is a morphism $\mfC \to \mfD$, there is also a morphism $\mfC_\square \to \mfC_\square$.
\end{lem}
\begin{proof}[Proof idea]
Let us only discuss the case $\square = \oplus$ and assume that we have a streaming setting morphism
whose underlying lax monoidal functor is $F : \cC \to \cD$. Then we consider $F_\oplus : \cC_\oplus \to \cD_\oplus$
defined using the universal property of $\cC_\oplus$ from the functor $\iota_\oplus \circ F : \cC \to \cD_\oplus$
such that
\[F_\oplus\left(\bigoplus_{u \in U} \iota_\oplus(C_u)\right) ~~=~~ \bigoplus_{u \in U} \iota_\oplus(F(C_u)) \]
so that $F_\oplus$ inherits a lax monoidal structure from $F$ by lifting the relevant maps functorially. For instance,
we have $\iota_\oplus(\unit) \to \iota_\oplus(F(\unit)) =
F_\oplus(\unit)$ for the unit. We leave checking the coherence diagram, the case of the binary tensor to the reader,
as well as checking that the rest of the structure of morphisms of streaming settings lift accordingly.
\end{proof}

\begin{lem}
\label{lem:distributive-law-tss}
Let $\mfC$ be a streaming setting. Then there are morphisms of streaming settings $\mfC_{\oplus\with} \to (\mfC_\oplus)_{\oplus\with}$
and $ (\mfC_\oplus)_{\oplus\with} \to \mfC_{\oplus\with}$.
\end{lem}
\begin{proof}[Proof idea]
The first morphism can be obtained using Lemma~\ref{lem:compl-funct} with the morphism of streaming setting $\mfC \to \mfC_\oplus$
corresponding to $\iota_\oplus$. The second morphism can be written as a composite
\[
\xymatrix@C=20mm{
(\mfC_\oplus)_{\oplus\with} \cong ((\mfC_\oplus)_\with)_\oplus \ar[r]^-{(\mathfrak{Dist}_\mfC)_\oplus} & ((\mfC_\with)_{\oplus})_\oplus \ar[r]^-{\mathfrak{Mu}_{\mfC_\with}} & (\mfC_{\with})_\oplus \cong \mfC_{\oplus\with}}\]
where $(\mathfrak{Dist}_\mfC)_\oplus$ is obtained from a morphism $\mathfrak{Dist}_\mfC : (\mfC_\oplus)_\with \to (\mfC_\with)_\oplus$ by Lemma~\ref{lem:compl-funct}.
$\mathfrak{Dist}_\mfC$ itself is built from a functor $D_\cC : (\cC_\oplus)_\with \to (\cC_\with)_\oplus$ defined on object following the mantra ``products distribute over coproducts''.
\[ D_\cC\left(\bigwith_{x \in X} \iota_\with\left(\bigoplus_{u \in U_x} \iota_\oplus(C_{x,u})\right)\right) ~~ = ~~ \bigoplus_{F \in \prod_x U_x} \iota_\oplus\left( \bigwith_{x \in X} \iota_\with(C_{x,u})\right)\]
On the other hand, the functor $M : \cC$ which is part of $\mathfrak{Mu}_\mfC$ is obtained by using the universal property of $(-)_\oplus$ so that
\[M\left(\bigoplus_{u \in U} \iota_\oplus \left(\bigoplus_{v \in V} \iota_\oplus(C_{u,v})\right)\right) ~~ = ~~ \bigoplus_{(u,v)} \iota_\oplus(C_{u,v})\]
\end{proof}

\subsection{Proof of the main result on trees}
\label{subsec:main-trees}

We can now give the proof of Theorem~\ref{thm:main-tree}, which can be summarized as the string of equalities
\[\begin{array}{lcl@{\quad}l}
\laml &=& 
\text{single-state $\mfLam$-BRTTs}
& \text{\small by Lemma~\ref{lem:lamlbrtt}} \\
&=&
\text{single-state $\REGISTERleoneconflict_\oplus$-BRTTs}
&(\dagger) \\
&=&
\text{$\REGISTERleoneconflict$-BRTTs}
& \text{\small by Theorem~\ref{thm:oplus-brtt-conservative}}\\
&=&
\text{regular tree functions}
& \text{\small by Proposition~\ref{prop:brtt-regular}}
  \end{array}\]

Note that applying Theorem~\ref{thm:oplus-brtt-conservative} in the last step requires that objects
of $\Registerleoneconflict$ have unitary support, which is easily checked for output alphabets
${\bf \Gamma}$ containing at least one letter of arity $\varnothing$.
It remains to justify step $(\dagger)$, i.e.,
that single-state $\mfLam$-BRTTs and $\REGISTERleoneconflict_\oplus$-BRTTs are equi-expressive.
By Lemma~\ref{lem:morph-BRTT}, it suffices to have morphisms of streaming settings
both ways $\mfLam \leftrightarrow \REGISTERleoneconflict_\oplus$ to conclude. They may be obtained
as the following composites (where we leave implicit the uses of Lemma~\ref{lem:compl-funct} and some easily inferrable steps).
\[
\xymatrix{
\mfLam
\ar@/^1pc/[r]^{\text{Cor.~\ref{cor:laml-tree-to-registeropluswith}}}
&
\REGISTER_{\oplus\with}
\ar@/^1pc/[r]^{\text{Lem.~\ref{lem:distributive-law-tss}}}
\ar@/^1pc/[l]^{\text{Cor.~\ref{cor:registeropluswith-to-laml-tree}}}
&
(\REGISTER_\oplus)_{\oplus\with}
\ar@/^1pc/[r]^{\text{Lem~\ref{lem:leone}}}
\ar@/^1pc/[l]
&
(\REGISTERleone_\oplus)_{\oplus\with}
\ar@/^1pc/[r]^{\text{Lem.~\ref{lem:distributive-law-tss}}}
\ar@/^1pc/[l]&
\REGISTERleone_{\oplus\with}
\cong
\ar@/^1pc/[l]
(\REGISTERleone_\with)_{\oplus}
\ar@/^1pc/[r]^-{\text{Lem.~\ref{lem:conflict-with-equiv}}}
&
\REGISTERleoneconflict_{\oplus}
\ar@/^1pc/[l]
}\]

\section{Conclusion \& further work}

We have proven that the tree (and string) functions definable at type
$\Treety_{\bf \Sigma}[A] \lin \Treety_{\bf \Gamma}$ in the $\laml$-calculus
correspond exactly to regular functions. To prove the non-trivial left-to-right
inclusion, we used a syntactic lemma for $\laml$-terms
allowing to compile those terms into
a kind of transducer model using purely linear $\laml$-terms as its memory ($\mfLam$-BRTTs).
We then showed in a principled way
that those transducers are equivalent to the more standard single-use-restricted
Bottom-up Ranked Tree Transducers
(expressed as $\REGISTERleoneconflict$-BRTTs in our framework) which capture regular tree functions.
Along the way, we revisited a few results on  string
transducers under the light of our categorical framework, exhibiting some points
of convergence with our semantics for $\laml$-terms.

There are a number of ways one could plan on expanding this work. We list some
of the most relevant problems, regarding the definable
functions between strings and trees in various $\lambda$-calculi, that we expect could be tackled using similar
methods or minor variations of our setting.

\subsubsection*{Removing the additive connectives}
$\laml$ features the additive connectives $\oplus$ and $\with$ (as well as the relevant
units $0$ and $\top$). A natural question is whether all regular functions are
still encodable if we remove access to those connectives and use the $\lama$-calculus, an \emph{affine}
version of $\laml$ without those connectives; for strings, this is first half of~\cite[Claim 6.2]{aperiodic}.

In a sense, we went further than that in showing that $\laml$-definable functions are regular,
since it can be checked that all $\lama$-definable string-to-string functions are also $\laml$-definable.
However, we did not give evidence for the converse direction, as our encoding of states of BRTTs use the
additive connectives. We believe this can be remedied by leveraging our previous
work on encoding sequential transducers~\cite[Sections 4-5]{aperiodic} (that
relies on the highly non-trivial Krohn-Rhodes theorem).

As for functions taking trees as input, we suspect that the
$\lama$-calculus is \emph{strictly less expressive} than copyless BRTTs
($\REGISTER$-BRTTs) -- in fact, that $\lama$-terms are
 even unable to recognize all regular tree \emph{languages}. We expect that this can be shown
 by interpreting purely affine $\lama$-terms into (a variant of) the geometry of
 interaction semantics mentioned in Example~\ref{ex:hines}.
Our intent is to deduce from this that $\lama$-definable languages can be recognized by \emph{tree-walking
automata} -- which are indeed a natural extension to trees of the two-way
automata from Example~\ref{ex:hines} -- and some regular tree languages are known to be unrecognizable by
such devices~\cite{bojanczyk2008tree}.

Anecdotally, we show in~\Cref{sec:app-with} that removing the additive
disjunction `$\oplus$' from the $\laml$-calculus while keeping the additive
conjunction `$\with$' does not affect the definability of tree functions. While
this is a routine application of a continuation-passing-style transformation,
the nice thing is that it can be entirely carried out at the level of streaming
settings.

\subsubsection*{Dropping commutativity}
The second half of~\cite[Claim 6.2]{aperiodic} raised the issue of definable string-to-string
transductions in the $\lamp$-calculus, the restriction of the $\lama$-calculus to a \emph{non-commutative} multiplicative structure.
By analogy with the main result of~\cite{aperiodic}, which was that
$\lamp$-definable languages are star-free (i.e.\ first-order definable) while $\lama$
defines regular
languages, the natural guess is that they correspond to \emph{first-order regular
  functions}. We expect non-commutative typing to translate into a
\emph{planarity} condition in the above-mentioned category of diagrams (close to a free (non-symmetric) compact-closed category~\cite{kelly1980coherence}).

The situation with trees looks more delicate, although trying to compare the
expressiveness of the
$\lamp$-calculus with the functional combinators described in~\cite{FOTree} might be a good starting point.

\subsubsection*{Duplicating the input}
Going to higher complexities, but without escalating to the towers of
exponentials definable in the simply-typed $\lambda$-calculus,
one natural question is: how expressive are $\laml$-terms of type $\Treety_{\bf \Sigma}[\kappa] \to
\Treety_{\bf \Gamma}$, for purely linear $\kappa$? We expect that this problem
could be tackled by reusing our semantic
interpretation of purely linear $\laml$-terms together with a suitable
refinement of Lemma~\ref{lem:laml-niceshape}. Concerning the case of strings,
our current conjecture is that the expressible functions have polynomial growth
and form a strict subclass of \emph{polyregular
  functions}~\cite{polyregular,polyregularMSO}.

\subsubsection*{Semantic interpretation of exponentials}
One frustrating aspect of our approach is that we need to rely on the normalization of $\laml$-terms
and a special-purpose syntactic lemma to dissect the term under consideration and extract the relevant purely
linear subterms before being able to go to semantic interpretations and more
high-level arguments.
The reason behind this is that our semantics does not interpret the non-linear arrow
`$\to$' or, equivalently, a version of the exponential modality `$\oc$' of linear logic.
It might be more pleasant to have more general semantic settings with those features that still allow
us to show our characterizations. Furthermore, such a setting would be \emph{necessary} if we hope to
tackle the question of the expressiveness of functions $\Treety_{\bf \Sigma}[A] \to \Treety_{\bf \Sigma}$
for general $A$ using semantic tools (this last variant of the problem is equivalent to the setting
of the simply-typed $\lambda$-calculus without linearity constraints).

\bibliographystyle{alphaurl}
\bibliography{bi}

\appendix

\section{Alur and D'Antoni's Bottom-Up Ranked Tree Transducers}
\label{sec:app-brtt}
We give here a self-contained definition of the notion of BRTT corresponding to
regular tree functions, as they were designed in~\cite{BRTT}. Since the
paper~\cite{BRTT} is mainly concerned with transducers over unranked trees, the
information concerning the definition of BRTTs is spread over its sections 2.1,
3.7 and 3.8. We restrict here
to binary trees (as in \cite[Sections 3.7 and 3.8]{BRTT}) and avoid using
multicategories.

\begin{defi}[{\cite[p.~31:36]{BRTT}}]
  The set $\BinTree(\Sigma)$ of \emph{binary trees} over the alphabet $\Sigma$,
  and the set $\dBinTree(\Sigma)$ of \emph{one-hole binary trees} are generated
  by the respective grammars
  \[ T,U ::= \trnode \mid a\trpair{T}{U}\quad (a \in \Sigma) \qquad\qquad
    T',U' ::= \square \mid a\trpair{T'}{U} \mid a\trpair{T}{U'} \quad (a \in
    \Sigma) \]

  That is, $\BinTree(\Sigma)$ consists of binary trees whose leaves are all
  equal to $\langle \rangle$ and whose nodes are labeled with letters in
  $\Sigma$. As for $\dBinTree(\Sigma)$, it contains trees with exactly one leaf
  labeled $\square$ instead of $\trnode$. This \enquote{hole} $\square$ is
  intended to be substituted by a tree: for $T' \in \dBinTree(\Sigma)$ and $U
  \in \BinTree(\Sigma)$, $T'[U]$ denotes $T'$ where $\square$ has been replaced
  by $U$.
\end{defi}
\begin{defi}[{\cite[p.~31:40]{BRTT}}]
  The \emph{binary tree expressions} ($E,F$ below, forming the set
  $\ExprBT(\Sigma,V,V')$) and \emph{one-hole binary tree expressions} ($E',F'$
  below, forming the set $\ExprdBT(\Sigma,V,V')$) over the variable sets $V$ and
  $V'$ are generated by the grammar (with $x \in V$, $x' \in V'$ and $a \in
  \Sigma$)
  \[ E,F ::= \trnode \mid x \mid a\trpair{E}{F} \mid E'[F] \qquad E',F' :=
    \square \mid x' \mid a\trpair{E'}{F} \mid a\trpair{E}{F'} \mid
    E'[F'] \]
  Given $\rho : V \to \BinTree(\Sigma)$ and $\rho' : V' \to \dBinTree(\Sigma)$,
  one defines $E(\rho,\rho') \in \BinTree(\Sigma)$ for $E \in \ExprBT(\Sigma)$
  and $E'(\rho,\rho') \in \dBinTree(\Sigma)$ for $E \in \ExprdBT(\Sigma)$ in the
  obvious way.
\end{defi}

\begin{defi}[{\cite[\S3.7]{BRTT}}]
  Let us fix an input alphabet $\Gamma$ and output alphabet $\Sigma$. The set of
  \emph{register assignments} over two disjoint sets $R,R'$, whose elements are
  called \emph{registers}, is
  \[ \assi(\Sigma,R,R') = \ExprBT(\Sigma,R_{\ttl\ttr},R'_{\ttl\ttr})^R \times
    \ExprdBT(\Sigma,R_{\ttl\ttr},R'_{\ttl\ttr})^{R'} \quad\text{where}\
    R_{\ttl\ttr} = R\times\{\ttl,\ttr\}\]
              A \emph{register tree transducer} consists of a \emph{finite} set $Q$ of
  \emph{states} with an \emph{initial state} $q_I \in Q$, two disjoint
  \emph{finite} sets $R,R'$ of \emph{registers}, a \emph{transition function}
  $\delta : Q \times Q \times \Gamma \to Q \times \assi(\Sigma,R,R')$ and an
  \emph{output function} $F : Q \to \ExprBT(\Sigma,R,R')$. To each tree $T \in
  \BinTree(\Gamma)$, it associates inductively a \emph{configuration} $\Conf(T)
  \in Q \times \BinTree(\Sigma)^R \times \dBinTree(\Sigma)^{R'}$:
  \begin{itemize}
  \item The base case is $\Conf(\trnode) = (q_I, (r \mapsto \trnode), (r'
    \mapsto \square))$.
  \item When $\Conf(T) = (q_\ttl,\rho_\ttl,\rho'_\ttl)$, $\Conf(U) =
    (q_\ttr,\rho_\ttr,\rho'_\ttr)$ and $\delta(c,q_\ttl,q_\ttr) =
    (q,(\varepsilon,\varepsilon'))$, we set $\Conf(c\trpair{T}{U}) = (q, (r
    \mapsto \varepsilon(r)(\rho,\rho')),(r' \mapsto
    \varepsilon'(r')(\rho,\rho')))$ where $\rho(r,d) = \rho_d(r)$ for $(r,d) \in
    R \times \{\ttl,\ttr\}$ and similarly for $\rho'$.
  \end{itemize}
  The function defined by the transducer is $T \in \BinTree(\Gamma) \mapsto
  F(q_{\mathrm{fin}}(T))(\rho_{\mathrm{fin}}(T), \rho'_{\mathrm{fin}}(T))$ where
  $(q_{\mathrm{fin}}(T),\rho_{\mathrm{fin}}(T), \rho'_{\mathrm{fin}}(T)) =
  \Conf(T)$ (recall that $F$ is the output function).
\end{defi}

\begin{figure}
\tikzstyle{vertex}=[circle,fill=black,minimum size=8pt,inner sep=0pt]
\begin{center}
  \begin{tikzpicture}
    \coordinate (root) at (4,6) {};
    \coordinate (l) at (2,5) {};
    \coordinate (ll) at (1,4) {};
    \coordinate (lr) at (3,4) {};
    \coordinate (lrl) at (2.5,3) {};
    \coordinate (lrr) at (3.5,3) {};
    \coordinate (r) at (6,5) {};
    \coordinate (rl) at (5,4) {};
    \coordinate (rr) at (7,4) {};

    \draw (0.75,4) rectangle (1.25,3.5);
    \draw (0.75,3.5) rectangle (1.25,3);
    \node[vertex] () at (1,3.75) {};
    \node[vertex] () at (1,3.25) {};

    \draw (2.25,3) rectangle (2.75,2.5);
    \draw (2.25,2.5) rectangle (2.75,2);
    \node[vertex] () at (2.5,2.75) {};
    \node[vertex] () at (2.5,2.25) {};

    \draw (3.25,3) rectangle (3.75,2.5);
    \draw (3.25,2.5) rectangle (3.75,2);
    \node[vertex] () at (3.5,2.75) {};
    \node[vertex] () at (3.5,2.25) {};

    \draw (4.75,4) rectangle (5.25,3.5);
    \draw (4.75,3.5) rectangle (5.25,3);
    \node[vertex] () at (5,3.75) {};
    \node[vertex] () at (5,3.25) {};

    \draw (6.75,4) rectangle (7.25,3.5);
    \draw (6.75,3.5) rectangle (7.25,3);
    \node[vertex] () at (7,3.75) {};
    \node[vertex] () at (7,3.25) {};
    \draw[thick] (root) -- (l);
    \draw[thick] (root) -- (r);
    \draw[thick] (l) -- (ll);
    \draw[thick] (l) -- (lr);
    \draw[thick] (lr) -- (lrl);
    \draw[thick] (lr) -- (lrr);
    \draw[thick] (r) -- (rl);
    \draw[thick] (r) -- (rr);

    \node[vertex,fill=blue] () at (root) {};

    \node[vertex,fill=red] () at (l) {};

    \node[vertex,fill=blue] () at (lr) {}; 
    \node[vertex,fill=red] () at (r) {};
  \end{tikzpicture}
\end{center}
\begin{center}
  \begin{tikzpicture}
    \coordinate (root) at (4,6) {};
    \coordinate (l) at (2,5) {};
    \coordinate (ll) at (1,4) {};
    \coordinate (lr) at (3,4) {};
    \coordinate (lrl) at (2.5,3) {};
    \coordinate (lrr) at (3.5,3) {};
    \coordinate (r) at (6,5) {};
    \coordinate (rl) at (5,4) {};
    \coordinate (rr) at (7,4) {};

    \draw (0.75,4) rectangle (1.25,3.5);
    \draw (0.75,3.5) rectangle (1.25,3);
    \node[vertex] () at (1,3.75) {};
    \node[vertex] () at (1,3.25) {};
    \draw (2,4) rectangle (4,2);
    \node[vertex,fill=blue] (a) at (3,3.8) {};
    \node[vertex] (b) at (2.2,2.2) {};
    \node[vertex] (c) at (3.8,2.2) {};
    \draw[thick] (a) -- (b);
    \draw[thick] (a) -- (c);
    \draw (2,2) rectangle (4,0);
    \node[vertex,fill=blue] (a) at (3,1.8) {};
    \node[vertex] (b) at (2.2,0.2) {};
    \node[vertex] (c) at (3.8,0.2) {};
    \draw[thick] (a) -- (b);
    \draw[thick] (a) -- (c);

    \draw (5,5) rectangle (7,3);
    \node[vertex,fill=red] (a) at (6,4.8) {};
    \node[vertex] (b) at (5.2,3.2) {};
    \node[vertex] (c) at (6.8,3.2) {};
    \draw[thick] (a) -- (b);
    \draw[thick] (a) -- (c);
    \draw (5,3) rectangle (7,1);
    \node[vertex,fill=red] (a) at (6,2.8) {};
    \node[vertex] (b) at (5.2,1.2) {};
    \node[vertex] (c) at (6.8,1.2) {};
    \draw[thick] (a) -- (b);
    \draw[thick] (a) -- (c);
    \draw[thick] (root) -- (l);
    \draw[thick] (root) -- (r);
    \draw[thick] (l) -- (ll);
    \draw[thick] (l) -- (lr);

    \node[vertex,fill=blue] () at (root) {};

    \node[vertex,fill=red] () at (l) {};
  \end{tikzpicture}
\end{center}
\begin{center}
  \begin{tikzpicture}
    \coordinate (root) at (4,6) {};
    \coordinate (l) at (2,5) {};
    \coordinate (ll) at (1,4) {};
    \coordinate (lr) at (3,4) {};
    \coordinate (lrl) at (2.5,3) {};
    \coordinate (lrr) at (3.5,3) {};
    \coordinate (r) at (6,5) {};
    \coordinate (rl) at (5,4) {};
    \coordinate (rr) at (7,4) {};

    \draw (5,5) rectangle (7,3);
    \node[vertex,fill=red] (a) at (6,4.8) {};
    \node[vertex] (b) at (5.2,3.2) {};
    \node[vertex] (c) at (6.8,3.2) {};
    \draw[thick] (a) -- (b);
    \draw[thick] (a) -- (c);
    \draw (5,3) rectangle (7,1);
    \node[vertex,fill=red] (a) at (6,2.8) {};
    \node[vertex] (b) at (5.2,1.2) {};
    \node[vertex] (c) at (6.8,1.2) {};
    \draw[thick] (a) -- (b);
    \draw[thick] (a) -- (c);

    \draw (0,5) rectangle (4,2.5);
    \node[vertex,fill=blue] (a) at (3,3.8) {};
    \node[vertex] (b) at (2.2,2.7) {};
    \node[vertex] (c) at (3.8,2.7) {};
    \draw[thick] (a) -- (b);
    \draw[thick] (a) -- (c);
    \node[vertex,fill=red] (d) at (2,4.8) {};
    \node[vertex] (e) at (1,3.8) {};
    \draw[thick] (d) -- (a);
    \draw[thick] (d) -- (e);
    \draw (0,2.5) rectangle (4,0);
    \node[vertex,fill=blue] (a) at (1,1.3) {};
    \node[vertex] (b) at (0.2,0.2) {};
    \node[vertex] (c) at (1.8,0.2) {};
    \draw[thick] (a) -- (b);
    \draw[thick] (a) -- (c);
    \node[vertex,fill=red] (d) at (2,2.3) {};
    \node[vertex] (e) at (3,1.3) {};
    \draw[thick] (d) -- (a);
    \draw[thick] (d) -- (e);
    \draw[thick] (root) -- (l);
    \draw[thick] (root) -- (r);
    \node[vertex,fill=blue] () at (root) {};
  \end{tikzpicture}
\end{center}
\caption{First few steps of the run of the bottom-up ranked tree transducer of
  Figure~\ref{brtt-example} over a tree whose alphabet of node labels is $\Sigma
  = \{{\color{blue}\bullet},{\color{red}\bullet}\} \cong \{a,b\}$.\protect\\
  The configuration at each subtree is represented
  by two boxes; the top (resp.\ bottom) box displays the contents of $r_1$
  (resp.\ $r_2$). (The single state is omitted from the visual representation of
  the configuration.)}
\label{fig:brtt-fig2}
\end{figure}
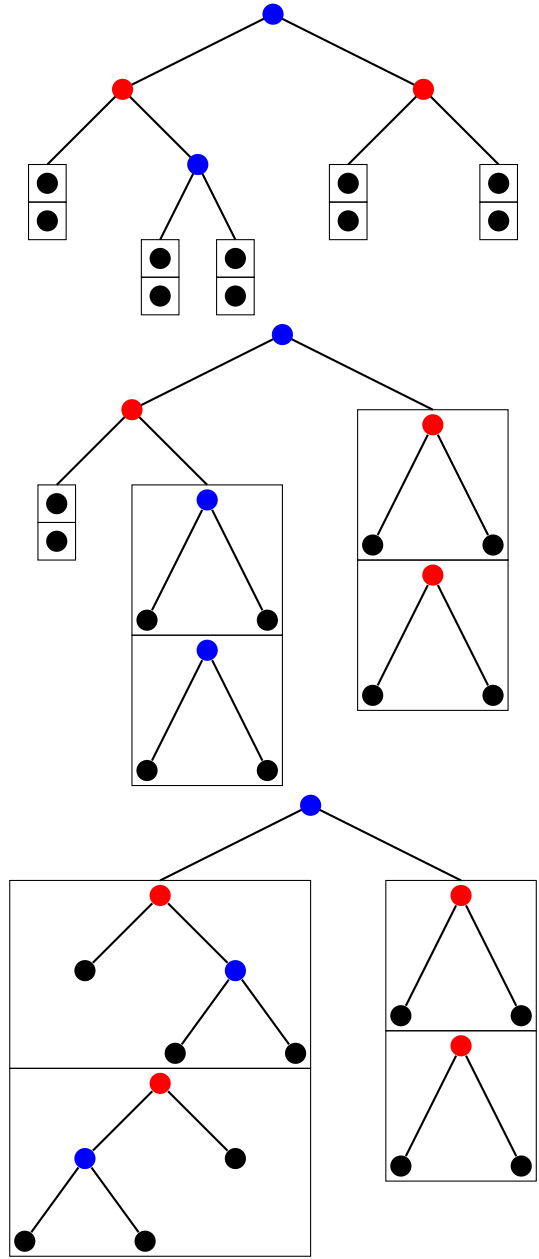

\begin{exa}[illustrated by Figure~\ref{fig:brtt-fig2}]
  \label{brtt-example}
  Let us consider a transducer over the alphabets $\Gamma = \Sigma = \{a,b\}$,
  with a single state ($|Q| = 1$) and two registers, both tree-valued (so $R =
  \{r_1, r_2\}$ and $R' = \varnothing$). This simplifies the transition function
  $\delta$ into a function $\Gamma \to
  \ExprBT(\Sigma,R_{\ttl\ttr},\varnothing)^R$ -- equivalently, we will consider
  $\delta : \Gamma \times R \to \ExprBT(\Sigma,R_{\ttl\ttr},\varnothing)$.

  We take $\delta(c,r_1) = c\trpair{(r_1)_\ttl}{(r_1)_\ttr}$ and $\delta(c,r_2)
  = c\trpair{(r_2)_\ttr}{(r_2)_\ttl}$ for $c \in \{a,b\}$, where $r_\ttl$ is a
  notation for $(r,\ttl) \in R_{\ttl\ttr} = R\times\{\ttl,\ttr\}$. If we write
  $\hat{r}_i(T)$ ($i \in \{1,2\}$) for the contents of the register $r_i$ at the
  end of a run of the transducer on $T \in \BT(\Gamma)$, then this $\delta$
  translates into:
  \[ \hat{r}_1(c\trpair{T}{U}) = c\trpair{\hat{r}_1(T)}{\hat{r}_1(U)} \qquad
    \hat{r}_2(c\trpair{T}{U}) = c\trpair{\hat{r}_2(U)}{\hat{r}_2(T)}\qquad (c
    \in \{a,b\})\]
  And the initial condition is $\hat{r}_1(\trnode) = \hat{r}_2(\trnode) =
  \trnode$. Therefore $\hat{r}_1(T) = T$ and $\hat{r}_2(T)$ is $T$
  \enquote{mirrored} by exchanging left and right; let us write $\hat{r}_2(T) =
  \mathtt{reverse}(T)$.

  The output function $F$ is also simplified into an expression in
  $\ExprBT(\Sigma,R,\varnothing)$. By taking $F = a\trpair{r_1}{r_2}$, we define
  a transducer computing the regular function $T \mapsto
  a\trpair{T}{\mathtt{reverse}(T)}$.
\end{exa}

To characterize regular functions, a condition must be imposed\footnote{Without
  this restriction, one could have outputs of exponential size, whereas regular
  functions have linearly bounded output. Filiot and
  Reynier~\cite{FiliotReynier} have shown that the analogous unrestricted model
  for strings corresponds to HDT0L systems; we are not aware of any similar
  result on trees.} on the register assignments: the \emph{single use
  restriction}. It is not intrinsic and depends on an additional piece of data,
namely a \emph{conflict relation} between the registers.
\begin{defi}[{\cite[\S2.1]{BRTT}}]
  \label{def-conflict}
  A \emph{conflict relation} is a binary reflexive and symmetric relation.

  An expression $E \in \ExprBT(\Sigma,V,V')\cup \ExprdBT(\Sigma,V,V')$ is said
  to be \emph{consistent with} a conflict relation~$\incoh$ over $V \cup V'$
  when each variable in $V \cup V'$ appears at most once in~$E$, and for all
  $x,y \in V \cup V'$, if $x \neq y$ and $x \incoh y$, then $E$ does not contain
  both $x$ and $y$.
    
  A register assignment $(\varepsilon,\varepsilon') \in \assi(\Sigma,R,R')$ is
  \emph{single use restricted} with respect to a conflict relation $\incoh$ over
  $R \cup R'$ when:
  \begin{itemize}
  \item all $\varepsilon(r)$ for $r \in R$ and all $\varepsilon'(r')$ for $r'
    \in R'$ are consistent with $\incoh$;
  \item if $x_1, x_2, y_1, y_2 \in R \cup R'$, $x_1 \incoh x_2$ and, for some $d
    \in \{\ttl,\ttr\}$, $(x_1, d)$ appears in\footnote{By $\varepsilon \cup
      \varepsilon'$ we mean the map on the \emph{disjoint} union $R \cup R'$
      induced in the obvious way by $\varepsilon$ and $\varepsilon'$.}
    $(\varepsilon \cup \varepsilon')(y_1)$ and $(x_2, d)$ appears in
    $(\varepsilon \cup \varepsilon')(y_2)$, then $y_1 \incoh y_2$ (note that
    this includes the case $x_1 = x_2$).
  \end{itemize}
\end{defi}
\begin{defi}
  A \emph{bottom-up ranked tree transducer (BRTT)} is a register tree transducer
  $(Q,q_I,R,R',\delta,F)$ endowed with a conflict relation $\incoh$ on $R \cup
  R'$, such that $F(q)$ is consistent with $\incoh$ and all register assignments
  in the image of $\delta$ are single use restricted w.r.t.\ $\incoh$.

  A \emph{regular tree function} is a function computed by a BRTT.

  When the conflict relation is trivial (i.e.\ coincides with equality), we say
  that the BRTT is \emph{copyless}. We also say that a register tree transducer
  is copyless if it becomes a BRTT when endowed with a trivial conflict
  relation.
\end{defi}

\section{Normalization of the $\laml$-calculus}
\label{sec:laml-normalization}

This section is devoted to proving that the $\laml$-calculus is strongly
normalizing. What it means is that any $\laml$-term $t$ admitting a typing
derivation $\Psi; \; \Delta \vdash t : A$ can be shown to be
$\beta\eta$-equivalent to a \emph{normal} term $u$.
The notion of normal ($\NF$) and \emph{neutral} ($\NE$) are
defined via the typing system presented in Figure~\ref{fig:laml-normal}.
The intuition is that a normal term cannot be $\beta$-reduced further, and
that neutral terms substituted in normal terms produce terms that stay normal.

The purpose of this section is to show that \emph{any} typed term $t$ can
be turned into a normal term $t'$ of the same type via a sequence of $\beta$-reductions.
Because $\laml$ features \emph{positive} type constructors $\tensor$ and $\oplus$,
the reality is not quite so straightforward: it is not sufficient to $\beta$-reduce
a term to reach a normal form. In certain circumstances, one must additionally use
$\eta$-conversion to make certain $\beta$-redexes appear and proceed with the
computation of normal terms.
This difficulty is rather well-known in the context of $\lambda$-calculus
extended with coproducts~\cite{Lindley-sums,AltenkirchDH-sums,gasche-thesis}.
However we are not aware of a text treating exactly $\laml$ (e.g., incorporating
additives, a native $\tensor$ and units), so we include a proof using reducibility
candidates along the lines of~\cite[Appendix A.1]{RoccaRoversi-sums}.

To describe said reducibility candidates, we first need to give an oriented version
of $\beta\eta$-equality.
The $\beta$-reduction relation $\to_\beta$ is obtained by closing the relations
given in Figure~\ref{fig:beta-red} by congruence. We write $\to_\beta^*$ for the
reflexive transitive closure relation. Much like with $=_\beta$, we assume that
terms related by $\to_\beta$ have the same type in the same context.

\begin{figure}
\begin{mdframed}
\[
\begin{array}{l !\; c !{\;\;} c}
\text{$\beta$-redexes} &
(\lam x. t) \; u ~\to_\beta~ t[u/x]
&
(\lam^\oc x. t) \; u ~\to_\beta~ t[u/x]
\\ &
\pi_1(\tuple{t,u}) ~\to_\beta~ t
&
\pi_2(\tuple{t,u}) ~\to_\beta~ u
\\ &
\case(\inl(t), x.u,x.v) ~\to_\beta~ u[t/x]
&
\case(\inr(t), x.u,x.v) ~\to_\beta~ v[t/x]
\\ &
\tlet{x \tensor y}{t \tensor u}{v} ~~\to_\beta~~ v[t/x][u/y]
&
\tlet{()}{()}{t} ~\to_\beta~ t
\\\\
\end{array}
\]
\end{mdframed}
\caption{$\beta$-redexes.}
\label{fig:beta-red}
\end{figure}

As it is not the case that every typable term $\beta$-reduces
to a normal form, we need to describe another set of reduction rules which involve $=_\eta$.
Those \emph{extrusion rules} are listed in Figure~\ref{fig:extr}; we also write
$\to_\extr$ for the congruence closure of the relation described there.
While the number of cases is daunting, it should be remarked that these rules are
obtained mechanically by considering the nesting of an eliminator for a positive type
(i.e., the $\tlet{\cdot}{\cdot}{\cdot}$ constructions ($\tensor, \unit$), $\case$ ($\oplus$) and $\abort$ ($0$))
within another eliminator (the aforementioned constructions plus function application ($\to, \lin$) and
projections $\pi_1,\pi_2$ ($\with$)). For a more careful discussion of (a subset) of these rules, we point
the reader to~\cite[Section 3.3]{gasche-thesis}.
We write $\to_\extr^*$ for the reflexive transitive closure of $\to_\extr$, $\to_{\beta\extr}$ for the
union of $\to_\extr$ and $\to_\beta$ and $\to_\extr$, and $\to_{\beta\extr}^*$ for its reflexive transitive
closure. With these notations, we can state the finer version of the normalization theorem for $\laml$.

\begin{thm}
\label{thm:laml-normalization}
For every term term $t$ such that $\Psi; \; \Delta \vdash t : \tau$, there exists $t'$ such that $t \to_{\beta\extr}^* t'$ and
$\Psi; \; \Delta \vdash_\NF t' : \tau$.
\end{thm}

Before embarking on the definitions of the reducibility candidates and the proof of
Theorem~\ref{thm:laml-normalization} itself, we first make a couple of observations
relating $\to_{\beta\extr}^*$ and $=_{\beta\eta}$.

\begin{lem}
Suppose that we terms $t$ and $t'$ with matching types and
that $t \to_{\beta\extr} t'$. Then, we have $t =_{\beta\eta} t'$.
Furthermore if $t$ is normal, so is $t'$. Similarly, if $t$ is neutral,
so is $t'$.
\end{lem}

\begin{lem}
If $t$ is normal, then there is no $t'$ such that $t \to_\beta t'$.
\end{lem}

\begin{figure}
\begin{mdframed}
\[
\begin{array}{c}
\footnotesize
\text{Nested $\tensor/\unit$ eliminator}\\
\\
\begin{array}{rcl}
\footnotesize
(\tlet{p}{t}{u}) \; v
&\to_\extr&
\tlet{p}{t}{(u \; v)}
\\
\pi_i(\tlet{p}{t}{u})
&\to_\extr&
\tlet{p}{t}{\pi_i(u)}
\\
\tlet{q}{\tlet{p}{t}{u}}{v}
&\to_\extr&
\tlet{p}{t}{\tlet{q}{u}{v}}
\\
\abort(\tlet{p}{t}{u})
&\to_\extr&
\tlet{p}{t}{\abort(u)}
\\
\case(\tlet{p}{t}{u},x.v,y.w)
&\to_\extr&
\tlet{p}{t}{\case(u,x.v,y.w)}
\\
\end{array}
\\
\\
\footnotesize
\text{Nested $0$ eliminator}\\
\\
\begin{array}{rcl}
\footnotesize
\abort(t) \; u
&\to_\extr&
\abort(t)
\\
\pi_i(\abort(t))
&\to_\extr&
\abort(t)
\\
\tlet{p}{\abort(t)}{u}
&\to_\extr&
\abort(t)
\\
\abort(\abort(t))
&\to_\extr&
\abort(t)
\\
\case(\abort(t),x.u,y.v)
&\to_\extr&
\abort(t)
\\
\end{array}
\\
\\
\footnotesize
\text{Nested $\oplus$ eliminator}\\
\\
\begin{array}{r@{~}c@{~}l}
\footnotesize
\case(t,x.u,y.v) \; w
&\to_\extr&
\case(t,x.u \; w,y. v \; w)
\\
\pi_i(\case(t,x.u,y.v))
&\to_\extr&
\case(t,x.\pi_i(u),y.\pi_i(v))
\\
\tlet{p}{\case(t,x.u,y.v)}{w}
&\to_\extr&
\case(t,x.\tlet{p}{u}{w},y.\tlet{p}{v}{w})
\\
\abort(\case(t,x.u,y.v))
&\to_\extr&
\case(t,x.\abort(u),y.\abort(v))
\\
\case(\case(t,x.u,y.v), x'.u',y'.v')
&\to_\extr&
\case(t,x.\case(u,x'.u',y'.v'),y.\case(v,x'.u',y'.v'))
\\
\end{array}
\\\\
\end{array}
\]
\end{mdframed}
\caption{The extrusion relation $\to_\extr$ ($i = 1,2$ and $p, q$ are patterns $()$ or $z \tensor z'$).}
\label{fig:extr}
\end{figure}

\begin{figure}[h]
\begin{mdframed}
\[
\begin{array}{c !\; c}
{\dfrac{\Psi ; \; \Delta \vdash_\NE t : \tau}{\Psi; \; \Delta \vdash_\NF t : \tau}}
\\\\
\dfrac{}{\Psi ; \; x : \tau \vdash_\NE x : \tau}
&
\dfrac{}{\Psi, x : \tau ; \; \cdot \vdash_\NE x : \tau}
\\\\
\begin{array}{c}
\dfrac{\Psi; \; \Delta, \; x : \tau \vdash_\NF t : \sigma}{\Psi; \; \Delta \vdash_\NF \lambda x. t : \tau \lin \sigma}
\\\\
\dfrac{\Psi,\; x : \tau; \; \Delta \vdash_\NF t : \sigma}{\Psi; \; \Delta \vdash_\NF \lambda^\oc x. t : \tau \to \sigma}
\end{array}
&
\dfrac{\Psi; \; \Delta \vdash_\NE t : \tau \lin \sigma \qquad \Psi; \; \Delta' \vdash_\NF u : \tau}{\Psi ; \; \Delta, \; \Delta' \vdash_\NE t \; u : \sigma}
\\\\
\dfrac{\tiny \Psi; \; \Delta \vdash_\NF t : \tau \qquad \Psi; \; \Delta' \vdash_\NF u : \sigma}{\Psi; \; \Delta, \;\Delta' \vdash_\NF t \tensor u : \tau \tensor \sigma}
&
\dfrac{
\tiny
\Psi; \; \Delta' \vdash_\NE t : \tau \tensor \sigma
\qquad
\Psi; \; \Delta, \; x : \tau, \; y : \sigma \vdash_X u : \kappa
}{\Psi ; \; \Delta, \; \Delta' \vdash_X \tlet{x \tensor y}{t}{u} : \kappa}
\\\\
\dfrac{}{\Psi; \; \cdot \vdash_\NF () : \unit}
&
\dfrac{
\Psi; \; \Delta' \vdash_\NE t : \unit
\qquad
\Psi; \; \Delta \vdash_X u : \kappa
}{\Psi ; \; \Delta, \; \Delta' \vdash_X \tlet{()}{t}{u} : \kappa}
\\\\\\
\dfrac{\Psi; \; \Delta \vdash_\NF t : \tau \qquad
\Psi; \; \Delta \vdash_\NF u : \sigma
}{\Psi ; \; \Delta \vdash_\NF \tuple{t, u} : \tau \with \sigma}
&
\dfrac{\Psi ;\; \Delta \vdash_\NE t : \tau \with \sigma}{\Psi ; \; \Delta \vdash_\NE \pi_1(t) : \tau}
\qquad
\dfrac{\Psi ;\; \Delta \vdash_\NE t : \tau \with \sigma}{\Psi ; \; \Delta \vdash_\NE \pi_2(t) : \sigma}
\\\\\\
\dfrac{\Psi; \; \Delta \vdash_\NF t : \tau}{\Psi; \; \Delta \vdash_\NF \inl(t) : \tau \oplus \sigma}
&
\dfrac{\Psi; \; \Delta \vdash_\NF t : \sigma}{\Psi; \; \Delta \vdash_\NF \inr(t) : \tau \oplus \sigma}
\\\\\\
\multicolumn{2}{c}{\dfrac{\Psi; \; \Delta \vdash_\NE t : \sigma \oplus \tau \qquad \Psi; \; \Delta', \; x : \sigma \vdash_X u : \kappa \qquad \Psi; \; \Delta', \; y : \tau \vdash_X v : \kappa}{\Psi; \; \Delta, \; \Delta' \vdash_X \case(t,x.u,y.v) : \kappa}}
\\\\
\dfrac{}{\Psi; \; \Delta \vdash_\NF \tuple{ } : \top}
&
\dfrac{\Psi; \; \Delta \vdash_\NE t : 0}{\Psi; \; \Delta, \; \Delta' \vdash_\NE \abort(t) : \tau}
\\\\\\
\end{array}
\]
\end{mdframed}
\caption{Normal forms for $\laml$-terms ($\vdash_\NF$ for normal forms and $\vdash_\NE$ for neutral forms and $X \in \{\NE, \NF\}$).}
\label{fig:laml-normal}
\end{figure}

Both of these Lemmas are proved by straightforward induction on the
relations $\to_\beta$ and $\to_{\beta\extr}$.

Another crucial ingredient is the confluence of the reduction relation $\to_{\beta\extr}$ (or Church-Rosser property).
Alas, this does not hold for syntactic equality (up to $\alpha$-equivalence. However, it holds up to \emph{commuting conversions},
an equivalence relation $\cocoeq$ inductively defined by the clauses in Figure~\ref{fig:coco} and closure under congruence.
We merely state the confluence property that we will use.

\begin{thm}
\label{thm:CR}
If we have $t \to^*_{\beta\extr} u$ and $t \to^*_{\beta\extr} v$, then there exists $u'$ and $v'$ such that
$u \to^*_{\beta\extr} u'$, $v \to^*_{\beta\extr} v'$ and $u' \cocoeq v'$.
\end{thm}

A first observation is these are compatible with $=_{\beta\eta}$ and
that neutral and normal term are preserved by commutative conversions.

\begin{lem}
\label{lem:coco-betaeta}
If $t \cocoeq t'$, then $t =_{\beta\eta} t'$.
\end{lem}

\begin{lem}
\label{lem:coco-nenf}
If $t$ is neutral (resp. normal) and $t \cocoeq t'$, then $t'$ is also neutral (resp. normal).
\end{lem}

\begin{figure}
\begin{mdframed}
\[
\begin{array}{c}
\begin{array}{rcl}
\footnotesize
\tlet{q}{u}{\tlet{p}{t}{v}}
&\cocoeq&
\tlet{p}{t}{\tlet{q}{u}{v}}
\\
\tlet{p}{t}{\abort(u)}
&\cocoeq&
\abort(\tlet{p}{t}{u})
\\
\tlet{p}{t}{\case(u,x.v,y.w)}
&\cocoeq&
\case(u,x.\tlet{p}{t}{v},y.\tlet{p}{t}{w})
\\
\case(u,x.\abort(t),y.\abort(v)) 
&\cocoeq&
\abort(\case(u,x.t,y.v))
\\
\end{array}
\\\\
\begin{array}{c}
\case(t,x.\case(u,x'.v,y'.w),y.\case(u,x'.v',y'.w'))\\
\cocoeq \\
\case(u,x'.\case(t,x.v,y.v'),y'.\case(t,x.w,y.w'))\\
\end{array}
\\
\end{array}
\]
\end{mdframed}
\caption{Commutative conversions $\cocoeq$ ($p, q$ are patterns $()$ or $z \tensor z'$ and both sides are assumed to be well-scoped).}
\label{fig:coco}
\end{figure}

We can now turn to the definition of the reducibility candidates,
where we write $t \to^*_{\beta\extr}\cocoeq t'$ when there
is some $t''$ such that $t \to^*_{\beta\extr} t''$ and $t'' \cocoeq t'$
hold.

\begin{defi}Define a judgment $\Psi; \; \Delta \models t : \tau$ by induction over
the type $\tau$ as follows:
\begin{itemize}
\item for $\tau = \basety, \unit, 0$ or $\top$, we have $\Psi; \; \Delta \models t : \tau$ if and only if there is
$t'$ such that $t \to^*_{\beta\extr} t'$ and $\Psi; \; \Delta \vdash_\NE t' : \tau$
\item $\Psi; \; \Delta \models t : \tau \lin \sigma$ holds if and only if
  \begin{itemize}
  \item there is $t'$ such that $t \to^*_{\beta\extr} t'$ and $\Psi; \; \Delta \vdash_\NF t' : \tau \lin \sigma$
  \item for every $u, \Delta'$ such that $\Psi; \; \Delta' \models u : \tau$, we have
$\Psi; \; \Delta, \; \Delta' \models t \; u : \sigma$
  \end{itemize}
\item $\Psi; \; \Delta \models t : \tau \to \sigma$ holds if and only if
  \begin{itemize}
  \item there is $t'$ such that $t \to^*_{\beta\extr} t'$ and $\Psi; \; \Delta \vdash_\NF t' : \tau \to \sigma$
  \item for every $u$ such that $\Psi; \; \cdot \models u : \tau$, we have
$\Psi; \; \Delta \models t \; u : \sigma$
  \end{itemize}
\item $\Psi; \; \Delta \models t : \tau \tensor \sigma$ holds if and only if
  \begin{itemize}
  \item there is $t'$ such that $t \to^*_{\beta\extr} t'$ and $\Psi; \; \Delta \vdash_\NF t' : \tau \tensor \sigma$
  \item if there are $t_1, t_2$ such that and $t \to^*_{\beta\extr} t_1 \tensor t_2$, then
  there are $\Delta_1$ and $\Delta_2$ such that $\Delta = \Delta_1, \Delta_2$ and
  \[\Psi; \; \Delta_1 \models t_1 : \tau \qquad \text{and} \qquad \Psi; \; \Delta_2 \models t_2 : \sigma\]
  \end{itemize}
\item $\Psi; \; \Delta \models t : \tau \with \sigma$ holds if and only if
  \begin{itemize}
  \item there is $t'$ such that $t \to^*_{\beta\extr} t'$ and $\Psi; \; \Delta \vdash_\NF t' : \tau \with \sigma$
  \item $\Psi; \; \Delta \models \pi_1(t) : \tau$
  \item $\Psi; \; \Delta \models \pi_2(t) : \sigma$
  \end{itemize}
\item $\Psi; \; \Delta \models t : \tau \oplus \sigma$ holds if and only if
  \begin{itemize}
  \item there is $t'$ such that $t \to^*_{\beta\extr} t'$ and $\Psi; \; \Delta \vdash_\NF t' : \tau \oplus \sigma$
  \item if there is $u$ such that $t \to^*_{\beta\extr} \cocoeq \inl(u)$, then $\Psi; \; \Delta \models u : \tau$
  \item if there is $v$ such that $t \to^*_{\beta\extr} \cocoeq \inr(v)$, then $\Psi; \; \Delta \models v : \sigma$
  \end{itemize}
\end{itemize}
\end{defi}

The set of terms $t$ such that $\Psi; \; \Delta \models t : \tau$ constitutes our set of reducibility candidates
at type $\tau$ in the context $\Psi; \; \Delta$. They are defined in such a way that if
$\Psi; \; \Delta \models t : \tau$, then there is $t'$ such that $t \to^*_{\beta\extr} t'$ and $\Psi; \; \Delta \vdash_\NF t' : \tau$. We shall be able to conclude this section if we show an \emph{adequacy lemma} stating that
every typable term lies in a reducibility candidate.
Before doing that, we first need a couple of stability properties: closure under anti-reduction, and the
fact that every neutral term lies in a reducibility candidate.

\begin{lem}
\label{lem:ne-not-nf}
If $t$ is a neutral term, then $t$ cannot be $\beta\eta$-equivalent to one of the following
\[\lam x. u \qquad \lam^\oc x. u \qquad (u,v) \qquad u \tensor v \qquad \langle \rangle \qquad () \qquad \inl(u) \qquad \inr(u)\]
\end{lem}
\begin{proof}
Trivial case-analysis.
\end{proof}

\begin{thm}
\label{thm:is-rc}
Suppose that $\Psi; \; \Delta \models t' : \tau$.
Then the following hold:
\begin{itemize}
\item There exists $t''$ such that $t' \to^*_{\beta\extr} t''$ and $\Psi; \; \Delta \vdash_\NF t'' : \tau$.
\item If we have $t \to^*_{\beta\extr} t'$, then we also have $\Psi; \; \Delta \models t : \tau$.
\item If $\Psi; \; \Delta \vdash_\NE t : \tau$, then $\Psi; \; \Delta \models t : \tau$.
\item If $t' \cocoeq t''$, then $\Psi; \; \Delta \models t'' : \tau$.
\item If $t' \to^*_{\beta\extr} t''$, then we also have $\Psi; \;\Delta \models t'' : \tau$.
\end{itemize}
\end{thm}
\begin{proof}
The first point can be proven via an easy case analysis on $\tau$ that we skip.
The second point we may prove by induction over $\tau$; let us sketch a few representative cases:
\begin{itemize}
\item If $\tau$ is $\basety$, $0$, $\unit$ or $\top$, this is immediate.
\item Suppose that $\tau ~=~ \sigma \lin \kappa$ and that $t \to^*_{\beta\extr} t'$. By definition of $\models$ at $\tau \lin \kappa$, there is some normal $t''$ such that
$t' \to^*_{\beta\extr} t''$, so we also have $t \to^*_{\beta\extr} t''$ by transitivity. Now suppose that we are given some
$u$ and $\Delta'$ such that $\Psi; \; \Delta' \models u : \tau$. By definition, we have $\Psi; \; \Delta, \; \Delta' \models t' \; u : \kappa$.
Using the induction hypothesis at $\kappa$ and the fact that $t \; u \to^*_{\beta\extr} t' \; u$, we thus have that $\Psi; \; \Delta, \; \Delta' \models t \; u : \kappa$.
We can thus conclude that $\Psi; \; \Delta \models t : \kappa$.
\item Suppose that $\tau ~=~ \sigma \oplus \kappa$ and that $t \to^*_{\beta\extr} t'$. By definition of $\models$, there is some normal $t''$ such that
$t' \to ^*_{\beta\extr} t''$, so we also have $t \to^*_{\beta\extr} t''$ by transitivity.
Now if we have $u$ (resp. $v$) such that $t \to_{\beta\extr}^*\cocoeq \inl(u)$ (resp. $\inr(v)$), then, thanks to confluence (Theorem~\ref{thm:CR}) we also have $t' \to^*_{\beta\extr} \cocoeq \inl(u)$ (resp. $\inr(v)$),
so we have $\Psi; \; \Delta \models u : \sigma$ (resp. $\Psi; \; \Delta \models u : \kappa$) by definition.
Therefore, we may conclude that $\Psi; \; \Delta \models t : \sigma \oplus \kappa$.
\end{itemize}
The third point is also proved via a straightforward induction over $\tau$ by leveraging Lemma~\ref{lem:ne-not-nf}.
The last two points follow from induction over $\tau$ combined with Theorem~\ref{thm:CR} and Lemma~\ref{lem:coco-nenf}.
\end{proof}

\begin{cor}
\label{cor:var-in-rc}
If $x$ is a variable of type $\tau$ in either $\Psi$ or $\Delta$, we have $\Psi; \; \Delta \models x : \tau$.
\end{cor}
\begin{proof}
Immediate as variables are neutral.
\end{proof}

\begin{lem}
\label{lem:rc-pos-elim}
Suppose that we have $\Psi; \; \Delta \vdash_{\NE} t : \tau \otimes \sigma$ and $\Psi; \; \Delta', x : \tau, y : \sigma \models u : \kappa$.
Then we have $\Psi; \; \Delta, \;\Delta' \models \tlet{x \tensor y}{t}{u} : \kappa$.

Similarly, if $\Psi; \; \Delta \vdash_{\NE} t : \tau \oplus \sigma$,
$\Psi; \; \Delta', x : \tau \models u : \kappa$ and
$\Psi; \; \Delta', y : \sigma \models u' : \kappa$, we have $\Psi; \; \Delta, \Delta' \models \case(t,x.u,y.u') : \kappa$.

Finally, if $\Psi; \; \Delta \vdash_\NE t : \unit$ and $\Psi; \; \Delta' \models u : \kappa$, we also have $\Psi; \; \Delta, \; \Delta' \models \tlet{()}{t}{u} : \kappa$.
\end{lem}
\begin{proof}
By induction over $\kappa$.
\end{proof}

\begin{thm}[Adequacy]
\label{thm:laml-adequacy}
Suppose that we have a non-linear context $\Psi ~=~ x_1 : \sigma_1, \; \ldots, \; x_k : \sigma_k$, a linear context
$\Delta ~=~ a_1 : \tau_1, \ldots, \; a_n : \tau_n$ and a term $v$ such that $\Psi; \; \Delta \vdash v : \kappa$ for some type $\kappa$.
Further, assume that $\Psi', \Delta'_1, \ldots, \Delta'_n$ and terms $t_1, \ldots t_k, u_1, \ldots, u_n$ such that
$\Psi'; \; \cdot \models t_i : \sigma_i$ for $1 \le i \le k$ and $\Psi'; \; \Delta'_j \models u_j : \tau_j$ for $1 \le j \le n$.
Then we have
\[\Psi'; \; \Delta'_1, \ldots, \; \Delta'_n \models v[t_1/x_1, \ldots, t_k/x_k, u_1/a_1, \ldots, u_n/a_n] : \kappa\]
\end{thm}
\begin{proof}
The proof goes by induction over the typing derivation $\Psi; \; \Delta \vdash v : \kappa$; we sketch a few representative subcases below.
To keep notations short, we write $\gamma$ (respectively $\delta$) instead of the sequence of assignments $t_1/x_1, \ldots, t_k/x_k$ (respectively $u_1/a_1, \ldots, u_n/a_n$) and
$\Delta' = \Delta'_1, \ldots, \; \Delta'_n$.
\begin{itemize}
\item If the last rule used in an axiom, the conclusion is immediate.
\item If the last rule used is a linear function application
\[\dfrac{\Psi; \; \Delta_1 \vdash v : \kappa \lin \kappa' \qquad \Psi; \; \Delta_2 \vdash v' : \kappa}{\Psi; \; \Delta_1, \; \Delta_2 \vdash v \; v' : \kappa'}\]
with $\delta_1, \delta_2$ and $\Delta''_1, \; \Delta''_2$ the obvious decomposition of $\delta$ and $\Delta'$,
the induction hypothesis yields
\[\Psi'; \; \Delta''_1 \models v[\gamma,\delta_1] : \kappa \lin \kappa' \qquad \text{and} \qquad \Psi'; \; \Delta''_2 \models v'[\gamma, \delta_2] : \kappa\]
By definition of $\models$ for type $\kappa \lin \kappa'$, we thus have $\Psi'; \; \Delta' \models v[\gamma, \delta_1] \; v'[\gamma, \delta_2] : \kappa'$, so we may conclude.
\item The case of non-linear function application is entirely analogous.
\item If the last rule is the typing of a linear $\lambda$-abstraction
\[\dfrac{\Psi; \; \Delta, \; c : \kappa \vdash v : \kappa'}{\Psi; \; \Delta \vdash \lam c. v : \kappa \lin \kappa'}\]
by the inductive hypothesis, we have $\Psi'; \; \Delta', \Delta'' \vdash v[\gamma,\delta,v'/c] : \kappa'$ for any $v'$ and $\Delta''$ such that $\Psi' ;\; \Delta'' \models v' : \kappa$.
We can prove the conjunct defining $\Psi'; \; \Delta' \models (\lam c.v)[\gamma,\delta] : \kappa \lin \kappa'$ as follows:
\begin{itemize}
\item First, by taking $\Delta'' = c : \kappa$ (Corollary~\ref{cor:var-in-rc}), we obtain that there exists some $v''$ such that $v[\gamma,\delta] \to^*_{\beta\extr} v''$ and $\Psi; \; \Delta, \; c : \kappa \vdash_\NF v'' : \kappa'$.
Therefore we have $\lam c. v[\gamma,\delta] \to^*_{\beta\extr} \lam c. v'$ and $\Psi; \; \Delta \vdash_\NF \lam c. v' : \kappa \lin \kappa'$.
\item Then, assume we have some $v'$ and $\Psi; \; \Delta'' \models v' : \kappa$, so that we have $\Psi'; \; \Delta', \Delta'' \vdash v[\gamma,\delta,v'/c] : \kappa'$ by the inductive hypothesis.
Because
\[(\lam c. v)[\gamma,\delta] \; v' = (\lam c. v[\gamma,\delta]) \; v' \to_{\beta} v[\gamma,\delta, v'/c]\]
we may apply Theorem~\ref{thm:is-rc} to conclude that $\Psi'; \; \Delta', \; \Delta'' \models (\lam c. v)[\gamma,\delta] \; v' : \kappa'$
\end{itemize}
\item The case of the non-linear $\lambda$-abstraction for $\to$ is similar.
\item If the last rule applied is an introduction of $\tensor$,
\[
\dfrac{\Psi; \; \Delta_1 \vdash t : \tau
\qquad
\Psi; \; \Delta_2 \vdash u : \sigma}{\Psi; \; \Delta_1, \; \Delta_2 \vdash t \tensor u : \tau \tensor \sigma}\]
call $\delta_1,\delta_2$ the splitting of $\delta$ according to the decomposition $\Delta = \Delta_1, \; \Delta_2$.
The induction hypothesis yields
\[\Psi'; \; \Delta'_1 \models t[\gamma, \delta_1] : \tau \qquad \text{and} \qquad \Psi'; \; \Delta'_2 \models u[\gamma, \delta_2] : \sigma\]
By definition it means that we have normal terms $t'$ and $u'$ such that $t[\gamma,\delta_1] \to^*_{\beta\extr} t'$ and $u[\gamma, \delta_2] \to^*_{\beta\extr} u'$,
so $(t \tensor u)[\gamma,\delta] \to^*_{\beta\extr} t \tensor u'$.
Now suppose that we have $(t \tensor u)[\gamma, \delta] \to_{\beta\extr}^* \cocoeq t'' \tensor u''$. It is not difficult to check (by induction over the length
of the reductions $\to_{\beta\extr}^*$ and derivation of $\cocoeq$) that we have $t[\gamma,\delta_1] \to_{\beta\extr}^*\cocoeq t''$ and $u[\gamma,\delta_2] \to_{\beta\extr}^*\cocoeq u''$.
So by Theorem~\ref{thm:is-rc}, we have that $\Psi; \; \Delta_1 \models t'' : \tau$ and $\Psi; \; \Delta_2 \models u'' : \sigma$,so we may conclude.
\item If the last rule applied is an elimination of $\tensor$
\[
\dfrac{
\Psi; \; \Delta_1 \vdash u : \tau \tensor \sigma
\qquad
\Psi; \; \Delta_2, x : \tau, y : \sigma \vdash t : \kappa
}{\Psi ; \; \Delta_1, \; \Delta_2 \vdash \tlet{x \tensor y}{u}{t} : \kappa}\]
with $\delta_1, \delta_2$ the obvious decomposition of $\delta$ along $\Delta_1, \; \Delta_2$,
the induction hypothesis applied to the first premise yields
$\Psi'; \; \Delta'_1 \models u[\gamma,\delta_1] : \tau \tensor \sigma$.
In particular, this means we have $u[\gamma,\delta_1] \to_{\beta\extr}^* u'$ such that
$\Psi'; \; \Delta'_1 \vdash_\NF u' : \tau \tensor \sigma$. By Theorem~\ref{thm:is-rc},
it suffices to show that $\tlet{x \tensor y}{u'}{t[\gamma,\delta_2]} \to_{\beta\extr}^* v$
such that $\Psi';\; \Delta'_1, \; \Delta'_2 \models v : \kappa$ to conclude.
We do so by going by induction over the judgment $\Psi'; \; \Delta'_1 \vdash_\NF u' : \tau \tensor \sigma$.
\begin{itemize}
\item If we have $\Psi'; \; \Delta'_1 \vdash_\NE u' : \tau \tensor \sigma$, then we may use the
outer inductive hypothesis $\Psi'; \; \Delta'_2, \; x : \tau, \; y : \sigma \models t[\gamma,\delta_2] : \kappa$
and apply Lemma~\ref{lem:rc-pos-elim}.
\item If we have $u' = \tlet{x' \tensor y'}{u''}{u'''}$, applying the induction hypothesis, we have some $v$ such
that
\[\tlet{x \tensor y}{u'''}{t[\gamma,\delta_2]} \to_{\beta\extr}^* v \qquad \text{and} \qquad \Psi';\;\Delta',x' : \tau', y' : \sigma' \models v : \kappa\]
We may thus conclude using the sequence of reductions
\[
\begin{array}{rll}
\tlet{x \tensor y}{\tlet{x' \tensor y'}{u''}{u'''}}{t[\gamma,\delta_2]}
&\to_\extr&
\tlet{x' \tensor y'}{u''}{\tlet{x \tensor y}{u'''}{t[\gamma,\delta_2]}}\\
&\to_{\beta\extr}^*& \tlet{x' \tensor y'}{u''}{v}
\end{array}
\]
and Lemma~\ref{lem:rc-pos-elim}
\item We proceed similarly  if $u' = \case(u'',x'.u''',y'.u'''')$, $\tlet{()}{u''}{u'''}$ or $\pi_i(u'')$.
\item Finally, if $u' = u'' \tensor u'''$, we apply the outer induction hypothesis with the substitution $\gamma,\delta_2, u''/x,u'''/y$ to conclude.
\end{itemize}
\end{itemize}
\end{proof}

\begin{proof}[Proof of Theorem~\ref{thm:laml-normalization}]
Instantiate Theorem~\ref{thm:laml-adequacy} in the case of a trivial substitution ($t_i = x_i$ and $u_j = a_j$)
using Corollary~\ref{cor:var-in-rc} and conclude with Theorem~\ref{thm:is-rc}.
\end{proof}

\section{Proof
Lemma~\ref{lem:laml-niceshape} (on
$\laml$-terms defining tree functions)}
\label{sec:laml-niceshape}

\begin{defi}
Write $\sqsubseteq_+$ for the least preorder relation over $\laml$ types satisfying the following
for every types $\tau$ and $\sigma$
\[
\tau, \sigma \sqsubseteq_+ \tau \tensor \sigma \quad
\tau, \sigma \sqsubseteq_+ \tau \oplus \sigma \quad
\tau, \sigma \sqsubseteq_+ \tau \with \sigma \quad
\sigma \sqsubseteq_+ \tau \lin \sigma \quad
\sigma \sqsubseteq_+ \tau \to \sigma
\]

We say that $\tau$ is a \emph{strictly positive subtype} of $\sigma$ whenever $\tau \sqsubseteq_+ \sigma$. 
\end{defi}

\begin{defi}
A context $\Psi; \; \Delta$ is called \emph{consistent} if there is no term $t$ such that $\Psi; \; \Delta \vdash t : 0$.
\end{defi}

\begin{lem}
\label{lem:consistent-ne}
A context $\Psi; \; \Delta$ is inconsistent if and only if there is a neutral term $t$ such that $\Psi; \; \Delta \vdash_\NE t : 0$.

Furthermore, if $\Psi; \; \Delta \vdash_\NE t : \tau$, the last typing rule applied has one premise $\Psi'; \; \Delta' \vdash_\NE u : \tau'$
and $\Psi; \; \Delta$ is consistent, then so is $\Psi'; \; \Delta'$.
\end{lem}
\begin{proof}
The first point is an easy corollary of Theorem~\ref{thm:laml-normalization}.
The second point follows from a case analysis, using the following facts:
\begin{itemize}
\item If $\Psi, \; \Psi'; \; \Delta, \; \Delta'$ is consistent, then so is $\Psi; \; \Delta$.
\item If $\Psi; \; \Delta, \; \Delta'$ is consistent and $\Psi; \; \Delta' \vdash t : \tau$, then $\Psi; \; \Delta, \; x : \tau$ is consistent.
\end{itemize}
\end{proof}

\begin{lem}
\label{lem:neutral-subtype-var}
If $\Psi; \; \Delta$ is consistent and $\Psi; \; \Delta \vdash_\NE t : \tau$, then
there is a variable in $\Psi; \; \Delta$ of type $\sigma$ with $\tau \sqsubseteq_+ \sigma$.
\end{lem}
\begin{proof}
By induction on the judgement $\Psi; \; \Delta \vdash_\NE t : \tau$.
\begin{itemize}
\item If the last rule applied was a variable lookup. \begin{mathpar}
\inferrule{ }{\Psi ; \; x : \tau \vdash_\NE x : \tau}
\and
\inferrule{ }{\Psi, x : \tau ; \; \cdot \vdash_\NE x : \tau}
\end{mathpar}
then the conclusion immediately follows.
\item The more interesting cases are those of the elimination rules for $\lin$, $\tensor$ and $\with$.
\begin{mathpar}
\small
\inferrule{\Psi; \; \Delta \vdash_\NE t : \tau \lin \sigma \qquad \Psi; \; \Delta' \vdash_\NF u : \tau}{\Psi ; \; \Delta, \; \Delta' \vdash_\NE t \; u : \sigma}
\and
\inferrule{
\Psi; \; \Delta' \vdash_\NE t : \tau \tensor \sigma
\qquad
\Psi; \; \Delta, \; x : \tau, \; y : \sigma \vdash_\NE u : \kappa
}{\Psi ; \; \Delta, \; \Delta' \vdash_\NE \tlet{x \tensor y}{t}{u} : \kappa}
\and
\inferrule{\Psi ;\; \Delta \vdash_\NE t : \tau \with \sigma}{\Psi ; \; \Delta \vdash_\NE \pi_1(t) : \tau}
\and
\inferrule{\Psi ;\; \Delta \vdash_\NE t : \tau \with \sigma}{\Psi ; \; \Delta \vdash_\NE \pi_2(t) : \sigma}
\and \inferrule{\Psi; \; \Delta \vdash_\NE t : \sigma \oplus \tau \qquad \Psi; \; \Delta', \; x : \sigma \vdash_\NE u : \kappa \qquad \Psi; \; \Delta', \; y : \tau \vdash_\NE v : \kappa}{\Psi; \; \Delta, \; \Delta' \vdash_\NE \case(t,x.u,y.v) : \kappa}
\and
\dfrac{
\Psi; \; \Delta' \vdash_\NE t : \unit
\qquad
\Psi; \; \Delta \vdash_X u : \kappa
}{\Psi ; \; \Delta, \; \Delta' \vdash_X \tlet{()}{t}{u} : \kappa}
\and
\dfrac{\Psi; \; \Delta \vdash_\NE t : 0}{\Psi \; \Delta \vdash_\NE \abort(t) : \tau}
\end{mathpar}
The treatment of $\with$ and $\unit$ is rather straightforward and $0$ is ruled out because $\Psi; \;\Delta$ is assumed to be consistent, so we only explain the inductive step for
$\lin$, $\tensor$ and $\oplus$.
\begin{itemize}
\item If the last rule applied is the elimination of a linear arrow
\[\dfrac{\Psi; \; \Delta \vdash_\NE t : \tau \lin \sigma \qquad \Psi; \; \Delta' \vdash_\NF u : \tau}{\Psi ; \; \Delta, \; \Delta' \vdash_\NE t \; u : \sigma}\]
then the induction hypothesis applied to the first premise means that there is a variable $x$ in $\Psi; \; \Delta$ of type $\kappa$
such that $\tau \lin \sigma \sqsubseteq_+ \kappa$, and we may conclude since $\sigma \sqsubseteq_+ \tau \lin \sigma$.
\item If the last rule applied is the elimination of a tensor product
\[\dfrac{
\Psi; \; \Delta' \vdash_\NE t : \tau \tensor \sigma
\qquad
\Psi; \; \Delta, \; x : \tau, \; y : \sigma \vdash_\NE u : \kappa
}{\Psi ; \; \Delta, \; \Delta' \vdash_\NE \tlet{x \tensor y}{t}{u} : \kappa}\]
then the induction hypothesis applied to the first premise yields a variable $z$ in $\Psi; \; \Delta, \; x : \tau, \; y : \sigma$ of type $\zeta$ with $\kappa \sqsubseteq_+ \zeta$.
If $z \notin \{x,y\}$, then $z$ occurs in $\Psi; \; \Delta$ and we may conclude. Otherwise, suppose that $z = x$; applying the induction hypothesis to the second premise, we
know that $\tau \tensor \sigma \sqsubseteq_+ \tau \sqsubseteq_+ \zeta'$ for a $\zeta'$ being the type of some variable in $\Delta'$ or a $\zeta' = \oc \zeta''$ with $\zeta''$ being the type of some variable in $\Psi$.
Therefore, we may conclude since ${\kappa} \sqsubseteq_+ {\tau} \sqsubseteq_+ {\tau \tensor \sigma} \sqsubseteq_+ {\zeta'}$.
The case of $z = y$ is treated similarly.
\end{itemize}
\item If the last rule applied is the elimination of a coproduct
\[\dfrac{\Psi; \; \Delta \vdash_\NE t : \sigma \oplus \tau \qquad \Psi; \; \Delta', \; x : \sigma \vdash_\NE u : \kappa \qquad \Psi; \; \Delta', \; y : \tau \vdash_\NE v : \kappa}{\Psi; \; \Delta, \; \Delta' \vdash_\NE \case(t,x.u,y.v) : \kappa}\]
then, by the inductive hypothesis applied to the second premise, there is $\zeta$ with
$\kappa \sqsubseteq_+ \zeta$ such that
\begin{enumerate}
\item either there is a variable in $\Psi;\; \Delta'$ with type $\zeta$
\item or $\zeta = \sigma$.
\end{enumerate}
In the first case, we may directly conclude.
Otherwise, the induction hypothesis applied to the first premise states that there is $\zeta''$
with $\sigma \oplus \tau \sqsubseteq_+ \zeta''$ so that $\zeta''$ is a type of some variable in $\Psi$
or $\Delta$. Hence, we have $\kappa \sqsubseteq_+ \sigma \sqsubseteq_+ \sigma \oplus \tau \sqsubseteq_+ \zeta''$
and we may conclude.
\end{itemize}
\end{proof}

\begin{prop}
\label{prop:laml-tree-nf}
Fix a ranked alphabet ${\bf \Sigma}$.
The map $t \mapsto \underline{t}$ taking trees $t \in \Tree({\bf\Sigma})$ to their Church encodings
is a bijection.
\end{prop}
\begin{proof}[Proof of Proposition~\ref{prop:laml-tree-nf}]
For the sake of this proof, let us assume that ranked alphabets $\bf \Sigma$
are ordered. Recall that if $t \in \Tree({\bf \Sigma})$, we write
$\underline{t} : \Treety_{\bf \Sigma}$ for its Church encoding.
When $\Sigma = \{ a_1, \ldots, a_k\}$, $\underline{t}$ has
shape $\lam^\oc a_1. \ldots. \lam^\oc a_k. t^\circ$. for some neutral
term $t^\circ$. Let us adopt this notation for a map $t \mapsto t^\circ$,
mapping trees $t$ to terms
$\widetilde{\bf \Sigma} ; \; \cdot \vdash t^\circ : \basety$, and
let us abbreviate the sequence of $\lambda$-abstractions $\lam^\oc a_1. \ldots. \lam^\oc a_k.$
as $\lam^\oc {\bf \Sigma}$ for arbitrary (ordered) ranked alphabets $\bf \Sigma$.

We use those conventions to show that the map $t \mapsto \underline{t}$ is surjective
(where it is understood that the codomain consists of terms \emph{up to $\beta\eta$-equivalence}.

\begin{lem}
\label{lem:laml-tree-body}
Fix a ranked alphabet $\widetilde{\bf \Sigma}$. For every typed normal term $u$, we have
\begin{enumerate}
\item if $\widetilde{\bf \Sigma}; \; \cdot \vdash_{\NF} u : \basety$, then there is $t \in \Tree({\bf \Sigma})$
such that $u = t^\circ$.
\item if $\widetilde{\bf \Sigma}; \; \cdot \vdash_{\NE} u : \basety \lin \ldots \lin \basety$ where the type
of $u$ has $k$ arguments, then there exists $a \in \Sigma$,
a list of trees $t_1, \ldots, t_{\card{\arity(a)}-k} \in \Tree({\bf \Sigma})$ such that
$u = a \; t_1^\circ \; \ldots \; t_{\card{\arity(a)}-k}^\circ$.
\end{enumerate}
\end{lem}
\begin{proof}
We proceed by induction over the typing judgment of the normal form $u$.
Many cases are easily seen to not arise (typically, constructor for various datatypes).
Most eliminators can also be ignored because of Lemma~\ref{lem:neutral-subtype-var}. For instance, suppose
$u = \case(v, x.w,y.w')$. Then it would means that we had some $\widetilde{\bf \Sigma} ; \; \cdot \vdash_\NE v : \tau \oplus \sigma$
for some $\tau, \sigma$, but that cannot be the case as $\tau \oplus \sigma$ is never a strictly positive subtype of some
$\basety \lin \ldots \lin \basety$.

As a consequence, there are only two cases of interest: the variable case and (linear) function application.
\begin{itemize}
\item If $u$ is a variable of type $\basety$, then the first result (and, as a consequence, the second) result is immediate: $u$
is the Church encoding of a tree with a single leaf. If it is a variable of type $\basety \lin \ldots \lin \basety$, the first claim
is vacuously true and the second is also immediate.
\item If $u$ is a function application $v \; w$, then we have that $\widetilde{\bf \Sigma}; \; \cdot \vdash_{\NE} v : \basety \lin \ldots \lin \basety$
in both cases. So we may apply the inductive hypothesis to obtain some $a \in \Sigma$ and trees $t_1, \ldots, t_l$ such that
$v = a \; t_1^\circ \; \ldots \; t_l^\circ$.
We also have that $\widetilde{\bf \Sigma}; \; \cdot \vdash_{\NF} : w :\basety$, so we have some tree $t_{l+1}$ such that $w =_{\beta\eta} t_{l+1}^\circ$.
Altogether, we thus have $u = a \; t_1^\circ \; \ldots \; t_{l+1}^\circ$ as expected of the second item.
If the first item is not vacuously true, we have that $l+1 = \card{A}$, and thus, $a \; t_1^\circ \; \ldots \; t_{l+1}^\circ = (a(t_1, \ldots, t_{l+1}))^\circ$
as required.
\end{itemize}
\end{proof}

Given two \emph{ordered} ranked alphabets ${\bf \Sigma}$ and ${\bf \Gamma}$, write ${\bf \Sigma} \tensor {\bf \Gamma}$ for the ordered ranked alphabet
with letters $\Sigma + \Gamma$ determined by $\inl(a) < \inr(b)$ for $a \in \Sigma$, $b \in \Gamma$ and where the order is lifted from ${\bf \Sigma}$ and
${\bf \Gamma}$ otherwise.

\begin{lem}
\label{lem:laml-tree-lam}
Fix ranked alphabets ${\bf \Sigma}, {\bf \Gamma}$.
If we have $\widetilde{\bf \Sigma}; \; \cdot \vdash u : \Treety_{\bf \Gamma}$, then there exists some $t \in \Tree({\bf \Sigma} \tensor {\bf \Gamma})$ such that
$u =_{\beta\eta} \lam^\oc {\bf \Gamma}. t^\circ$.
\end{lem}
\begin{proof}
First, we use Theorem~\ref{thm:laml-normalization} to suppose that $u$ is under normal form, and we proceed by induction over the size of ${\bf \Gamma}$.
If it is empty, then the result follows from the first item of Lemma~\ref{lem:laml-tree-body}.
Otherwise ${\bf \Gamma} = {\bf S} \tensor {\bf \Gamma'}$ for some singleton alphabet ${\bf S}$ with letter $b$.
Then, a quick case analysis shows that, as in Lemma~\ref{lem:laml-tree-body}, most cases can be ignored due to the typing of $u$, and because
of considerations based on Lemma~\ref{lem:neutral-subtype-var}.
There is only one interesting case which is the non-linear $\lambda$-abstraction.
\[
\dfrac{\widetilde{\bf \Sigma}, \; \widetilde{\bf S}; \; \cdot \vdash v : \Treety_{\bf \Gamma}}{\widetilde{\bf\Sigma}; \; \cdot \vdash \lam^\oc b. v : \Treety_{\bf S \tensor \bf \Gamma}}\]
We may apply the induction hypothesis as $\widetilde{{\bf \Sigma} \tensor {\bf S}} = \widetilde{\bf \Sigma}, \widetilde{\bf S}$ and get that $u =_{\beta\eta} \lam^\oc b. v = \lam^\oc b. \lam^\oc {\bf \Gamma}. t^\circ = \lam^\oc {\bf S} \tensor {\bf \Gamma}. t^\circ$.
\end{proof}

By instantiating this latest lemma with ${\bf \Gamma}$ empty, we can thus deduce that the map $t \mapsto \underline{t}$ is surjective.
We may also show that it is injective by exhibiting a left-inverse map, using a semantic interpretation of $\laml$ into $\Set$:
with ${\bf \Sigma}$ fixed, use the cartesian-closed structure and coproducts to interpret $\laml$ with the interpretation of $\basety$ being $\Tree({\bf \Sigma})$.
This yields a map from terms $t : \Treety_{\bf \Sigma}$ to set-theoretic functions $(\Tree({\bf \Sigma}) \to \ldots \to \Tree({\bf \Sigma})) \to \ldots \to \Tree({\bf \Sigma})$,
where the arguments correspond to the arity of tree constructors; feed the actual constructors to this function to recover a tree in $\Tree({\bf \Sigma})$.

It is straightforward to check that this map is indeed a left inverse of $t \mapsto \underline{t}$, by induction over $t$.
Hence the map $t \mapsto \underline{t}$ is bijective.
\end{proof}

\begin{lem}
\label{lem:laml-niceshape-partial}
Let $\tau = \kappa_1 \to \ldots \to \kappa_k \to \kappa'$ be a type
and $s$ a distinguished variable of type $\tau$.
Let $\bf \widetilde\Sigma$ be a ranked alphabet such that ${\bf \widetilde\Sigma}; \; s : \tau$ is consistent.
Then, if there is $k' < k$ such that ${\bf \widetilde\Sigma}; \; s : \tau \vdash_\NE t : \kappa_{k'+1} \to \ldots \to \kappa_k \to \kappa$, there are also
terms $d_1, \ldots, d_{k'}$ such that
\[t =_\lamlequiv s \; d_1 \; \ldots \; d_{k'} \quad \text{and} \quad
{\bf \widetilde\Sigma}; \; \cdot \vdash_\NF d_i : \kappa_i \qquad \text{for $i \in \{1, \ldots, k'\}$}\]
\end{lem}
\begin{proof}
By induction over $k'$. Note that $t$ being neutral is essential here.
\end{proof}

\begin{lem}
\label{lem:laml-niceshape-invariant}
Let $\tau = \kappa_1 \to \ldots \to \kappa_k \to \kappa'$ be a type
with $\kappa'$ purely linear
and $s$ a distinguished variable of type $\tau$.
Let $\bf \widetilde\Sigma$ be a ranked alphabet such that ${\bf \widetilde\Sigma}; \; s : \tau$ is consistent and
$t$ be a term such that ${\bf\widetilde\Sigma}; \; s : \tau \vdash_\NE t : \sigma$ for some
$\sigma$ such that $\sigma \sqsubseteq_+ \kappa'$ or of the shape $\basety \lin \ldots \lin \basety$.
Then, there are terms $o$, $d_1$, \dots, $d_k$ such that $t =_\lamlequiv o \; (s \; d_1 \; \ldots \; d_k)$ and
\[{\bf\widetilde\Sigma}; \; \cdot \vdash o : \kappa' \lin \sigma \qquad \qquad {\bf\widetilde\Sigma}; \; \cdot \vdash_\NF d_i : \kappa_i \quad \text{for $i \in \{1, \ldots, k\}$}\]
\end{lem}
\begin{proof}
We proceed by induction over a derivation of ${\bf\widetilde\Sigma}; \; s : \tau \vdash_\NE t : \sigma$.
Note that to apply the induction hypothesis, we need to ensure that every context under consideration is consistent.
We keep this check implicit as it always follows from Lemma~\ref{lem:consistent-ne}.
\begin{itemize}
\item If the last rule applied is a variable lookup, then the term in question must be $s$ itself.
Furthermore, we must have $k = 0$, so we may simply take $o = \lam x. x$ to conclude.
\item If the last rule considered is the following instance of the application rule
\[\dfrac{{\bf\widetilde\Sigma};  \; s : \tau \vdash_\NE t : \sigma' \lin \sigma \qquad {\bf\widetilde\Sigma}; \; \cdot \vdash_\NF u : \sigma'}{{\bf\widetilde\Sigma} ; \; s : \tau \vdash_\NE t \; u : \sigma}\]
then, by Lemma~\ref{lem:neutral-subtype-var} (applied on the first premise), it means that we have 
\[\text{either}~~ \sigma' \lin \sigma \sqsubseteq_+ \tau \qquad \text{or}~~ \sigma \lin \sigma' = \basety \lin \ldots \lin \basety\]
In the first case, we can further see that $\sigma' \lin \sigma \sqsubseteq_+ \kappa'$, so in both cases the induction hypothesis can be applied to the first premise to yield
some $o'$ and $d_1, \ldots, d_k$ such that $o' \; (s \; d_1 \ldots d_k) =_\lamlequiv t$, and we may set $o = \lam s. \; o' \; s \; u$ to conclude.
\item If the last rule considered is the following instance of the application rule
\[\dfrac{{\bf\widetilde\Sigma};  \; s : \tau \vdash_\NE t : \sigma' \to \sigma \qquad {\bf\widetilde\Sigma}; \; \cdot \vdash_\NF u : \sigma'}{{\bf\widetilde\Sigma} ; \; s : \tau \vdash_\NE t \; u : \sigma}\]
then, by Lemma~\ref{lem:neutral-subtype-var} (applied on the first premise), it means that we have 
\[\text{either}~~ \sigma' \to \sigma \sqsubseteq_+ \tau \qquad \text{or}~~ \sigma \to \sigma' = \basety \lin \ldots \lin \basety\]
The second alternative is absurd, and the first leads to $\sigma = \kappa'$ and $\sigma' = \kappa_k$.
Therefore, we may apply Lemma~\ref{lem:laml-niceshape-partial} to get terms $d_1, \ldots, d_{k-1}$ in normal form such that $t =_{\beta\eta} s \; d_1\; \ldots \;d_{k-1}$. We then set $d_k$ to be $u$ and $o$ to be the identity to conclude.
\item If the last rule considered is the other instance of the application rule
\[\dfrac{{\bf \widetilde\Sigma}; \; \Delta \vdash_\NE t : \sigma' \lin \sigma \qquad \widetilde{\bf \Sigma}; \; \Delta', \; s : \tau \vdash_\NF u : \sigma'}{\widetilde{\bf \Sigma} ; \; \Delta, \; \Delta', \; s : \tau \vdash_\NE t \; u : \sigma}\]
by Lemma~\ref{lem:neutral-subtype-var} applied to the first premise, we know that $\sigma' = \basety$.
Therefore, we may apply the induction hypothesis to the second premise to obtain
$d_1, \ldots, d_k$ and $o'$ such that $u =_\lamlequiv o' \; (s \; d_1 \; \ldots\; d_k)$, in which case, $t \; u =_\lamlequiv (\lam x.\; t \; (o \; x))\; (s \; d_1 \; \ldots \; d_k)$. We conclude by setting $o = \lam z.\;t \; (o' \; z)$.
\item If the last rule considered is the following instance of the elimination of tensor products
\[\dfrac{
{\bf\widetilde\Sigma}; \; s : \tau \vdash_\NE t : \zeta_1 \tensor \zeta_2
\qquad
{\bf\widetilde\Sigma}; \; x_1 : \zeta_1, \; x_2 : \zeta_2 \vdash_\NE u : \sigma
}{{\bf\widetilde\Sigma} ;  \; s : \tau \vdash_\NE \tlet{x \tensor y}{t}{u} : \sigma}\]
then, by the induction hypothesis (which is applicable because of Lemma~\ref{lem:neutral-subtype-var}),
there are $d_1, \ldots, d_k$ and $o'$ such that $t =_\lamlequiv o' \; (s \; d_1 \; \ldots\; d_k)$, in which case,
we conclude by setting $o = \lam z. \; \tlet{x \tensor y}{o' \; z}{u}$.
\item The last rule considered cannot be the following instance of the elimination of tensor products
\[\dfrac{
{\bf\widetilde\Sigma}; \; \cdot \vdash_\NE t : \zeta_1 \tensor \zeta_2
\qquad
{\bf\widetilde\Sigma};  \; s : \tau, \; x_1 : \zeta_1, \; x_2 : \zeta_2 \vdash_\NE u : \sigma
}{{\bf\widetilde\Sigma} ;  \; s : \tau \vdash_\NE \tlet{x \tensor y}{t}{u} : \sigma}\]
as Lemma~\ref{lem:neutral-subtype-var} would require that $\zeta_1 \tensor \zeta_2$ be a strictly positive subtype of some $\basety \lin \ldots \lin \basety$.
\item If the last rule considered types a projection
\[
\dfrac{{\bf\widetilde\Sigma} ; \; s : \tau \vdash_\NE t : \sigma_1 \with \sigma_2}{{\bf\widetilde\Sigma} ;  \; s : \tau \vdash_\NE \pi_i(t) : \sigma_i}
\]
then the induction hypothesis yields
terms $o'$ and $d_1, \ldots, d_k$ such that $t =_\lamlequiv o' \; (s \; d_1 \ldots d_k)$. We may set $o = \lam z. \; \pi_i(o' \; z)$ to conclude.
\item If the last rule applied is an elimination of a coproduct
\[\dfrac{{\bf\widetilde\Sigma}; \; s : \tau \vdash_\NE t : \zeta_1 \oplus \zeta_2 \qquad {\bf\widetilde\Sigma}; \; x : \zeta_1 \vdash_\NE u : \sigma \qquad {\bf\widetilde\Sigma}; \; y : \zeta_2 \vdash_\NE v : \sigma}{{\bf\widetilde\Sigma};  \; \cdot \vdash_\NE \case(t,x.u,y.v) : \sigma}\]
then, the induction hypothesis (applicable because of Lemma~\ref{lem:neutral-subtype-var}) yields 
terms $o'$ and $d_1, \ldots, d_k$ such that $t =_\lamlequiv o' \; (s \; d_1 \ldots d_k)$. We may set $o = \lam z.\;\case(o' \; z, x.u, y.v)$ to conclude.
\item The last rule applied cannot be one of the following instances of the elimination of a coproduct because of Lemma~\ref{lem:neutral-subtype-var} applied to the first premise:
\[\dfrac{{\bf\widetilde\Sigma}; \; \cdot \vdash_\NE t : \zeta_1 \oplus \zeta_2 \qquad {\bf\widetilde\Sigma}; \; s : \tau, \; x : \zeta_1 \vdash_X u : \sigma \qquad {\bf\widetilde\Sigma}; \; y : \zeta_2 \vdash_X v : \sigma}{{\bf\widetilde\Sigma};  \; s : \tau \vdash_X \case(t,x.u,y.v) : \sigma}\]
\[\dfrac{{\bf\widetilde\Sigma}; \; \cdot \vdash_\NE t : \zeta_1 \oplus \zeta_2 \qquad {\bf\widetilde\Sigma}; \; \; x : \zeta_1 \vdash_X u : \sigma \qquad {\bf\widetilde\Sigma}; \; s : \tau, \; y : \zeta_2 \vdash_X v : \sigma}{{\bf\widetilde\Sigma};  \; s : \tau \vdash_X \case(t,x.u,y.v) : \sigma}\]
\item If the last rule applied is an elimination of $\unit$
\[\dfrac{{\bf\widetilde\Sigma}; \; s : \tau \vdash_\NE t : \unit \qquad {\bf\widetilde\Sigma}; \; \cdot \; \vdash_\NE u : \sigma}{{\bf\widetilde\Sigma};  \; \cdot \; \vdash_\NE \tlet{()}{t}{u} : \sigma}\]
then, the induction hypothesis (applicable because of Lemma~\ref{lem:neutral-subtype-var}) yields 
terms $o'$ and $d_1, \ldots, d_k$ such that $t =_\lamlequiv o' \; (s \; d_1 \ldots d_k)$. We may set $o = \lam z.\;\tlet{()}{o' \; z}{u}$ to conclude.
\item The last rule applied cannot be one of the other instances of $\unit$ because of Lemma~\ref{lem:neutral-subtype-var}.
\item Finally, since the context under consideration are assumed to be consistent, the last rule applied cannot be an elimination of $0$.
\end{itemize}
\end{proof}

\begin{lem}
\label{lem:inconsistent-empty}
Let ${\bf \Sigma}$ and ${\bf \Gamma}$ be ranked alphabets.
If the context ${\bf \widetilde\Gamma};\; s : \Treety_{\bf \Sigma}[\kappa]$ is inconsistent, then
$\Tree({\bf \Sigma}) = \varnothing$.
\end{lem}
\begin{proof}
Let us write $\interp{-}$ for a semantic interpretation of $\laml$ types as (classical) propositions
following the usual type/proposition mapping (i.e., $\interp{\tau \lin \sigma} = \interp{\tau \to \sigma} = \interp{\tau} \Rightarrow \interp{\sigma}$,
$\interp{\tau \tensor \sigma} = \interp{\tau \with \sigma} = \interp{\tau} \wedge \interp{\sigma}$, $\interp{0} = \bot$, \ldots) and such that
\[\interp{\basety} ~~\Leftrightarrow~~ \Tree({\bf \Gamma}) \neq \varnothing\]
It is easy to check that, under the usual conjunctive interpretation of contexts, if $\Psi; \; \Delta \vdash t : \tau$, then
$\interp{\Psi} \wedge \interp{\Delta} \Rightarrow \interp{\tau}$ holds.
Further, our choice for $\interp{\basety}$ means that $\interp{\bf \widetilde\Gamma}$ holds, so, if our
context ${\bf \widetilde\Gamma};\; s : \Treety_{\bf \Sigma}[\kappa]$ is inconsistent, $\neg \interp{\Treety_{\bf \Sigma}[\kappa]}$ holds.
Now, assume that $\Tree({\bf \Sigma})$ has an inhabitant $t$. There is a corresponding Church encoding $\underline{t}$
of type $\Treety_{\bf \Sigma}$, which has also type $\Treety_{\bf \Sigma}[\kappa]$. Hence, $\interp{\Treety_{\bf \Sigma}[\kappa]}$ also holds,
leading to a contradiction.
\end{proof}

\begin{proof}[Proof of Lemma~\ref{lem:laml-niceshape}]
Assume that we have a closed $\laml$-term $t$ of type $\Treety_{\bf \Sigma}[\kappa] \lin \Treety_{\bf \Gamma}$ with
$\Tree({\bf \Sigma}) \neq \varnothing$, so that we may safely assume ${\bf\widetilde\Gamma};\; s : \Treety_{\bf \Sigma}[\kappa]$ to be consistent.
By $\eta$-expansion, we have  is of the shape
\[t =_\lamlequiv \lam w. \lam^\oc b_1. \ldots \lam^\oc b_k. t \; w \; b_1 \ldots b_k\]
and ${\bf \widetilde\Gamma}; \; s : \Treety_{\bf \Sigma}[\kappa] \vdash t \; w \; b_1 \ldots b_k : \basety$.
By normalization of $\laml$, there is $t'$ such that
\[t' =_\lamlequiv t \; w \; b_1 \ldots b_k \qquad \text{and} \qquad {\bf \widetilde\Gamma}; \; s : \Treety_{\bf \Sigma}[\kappa] \vdash_\NF t' : \basety\]
For the latter, note that an easy case analysis shows that every (open) term of type $\basety$ is in fact neutral, so we may conclude by applying
Lemma~\ref{lem:laml-niceshape-invariant} and the fact that $=_\lamlequiv$ is a congruence over terms.
\end{proof}

\section{Proof of Theorem~\ref{thm:dial-haslin} (monoidal closure)}
\label{sec:app-dial}
Here we give a detailed proof of Theorem~\ref{thm:dial-haslin} by establishing
a tensor/hom adjunction
\[\Hom{\cC_{\oplus\with}}{C \tensor A}{B} \cong \Hom{\cC_{\oplus\with}}{C}{A \lin B}\]
\emph{natural in $C$}.
While we will not devote too much space to
checking naturality down to low-level details, we have attempted to break down
this verification into manageable chunks. In particular, our notations are going
to stress the functoriality of the various intermediate expressions in our
computation. Some lemmas will be useful for this.

\begin{lem}
  \label{lem:lift-nat-with}
  We can lift every functor $F : \cC^{\op} \to \Set$ to a functor
  $\langle\with\rangle F : (\cC_{\with})^{\op} \to \Set$ such that for
  $(C_{i}) \in \Obj(C)^{I}$,
  \[ \langle\with\rangle F\left(\bigwith_{i \in I} C_{i}\right) \quad=\quad \sum_{i \in I} F(C_{i}) \]
  and in such a way that if $F, G : \cC^{\op} \to \Set$ are naturally
  isomorphic, so are $\langle\with\rangle F$ and $\langle\with\rangle G$.
\end{lem}
\begin{proof}
  Since $(\cC_{\with})^{\op} = (\cC^{\op})_{\oplus}$ by duality, the existence
  of $\langle\with\rangle F$ is simply the universal property of the
  \emph{coproduct} completion! Concretely, consider a morphism
  \[ f = \displaystyle \left(\varphi_{f} : I \to J,\; \left(f_{i} : C'_{\varphi_{f}(i)} \to C_{i}\right)_{{i \in I}}\right) \in \varHom{\cC_{\with}}{\bigwith_{j \in J} C'_{j}}{\bigwith_{i \in I} C_{i}} \]
  Its functorial image is given by
  \[ \langle\with\rangle F(f) \quad:\quad (i,x) \in \sum_{i \in I} F(C_{i}) \quad\mapsto\quad (\varphi_{f}(i), F(f_{i})(x)) \in \sum_{j \in J} F(C'_{j}) \]
  Using this explicit expression, one can check that if
  $\eta_X : F(X) \xrightarrow{\;\sim\;} G(X)$ is a natural isomorphism (for $X$
  varying over $\cC^{\op}$), then
  \[ (i,x) \in \sum_{i \in I} F(C_{i}) \quad\mapsto\quad (i,\eta_{C_{i}}(x)) \in \sum_{i \in I} G(C_{i}) \]
  defines a natural isomorphism between $\langle\with\rangle F$ and
  $\langle\with\rangle G$.
\end{proof}
\begin{lem}
  \label{lem:lift-nat-oplus}
  We can lift every functor $F : (\cC_{\with})^{\op} \to \Set$ to
  $\langle\oplus\rangle F : (\cC_{\oplus\with})^{\op} \to \Set$ in such a way
  that if $F, G : (\cC_{\with})^{\op} \to \Set$ are naturally isomorphic, so are
  $\langle\oplus\rangle F$ and $\langle\oplus\rangle G$, and with the action on
  objects
  \[ \langle\oplus\rangle F\left(\bigoplus_{u \in U} \bigwith_{x \in X_{u}} C_{x}\right) \quad=\quad \prod_{u \in U} F\left(\bigwith_{x \in X_{u}} C_{x} \right) \]
\end{lem}
\begin{proof}
  Write $\cC_{\oplus\with} \cong \cD_{\oplus}$ with $\cD = \cC_{\with}$. The
  statement is a special case of a property that works for any category $\cD$;
  analogously (and dually) to the previous lemma, it corresponds to the
  universal property of the product completion.
\end{proof}

Before we can state the next lemma, we first need a few definitions. Let
\[\iota_\with^\oplus : \cC_\oplus \to \cC_{\oplus\with} \quad\text{such that}\quad  \iota_{\oplus\with} = \iota_\with^\oplus \circ \iota_{\oplus} \]
be the full and faithful embedding be obtained by applying the universal
property of $\cC_\oplus$ to the coproduct-preserving functor
$\iota_{\oplus\with}$. Furthermore, the isomorphism
$\cC_{\oplus\with} \cong (\cC_{\with})_{\oplus}$ entails the existence of
another full and faithful functor
\[ \iota'_{\oplus} : \cC_{\with} \to \cC_{\oplus\with} \quad\text{such
    that}\quad \iota_{\oplus\with} = \iota'_\oplus \circ \iota_{\with} \]
corresponding to the canonical embedding $\cD \to \cD_{\oplus}$ for
$\cD = \cC_\with$. Both those functors have straightforward explicit
constructions, whose actions on objects are
\[ \iota_\with^\oplus\left( \bigoplus_{i \in I} C_i \right) \quad=\quad \bigoplus_{i\in I} \bigwith_{j \in \{*\}} C_i \qquad\qquad \iota'_\oplus\left( \bigwith_{i \in I} C_i \right) \quad=\quad \bigoplus_{j \in \{*\}} \bigwith_{i \in I} C_i \]
(where the (co)product notation is syntactic sugar for indexed families).
\begin{lem}
  \label{lem:with-left-sum}
  For any $D \in \Obj(\cC_\oplus)$, there is a natural isomorphism (with domain
  $\cC_{\with}$)
  \[ \Hom{\cC_{\oplus\with}}{\iota'_{\oplus}(-)}{\iota^{\oplus}_{\with}(D)} \quad\cong\quad \langle\with\rangle\left[ \Hom{\cC_{\oplus}}{\iota_{\oplus}(-)}{D} \right] \]
\end{lem}
\begin{proof}
  Let $\displaystyle C = \bigwith_{i \in I} C_i$ and
  $\displaystyle D = \bigoplus_{k \in K} D_{k}$ for $C_{i},D_{k} \in \Obj(\cC)$.
  By definition,
  \begin{align*}
    \Hom{\cC_{\oplus}}{\iota'_{\oplus}(C)}{\iota_{\with}^{\oplus}(D)}
    &= \prod_{p \in \{*\}} \sum_{k\in K} \prod_{q \in \{*\}} \sum_{i \in I} \Hom{\cC}{C_{i}}{D_{k}} \\
    &\cong \sum_{i \in I} \sum_{k\in K}  \Hom{\cC}{C_{i}}{D_{k}} \\
    &\cong \sum_{i \in I} \Hom{\cC_{\oplus}}{\iota_{\oplus}(C_{i})}{D}
  \end{align*}
  This establishes the isomorphism for each object of $\cC_{\with}$. The
  naturality condition can be recast as the commutation of the following
  diagram:
  \[
  \xymatrix@C=10mm{
    \Hom{\cC_{\oplus\with}}{\iota'_{\oplus}(C)}{\iota^{\oplus}_{\with}(D)}
    \ar[d]_{- \circ f} & \ar[l]_{\qquad\sim} \displaystyle\sum_{i,k}
    \Hom{\cC}{C_{i}}{D_{k}} \ar[r]^{\sim\qquad} & \displaystyle\sum_{i \in I}
  \Hom{\cC_{\oplus}}{\iota_{\oplus}(C_{i})}{D}
  \ar[d]^{(i,g) \mapsto (\varphi_{f}(i), g \circ f_{i})} \\
  \Hom{\cC_{\oplus\with}}{\iota'_{\oplus}(C')}{\iota^{\oplus}_{\with}(D)}
  \ar[r]^{\sim} & \displaystyle\sum_{j,k} \Hom{\cC}{C'_{j}}{D_{k}}
  & \ar[l]_{\sim\qquad} \displaystyle\sum_{j \in J}
   \Hom{\cC_{\oplus}}{\iota_{\oplus}(C'_{j})}{D}
 }\]
Fortunately, it can be checked from the definition of composition in
$\cC_{\oplus}$ and $\cC_{\oplus\with}$ that both paths in this diagrams denote
the same map, namely $(i,k,h) \mapsto (\varphi_{f}(i),k,h\circ f_{i})$.
\end{proof}

\begin{proof}[Proof of Theorem~\ref{thm:dial-haslin}]
  Let $A,B,C \in \Obj(\cC_{\oplus\with})$. By definition, we can write
\[ A = \displaystyle\bigoplus_{u \in U} \bigwith_{x \in X_u} A_x \qquad
  B = \displaystyle\bigoplus_{v \in V} \bigwith_{y \in Y_v} B_{y} \qquad
    C = \displaystyle\bigoplus_{w \in W} \bigwith_{z \in Z_{w}} C_{z}
  \]
  where $A_{x},B_{y},C_{z}\in\Obj(\cC)$. Using the definition of $\otimes$ in $\cC_{\oplus\with}$ and of $\Hom{\cC_{\oplus\with}}{-}{-}$,
  \begin{align*}
    \Hom{\cC_{\oplus\with}}{C\tensor A}{B}
    &= {\displaystyle \varHom{\cC_{\oplus\with}}{\bigoplus_{(w,u) \in W \times U} \bigwith_{(z,x) \in Z_w \times X_u} C_{z} \tensor A_{x}}{\bigoplus_{v \in V} \bigwith_{y \in Y_v} B_{y}}}
    \\
    &= {\displaystyle \prod_{w, u} \sum_{v \in V} \prod_{y \in Y_v} \sum_{z, x} \Hom{\cC}{C_{z} \tensor A_{x}}{\;B_{y}}}
    \\
    &\cong {\displaystyle \prod_{w \in W} \prod_{u \in U}  \sum_{v \in V} \prod_{y \in Y_v} \sum_{x \in X_{u}} \sum_{z \in Z_{w}} \Hom{\cC}{C_{z} \tensor A_{x}}{\;B_{y}}}
  \end{align*}
  Using the operators from Lemmas~\ref{lem:lift-nat-with}
  and~\ref{lem:lift-nat-oplus}, the last expression above can be turned into a
  functor in $C$, and the isomorphism is then natural in $C$, that is:
  \[ \Hom{\cC_{\oplus\with}}{-\tensor A}{B} \quad\cong\quad \langle\oplus\rangle \left[ \prod_{u \in U} \sum_{v \in V} \prod_{y \in Y_v} \sum_{x \in X_{u}} \langle\with\rangle \left[ \Hom{\cC}{- \tensor A_{x}}{\;B_{y}} \right] \right] \]
  (Any reader who is not convinced that rearranging the indexing of
  sums/products is innocuous with respect to the morphisms in
  $\cC_{\oplus\with}$ can check the naturality by a brute-force computation.)

  For any $A',B' \in \Obj(\cC)$, we have a sequence of natural isomorphisms
  \begin{align*}
    \Hom{\cC}{- \tensor A'}{B'} &\cong \Hom{\cC_{\oplus}}{\iota_{\oplus}(- \tensor A')}{\iota_{\oplus}(B')}\\
  &\cong \Hom{\cC_{\oplus}}{\iota_{\oplus}(-) \tensor \iota_{\oplus}(A')}{\iota_{\oplus}(B')}\\
  &\cong \Hom{\cC_{\oplus}}{\iota_{\oplus}(-)}{ \iota_{\oplus}(A')\lin\iota_{\oplus}(B')}
  \end{align*}
  by using the strong monoidality and full faithfulness of $\iota_{\oplus}$ as
  well as the assumption on internal homsets in $\cC_{\oplus}$.
  Lemma~\ref{lem:lift-nat-with} allows us to lift natural isomorphisms through
  $\langle\with\rangle$; combined with Lemma~\ref{lem:with-left-sum}, this gives
  us
  \begin{align*}
    \langle\with\rangle \left[ \Hom{\cC}{- \tensor A'}{B'} \right]
    &\cong \langle\with\rangle \left[ \Hom{\cC_{\oplus}}{\iota_{\oplus}(-)}{\iota_{\oplus}(A')\lin\iota_{\oplus}(B')} \right]\\
    &\cong \Hom{\cC_{\oplus\with}}{\iota'_{\oplus}(-)}{\iota^{\oplus}_{\with} (\iota_{\oplus}(A')\lin\iota_{\oplus}(B'))}
  \end{align*}
  Let us introduce the abbreviation
  $A' \rightrightarrows B' = \iota_{\with}^{\oplus}(\iota_{\oplus}(A')\lin\iota_{\oplus}(B'))$.
  Using Proposition~\ref{prop:yondistr} and the dual of
  Remark~\ref{rem:coprod-univ} for products, we have
\begin{align*}
  \Hom{\cC_{\oplus\with}}{- \tensor A}{B}
  &\cong \langle\oplus\rangle\left[  \prod_{u \in U} \sum_{v \in V} \prod_{y \in Y_v} \sum_{x \in X_u} \varHom{\cC_{\oplus\with}}{\iota'_\oplus(-)}{A_x \rightrightarrows B_y} \right] \\
  &\cong \langle\oplus\rangle\left[  \prod_{u \in U} \sum_{v \in V} \prod_{y \in Y_v} \varHom{\cC_{\oplus\with}}{\iota'_\oplus(-)}{\bigoplus_{x \in X_u} A_x \rightrightarrows B_y} \right] \\
      &\cong \;\cdots\\
  &\cong \langle\oplus\rangle\left[  \varHom{\cC_{\oplus\with}}{\iota'_\oplus(-)}{\overbrace{\bigwith_{u \in U} \bigoplus_{v \in V} \bigwith_{y \in Y_v} \bigoplus_{x \in X_u} A_x \rightrightarrows B_y}^{
A \lin B
}} \right] \\
\end{align*}

To conclude, it suffices to show that the last functor in the above computation
is naturally isomorphic to $\Hom{\cC_{\oplus\with}}{-}{A \lin B}$. The
isomorphism can seen as an instance of the universal property of the coproduct,
cf.\ Remark~\ref{rem:coprod-univ} which unfortunately only states naturality in
the \enquote{wrong} argument for our purposes. The remaining naturality
condition can be verified routinely; we leave it to the reader.
\end{proof}

\section{Proof of Theorem~\ref{thm:functor-from-sr}
  (building functors from $\Sr$)}
\label{sec:coherence-sr}

Let us fix a an symmetric monoidal category $(\cC,\otimes,\unit)$ with an
internal monoid $(M,\mu,\eta)$.
We call the functors built from the monoidal structure on $\cC$
\enquote{tensorial functors} since the monoidal product $\otimes$ is sometimes
called a tensor product. They play a central role in our proof and are more
precisely defined as follows.
\begin{defi}
  The \emph{tensorial expressions} over a finite indexing set $I$ are freely
  inductively generated as follows:
  \begin{itemize}
  \item $\unit$ is a tensorial expression over $\varnothing$;
  \item $i$ is a tensorial expression over $\{i\}$;
  \item if $e,e'$ are tensorial expressions over $I$ and $I'$ respectively, with
    $I \cap I' = \varnothing$, then $(e \otimes e')$ is a tensorial expression
    over $I \cup I'$.
  \end{itemize}
  In other words, a tensorial expression over $I$ is a binary tree whose leaves
  are labeled either by $\unit$ or by $i \in I$, such that each element $i$
  appears exactly once.

  A tensorial expression $e$ over $I$ induces a functor $F_e : \cC^I \to \cC$ in
  an obvious way. The functors thus obtained are called \emph{tensorial
    functors}. A tensorial functor $F_e : \cC^I \to \cC$ is \emph{ordered} when
  $I$ is endowed with the total order that corresponds to the infix order on the
  $I$-labeled leaves of the expression $e$.
\end{defi}

Of course, the basic intuition is that over a given totally ordered indexing
set, two ordered tensorial functors express \enquote{the same thing} up to
inessential bracketing, while two unordered tensorial functors should be morally
the same \enquote{up to permutation}. This is expressed formally as a natural
isomorphism between functors, but \emph{Mac Lane's coherence theorem} gives us
something stronger: the natural isomorphisms can be chosen \emph{canonically}.

\begin{thm}[Coherence for symmetric monoidal categories~{\cite[\S{}XI.1
    (Theorem 1)]{CWM}}]
  There exists a map that sends each triple $(I,e,e')$, where $I$ is a finite
  set and $e,e'$ are tensorial expressions over $I$, to a
  natural isomorphism from the tensorial functor $F_e$ to $F_{e'}$, which we
  call a \emph{canonical isomorphism}, such that:
  \begin{itemize}
  \item identities, associators, unitors and symmetries are canonical
    isomorphisms, i.e.\ are in the image of this map;
  \item composing the image of $(I,e,e')$ by this map with the image of
    $(I,e',e'')$ yields the image of $(I,e,e'')$ -- in other words, canonical
    isomorphisms are closed under composition;
  \item canonical isomorphisms are also closed under monoidal product.
  \end{itemize}
\end{thm}
If $F, G : \cC^I \to \cC$ are \emph{ordered} tensorial functors for the same
total order on $I$, the construction of the canonical isomorphism between $F$
and $G$ works in any monoidal category, not necessarily symmetric. Indeed,
general monoidal categories enjoy a coherence theorem of their
own~\cite[\S{}VII.2]{CWM}, which involves its own canonical isomorphisms; thanks
to the uniqueness clauses in the coherence theorems with and without symmetry,
one can check that the two notions of canonical isomorphism coincide for ordered
tensorial functors. The ordered case is important for us because even though our
monoidal category is symmetric, the internal monoid $(M,\mu,\eta)$ that we are
given need not be \emph{commutative}. What the axioms for monoid objects
\emph{do} state, however, is a suitable form of associativity, a consequence of
which is that $n$-ary products are somehow \enquote{independent of bracketing}.
Formally speaking:
\begin{defi}
  To any tensorial expression $e$, we associate an \enquote{$e$-fold monoid
    multiplication} morphism $\mu^{(e)} : F_e((M)_{i \in I}) \to M$ inductively:
  \[ \mu^{(\unit)} = \eta \qquad \mu^{(i)} = \id\ \text{for}\ I = \{i\} \qquad
    \mu^{(e \otimes e')} = \mu \circ \left(\mu^{(e)} \otimes \mu^{(e')}\right)\]
\end{defi}
\begin{thm}[General associativity law~{\cite[\S{}VII.3]{CWM}}]
  \label{thm:assoc-internal-monoid}
  Let $I$ be a totally ordered finite set, $F_e, F_{e'} : \cC^I \to \cC$ be two
  \emph{ordered} tensorial functors and $\Xi : F_e \Rightarrow F_{e'}$ be their
  canonical natural isomorphism. Then $\mu^{(e)} = \mu^{(e')} \circ
  \Xi_{\vec{M}}$ where $\vec{M} = (M)_{i \in I}$.
\end{thm}
We find it convenient to work directly with tensorial functors in the rest of
this appendix, leaving the tensorial expressions that define them implicit.
Therefore, we write $\mu^{(F)}$ instead of $\mu^{(e)}$ when $F = F_e$; strictly
speaking, two expressions could define the same functor for accidental reasons,
but for our purposes, the right choice of $e$ can always be inferred from the
context. Similarly, we shall speak of \emph{the} canonical isomorphism between
two tensorial functors $\cC^I \to \cC$.
In accordance with \Cref{subsec:notations}, we write
\[ \bigotimes_{i=1}^n Y_i = (\dots(Y_1 \otimes Y_2) \otimes \dots) \otimes
  Y_n \qquad M^{\otimes n} = \bigotimes_{i=1}^n M \]
and this will be used to define ordered tensorial functors. Furthermore, for
each finite set $I$, we fix an arbitrary choice of tensorial functor
$\bigotimes_{i \in I} (-) : \cC^I \to \cC$, and denote by $\bigotimes_{i \in I}
Y_i$ the image of $(Y_i)_{i \in I}$ by this functor.

After these general preliminaries, let us focus on the specific study of the
category $\Sr(\Gamma)$ for a fixed finite alphabet $\Gamma$.

\begin{defi}
  For $t \in [R \to_\Sr R']$ -- recall that this implies $t : R' \to (\Gamma +
  R)^*$, let $\partial(t) \subseteq R$ be the set of register variables that do
  not occur in any $t(r')$ for $r' \in R'$.
\end{defi}

\begin{defi}
  Given $w \in (\Gamma + R)^*$, we write
\[ \bigotimes_{r \looparrowleft w} Y_r \;=\; \bigotimes_{i=1}^{\len{w}}
  \begin{cases}
    \unit & \text{when}\ w[i] \in \inl(\Gamma)\\
    Y_r & \text{when}\ w[i] = \inr(r)\ \text{for}\ r \in R
  \end{cases}
  \qquad\text{and}\qquad M^{\looparrowleft w} \;=\; \bigotimes_{r \looparrowleft w} M
\]
\end{defi}
\begin{lem}
  \label{lem:coherence-sr-dep}
  For $t \in [R \to_\Sr R']$, we have a canonical isomorphism
  \[ \bigotimes_{r \in R} Y_{r} \quad\cong\quad \bigotimes_{r\in\partial(t)}
    Y_{r} \otimes \bigotimes_{r' \in R'} \bigotimes_{r\looparrowleft t(r')}
    Y_{r} \]
\end{lem}
\begin{proof}
  To apply Mac Lane's coherence theorem, we just have to check that the
  right-hand side defines a tensorial functor with indexing set $R$. This
  amounts to the equality
  \[ R = \partial(t) \cup \bigcup_{r' \in R'} \{ r \mid \exists i.\, t(r')[i] =
    \inr(r) \} \]
  where, to ensure that no index in $R$ is repeated, the union must be
  \emph{disjoint} and the letters of $t(r')$ that are in $\inr(R)$ must all be
  distinct. Those conditions are consequences of copylessness, while the
  equality itself is essentially the definition of $\partial(t)$.
\end{proof}

We are now in a position to build the functor promised in
Theorem~\ref{thm:functor-from-sr}. Following the statement of this theorem, we
fix a family of morphisms $(m_c) \in \Hom{\cC}{\unit}{M}^\Gamma$, and assume
that $(\cC,\otimes,\unit)$ is affine. Thus, we may write $\tuple{}_A : A \to
\unit$ for the terminal morphism from $A$, omiting the subscript $A$ when it can
be inferred from the context.

\begin{defi}
  \label{def:coherence-sr}
  We define a map on objects $F : R \in \Obj(\Sr) \mapsto M^{\otimes R} \in
  \Obj(\cC)$.

  As for morphisms, given a register transition $t \in [R \to_\Sr R']$, we set
  $F(t)$ to be
  \[ F(R) = M^{\otimes R} \xrightarrow{\;\sim\;} M^{\otimes\partial(t)}
    \otimes \bigotimes_{r'\in R'} \bigotimes_{i=1}^{\len{t(r')}} X_{r',i}
    \xrightarrow{\;\tuple{}\otimes\widetilde{F}(t)\;} \unit \otimes M^{\otimes R'}
    \xrightarrow{\;\sim\;} M^{\otimes R'} = F'(R)\]
  where the left arrow instantiates the canonical isomorphism of
  Lemma~\ref{lem:coherence-sr-dep}, with
  \[ X_{r',i} =
    \begin{cases}
      \unit & \text{when}\ t(r')[i] \in \inl(\Gamma)\\
      M & \text{otherwise, i.e.}\ t(r')[i] \in \inr(R)
    \end{cases}
    \qquad\text{so that}\qquad
    \bigotimes_{i=1}^{\len{t(r')}} X_{r',i} = M^{\looparrowleft t(r')} \]
  and in the middle arrow, $\tuple{} : M^{\otimes\partial(t)} \to \unit$ is the
  terminal morphism and
  \[ \widetilde{F}(t) = \bigotimes_{r' \in R'} \left( \widetilde{F}(t)_{r'} :
      \bigotimes_{i=1}^{\len{t(r')}} X_{r',i}
      \xrightarrow{\;\;\bigotimes_i f_{r',i} \;\;} \bigotimes_{i=1}^{\len{t(r')}}
      M = M^{\otimes\len{t(r')}} \xrightarrow{\;\;\mu^{(\len{t(r')})}\;\;} M
    \right) \]
  where for $i \in \{1,\ldots,\len{t(r')}\}$, we pick
  \[ f_{r',i} =
    \begin{cases}
      m_c & \text{when}\ t(r')[i] = \inl(c)\ \text{for}\ c \in \Gamma\\
      \id_M & \text{otherwise, i.e.}\ t(r')[i] \in \inr(R)
    \end{cases}
  \]
  (recall that $m_c : \unit \to M$ is the prescribed functorial image, by the
  $F$ that we are defining, of the register transition $\widehat{c} \in [\varnothing
  \to_\Sr \{\bullet\}]$).
\end{defi}

The tedious part in proving Theorem~\ref{thm:functor-from-sr} is checking that
the above definition works.

\begin{prop}
  \label{prop:coherence-sr-is-functor}
  The operation $F$ introduced in Definition~\ref{def:coherence-sr} is a functor.
\end{prop}

\begin{figure}
  \centering
  \[\renewcommand{\labelstyle}{\textstyle}
    \xymatrix@C=18mm@R=12mm{
      M^{\tensor R} \ar[r]^-{F(t)} \ar[d]^{\sim}
      \ar@{}[dr] | {\text{definition of}\ F(t)}
      & M^{\tensor R'} \ar[r]^-{F(t')} \ar@{}[dr] | {\text{naturality}}
      & M^{\tensor R''} \\
      \displaystyle M^{\otimes\partial(t)} \otimes \bigotimes_{r'\in R'}
      M^{\looparrowleft t(r')} \ar@{=}[d]
      \ar@{}[drr] | {\text{bifunctoriality of $\otimes$}}
      \ar[r]^-{\tuple{} \otimes \widetilde{F}(t)}
      & \unit \otimes M^{\otimes R'} \ar[r]^-{\unit\otimes F(t')} \ar[u]^-{\sim}
      & \unit \otimes M^{\otimes R''} \ar[u]^-{\sim} \ar@{=}[d] \\
      \displaystyle M^{\otimes\partial(t)} \otimes \bigotimes_{r'\in R'}
      M^{\looparrowleft t(r')} \ar[d]^-{\sim}
      \ar@{}[drr] |
      {\tuple{}\otimes\text{(Lemma~\ref{lem:coherence-sr-is-functor} /
          \Cref{fig:coherence-sr-secondary})}}
      \ar[rr]^-{\tuple{}\otimes(F(t')\circ \widetilde{F}(t))}
      & & \unit \otimes M^{\otimes R''}\\
      \displaystyle M^{\otimes\partial(t)} \otimes \left(
        N \otimes \bigotimes_{r'' \in R''} M^{\looparrowleft (t \circ t')(r'')}
      \right)
      \ar[d]^-{\sim} \ar@{}[drr] | {\text{naturality of the associator}}
      \ar[rr]^-{\tuple{}\otimes(\tuple{}\otimes\widetilde{F}(t' \circ t))} & &
      \displaystyle \unit\otimes\left(\unit \otimes M^{\otimes R''}\right)
      \ar[u]^-{\sim}\\
      \displaystyle \left( M^{\otimes\partial(t)} \otimes N  \right) \otimes
      \bigotimes_{r'' \in R''} M^{\looparrowleft (t \circ t')(r'')}
      \ar[d]^-{\sim} \ar@{}[drr] | {\text{$\unit\otimes\unit$ and $\unit$ are
          terminal (affineness assumption)}}
      \ar[rr]^-{(\tuple{}\otimes\tuple{})\otimes\widetilde{F}(t' \circ t)} & &
      \displaystyle (\unit\otimes\unit) \otimes M^{\otimes R''} \ar[u]^-{\sim}\\
      \displaystyle M^{\otimes\partial(t' \circ t)} \otimes
      \bigotimes_{r'' \in R''} M^{\looparrowleft (t \circ t')(r'')}
      \ar[rr]^-{\tuple{}\otimes\widetilde{F}(t' \circ t)} & &
      \displaystyle \unit \otimes M^{\otimes R''} \ar[u]^-{\sim}
    }\]
  \caption{The commutativity of the outer square of this diagram establishes
    Proposition~\ref{prop:coherence-sr-is-functor}. The text in the inner
    squares explains why they commute; the proof text for
    Proposition~\ref{prop:coherence-sr-is-functor} defines $N$ and gives further
    justifications (in particular for the existence of the canonical
    isomorphisms denoted by $\xrightarrow{\;\sim\;}$).}
  \label{fig:coherence-sr-primary}
\end{figure}

\begin{proof}
  Let $t \in [R \to_\Sr R']$ and $t' \in [R' \to_\Sr R'']$; we want to reason on
  $F(t') \circ F(t)$ to show that it is equal to $F(t' \circ t)$. Beware: we
  write $t' \circ t$ for composition of (copyless) register transitions in the
  category $\Sr$, and will employ the notation $t(r')$ for set-theoretic
  application ($t : R' \to (\Gamma + R)^*$), but we do \emph{not} have $(t'
  \circ t)(r') = t'(t(r'))$ -- indeed, the two sides of the equality are not
  even well-defined for $r' \in R'$! Since $t' \circ t$ is in $[R \to_\Sr R'']$,
  it is a set-theoretic map $R'' \to (\Gamma + R)^*$.

  To prove this, we first reduce our goal to
  Lemma~\ref{lem:coherence-sr-is-functor}, and then prove that lemma. The
  reduction is given by the commutative diagram of
  \Cref{fig:coherence-sr-primary} (with $N$ to be defined later). Indeed, while
  the morphism on the top is $F(t') \circ F(t)$, the one defined by the three
  other sides of the outermost square is
  \[ M^{\otimes R} \xrightarrow{\;\sim\;} M^{\otimes\partial(t'\circ t)}
    \otimes \bigotimes_{r'' \in R''} M^{\looparrowleft (t' \circ t)(r'')}
    \xrightarrow{\;\tuple{}\otimes\widetilde{F}(t' \circ t)\;} \unit \otimes
    M^{\otimes R''} \xrightarrow{\;\sim\;} M^{\otimes R''} \]
  thanks to the closure of canonical isomorphisms of tensorial functors under
  composition. Since these canonical isomorphisms are also unique, we can equate
  that with the expression for $F(t' \circ t)$ given in
  Definition~\ref{def:coherence-sr}, hence $F(t') \circ F(t) = F(t' \circ t)$.

  The main justifications that are missing from \Cref{fig:coherence-sr-primary}
  are the commutation of one square treated in
  Lemma~\ref{lem:coherence-sr-is-functor}, and the existence of the various
  canonical isomorphisms involved. The ones in the right column are
  just unitors, so let us focus on the left column.
  \begin{itemize}
  \item The top left isomorphism is an instance of Lemma~\ref{lem:coherence-sr-dep}.
  \item The next one (skipping the equality) is provided by
    Lemma~\ref{lem:coherence-sr-is-functor}.
  \item Then the penultimate one is just the inverse of an associator (see
    \Cref{subsec:prelim-cat}), of the form
    \[ \alpha^{-1}_{A,B,C} \;:\; A \otimes (B \otimes C) \;\to\; (A \otimes B)
      \otimes C\]
  \item Finally, the last one reduces to $M^{\otimes\partial(t)} \otimes N \cong
    M^{\otimes\partial(t' \circ t)}$.
  \end{itemize}
  To prove the latter, we must first clarify it by defining $N$ consistently
  with Lemma~\ref{lem:coherence-sr-is-functor}:
  \[ N = \bigotimes_{r' \in \partial(t')} M^{\looparrowleft t(r')} \]
  We then see that the canonical isomorphism that we are looking for is the
  instantiation to the constant family $(M)_{r \in \partial(t' \circ t)}$ of the
  following natural isomorphism in $Y_r$ for $r \in \partial(t' \circ t)$:
  \[ \bigotimes_{r\in\partial(t)} Y_r \;\otimes\;
    \bigotimes_{r'\in\partial(t')} \bigotimes_{r \looparrowleft t(r')} Y_r
    \quad\cong\quad \bigotimes_{r\in\partial(t' \circ t)} Y_r \]
  This is a consequence of an elementary combinatorial fact:
  \[ \partial(t) \cup \bigcup_{r'\in\partial(t')} \{r \mid \exists i \in
    \{1,\ldots,\len{t(r')}\}.\; t(r')[i] = \inr(r)\} \quad=\quad
    \partial(t'\circ t) \]
  Informally speaking, this means that a register is thrown away by $t' \circ t$
  if and only if is either thrown away by $t$, or used by $t$ to compute a value
  that is then discarded by $t'$. Furthermore, thanks to copylessness, the union
  is \emph{disjoint} and for each $r$ in the union over $r' \in \partial(t')$,
  there is a \emph{single} $(r',i)$ such that $t(r')[i] = \inr(r)$; those
  properties are necessary to get the natural isomorphism. We leave a formal
  proof of this combinatorial identity to the reader.

  This being done, all that remains to conclude our proof is to show
  Lemma~\ref{lem:coherence-sr-is-functor} below.
\end{proof}

\begin{lem}
  \label{lem:coherence-sr-is-functor}

  The diagram of \Cref{fig:coherence-sr-secondary} commutes, with the following
  definitions:
  \[ N = \bigotimes_{r' \in \partial(t')} M^{\looparrowleft t(r')} \qquad O_{r''} =
    \bigotimes_{r' \looparrowleft t'(r'')} M^{\looparrowleft t(r')} \qquad
    \psi_{r''} = \bigotimes_{r'\looparrowleft t'(r'')} \widetilde{F}(t)_{r'} \]
\end{lem}

\begin{figure}
  \centering
  \[\renewcommand{\labelstyle}{\textstyle}
    \xymatrix@C=22mm@R=15mm{
      \displaystyle \bigotimes_{r'\in R'} M^{\looparrowleft t(r')} \ar[d]^{\sim}
      \ar[r]^-{\widetilde{F}(t) = \bigotimes_{r' \in R'} \widetilde{F}(t)_{r'}}
      \ar@{}[dr] | {\text{naturality / Lemma~\ref{lem:coherence-sr-dep}}}
      & M^{\tensor R'} \ar[d]^-{\sim}
      \ar[r]^-{F(t')} \ar@{}[dr] | {\quad\text{definition of}\ F(t')}
      & M^{\tensor R''} \\
      \displaystyle N \otimes \bigotimes_{r'' \in R''} O_{r''}
      \ar[r]^-{(\ldots)\otimes\bigotimes_{r''} \psi_{r''}} \ar@{=}[d]
      \ar@{}[drr] | {\text{$\otimes$ is a bifunctor and $\unit$ is a terminal object}}
      & \displaystyle M^{\otimes\partial(t')} \otimes \bigotimes_{r''\in R''}
      M^{\looparrowleft t'(r'')} \ar[r]^-{\tuple{}\otimes\widetilde{F}(t')}
      & \unit \otimes M^{\otimes R''} \ar[u]^-{\sim} \ar@{=}[d] \\
      \displaystyle N \otimes \bigotimes_{r'' \in R''} O_{r''} \ar[d]^-{\sim}
      \ar[rr]^-{\tuple{}\otimes\bigotimes_{r'' \in R''}
        (\widetilde{F}(t')_{r''}\circ\psi_{r''})}
      \ar@{}[drr] | {\text{manipulations on monoid objects, explained in main text}}
      & & \unit \otimes M^{\otimes R''} \ar@{=}[d] \\
      \displaystyle N \otimes \bigotimes_{r'' \in R''} M^{\looparrowleft (t \circ t')(r'')}
      \ar[rr]^-{\tuple{}\otimes\bigotimes_{r'' \in R''}
        \widetilde{F}(t'\circ t)_{r''} \;=\; \tuple{}\otimes F(t'\circ t)}
      & & \unit \otimes M^{\otimes R''}
    }\]
  \caption{Lemma~\ref{lem:coherence-sr-is-functor} defines $N$, $O_{r''}$,
    and $\psi_{r''}$, and proves that this diagram commutes.}
  \label{fig:coherence-sr-secondary}
\end{figure}

\begin{proof}
  In addition to the information already given in
  \Cref{fig:coherence-sr-secondary}, there are two things to justify about this
  figure: the top-left naturality square, and the commutativity of the bottom
  square. Let us tackle the former. Recall that Lemma~\ref{lem:coherence-sr-dep}
  gives us the canonical isomorphism
  \[ \bigotimes_{r' \in R'} Y_{r'} \quad\cong\quad \bigotimes_{r'\in\partial(t')}
    Y_{r'} \;\otimes\; \bigotimes_{r'' \in R''}
    \bigotimes_{r'\looparrowleft t'(r'')} Y_{r'} \]
  which is natural in $Y_{r'}$ for $r ' \in R'$. By instantiating with $Y_{r'} =
  M^{\looparrowleft t(r')}$, we get
  \[ \bigotimes_{r' \in R'} M^{\looparrowleft t(r')} \quad\cong\quad
    \bigotimes_{r'\in\partial(t')} M^{\looparrowleft t(r')} \;\otimes\;
    \bigotimes_{r''\in R''} \bigotimes_{r'\looparrowleft t'(r'')}
    M^{\looparrowleft t(r')} \quad=\quad N \otimes \bigotimes_{r'' \in R''}
    O_{r''} \]
  by definition of $N$ and $O_{r''}$. Hence the vertical canonical isomorphism
  at the top left of \Cref{fig:coherence-sr-secondary}; as for the top middle,
  it is the instantiation of the same natural isomorphism with $Y_{r'} = M$, as
  already observed in Definition~\ref{def:coherence-sr}. To make the top-left
  square a naturality square, we should then have
  \[ (\ldots) \otimes \bigotimes_{r'' \in R''} \psi_{r''} \quad=\quad
    \bigotimes_{r'\in\partial(t')} \widetilde{F}(t)_{r'} \;\otimes\;
    \bigotimes_{r'' \in R''} \bigotimes_{r'\looparrowleft t'(r'')}
    \widetilde{F}(t)_{r'} \]
  which is indeed the case with our definition of $\psi_{r''}$.

  Our next and final task is the commutativity of the bottom square of
  \Cref{fig:coherence-sr-secondary}. Thanks to the bifunctoriality of $\otimes$,
  it reduces to a simpler commutative diagram:
  \[\renewcommand{\labelstyle}{\textstyle}
    \xymatrix@C=30mm@R=10mm{
      \displaystyle
      O_{r''} = \bigotimes_{r'\looparrowleft t'(r'')}M^{\looparrowleft t(r')}
      \ar[d]_-{\sim} \ar[r]^-{\widetilde{F}(t')_{r''}\circ\psi_{r''}} & M\\
      \displaystyle M^{\looparrowleft (t \circ t')(r'')}
      \ar[ur]_-{\quad\widetilde{F}(t'\circ t)_{r''}}
    } \qquad\text{for}\ r'' \in R''\]
  Let $r'' \in R''$. To show that this indeed commutes, we start by writing out
  $\widetilde{F}(t')_{r''}\circ\psi_{r''}$ as
  \[ \bigotimes_{r'\looparrowleft t'(r'')}M^{\looparrowleft t(r')}
    \xrightarrow{\;\bigotimes_{r'}\widetilde{F}(t)_{r'}\;} M^{\looparrowleft t'(r'')}
    = \bigotimes_{j=1}^{\len{t'(r'')}} X'_{r'',j}
    \xrightarrow{\;\bigotimes_j f'_{r'',j}\;} M^{\otimes\len{t'(r'')}}
    \xrightarrow{\;\mu^{(\len{t'(r'')})}\;} M \]
  where, analogously to the $X_{r',i}$ and $f_{r',i}$ involved in the
  definition of $\widetilde{F}(t)_{r'}$, we take
  \[ (X'_{r'',j},\; (f'_{r'',j} : X'_{r'',j} \to M)) =
    \left\{\begin{array}{lll}
             (\unit, & \!\!\!m_c) & \text{when}\ t'(r'')[j] = \inl(c)\
                                    \text{for}\ c\in\Gamma\\
             (M, & \!\!\!\id_M) & \text{otherwise, i.e.}\ t'(r'')[j] \in \inr(R')
           \end{array}\right.
  \]
  The leftmost arrow in the above sequence is how we defined $\psi_{r''}$ in the
  lemma statement, while the composition of the two others equals
  $\widetilde{F}(t')_{r''}$ by definition.

  To manipulate this, let us introduce a new notation:
  \[ \bigotimes^{\mathtt{foldMap}}_{w} f \quad=\quad \bigotimes_{i=1}^{|w|}
    f(w_i) \qquad\text{for}\ w = w_1 \dots w_n\
    \text{and}\ f\ \text{a function} \]
  Note that in general, this may lead to monoidal products with repeated
  factors: there is no a priori guarantee that this defines a tensorial functor.
  The function $f$ will often be expressed as a copairing using the following
  notation:
  \[ f : \begin{cases}
      \inl(y) &\!\!\!\mapsto f_1(y)\\
      \inr(z) &\!\!\!\mapsto f_2(z)
    \end{cases} \quad\iff\quad f : x \in Y + Z \mapsto
    \begin{cases}
      f_1(y) & \text{when}\ x = \inl(y)\ \text{for}\ y \in Y\\
      f_2(z) & \text{when}\ x = \inr(z)\ \text{for}\ z \in Z
    \end{cases}
  \]
  In the case that we are interested in right now, we have by functoriality of
  $\otimes$:
  \[ \widetilde{F}(t')_{r''} \circ \psi_{r''} \;=\;
    \mu^{(\len{t'(r'')})} \circ \bigotimes^{\mathtt{foldMap}}_{t'(r'')}
    \begin{cases}
      \inl(c)&\!\!\!\!\mapsto m_c\\
      \inr(r')&\!\!\!\!\mapsto \widetilde{F}(t)_{r'} = \mu^{(\len{t(r')})} \circ
      \displaystyle\bigotimes_{t(r')}^{\mathtt{foldMap}}
              \begin{cases}
                \inl(a) &\!\!\!\!\mapsto m_{a} \\
                \inr(r) &\!\!\!\!\mapsto \id_M
              \end{cases}
    \end{cases}
  \]
  To simplify this, we introduce the following tensorial functor:
  \[  L_{r''}((Y_p)_{p \in P_{r''}}) \quad=\quad \bigotimes_{j=1}^{\len{t'(r'')}}
    \begin{cases}
      \qquad\;\, Y_{j,1} & \text{when}\ t'(r'')[j] \in \inl(\Gamma)\\
      \displaystyle\bigotimes_{i=1}^{\len{t(r')}} Y_{j,i} & \text{when}\
      t'(r'')[j] = \inr(r')
    \end{cases}\]
  \[ \text{where}\qquad P_{r''} \quad=\quad \sum_{j=1}^{\len{t'(r'')}} \begin{cases}
      \{1\} & \text{when}\ t'(r'')[j] \in \inl(\Gamma)\\
      \{1,\ldots,\len{t(r')}\} & \text{when}\ t'(r'')[j] = \inr(r')
    \end{cases} \]
  (using the dependent sum operation, cf.\ \Cref{subsec:notations}, so that
  $P_{r''} \subset \bN^2$). There is a unique total order on $P_{r''}$ that
  makes $L_{r''}$ into an \emph{ordered} tensorial functor: the
  \emph{lexicographical order} inherited from $\bN^2$. Thus, we may meaningfully
  speak of $\mu^{(L_{r''})} : L_{r''}((M)_{p \in P_{r''}}) \to M$, the
  $L_{r''}$-ary monoid multiplication, whose inductive definition leads to:
  \[ \mu^{(L_{r''})} \quad=\quad  \mu^{(\len{t'(r'')})} \circ
    \bigotimes^{\mathtt{foldMap}}_{t'(r'')}\begin{cases}
      \inl(c)&\!\!\!\mapsto \mu^{(1)} = \id_M\\
      \inr(r')&\!\!\!\mapsto \mu^{(\len{t(r')})}
    \end{cases}\]
  Now that this is defined, we can state the following equation, that directly
  follows by functoriality from the previous expressions of
  $\widetilde{F}(t')_{r''} \circ \psi_{r''}$ and of $\mu^{(L_{r''})}$:
  \[ \widetilde{F}(t')_{r''} \circ \psi_{r''} \quad=\quad \mu^{(L_{r''})} \circ
    \bigotimes^{\mathtt{foldMap}}_{t'(r'')} \begin{cases}
      \inl(c)&\!\!\!\!\mapsto\quad m_c\\
      \inr(r')&\!\!\!\!\mapsto
      \displaystyle\bigotimes_{t(r')}^{\mathtt{foldMap}}
                \begin{cases}
                  \inl(a) &\!\!\!\!\mapsto m_{a} \\
                  \inr(r) &\!\!\!\!\mapsto \id_M
                \end{cases}
    \end{cases}
  \]
  For the next step, recall that by Definition~\ref{def:reg-trans-compo}, $(t'
  \circ t)(r'') = t^\ddagger(t'(r''))$ (which is \emph{not} $t'(t(r''))$, as
  previously emphasized, since this `$\circ$' is composition in $\Sr(\Gamma)$),
  where $t^\ddagger$ is the morphism of free monoids such that
  $t^\dagger(\inl(c)) = \inl(c)$ and $t^\dagger(\inr(r')) = t(r')$. Thus, the
  $j$-th term of the dependent sum above is equal to
  $\{1,\ldots,\len{t^\dagger(t'(r'')[j])}\}$, and from the morphism property
  \[ (t' \circ t)(r'') = t^\ddagger(t'(r'')) = t^\ddagger(t'(r'')[1]) \cdot
    \ldots \cdot t^\ddagger(t'(r'')[\len{t'(r'')}])\]
  we obtain an bijection $\xi_{r''} : P_{r''} \xrightarrow{\sim} \{1,\ldots,(t'
  \circ t)(r'')\}$ such that for all $(j,i) \in P_{r''}$,
  \[ (t' \circ t)(r'')[\xi(j,i)]
    = t^\ddagger(t'(r'')[j])[i] \\
    = \begin{cases}
      \inl(c) & \text{when}\ t'(r'')[j] = \inl(c)\\
      t(r')[i] & \text{when}\ t'(r'')[j] = \inr(r')
    \end{cases}
  \]
  Thanks to this, we have
  \[ \widetilde{F}(t')_{r''} \circ \psi_{r''} \;=\;
    \mu^{(L_{r''})} \circ L_{r''} \left( \left(
        \begin{cases}
          m_c & \text{when}\ (t' \circ t)(r'')[\xi(p)] = \inl(c)\ \text{for}\
          c\in\Gamma \\
          \id_M & \text{otherwise, i.e.}\ (t' \circ t)(r'')[\xi(p)] \in \inr(R)
        \end{cases}
      \right)_{p \in P_{r''}}  \right) \]
  It is not hard to see that the bijection $\xi$ is \emph{monotone} (taking the
  lexicographical order on $P_{r''}$ as we did earlier). Therefore, the natural
  isomorphism
  \[ \Xi_{(Y_p)_{p \in P_{r''}}} : L\left( (Y_p)_{p \in P_{r''}} \right)
    \xrightarrow{\;\sim\;} \bigotimes_{k=1}^{\len{(t'\circ t)(r'')}}
    Y_{\xi^{-1}(k)}\]
  induced by $\xi$ is a \emph{canonical isomorphism of ordered tensorial
    functors} (i.e.\ it only uses associators and unitors, not symmetries).
  Thus, we may apply the general associativity law for internal monoids
  (Theorem~\ref{thm:assoc-internal-monoid}):
  $\mu^{(L_{r''})} = \mu^{(\len{(t' \circ t)(r'')})} \circ \Xi$ (we omit the
  subscript of $\Xi$ for convenience). By naturality of $\Xi$, we then have
  \[ \widetilde{F}(t')_{r''} \circ \psi_{r''}
    \quad=\quad \mu^{(\len{(t' \circ t)(r'')})} \;\circ\; \left(
      \bigotimes_{(t' \circ t)(r'')}^{\mathtt{foldMap}}
      \begin{cases}
        \inl(c) &\!\!\!\mapsto m_c\\
        \inr(r) &\!\!\!\mapsto \id_M
      \end{cases}
    \right) \;\circ\; \Xi \]
  By definition, this is equal to $\widetilde{F}(t' \circ t) \circ \Xi$, which
  is what we needed to conclude the proof.
\end{proof}

\begin{prop}
  The functor $F$ is strong monoidal.
\end{prop}
\begin{proof}
  From the definition, we immediately get
  \[ F(R + R') = M^{\otimes(R+R')} \cong M^{\otimes R} \otimes
    M^{\otimes R'} =  F(R) \otimes F(R') \]
  as an instance of a canonical isomorphism
  \[ \bigotimes_{x \in R + R'} Y_x  \quad\cong\quad \bigotimes_{r \in R} Y_{\inl(r)}
    \otimes \bigotimes_{r' \in R'} Y_{\inr(r')} \]
  which is natural in $Y_x$ for $x \in R + R'$. One can then verify that this
  family of isomorphisms $F(R + R') \cong F(R) \otimes F(R')$ is natural in $R$
  and $R'$.
\end{proof}

To finish proving Theorem~\ref{thm:functor-from-sr}, it suffices to carry out
some short explicit computations to check that this functor $F$ satisfies the
claimed equalities. We leave this to the reader.

\section{Equivalence with $\lam\ell^{\with}$-definable tree functions}
\label{sec:app-with}

One natural question is whether of Theorems~\ref{thm:main-string} and~\ref{thm:main-tree} for variations
of $\laml$. Indeed, we expect that for strings, the equivalence between $\laml$-definability and regular
functions still holds if we forbid $\oplus$ and $\with$ in the former (see our discussion in~\cite[Claim 6.2]{aperiodic}). We do not attempt to prove this here, as we expect this would require a notable development exploiting the
Krohn-Rhodes theorem.
However, calling $\lam\ell^\with$ the restriction of $\laml$ forbidding the use of sum types, we have
a straightforward proof of the following.

\begin{thm}
\label{thm:main-with}
Tree and string functions definable in $\lam\ell^\with$ are exactly the regular functions.
\end{thm}

Of course, since $\laml$ is more expressive than $\lam\ell^\with$, one inclusion follows from our main theorems.
The converse can be done at the level of streaming settings, by first considering the analogue of $\Lamcat$ and
$\mfLam$ for $\lam\ell^\with$ and using a trick reminiscent of \emph{continuation passing style
transformation}~\cite{discoveryCPS} in order to simulate coproducts with a combination of products and higher-order functions.

We do not spell out the full details of this proof, but give the key lemmas, both in the case of strings and trees.

\begin{lem}
\label{lem:string-cps}
Let $\mfC$ be a symmetric monoidal closed string streaming setting with products such that $\initty$ is the monoidal unit.
There is a streaming setting morphism $\mfC_\oplus \to \mfC$. 
\end{lem}
\begin{proof}
First note that the ``negation'' $N(C) = C \lin \retty$ is a contravariant functor $\cC \to \cC^\op$ and that the coproduct
completion $(-)_\oplus$ is itself functorial over $\Cat$.
We thus define the underlying functor $NN : \cC_\oplus \to \cC$ of our streaming setting morphism as a composite
\[
\xymatrix@C=20mm{
\cC_\oplus \ar[r]^{N_\oplus} & (\cC^\op)_\oplus \ar[r] & \cC^\op \ar[r]^N & \cC}
\]
where the middle arrow is obtained with by the universal property of the free completion (recall that products of $\cC$
become coproducts of $\cC^\op$). On objects we thus have
\[N \quad: \qquad \bigoplus_{u \in U} C_u \quad \longmapsto \quad \left(\bigwith_{u \in U} C_u \lin \retty\right) \lin \retty\]
The map $i : \unit \to \iota_\oplus((\unit \lin \retty) \lin \retty)$ is induced by the transposition of the evaluation map $\Lambda(\ev \circ \gamma)$ and $o : \iota_\oplus((\retty \lin \retty) \lin \retty) \to \retty$
is the evaluating of its argument on the constant map $\tilde{\id} : \unit \to \retty \lin \retty$.
To check that this triple is indeed a morphism of streaming settings,
it suffices to prove that the following diagram commutes in $\cC$ for any $f \in \Hom{\cC}{\unit}{\retty}$:
\[
\xymatrix@C=30mm{
\unit \ar[r]^f \ar[d]_{\Lambda(\ev \circ \gamma)} & \retty \\
(\unit \lin \retty) \lin \retty \ar[r]^{N(N(f))} & (\retty \lin \retty) \lin \retty \ar[u]_{\ev \circ (\id \tensor \tilde{\id}) \circ \rho^{-1}}
}
\]
Unravelling the defininition of $N(N(f))$, this follows from the elementary equational properties
of symmetric monoidal categories.
\end{proof}

\begin{lem}
Let $\mfC$ be a symmetric monoidal closed \emph{tree} streaming setting with products such that $\initty$ is the monoidal unit.
There is a streaming setting morphism $\mfC_\oplus \to \mfC$. 
\end{lem}
\begin{proof}[Proof idea]
The underlying functor $NN$ is exactly the same. To complete the construction in the case of trees, one needs to append a suitable lax monoidal structure
to $NN$, that is, maps
\[
\xymatrix@C=15mm{
\unit \ar[r]^-{m^0} & NN(\unit) & \text{and} & NN(A) \tensor NN(B) \ar[r]^-{m^1_{A,B}} & NN(A \tensor B)
}
\]
with $m^1_{A,B}$ natural in $A$ and $B$ and such that the diagrams in~\cite[Section 5.1]{mellies09ps} commute;
we only sketch the definitions of those maps. For $m^0 : \unit \to ((\unit \lin \retty) \lin \retty)$,
we may take the usual $\Lambda(\ev \circ \gamma)$. As for $m^2_{-,-}$ for the objects $\bigoplus_{i \in I} A_i$ and
$\bigoplus_{j \in J} B_j$, it should be a map
\[\left(\left(\bigwith_{i} A_i \lin \retty\right) \lin \retty \right) \tensor
\left(\left(\bigwith_{j} B_j \lin \retty\right) \lin \retty \right) \longto 
\left(\left(\bigwith_{i,j} A_i \tensor B_j \lin \retty\right) \lin \retty \right)\]
There are two distinct options one might take, which intuitively correspond to the two
intuitive way of proving $\neg\neg A \tensor \neg\neg B \vdash \neg \neg (A \tensor B)$ in linear
logic; one can proceed for instance with the one which is biased towards the left.
\end{proof}

\end{document}